\documentclass[12pt]{article}
\parskip=\medskipamount

\DeclareFontEncoding{LS1}{}{}
\DeclareFontSubstitution{LS1}{stix}{m}{n}
\DeclareSymbolFont{stixsymbols}{LS1}{stixscr}{m}{n}
\SetSymbolFont{stixsymbols}{bold}{LS1}{stixscr}{b}{n}
\DeclareMathSymbol{\kay}{\mathalpha}{stixsymbols}{"6B}
\DeclareMathSymbol{\hay}{\mathalpha}{stixsymbols}{"68}

\usepackage{xcolor}

\usepackage{verbatim}
\usepackage{tensor}
 
\usepackage{amsmath,amssymb,calc, amsthm,bbm, epsfig,psfrag, mathtools}
\usepackage{latexsym,bm,amsfonts}

\newtheorem{thm}{Theorem}[section]

\usepackage{graphicx, enumerate}

 \usepackage[all,cmtip]{xy}
\usepackage[numbers,sort&compress]{natbib}
\usepackage[dvipsnames]{xcolor}
\usepackage{mathrsfs} 
\usepackage{graphicx,tikz,tikz-cd}
\usepackage{quiver}
\usetikzlibrary{shadows}
\usetikzlibrary{shapes.arrows}

\usepackage{braket}
\usepackage{physics}

\allowdisplaybreaks

\makeatletter
\newcommand*\bigcdot{\mathpalette\bigcdot@{.5}}
\newcommand*\bigcdot@[2]{\mathbin{\vcenter{\hbox{\scalebox{#2}{$\m@th#1\bullet$}}}}}
\makeatother
\newcommand{\deq}{\stackrel{\bigcdot}{=}}

\usepackage{bbold}

\newcommand{\Comment}[1]{{}}
\definecolor{darkblue}{rgb}{0.15,0.35,0.55}
\definecolor{reddish}{rgb}{0.65, 0.2, 0.2}
\usepackage[linktocpage=true]{hyperref}
\hypersetup{
colorlinks=true,
citecolor=darkblue,
linkcolor=reddish,
urlcolor=darkblue,
pdfauthor={},
pdftitle={},
pdfsubject={}
}

\makeatletter
\renewcommand\section{\@startsection {section}{1}{\z@}%
                                   {-3.5ex \@plus -1ex \@minus -.2ex}
                                   {2.3ex \@plus.2ex}%
                                   {\normalfont\large\bfseries}}
\renewcommand\subsection{\@startsection{subsection}{2}{\z@}%
                                     {-3.25ex\@plus -1ex \@minus -.2ex}%
                                     {1.5ex \@plus .2ex}%
                                     {\normalfont\bfseries}}
\makeatother

\newfont{\goth}{ygoth.tfm scaled 1200}                   

 \numberwithin{equation}{section}

\newcommand{\overbar}[1]{\mkern 1.5mu\overline{\mkern-1.5mu#1\mkern-1.5mu}\mkern 1.5mu}

\def\TT{{T\overbar{T}}}

\begin{document}
\begin{titlepage}
\begin{flushright}
\today
\end{flushright}
\vspace{5mm}


\begin{center}
{\Large \bf 
The classical Yangian symmetry of
\\
Auxiliary Field Sigma Models
}
\end{center}

\begin{center}

{\bf
Daniele Bielli${}^{a,b}$,
Christian Ferko${}^{c,d}$,
Michele Galli${}^{e}$,\\
Gabriele Tartaglino-Mazzucchelli${}^{e}$
} \\
\vspace{5mm}

\footnotesize{
${}^{a}$
{\it 
Asia Pacific Center for Theoretical Physics (APCTP), \\ 
Postech, Pohang 37673, 
Korea}
 \\~\\
${}^{b}$
{\it 
High Energy Physics Research Unit, Faculty of Science \\ 
Chulalongkorn University, Bangkok 10330, Thailand
}
 \\~\\
${}^{c}$
{\it 
Department of Physics, Northeastern University, Boston, MA 02115, USA
}
 \\~\\
${}^{d}$
{\it 
The NSF Institute for Artificial Intelligence
and Fundamental Interactions
}
  \\~\\
${}^{e}$
{\it 
School of Mathematics and Physics, University of Queensland,
\\
 St Lucia, Brisbane, Queensland 4072, Australia}
}
\vspace{6mm}
~\\
\texttt{d.bielli4@gmail.com,
c.ferko@northeastern.edu,
m.galli@uq.edu.au,
g.tartaglino-mazzucchelli@uq.edu.au
}\\
\vspace{2mm}

\end{center}

\begin{abstract}
\baselineskip=14pt

\noindent Integrable field theories exhibit infinitely many symmetries which underlie their solvability, but the structure of these symmetries can become obscured after performing an integrable deformation such as $\TT$ or an auxiliary field deformation. In this paper, we present a systematic organizing principle for understanding deformed charges and their Yangian structure in a broad class of integrable sigma models and their auxiliary field deformations. We generalize the recursive procedure of Brezin, Itzykson, Zinn-Justin, and Zuber (BIZZ) for generating non-local charges, and give sufficient conditions under which the resulting charges obey a Yangian algebra. We apply these results to many examples of integrable sigma models and their auxiliary field deformations, finding a Yangian algebra and Maillet bracket structure in all cases. This offers a unified explanation for the persistence of Hamiltonian integrability and Yangian symmetry across a wide landscape of deformed sigma models.
\end{abstract}
\vspace{5mm}

\vfill
\end{titlepage}

\newpage
\renewcommand{\thefootnote}{\arabic{footnote}}
\setcounter{footnote}{0}

\tableofcontents{}
\vspace{1cm}
\bigskip\hrule


\allowdisplaybreaks

\section{Introduction}
It is one of the mantras of modern physics that the more symmetry a theory possesses, the more tractable it is, since the conservation laws arising from such symmetry often provide simplifications and ultimately better control over the theory. The defining feature of integrable (quantum) field theories (IQFTs) is the presence of an infinite tower of conserved charges, which render such systems extremely powerful theoretical laboratories where it is sometimes possible to leverage the exceptional amount of symmetry in order to exactly compute quantities of interest. It is indeed often the case that observables computed with integrability methods are completely inaccessible otherwise, or perhaps can only be computed in some limit using perturbative techniques.

One of the beautiful algebraic structures that often emerges in IQFTs is the Yangian algebra \cite{Drinfeld:1985rx,Drinfeld:1986in,Drinfeld:1987sy} -- see also the comprehensive reviews \cite{Loebbert:2016cdm,MacKay:2004tc,Torrielli:2010kq,Torrielli:2011gg,Bernard:1992ya} -- an infinite-dimensional extension of a Lie algebra which nicely encodes the hidden symmetries of integrable systems and was introduced along the search for solutions to the quantum Yang-Baxter equation \cite{Yang:1967bm,Yang:1968rm,Baxter:1972hz}. In certain two-dimensional IQFTs, a Yangian symmetry arises from conserved non-local charges \cite{Luscher:1977uq,deVega:1984wk,deVega:1983ogx,Bernard:1990jw}, imposing strong constraints on physically relevant quantities, among which the two-particle S-matrix, that plays a fundamental role in these settings \cite{Zamolodchikov:1978xm,Dorey:1996gd,Bombardelli:2016scq}. In the context of the AdS/CFT correspondence, integrability has proven to be extremely powerful to probe the holographic duality between AdS$_5\times S^5$ and ${\cal N}=4$ SYM -- see for example \cite{Metsaev:1998it,Minahan:2002ve,Frolov:2003qc,Bena:2003wd,Dolan:2003uh,Dolan:2004ps,Dolan:2004ys,Kazakov:2004qf,Beisert:2005fw,Beisert:2005tm,Drummond:2008vq,Drummond:2009fd} and the reviews \cite{Arutyunov:2009ga,Beisert:2010jr,Bombardelli:2016rwb,Demulder:2023bux} for a more complete list of references on these topics. It played a crucial role in understanding the worldsheet scattering theory of strings in AdS$_5\times S^5$, where one can \textit{fully fix} the form of the bound states' S-matrix by requiring invariance under the Yangian of centrally extended $\mathfrak {su}(2|2)$ \cite{Beisert:2007ds,deLeeuw:2008dp, Arutyunov:2009mi}. In light of all this, it is highly desirable, whenever one is presented with an IQFT, to understand if and how a concrete realisation of the Yangian symmetry algebra can be determined.

In spite of an intrinsically quantum mechanical definition, Yangian symmetries have also turned out to play an interesting role in the realm of classical physics. Starting from a family of two-dimensional classically integrable sigma models known as Principal Chiral Models (PCM) it has indeed been shown \cite{MacKay:1992he} that the Poisson brackets of certain non-local charges generate, quite non-trivially, an algebra which can be interpreted as the classical counterpart of the Yangian. These non-local charges, first discovered by L{\"u}scher and Pohlmeyer in \cite{Luscher:1977rq} were then given a systematic derivation in terms of the Brezin, Itzykson, Zinn-Justin, and Zuber (BIZZ) construction \cite{BREZIN1979442} and were subsequently studied and related to the Yangian in various other examples of sigma models \cite{Abdalla:1993sc,Saltini:1995xr,Hatsuda:2004it,Kawaguchi:2010jg,Kawaguchi:2011mz,Kawaguchi:2012ve,Itsios:2014vfa,Klose:2016qfv}.

In this work we consider a generalisation of the BIZZ construction \cite{BREZIN1979442} and successively extend the argument of \cite{MacKay:1992he} by focussing on a certain class of current algebras and finding necessary conditions for this to give rise to a classical Yangian. The guiding principle in our approach is the search for a Yangian symmetry underlying infinite families of integrable irrelevant deformations of various classes of two-dimensional sigma models.  Such deformations include $T\overline  T $ \cite{Zamolodchikov:2004ce, Cavaglia:2016oda}, root-$T\overline  T$ \cite{Ferko:2022cix} and more generally functions of the stress-tensor and higher-spin conserved currents of the Smirnov-Zamolodchikov type \cite{Smirnov:2016lqw}.

By now these deformations have been extensively studied, and in particular the $T \overline T$ deformation of a two-dimensional QFT has been shown to exhibit several unique properties among theories with irrelevant interactions. For example the finite volume spectrum of a $T\overline T$-deformed QFT can be computed from the undeformed theory via a hydrodynamic equation \cite{Cavaglia:2016oda}. Similarly, many other quantities get deformed in a surprisingly controlled fashion, including the torus partition functions \cite{Cardy:2018sdv}, correlators \cite{Cardy:2018sdv,Li:2026ecl,Guica:2019vnb}, classical Lagrangians \cite{Cavaglia:2016oda}, and more. Interestingly, it has also been shown that the Virasoro symmetry of a CFT survives the $T\overline T$ deformation in a specific way: namely, despite breaking conformal symmetry, the generators of the conformal algebra get deformed into non-local charges which nonetheless generate a Virasoro algebra \cite{Chen:2025jzb}. We will show that a very similar mechanism underpins the Yangian algebra generated by a generalisation of the BIZZ charges, in a much larger class of irrelevant deformations which includes $T\overline T$. We will do so by using the recently constructed auxiliary field realisation of these deformations \cite{Ferko:2024ali,Bielli:2024ach,Bielli:2025uiv,Baglioni:2025tsc}, which represent a natural tool to ``hide'' all the non-locality of the deformation at the price of working partially on-shell. The paper is organised as follows:

In Section \ref{section:BIZZ-charges} we begin by reviewing the construction of the BIZZ tower of conserved charges, which relies on the existence of a flat and conserved current $L$, and proceed by extending it to the case where the system is characterised by two separate currents $A$ and $B$, respectively flat and conserved and restricted to satisfy certain relations.\\
\indent
In Section \ref{section:Yangians} we briefly recall some defining properties of the Yangian and then extend the derivation of \cite{MacKay:1992he}, also taking inspiration from \cite{Klose:2016qfv},  by applying the generalised BIZZ construction of Section \ref{section:BIZZ-charges} to a chosen broad class of current algebras satisfied by $A$ and $B$, determining under which conditions these give rise to a classical Yangian.\\
\indent
In Section \ref{section:application-to-AFSM} we study in detail several classes of sigma models and their auxiliary field deformations, showing their relation to the generalised BIZZ construction and how they satisfy the conditions of the theorems in Sections \ref{section:BIZZ-charges} and \ref{section:Yangians}, which are in fact inspired by such theories. This way we constructively prove the existence of a classical Yangian symmetry in a wealth of deformed non-linear sigma models.
We also discuss a few interesting observations resulting from this long survey of theories and finally study the Maillet structure of the Lax connections arising from the current algebra of Sections \ref{section:BIZZ-charges} and \ref{section:Yangians}.
\\
\indent
In Section \ref{section:conclusion} we finally reflect on the results presented in the main body and conclude by discussing some possible future research directions.\\
\indent
Technical aspects are collected in a few appendices. In Appendix \ref{appendix:A} we summarise notation, conventions and useful relations about Lie groups and algebras. Appendix \ref{appendix:B} contains the details of the proof highlighted in Section \ref{section:Yangians} and
Appendix \ref{appendix:C} contains detailed derivations of the Poisson brackets characterising the models described in Section \ref{section:application-to-AFSM}.

\section{Generalised Brezin, Itzykson, Zinn-Justin, Zuber currents}\label{section:BIZZ-charges}
There are a few ways to obtain towers of conserved non-local charges in integrable models, and the algebra they satisfy is typically a model-dependent feature which should be studied on a case-by-case basis. One possibility is expanding the monodromy matrix built from the Lax connection -- see for example the book/reviews \cite{Babelon:2003qtg,Torrielli:2016ufi,Loebbert:2016cdm,Driezen:2021cpd} and \cite{Itsios:2014vfa} for relations to the Yangian. Another way, particularly useful in the context of sigma models, is to begin with a ``seed'' conserved quantity and attempt to systematically generate more conserved quantities from the seed using a recursive/inductive procedure. This second method was pursued by Brezin, Itzykson, Zinn-Justin and
Zuber \cite{BREZIN1979442} and successively related to the Yangian for principal chiral models in \cite{MacKay:1992he}. In this section we first review this construction, which relies on the existence of a flat and conserved current, and successively generalise it by taking inspiration from AF sigma models, namely by considering the possibility that a system might be characterised by two separate currents, one which is flat and one which is conserved, constrained to satisfy certain commutator identities.

\subsection{Short review of the BIZZ construction}
Let us begin by briefly reviewing the derivation of non-local conserved quantities developed in \cite{BREZIN1979442}. We use the language of $\mathfrak{g}$-valued forms on a $2d$ Lorentzian manifold and refer to Appendix \ref{appendix:A} for further details about conventions and useful identities.
\begin{thm}[BIZZ]\label{BIZZ-theorem}
Consider a Lie-algebra-valued 1-form $L$ satisfying 
\begin{equation}\label{j-properties-definition}
\begin{aligned}
\text{Conservation:}& \qquad
0=\mathrm{d}(\star L) \, ,
\\
\text{Flatness:}& \qquad 0=\mathrm{d}L+\tfrac{1}{2}[L,L] \,\, .
\end{aligned}
\end{equation}
Then the following tower of currents is conserved:
\begin{equation}\label{original-BIZZ-tower}
\begin{aligned}
J^{(0)}&:=\beta \, L
\quad \quad \quad \quad \,\,\, \forall \, \beta \in \mathbb{R} \, ,
\\
J^{(1)}&:=\star L+\tfrac{1}{2\beta}[L,\chi^{(0)}] \, ,
\\
J^{(n)}&:=\nabla_{L}\chi^{(n-1)} \qquad \forall \, n \geq 2 \qquad \text{with} \qquad \nabla_{L}:=\mathrm{d}+[L,-] \,\, ,
\end{aligned}
\end{equation}
with $\chi^{(i)}$ Lie-algebra-valued 0-form potentials locally defined via current conservation
\begin{equation}\label{BIZZ-potentials}
\mathrm{d}(\star J^{(i)})=0 
\qquad \Longleftrightarrow \qquad
\star J^{(i)}=\mathrm{d}\chi^{(i)} \,\, , 
\qquad \qquad \forall \, i\geq 0
\,\, .
\end{equation}
\end{thm}
\begin{proof}
Conservation of $J^{(0)}$ follows trivially from \eqref{j-properties-definition}, while for $J^{(1)}$ one finds
\begin{equation}
\begin{aligned}
\mathrm{d}(\star J^{(1)})&=\mathrm{d}(L+\tfrac{1}{2\beta}[\star L,\chi^{(0)}])
\\
&=\mathrm{d}L-\tfrac{1}{2\beta}[\star L, \mathrm{d}\chi^{(0)}]
\\
&=\mathrm{d}L-\tfrac{1}{2\beta}[\star L,\star J^{(0)}]
\\
&=\mathrm{d}L+\tfrac{1}{2\beta}[ L, J^{(0)}]
\\
&=-\tfrac{1}{2}[L,L]+\tfrac{1}{2}[L,L]
\\
&=0 \,\, ,
\end{aligned}
\end{equation}
where in the first line we used conservation of $L$, in the second the definition of $\chi^{(0)}$ via $\mathrm{d}\chi^{(0)}=\star J^{(0)}$, in the third the Hodge-star identity \eqref{form-commutator-star-identities}, in the fourth the definition of $J^{(0)}$, and finally the flatness of $L$.
Conservation of $J^{(n)}$ for $n\geq 2$ follows instead from
\begin{equation}\label{BIZZ-important-relations}
\mathrm{d}(\star \nabla_{L}f)=\nabla_{L}(\star \mathrm{d}f)
\qquad \text{and} \qquad 
\nabla_{L}(\nabla_{L}f)=0 \qquad \forall \, f\in \Omega^{0}(\Sigma,\mathfrak{g}) \,\, ,
\end{equation}
as it can be verified explicitly:
\begin{equation}
\mathrm{d}(\star J^{(n)})=\mathrm{d}(\star \nabla_{L}\chi^{(n-1)})=\nabla_{L}(\star \mathrm{d}\chi^{(n-1)})=\nabla_{L}(J^{(n-1)})=\nabla_{L}(\nabla_{L}\chi^{(n-2)})=0 \,\, .
\end{equation}
This establishes that the tower of currents \eqref{original-BIZZ-tower} is indeed conserved,
\begin{equation}
\mathrm{d}(\star J^{(n)})=0 \qquad \forall \, n\in \mathbb{N} \,\, .
\end{equation}
The first identity in \eqref{BIZZ-important-relations} follows from the definition of $\nabla_{L}$ and conservation of $L$:
\begin{equation}
\begin{aligned}
\mathrm{d}(\star \nabla_{L}f)&=\mathrm{d}(\star \mathrm{d}f+[\star L,f])
\\
&=\mathrm{d}(\star \mathrm{d}f)+[\mathrm{d}(\star L),f]-[\star L,\mathrm{d}f]
\\
&=\mathrm{d}(\star \mathrm{d}f)+[L,\star \mathrm{d}f]
\\
&=\nabla_{L}(\star \mathrm{d}f) \,\, ,
\end{aligned}
\end{equation}
while the second relation in \eqref{BIZZ-important-relations} is a result of Jacobi identity \eqref{useful-Jacobi-1},
\begin{equation}
\begin{aligned}
\nabla_{L}(\nabla_{L}f)&=\nabla_{L}(\mathrm{d}f+[L,f])
\\
&=\mathrm{d}^2f+[L,\mathrm{d}f]+[\mathrm{d}L,f]-[L,\mathrm{d}f]+\bigl[L,[L,f]\bigr]
\\
&=[\mathrm{d}L,f]+\bigl[L,[L,f]\bigr]
\\
&=-\tfrac{1}{2}\bigl[[L,L],f\bigr]+\bigl[L,[L,f]\bigr]
\\
&=0 \,\, .
\end{aligned}
\end{equation}\qedhere
\end{proof}
The BIZZ construction of Theorem \ref{BIZZ-theorem} relies on the assumption \eqref{j-properties-definition}, namely on the existence of a current which is both conserved and flat. These conditions are known to arise in various classes of $2d$ integrable models, characterised by the existence of a Lax connection depending on the current $L$ and on a complex spectral parameter $z$ as
\begin{equation}\label{BIZZ-Lax}
\mathfrak{L}=\frac{L+z \,\star L}{1-z^2} \,\, ,
\end{equation}
and whose flatness precisely encodes the conditions \eqref{j-properties-definition}:
\begin{equation}\label{BIZZ-Lax-flatness}
\begin{aligned}
\mathrm{d}\mathfrak{L}+\tfrac{1}{2}[\mathfrak{L},\mathfrak{L}]=&\frac{1}{1\!-\!z^2}\Bigl( \mathrm{d}L+z\,\mathrm{d}(\star L)+\frac{1}{2(1\!-\!z^2)} [L+z\,\star L,L+z\,\star L]  \Bigr)
\\
=&\frac{1}{1\!-\!z^2}\Bigl( \mathrm{d}L+z\,\mathrm{d}(\star L)+\frac{1}{2(1\!-\!z^2)} \bigl( [L,L]+2z\,[L,\star L]-z^2[L,L]  \bigr)\Bigr)
\\
=&\frac{1}{1\!-\!z^2}\Bigl( \mathrm{d}L+\tfrac{1}{2}[L,L]  +z\,\mathrm{d}(\star L)\Bigr) \,\, .
\end{aligned}
\end{equation}

\subsection{Generalised BIZZ construction}\label{subsec:generalised-BIZZ}
We consider now a possible extension of the BIZZ derivation reviewed in the previous subsection. The generalisation consists in replacing the initial assumption of having a single flat and conserved current, with that of having two currents which are separately flat and conserved, restricted by certain commutator identities. In Section \ref{section:application-to-AFSM} these conditions will find application in families of deformed $2d$ sigma models, which inspired them.

\begin{thm}\label{gBIZZ-theorem}
Consider two Lie-algebra valued 1-forms $A$ and $B$ satisfying
\begin{equation}\label{A,B-definition}
\begin{aligned}
\text{Conservation:}& \qquad
0=\mathrm{d}(\star B)  \, ,
\\
\text{Flatness:}& \qquad 0=\mathrm{d}A+\tfrac{1}{2}[A,A] \, ,
\end{aligned}
\end{equation}
and commutation relations
\begin{equation}\label{A,B-commutators}
[A,A]=[B,B]
\qquad \text{and} \qquad 
[A,\star B]=[B,\star A]=0 \,\, .
\end{equation}
Then the following currents generalise the BIZZ tower \eqref{original-BIZZ-tower} and are conserved:
\begin{equation}\label{generalised-BIZZ-tower}
\begin{aligned}
J^{(0)}&:=\beta \, B \qquad \qquad \qquad \qquad \qquad \qquad \quad \quad \,\, \forall \, \beta \in \mathbb{R}\,\, ,
\\
J^{(1)}&:= \star A + \tfrac{1}{2\beta}[B,\chi^{(0)}] \, ,
\\
J^{(2)}&:=B+\tfrac{1}{2\beta}[\star A,\chi^{(0)}]+[B,\chi^{(1)}] \, ,
\\
J^{(n)}&:=J^{(n-2)}+[\star A,\chi^{(n-2)}]+[B,\chi^{(n-1)}] \qquad \forall \, n\geq 3 \,\, ,
\end{aligned}
\end{equation}
with $\chi^{(i)}$ local potentials defined as in \eqref{BIZZ-potentials}
\begin{equation}\label{gBIZZ-potentials}
\mathrm{d}(\star J^{(i)})=0 
\qquad \Longleftrightarrow \qquad
\star J^{(i)}=\mathrm{d}\chi^{(i)} \,\, , 
\qquad \qquad \forall \, i\geq 0
\,\, .
\end{equation}
\end{thm}

\begin{proof} It is straightforward to recognise that for $A=B=L$ the commutators \eqref{A,B-commutators} are identically satisfied and the conditions \eqref{A,B-definition} become \eqref{j-properties-definition}, such that the tower \eqref{generalised-BIZZ-tower} reduces to \eqref{original-BIZZ-tower}. The proof of conservation can be divided into two sectors.

\textbf{Conservation of $J^{(0)},J^{(1)},J^{(2)}$.} \qquad This can be verified by direct computation: for $J^{(0)}$ it follows trivially from \eqref{A,B-definition}, while for $J^{(1)}$ one finds 
\begin{equation}
\begin{aligned}
\mathrm{d}(\star J^{(1)})&=\mathrm{d}(A+\tfrac{1}{2\beta}[\star B,\chi^{(0)}])
\\
&=\mathrm{d}A+\tfrac{1}{2\beta}[\mathrm{d}(\star B),\chi^{(0)}]-\tfrac{1}{2\beta}[\star B,\mathrm{d}\chi^{(0)}]
\\
&=-\tfrac{1}{2}[A,A]-\tfrac{1}{2\beta}[\star B,\star J^{(0)}]
\\
&=-\tfrac{1}{2}[A,A]+\tfrac{1}{2\beta}[ B, J^{(0)}]
\\
&=-\tfrac{1}{2}[A,A]+\tfrac{1}{2}[ B, B]
\\
&=0 \,\, ,
\end{aligned}
\end{equation}
where in going from the third to the fourth line we picked up a sign on the second term using \eqref{form-commutator-star-identities} and in the last step the two terms cancel due to \eqref{A,B-commutators}. Similarly, for $J^{(2)}$ 
\begin{equation}
\begin{aligned}
\mathrm{d}(\star J^{(2)})&=\mathrm{d}(\star B+\tfrac{1}{2\beta}[A,\chi^{(0)}]+[\star B,\chi^{(1)}])
\\
&=\tfrac{1}{2\beta}[\mathrm{d}A,\chi^{(0)}]-\tfrac{1}{2\beta}[A,\mathrm{d}\chi^{(0)}]-[\star B,\mathrm{d}\chi^{(1)}]
\\
&=\tfrac{1}{2\beta}[\mathrm{d}A,\chi^{(0)}]-\tfrac{1}{2\beta}[A,\star J^{(0)}]-[\star B,\star J^{(1)}]
\\
&=\tfrac{1}{2\beta}[\mathrm{d}A,\chi^{(0)}]-\tfrac{1}{2\beta}[A,\star J^{(0)}]+[ B, J^{(1)}]
\\
&=-\tfrac{1}{4\beta}\bigl[[A,A],\chi^{(0)}\bigr]-\tfrac{1}{2}[A,\star B]+[B,\star A]+\tfrac{1}{2\beta}\bigl[B,[B,\chi^{(0)}]\bigr]
\\
&=\tfrac{1}{4\beta}\bigl[\chi^{(0)},[A,A]\bigr]+\tfrac{1}{2\beta}\bigl[B,[B,\chi^{(0)}]\bigr]
\\
&=0 \,\, ,
\end{aligned}
\end{equation}
where in the first line we exploited conservation of $B$, in the second the definition of the potentials \eqref{BIZZ-potentials}, in the third the identities \eqref{form-commutator-star-identities}, in the fourth the commutators \eqref{A,B-commutators}, and in the last again \eqref{A,B-commutators} and the Jacobi identity \eqref{useful-Jacobi-1}.

\textbf{Conservation of $J^{(n)}$ for $n\geq3$.} \qquad To show conservation of the remaining currents in the tower we proceed in three steps. First, notice that conservation for $n\geq3$ reads
\begin{equation}
\begin{aligned}
\mathrm{d}(\star J^{(n)})= [\star A,J^{(n-2)}]+[B,J^{(n-1)}]+\tfrac{1}{2}\bigl[\chi^{(n-2)},[A,A]\bigr] =: X^{(n-1)} \,\, .
\end{aligned}
\end{equation}
By direct computation one indeed finds that
\begin{equation}
\begin{aligned}
\mathrm{d}(\star J^{(n)})&=\mathrm{d}(\star J^{(n-2)}+[A,\chi^{(n-2)}]+[\star B,\chi^{(n-1)}])
\\
&=[\mathrm{d}A,\chi^{(n-2)}]-[A,\mathrm{d}\chi^{(n-2)}]-[\star B,\mathrm{d}\chi^{(n-1)}]
\\
&=-\tfrac{1}{2}\bigl[ [A,A],\chi^{(n-2)}\bigr]-[A,\star J^{(n-2)}]-[\star B,\star J^{(n-1)}]
\\
&=\tfrac{1}{2}\bigl[\chi^{(n-2)},[A,A]\bigr]+[\star A, J^{(n-2)}]+[B, J^{(n-1)}]
\\
&=: X^{(n-1)} \,\, ,
\end{aligned}
\end{equation}
where we rearranged using \eqref{BIZZ-potentials}, flatness \eqref{A,B-definition} of $A$, and commutator identities \eqref{form-commutator-star-identities}. The second step consists of noting that, thanks to the recursive definition of the currents for $n \geq 3$, the conservation conditions are recursively related as
\begin{equation}
X^{(n+1)}=X^{(n-1)} \qquad \forall \, n\geq 3 \,\, .
\end{equation}
This can once more be checked explicitly:
\begin{equation}
\begin{aligned}
X^{(n+1)}=&+[\star A,J^{(n)}]+[B,J^{(n+1)}]+\tfrac{1}{2}\bigl[\chi^{(n)},[A,A]\bigr] 
\\
=&+[\star A,J^{(n-2)}]+\bigl[\star A,[\star A,\chi^{(n-2)}] \bigr]+\bigl[\star A,[B,\chi^{(n-1)}]  \bigr]
\\
&+[B,J^{(n-1)}]+\bigl[B,[\star A,\chi^{(n-1)}] \bigr]+\bigl[B,[B,\chi^{(n)}] \bigr]+\tfrac{1}{2}\bigl[\chi^{(n)},[A,A]\bigr]
\\
=&+\bigl[\star A,[B,\chi^{(n-1)}]  \bigr]+\bigl[B,[\star A,\chi^{(n-1)}] \bigr]
\\
&+\bigl[B,[B,\chi^{(n)}] \bigr]+\tfrac{1}{2}\bigl[\chi^{(n)},[A,A]\bigr]+X^{(n-1)}
\\
=&X^{(n-1)} \,\, .
\end{aligned}
\end{equation}
In the first line we substituted the recursive definition of $J^{(n)}$ and $J^{(n+1)}$ and in the second we collected the terms defining $X^{(n-1)}$, after noting that 
\begin{equation}
\bigl[\star A,[\star A,\chi^{(n-2)}] \bigr]=-\bigl[ A,[ A,\chi^{(n-2)}] \bigr] =+\tfrac{1}{2}\bigl[ \chi^{(n-2)},[A,A]\bigr] \,\, ,
\end{equation}
as a result of \eqref{form-commutator-star-identities} and the Jacobi identity \eqref{useful-Jacobi-1}. In going to the last line we exploited that the four extra terms cancel pairwise as a result of the Jacobi identities \eqref{useful-Jacobi-1}, \eqref{useful-Jacobi-2}, and the commutator identity \eqref{A,B-commutators}. At this point we have shown that as a result of the recursive definition \eqref{generalised-BIZZ-tower} for currents $J^{(n)}$ with $n\geq 3$, their conservation conditions also exhibit a recursive structure of the form
\begin{equation}
\mathrm{d}(\star J^{(n)})=X^{(n-1)}=X^{(n-3)}=X^{(n-5)} = \ldots \, ,
\end{equation}
and to show conservation of the whole tower it is sufficient to check the vanishing of the first non-trivial even and odd cases, namely $X^{(2)}=0$ and $X^{(3)}=0$, which constitutes the third and final step of the proof. For $X^{(2)}$ one finds that
\begin{equation}
\begin{aligned}
X^{(2)}=&+[\star A,J^{(1)}]+[ B,J^{(2)}]+\tfrac{1}{2}\bigl[\chi^{(1)},[A,A]\bigr]
\\
=&+[\star A,\star A]+\tfrac{1}{2\beta}\bigl[\star A,[B,\chi^{(0)}]\bigr]
\\
&+[B,B]+\tfrac{1}{2\beta}\bigl[B,[\star A,\chi^{(0)}]\bigr]+\bigl[B,[B,\chi^{(1)}]\bigr]+\tfrac{1}{2}\bigl[\chi^{(1)},[A,A]\bigr]
\\
=&-[A,A]+[B,B]+
\\
&+\tfrac{1}{2\beta}\bigl[\star A,[B,\chi^{(0)}]\bigr]+\tfrac{1}{2\beta}\bigl[B,[\star A,\chi^{(0)}]\bigr]
\\
&-\tfrac{1}{2}\bigl[\chi^{(1)},[B,B]\bigr]+\tfrac{1}{2}\bigl[\chi^{(1)},[A,A]\bigr]
\\
=&0 \,\, ,
\end{aligned}
\end{equation}
where in the first line we substituted the definitions \eqref{generalised-BIZZ-tower} of $J^{(1)}$ and $J^{(2)}$, in the second we exploited the identities \eqref{form-commutator-star-identities}, and in the third we noticed that terms cancel pairwise after using \eqref{A,B-commutators} and the Jacobi identity \eqref{useful-Jacobi-2}. Similarly, for $X^{(3)}$ one finds
\begin{equation}
\begin{aligned}
X^{(3)}=&+[\star A,J^{(2)}]+[B,J^{(3)}]+\tfrac{1}{2}\bigl[\chi^{(2)},[A,A]\bigr]
\\
=&+[\star A,B]+\tfrac{1}{2\beta}\bigl[\star A,[\star A,\chi^{(0)}]\bigr]+\bigl[\star A,[B,\chi^{(1)}]\bigr]+\tfrac{1}{2}\bigl[\chi^{(2)},[A,A]\bigr]
\\
&+[B,J^{(1)}]+\bigl[B,[\star A,\chi^{(1)}]\bigr]+\bigl[B,[B,\chi^{(2)}]\bigr]
\\
=&-\tfrac{1}{2\beta}\bigl[A,[A,\chi^{(0)}]\bigr]+\bigl[\star A,[B,\chi^{(1)}]\bigr]+\tfrac{1}{2}\bigl[\chi^{(2)},[A,A]\bigr]
\\
&+[B,\star A]+\tfrac{1}{2\beta}\bigl[B,[B,\chi^{(0)}]\bigr]+\bigl[B,[\star A,\chi^{(1)}]\bigr]+\bigl[B,[B,\chi^{(2)}]\bigr]
\\
=&+\tfrac{1}{4\beta}\bigl[\chi^{(0)},[A,A]\bigr]-\tfrac{1}{4\beta}\bigl[\chi^{(0)},[B,B]\bigr]+
\\
&+\bigl[\star A,[B,\chi^{(1)}]\bigr]+\bigl[B,[\star A,\chi^{(1)}]\bigr]
\\
&+\tfrac{1}{2}\bigl[\chi^{(2)},[A,A]\bigr]-\tfrac{1}{2}\bigl[\chi^{(2)},[B,B]\bigr]
\\
=&0 \,\, ,
\end{aligned}
\end{equation}
where in the first line we substituted the definitions of $J^{(2)}$ and $J^{(3)}$ in \eqref{generalised-BIZZ-tower}, and in the second we substituted the expression for $J^{(1)}$ in \eqref{generalised-BIZZ-tower} and simplified a few terms using \eqref{A,B-commutators}, \eqref{form-commutator-star-identities}. In the third line we rearranged terms using the Jacobi identity \eqref{useful-Jacobi-1} and in the fourth one we just exploited cancellations resulting from \eqref{A,B-commutators} and \eqref{useful-Jacobi-2}.
\end{proof}
The generalised BIZZ construction in Theorem \ref{gBIZZ-theorem} represents a straightforward extension of Theorem \ref{BIZZ-theorem}, where the flat and conserved current $L$ is replaced by two currents $A$ and $B$ which are respectively only flat and only conserved, while being at the same time not completely independent, but related via the conditions \eqref{A,B-commutators}. The latter requirement is not only crucial in showing the conservation of the tower \eqref{generalised-BIZZ-tower}, as seen in the above proof, but also allows us to recover the flatness and conservation \eqref{A,B-definition} of $A$ and $B$ from the flatness of a Lax connection which represents the natural generalisation of \eqref{BIZZ-Lax},
\begin{equation}\label{gBIZZ-Lax}
\mathfrak{L}=\frac{A+z \,\star B}{1-z^2} \,\, ,
\end{equation}
as one can see by following steps analogous to \eqref{BIZZ-Lax-flatness} and using the condition \eqref{A,B-commutators}:
\begin{equation}
\begin{aligned}
\mathrm{d}\mathfrak{L}+\tfrac{1}{2}[\mathfrak{L},\mathfrak{L}]=&\frac{1}{1\!-\!z^2}\Bigl( \mathrm{d}A+z\,\mathrm{d}(\star B)+\frac{1}{2(1\!-\!z^2)} [A+z\,\star B,A+z\,\star B]  \Bigr)
\\
=&\frac{1}{1\!-\!z^2}\Bigl( \mathrm{d}A+z\,\mathrm{d}(\star B)+\frac{1}{2(1\!-\!z^2)} \bigl( [A,A]+2z\,[A,\star B]-z^2[B,B]  \bigr)\Bigr)
\\
=&\frac{1}{1\!-\!z^2}\Bigl( \mathrm{d}A+\frac{1}{2}[A,A]  +z\,\mathrm{d}(\star B)\Bigr) \,\, .
\end{aligned}
\end{equation}
Clearly, the above construction reduces to the standard BIZZ result of Theorem \ref{BIZZ-theorem} when both $A$ and $B$ coincide with $L$, in which case the conditions \eqref{A,B-commutators} are trivially satisfied. While in Section \ref{section:Yangians} we will try to keep a model-independent point of view and consider a broad class of generalised BIZZ currents encoding a classical Yangian algebra, in Section \ref{section:application-to-AFSM} we will look at various explicit examples of $2d$ integrable sigma models exhibiting the structure of Theorem \ref{BIZZ-theorem}, as well as to a set of deformations illustrating Theorem \ref{gBIZZ-theorem}.

\section{Classical Yangians for generalised BIZZ currents}\label{section:Yangians}
An important application of the Theorem \ref{BIZZ-theorem} of Brezin, Itzykson, Zinn-Justin, Zuber is that the tower of conserved currents \eqref{original-BIZZ-tower} can be exploited to show the emergence of an underlying classical Yangian algebra -- understood as the Poisson algebra of appropriately defined charges -- for certain theories which are characterised by a flat and conserved current \eqref{j-properties-definition}, such as principal chiral models \cite{MacKay:1992he} or symmetric-space sigma models \cite{Klose:2016qfv}. The existence of such a large symmetry, briefly reviewed below in subsection \ref{subsec:Yangian-first-realisation}, does not however only depend on the existence of the BIZZ currents, but is also tightly connected to the Poisson bracket structure of the flat and conserved current $L$ out of which all the elements in the tower \eqref{original-BIZZ-tower} ultimately depend. For this reason, knowing the BIZZ currents does not automatically allow us to conclude that a certain theory does enjoy a Yangian symmetry, and a case-by-case proof of such a property should in fact be sought. In light of these observations, our main objective throughout this section will be showing how a classical Yangian symmetry can also emerge in theories characterised by a generalised BIZZ construction of Theorem \ref{gBIZZ-theorem}, when the Poisson brackets of the $A$ and $B$ currents are assumed to take a specific, yet quite broad, structure. In particular, as will become clear in Section \ref{section:application-to-AFSM}, such a structure for the Poisson brackets is meant to at least encompass, and is inspired by, quite a few examples of so-called auxiliary field sigma models. The latter represent infinite families of integrable deformations of known integrable theories and our analysis will show the existence of an underlying classical Yangian symmetry for any element of such families, hence including the undeformed settings as special cases.

\subsection{Short review of the Yangian's first realisation}\label{subsec:Yangian-first-realisation}
There exist various ways of defining a Yangian \cite{Drinfeld:1985rx,Drinfeld:1986in,Drinfeld:1987sy}; see e.g. the comprehensive reviews \cite{Loebbert:2016cdm,MacKay:2004tc,Torrielli:2010kq,Torrielli:2011gg,Bernard:1992ya}. For this reason, before showing in the next subsection \ref{subsec:Yangian-theorem} how these can arise from certain classes of generalised BIZZ currents, it will be useful to briefly summarise the definition to which we will be referring, which is also known as the ``first realisation''. Our analysis is inspired by and closely follows the approaches \cite{MacKay:1992he,Klose:2016qfv} and we also recommend \cite{Itsios:2014vfa} as an excellent guidebook with various important explicit calculations. 

The Yangian $Y(\mathfrak{g})$ of a finite-dimensional simple Lie algebra $\mathfrak{g}$ -- see Appendix \ref{appendix:A} for notation -- is an infinite-dimensional algebra defined via the relations between the level-0 generators $T^{A}$, namely the Lie algebra, and level-1 generators $M^{A}$
\begin{equation}\label{HM-Yangian-general-level0-and-level1-brackets}
[T^{A},T^{B}]=f^{AB}{}_{C}T^{C} 
\qquad \text{and} \qquad 
[ T^{A},M^{B} ]=f^{AB}{}_{C}M^{C} \,\, .
\end{equation}
Higher-level generators can be iteratively obtained from commutators of the lower ones, subject to the following consistency condition, known as first Serre relation\footnote{\label{footnote:symmetrisation-of-indices}Square brackets $[A_{1}...A_{n}]$ on the left hand side of \eqref{HM-Serre-relations} denote antisymmetrisation of indices and curved brackets $(A_{1}...A_{n})$ on the right hand side of \eqref{HM-Serre-relations} and \eqref{second-serre}  denote summation over all permutations of the enclosed generators. Both are understood here with an overall counting factor $\tfrac{1}{n!}$.}
\begin{equation}\label{HM-Serre-relations}
\begin{aligned}
f_{D}{}^{[AB}[M^{C]},M^{D}]=&\tfrac{\alpha^2}{12}a^{ABC}{}_{PQR}T^{(P}T^{Q}T^{R)},  
\\
\text{with} \qquad a^{ABC}{}_{PQR}:=&f^{A}{}_{PI}f^{B}{}_{QJ}f^{C}{}_{RK}f^{IJK} 
\quad \text{and} \quad \alpha\in \mathbb{R}.
\end{aligned}
\end{equation}
In fact, the Yangian is equipped with a Hopf algebra structure -- we refer to the papers and reviews mentioned above for more details on this topic -- and the condition \eqref{HM-Serre-relations} arises from the requirement that the coproduct should be a homomorphism. This should furthermore be supplemented by the second Serre relation
\begin{equation}\label{second-serre}
\begin{aligned}
f_{E}{}^{AB}&\bigl[[M^{C},M^{D}],M^{E}\bigr]+f_{E}{}^{CD}\bigl[[M^{A},M^{B}],M^{E}\bigr]
\\
=&\tfrac{\alpha^2}{12}(a^{ABE}{}_{PQR}f_{E}{}^{CD}+a^{CDE}{}_{PQR}f_{E}{}^{AB})T^{(P}T^{Q}M^{R)} ,
\end{aligned}
\end{equation}
which is however implied by \eqref{HM-Serre-relations} -- see e.g. \cite{Itsios:2014vfa} for a proof -- unless one is considering the case of $\mathfrak{g}=\mathfrak{sl}(2)$, where \eqref{HM-Serre-relations} trivialises and \eqref{second-serre} becomes the relevant constraint.

\subsection{A class of generalised BIZZ currents with Yangian symmetry}\label{subsec:Yangian-theorem}
As previously anticipated, in this subsection we are going to show how the classical version of the Yangian structure \eqref{HM-Yangian-general-level0-and-level1-brackets} and \eqref{HM-Serre-relations} can emerge from a setup involving a flat current $A$ and a conserved current $B$ characterising the generalised BIZZ construction of Theorem \ref{gBIZZ-theorem}. This is something that cannot be shown in full generality, as the recovery of the desired relations ultimately depends on the Poisson brackets obeyed by the components of the currents, and for this reason we will have to restrict our scope by making some assumptions on such a structure.
Our restriction will only be motivated a posteriori, eventually through the survey of examples discussed in Section \ref{section:application-to-AFSM}, which will all be included in the chosen Poisson brackets and hence exhibit a classical Yangian symmetry.

\begin{thm}\label{Yangian-theorem} Consider two Lie-algebra valued 1-forms $A, B \in \mathfrak{g}$, with $\mathfrak{g} \neq \mathfrak{sl} ( 2 )$, satisfying the generalised BIZZ construction of Theorem \ref{gBIZZ-theorem} and the Poisson brackets
\begin{align}\label{Yangian-general-brackets}
\{ B_{\tau}{}^{A}(\sigma),B_{\tau}{}^{B}(\sigma') \} =& m_{1}f^{AB}{}_{C} \, B_{\tau}{}^{C}(\sigma)\delta(\sigma-\sigma')+m_{2}\gamma^{AB}\partial_{\sigma}\delta(\sigma-\sigma') \, ,
\notag \\
\{ B_{\tau}{}^{A}(\sigma),A_{\sigma}{}^{B}(\sigma') \} =& m_{3}f^{AB}{}_{C} \, A_{\sigma}{}^{C}(\sigma)\delta(\sigma-\sigma')+m_{4}C^{AB}(\sigma') \, \partial_{\sigma}\delta(\sigma-\sigma') \, ,
\\
\{ A_{\sigma}{}^{A}(\sigma),A_{\sigma}{}^{B}(\sigma') \} =& f^{AB}{}_{C}\Bigl(m_{5}A_{\sigma}{}^{C}+m_{6}B_{\tau}{}^{C}\Bigr)\delta(\sigma-\sigma')+m_{7}\gamma^{AB}\partial_{\sigma}\delta(\sigma-\sigma') \,\, ,
\notag
\end{align}
with $m_{1},...,m_{7} \in \mathbb{R}$ constant coefficients, $\gamma^{AB}$ the Cartan-Killing form on the underlying Lie-algebra and $C^{AB}(\sigma)$ a possibly field-dependent symmetric tensor.\footnote{Notice that compatibility of the brackets \eqref{Yangian-general-brackets} with Jacobi identities imposes constraints on $C^{AB}(\sigma')$. We refer to Appendix \ref{appendix:B} for further details and comments on this point.
} Then, the charges 
\begin{equation}\label{Yangian-charges-def}
Q^{(i)A}=\int_{-\infty}^{+{\infty}} \mathrm{d}\sigma \,  J_{\tau}^{(i)A} \,\, ,
\end{equation}
with $J_{\tau}^{(i)A}$ the timelike components of the generalised BIZZ currents \eqref{generalised-BIZZ-tower}, obey the classical version $Y_{\text{C}}(\mathfrak{g})$ of the Yangian algebra in subsection \ref{subsec:Yangian-first-realisation}, namely
\begin{equation}\label{Yangian-general-level0-and-level1-brackets}
\{Q^{(0)A},Q^{(0)B}\}=f^{AB}{}_{C}Q^{(0)C} 
\qquad , \qquad 
\{ Q^{(0)A},Q^{(1)B}  \}=f^{AB}{}_{C}Q^{(1)C} \,\, ,
\end{equation}
and \footnote{With the conventions in footnote \ref{footnote:symmetrisation-of-indices}, for classically commuting quantities one has $Q^{(0)(P}Q^{(0)Q}Q^{(0)R)}=Q^{(0)P}Q^{(0)Q}Q^{(0)R}$, which differs by a factor $1/6$ from e.g. \cite{MacKay:1992he} or \cite{Itsios:2014vfa}.}
\begin{equation}\label{Serre-relations}
f_{D}{}^{[AB}\{Q^{(1)C]},Q^{(1)D}\}=\tfrac{\alpha^2}{12}f^{A}{}_{PI}f^{B}{}_{QJ}f^{C}{}_{RK}f^{IJK}Q^{(0)P}Q^{(0)Q}Q^{(0)R} 
\quad \text{with} \quad
\alpha:=m_{1}^2 \,\, ,
\end{equation}
provided that $\beta$ in \eqref{gBIZZ-theorem} is chosen as $\beta:=m_{1}^{-1}$ and the following relations are satisfied:
\begin{equation}\label{Yangian-theorem-condition}
\begin{aligned}
m_{3} =&m_{1} \qquad \quad \text{and} \qquad\quad f_{D}{}^{[AB}\mathcal{V}_{1}^{C]D}=0 \,\, , 
\qquad\quad \text{with} 
\\
\mathcal{V}_{1}^{CD}:=&
\int_{-\infty}^{+\infty}\mathrm{d}\sigma' \,   B_{\tau}{}^{E}(\sigma')\Bigl(f_{E}{}^{CF}C_{F}{}^{D}(\sigma')-f_{E}{}^{DF}C_{F}{}^{C}(\sigma')\Bigr) \,\, .
\end{aligned}
\end{equation}
\end{thm}
To prove Theorem \ref{Yangian-theorem} we will closely follow \cite{MacKay:1992he}, supplemented by the analysis in Section 5.2 of \cite{Klose:2016qfv} to deal with the presence of $C^{AB}(\sigma')$ in the relevant brackets \eqref{Yangian-general-brackets}. Studying the structures \eqref{Yangian-general-level0-and-level1-brackets}, \eqref{Serre-relations} requires us to consider the first two currents in \eqref{generalised-BIZZ-tower},
\begin{equation}\label{Yangian-currents-we-need}
J_{\tau}^{(0)}=\beta \, B_{\tau}
\qquad \text{and} \qquad 
J_{\tau}^{(1)}=s \, A_{\sigma}+\tfrac{1}{2\beta}[B_{\tau},\chi^{(0)}] \,\, ,
\end{equation}
where $s=\pm 1$ represents a sign ($s^2=1$) which depends on the choice of convention for the Hodge star operator, and computing the following brackets:
\begin{align}\label{Yangian-brackets-to-compute}
\{Q^{(0)A}(\sigma),Q^{(0)B}(\sigma')\}=&\int_{-\infty}^{+\infty}\!\!\!\!\!\!\mathrm{d}\sigma
\mathrm{d}\sigma' \, \beta^2 \{B_{\tau}{}^{A}(\sigma),B_{\tau}{}^{B}(\sigma')\} \, ,
\\
\{Q^{(0)A}(\sigma),Q^{(1)B}(\sigma')\}=& \int_{-\infty}^{+\infty} \!\!\!\!\!\! \mathrm{d}\sigma
\mathrm{d}\sigma' \, \beta s \, \{B_{\tau}{}^{A}(\sigma),A_{\sigma}{}^{B}(\sigma')\}+\tfrac{1}{2} \, \{B_{\tau}{}^{A}(\sigma),[B_{\tau}(\sigma'),\chi^{(0)}(\sigma')]^{B}\} \, ,
\notag \\
\{Q^{(1)A}(\sigma),Q^{(1)B}(\sigma')\}=& \int_{-\infty}^{+\infty} \!\!\!\!\!\! \mathrm{d}\sigma
\mathrm{d}\sigma' \, s^2 \, \{A_{\sigma}{}^{A}(\sigma),A_{\sigma}{}^{B}(\sigma')\} \, ,
\notag \\
+&\tfrac{s}{2\beta} \, \{A_{\sigma}{}^{A}(\sigma),[B_{\tau}(\sigma'),\chi^{(0)}(\sigma')]^{B}\}+\tfrac{s}{2\beta} \, \{[B_{\tau}(\sigma),\chi^{(0)}(\sigma)]^{A},A_{\sigma}{}^{B}(\sigma')\}
\notag \\
+&\tfrac{1}{4\beta^2}
\{ [B_{\tau}(\sigma),\chi^{(0)}(\sigma)]^{A} ,[B_{\tau}(\sigma'),\chi^{(0)}(\sigma')]^{B}  \} \,\, .
\notag
\end{align}
Before proceeding, however, a few important comments are in order:
\begin{enumerate}
\item The potential $\chi^{(0)}$ was introduced in \eqref{gBIZZ-potentials} via the relation $J^{(0)}=\star \mathrm{d}\chi^{(0)}$, which rewritten in terms of the current $B$, and in components, reads \cite{MacKay:1992he}
\begin{equation}\label{Yangian-B-chi0-relation}
\!\!B_{\tau}\!=\!\tfrac{s}{\beta}\,\partial_{\sigma}\chi^{(0)} \quad \longleftrightarrow \quad \chi^{(0)}\!(\sigma)\!=\!\beta s\!\!\int_{-\sigma_{1}}^{+\sigma_{2}}\!\!\!\!\!\!\!\!\!\!\mathrm{d}\sigma' \, \theta(\sigma\!-\!\sigma')B_{\tau}(\sigma')\quad \text{for} \quad \sigma\!\in\![-\sigma_{1},+\sigma_{2}] \, ,
\end{equation}
with $s$ the same sign as in \eqref{Yangian-currents-we-need} and $\theta(x)$ the Heaviside function defined as
\begin{equation}\label{Heaviside-definition}
\theta(x)=
\begin{cases}
1 \qquad x >0
\\
\tfrac{1}{2} \qquad x=0
\\
0 \qquad x<0
\end{cases}
\qquad \text{satisfying} \qquad
\begin{cases}
\theta(x)+\theta(-x)=1 
\\
\partial_{x}\theta(x)=\delta(x)
\\
\partial_{x}\theta(x-x_{0})=\delta(x-x_{0})
\end{cases}
\,\, .
\end{equation}
This rewriting of the potential in terms of $B_{\tau}$ allows us to cast each contribution to the Poisson brackets \eqref{Yangian-brackets-to-compute} in terms of combinations of the brackets 
\begin{equation}\label{key-brackets}
\{B_{\tau}{}^{A}(\sigma),B_{\tau}{}^{B}(\sigma')\} \quad, \quad \{B_{\tau}{}^{A}(\sigma),A_{\sigma}{}^{B}(\sigma')\} \quad , \quad \{A_{\sigma}{}^{A}(\sigma),A_{\sigma}{}^{B}(\sigma')\} \,\, ,
\end{equation}
and the emergence of a classical Yangian structure encoded in the form \eqref{Yangian-general-level0-and-level1-brackets} and \eqref{Serre-relations} is hence dependent on the detailed structure of \eqref{key-brackets}. This naturally brings up an intriguing fundamental question: what is the most general form of these brackets which remains compatible with the emergence of a Yangian? Finding an answer is certainly a quite complicated task, which we will not dare to tackle here. Having in mind the applications in Section \ref{section:application-to-AFSM}, we will be restricting ourselves to \eqref{Yangian-general-brackets}.
\item The presence of a non-ultralocal term in the second line of \eqref{Yangian-general-brackets}, characterised by a symmetric tensor $C^{AB}(\sigma')$ possibly dependent on the spatial coordinates, requires extra care, as noted in \cite{Klose:2016qfv}. We can proceed with a similar approach, namely by formally putting the contributions to each charge \eqref{Yangian-charges-def} on finite $\sigma$-domains, ultimately sending the limits of integration to infinity following an ordering prescription which allows us to get rid of unwanted boundary terms generated by the non-ultralocal structure, ultimately recovering \eqref{Yangian-charges-def}. We will hence effectively consider the charges
\begin{equation}\label{charges-with-boundaries}
\begin{aligned}
\!\!\!\!\! Q^{(0)A}\!=\!\!\int_{-\sigma_{1}}^{+\sigma_{2}} \!\!\!\!\!\!\! \mathrm{d}\sigma \, \beta B_{\tau}{}^{A}(\sigma) 
\,\,\,;\,\,\,
Q^{(1)A}\!=\!\!\int_{-\sigma_{5}}^{+\sigma_{6}} \!\!\!\!\!\!\! \mathrm{d}\sigma \, s A_{\sigma}{}^{A}(\sigma)+\!\!\int_{-\sigma_{3}}^{+\sigma_{4}} \!\!\!\!\!\!\! \mathrm{d}\sigma \, \tfrac{1}{2\beta}[ B_{\tau}(\sigma),\chi^{(0)}(\sigma)]^{A} \,\, ,
\end{aligned}
\end{equation}
with the labelling of the cutoffs $\sigma_i$ following the structure in \cite{Klose:2016qfv} for better comparison. We will also refer to these limits of integration as ``boundaries''.
\item For later purposes, it is also important to take a closer look at the relation between the potential $\chi^{(0)}$ in \eqref{Yangian-B-chi0-relation} and the first Yangian charge $Q^{(0)}$ in \eqref{charges-with-boundaries}. Substituting the first relation in \eqref{Yangian-B-chi0-relation} into $Q^{(0)}$ in \eqref{charges-with-boundaries} leads to
\begin{equation}
 Q^{(0)A}=s\int_{-\sigma_{1}}^{+\sigma_{2}} \!\!\!\!\!\!\! \mathrm{d}\sigma \, \partial_{\sigma} \chi^{(0)A}(\sigma) = s\Bigl( \chi^{(0)A}(+\sigma_{2})-\chi^{(0)A}(-\sigma_{1})\Bigr) \,\, ,
\end{equation}
while evaluating the second relation in \eqref{Yangian-B-chi0-relation} at the boundaries one finds
\begin{equation}
\begin{aligned}
\chi^{(0)}(+\sigma_{2})&=\beta s\int_{-\sigma_{1}}^{+\sigma_{2}}\!\!\!\!\!\!\!\!\!\!\mathrm{d}\sigma' \, \theta(\sigma_{2}\!-\!\sigma')B_{\tau}(\sigma')=\beta s\int_{-\sigma_{1}}^{+\sigma_{2}}\!\!\!\!\!\!\!\!\!\!\mathrm{d}\sigma' \, B_{\tau}(\sigma')=sQ^{(0)} \, ,
\\
\chi^{(0)}(-\sigma_{1})&=\beta s\int_{-\sigma_{1}}^{+\sigma_{2}}\!\!\!\!\!\!\!\!\!\!\mathrm{d}\sigma' \, \theta(-\sigma_{1}\!-\!\sigma')B_{\tau}(\sigma')=0 \,\, ,
\end{aligned}
\end{equation}
since $\sigma'\in[-\sigma_{1},+\sigma_{2}]$ and hence $\theta(\sigma_{2}-\sigma')=1 $ for all $\sigma'$ and $\theta(-\sigma_{1}-\sigma')=0$ for all $\sigma'$. 
These relations connect the charge $Q^{(0)}$ to the value of the potential $\chi^{(0)}$ at the boundaries, which should ultimately be sent to infinity, leading to
\begin{equation}\label{relation-charges-potential}
Q^{(0)}=s\chi^{(0)}(+\infty) 
\qquad \text{and} \qquad 
\chi^{(0)}(-\infty) = 0 \,\, .
\end{equation}
\item The Serre relations \eqref{Serre-relations} will be recovered by first computing the third Poisson bracket in \eqref{Yangian-brackets-to-compute} and successively contracting the result with the structure constants and antisymmetrising the indices appropriately. In this procedure, the result of the third bracket in \eqref{Yangian-brackets-to-compute} is characterised by two main tensorial structures
\begin{equation}\label{level1-bracket-tensorial-structures-def}
C_{E_{1}..E_{n}}{}^{AB}:=X_{E_{1}...E_{n}}{}^{P}f_{P}{}^{AB}
\qquad \text{and} \qquad
T_{CFG}{}^{AB}:=f_{CD}{}^{[A}f_{EF}{}^{B]}f^{DE}{}_{G}
 \,\, ,
\end{equation}
with $X_{E_{1}...E_{n}}{}^{P}$ unspecified and possibly carrying any number $n$ of lower indices. $C$-structures enjoy the crucial property of being ``central'', in the sense that they vanish when contracted with the structure constants and antisymmetrised,
\begin{equation}\label{central-tensors-def}
f_{B}{}^{[CD}C_{E_{1}...E_{n}}{}^{A]B}=0 \,\, ,
\end{equation}
as shown in \eqref{vanishing-central-terms} and needed for the Serre relations \eqref{Serre-relations}. The second relevant structure in \eqref{level1-bracket-tensorial-structures-def}, the $T_{CFG}{}^{AB}$ tensor, is by definition antisymmetric in the upper indices and has the property of being fully symmetric, up to central contributions, in the lower three indices, as shown in  \eqref{T-sym-first-two-indices}, \eqref{T-sym-last-two-indices} and \eqref{T-sym-first-and-last-indices}:
\begin{equation}\label{T-symmetry}
T_{CFG}{}^{AB}=T_{(CFG)}{}^{AB} + C_{CFG}{}^{AB} \,\, .
\end{equation}
Importantly, while this symmetry follows from the antisymmetrisation in the upper indices, when the lower indices are contracted with three copies of the same object, $T$ becomes automatically antisymmetric in the upper indices:
\begin{equation}\label{T-tensor-contraction}
\!T_{PQR}{}^{CD}O^{P}O^{Q}O^{R}\!=\!f_{PE}{}^{[C}f_{FQ}{}^{D]}f^{EF}{}_{R}O^{P}O^{Q}O^{R}\!=\!f_{PE}{}^{C}f_{FQ}{}^{D}f^{EF}{}_{R}O^{P}O^{Q}O^{R} \, .
\end{equation}
Contracting with the structure constants and antisymmetrising one finds
\begin{equation}\label{Serre-structure-T-tensor-contraction}
\begin{aligned}
f_{D}{}^{[AB}T_{PQR}{}^{C]D}O^{P}O^{Q}O^{R}&=-f^{A}{}_{PI}f^{B}{}_{QJ}f^{C}{}_{RK}f^{IJK}O^{P}O^{Q}O^{R}
\\
&=-a^{ABC}{}_{PQR}O^{P}O^{Q}O^{R}
\,\, ,
\end{aligned}
\end{equation}
as shown in \eqref{T-contraction-intermediate}, \eqref{T-contraction-antisymmetrisation} and needed for the Serre relations \eqref{Serre-relations}.
\end{enumerate}
We now have all the necessary ingredients to prove Theorem \ref{Yangian-theorem}. Since we will proceed by direct evaluation of the brackets \eqref{Yangian-general-level0-and-level1-brackets} and \eqref{Serre-relations}, and various contributions will have to be taken into account, it is convenient to first provide a brief roadmap of the calculation.
\begin{itemize}
\item $\{Q^{(0)A},Q^{(0)B}\}$ has two contributions. One of them vanishes identically, while the other straightforwardly leads to the desired result upon choosing $\beta:=m_{1}^{-1}$ in \eqref{generalised-BIZZ-tower}.
\item $\{Q^{(0)A},Q^{(1)B}\}$ has contributions $Y_{1}^{AB},Y_{2}^{AB},X_{1}^{AB},X_{2}^{AB}$. The last two are boundary terms and cannot contribute to $Q^{(1)C}$, which must be recovered from the first two.
\begin{itemize}
\item $Y_{1}^{AB}$ provides the first part of $Q^{(1)C}$ in \eqref{charges-with-boundaries} and imposes the ordering conditions $\sigma_{6}<\sigma_{2}$ and $\sigma_{5}<\sigma_{1}$ on the boundaries.
\item $Y_{2}^{AB}$ provides the second part of $Q^{(1)C}$ in \eqref{charges-with-boundaries} and imposes that $\sigma_{4}<\sigma_{2}$ and $ \sigma_{3}<\sigma_{1}$. Combining this term with $Y_{1}^{AB}$ further requires $m_{3}=m_{1}$.
\item $X_{1}^{AB}$ and $X_{2}^{AB}$ separately vanish provided that $(\sigma_{6}<\sigma_{2} \, , \, \sigma_{5}<\sigma_{1})$ and $(\sigma_{4}<\sigma_{2} \, , \, \sigma_{3}<\sigma_{1})$, which is compatible with the requirements imposed by $Y_{1}^{AB},Y_{2}^{AB}$.
\end{itemize}
\item $\{Q^{(1)A},Q^{(1)B}\}$ contains three contributions $U^{AB},V^{AB},Z^{AB}$.
\begin{itemize}
\item $U^{AB}$ is immediately found to be central in the sense of \eqref{central-tensors-def} and hence does not contribute to the Serre relations.
\item $Z^{AB}$ is further divided into two more contributions $Z_{1}^{AB}$ and $Z_{2}^{AB}$.
\begin{itemize}
\item $Z_{1}^{AB}$ precisely leads to the Serre relations. This means that the remaining terms, namely $Z_{2}^{AB}$ and $V^{AB}$, should either vanish or cancel each other or be central, not to spoil the desired result.
\item $Z_{2}^{AB}$ splits into two further terms; one of them vanishes identically while the other is central. Altogether neither of them affects the Serre relations.
\end{itemize}
\item $V^{AB}$ contains in turn three new terms
$\mathcal{V}_{1}^{AB},\mathcal{V}_{2}^{AB},\mathcal{V}_{3}^{AB}$.
\begin{itemize}
\item $\mathcal{V}_{3}^{AB}$ is central and does not contribute to the Serre relations.
\item $\mathcal{V}_{2}^{AB}$ is a boundary term which vanishes identically if $\sigma_{6}<\sigma_{4}$ and $\sigma_{5}<\sigma_{3}$. This requirement is compatible with the previous ones and fully fixes the ordering to be used in sending the boundaries to infinity as
\begin{equation}
\sigma_{6} <\sigma_{4}<\sigma_{2}
\qquad \text{and} \qquad 
\sigma_{5} <\sigma_{3}<\sigma_{1} \,\, .
\end{equation}
\item The term $\mathcal{V}_{1}^{AB}$ is the only one remaining and potentially spoiling the Serre relations, leading to the constraint \eqref{Yangian-theorem-condition}.
\end{itemize}
\end{itemize}
\end{itemize}

\begin{proof} For a smoother presentation, here we quote the expressions for the various contributions highlighted above, referring to Appendix \ref{appendix:B} for the more technical computations.

\textbf{Calculation of $\{Q^{(0)A}(\sigma),Q^{(0)B}(\sigma')\}$.} \quad Using \eqref{charges-with-boundaries} and \eqref{Yangian-general-brackets} one finds
\begin{align}
\{Q^{(0)A}(\sigma),Q^{(0)B}(\sigma')\}=&\beta^2\int_{-\sigma_{1}}^{+\sigma_{2}}\!\!\!\!\!\!\!\!\!\!\mathrm{d}\sigma\mathrm{d}\sigma' \, \{ B_{\tau}{}^{A}(\sigma),B_{\tau}{}^{B}(\sigma') \}
\notag\\
=& m_{1}\beta^2 f^{AB}{}_{C} \!\!\int_{-\sigma_{1}}^{+\sigma_{2}}\!\!\!\!\!\!\!\!\!\!\mathrm{d}\sigma\mathrm{d}\sigma' \, B_{\tau}{}^{C}(\sigma)\delta(\sigma\!-\!\sigma')+m_{2}\beta^2\gamma^{AB}\!\!\int_{-\sigma_{1}}^{+\sigma_{2}}\!\!\!\!\!\!\!\!\!\!\mathrm{d}\sigma\mathrm{d}\sigma' \,\partial_{\sigma}\delta(\sigma\!-\!\sigma')
\notag\\
=&m_{1}\beta^2 f^{AB}{}_{C}\!\!\int_{-\sigma_{1}}^{+\sigma_{2}}\!\!\!\!\!\!\!\!\!\!\mathrm{d}\sigma \, B_{\tau}{}^{C}(\sigma)+m_{2}\beta^2\gamma^{AB}\!\!\int_{-\sigma_{1}}^{+\sigma_{2}}\!\!\!\!\!\!\!\!\!\!\mathrm{d}\sigma'[\delta(+\sigma_{2}\!-\!\sigma')\!-\!\delta(-\sigma_{1}\!-\!\sigma')]
\notag\\
=&m_{1}\beta f^{AB}{}_{C}Q^{(0)C} \,\, ,
\end{align}
where in the last step we recombined the first integral into the definition of $Q^{(0)}$ and noticed that the two $\delta$-function contributions in the second integral cancel each other. The final expression agrees with the first relation in \eqref{Yangian-general-level0-and-level1-brackets} after choosing $\beta=m_{1}^{-1}$.

\textbf{Calculation of $\{Q^{(0)A}(\sigma),Q^{(1)B}(\sigma')\}$.} \quad
Using \eqref{charges-with-boundaries} and \eqref{Yangian-general-brackets} one arrives, after some manipulations described in \eqref{Q0-Q1-bracket-intermediate}, at the expressions
\begin{equation}\label{Q0Q1-first-step}
\begin{aligned}
\{Q^{(0)A},Q^{(1)B}\}=&Y_{1}^{AB}+Y_{2}^{AB}+X_{1}^{AB}+X_{2}^{AB} \, , \qquad \text{with}
\\
Y_{1}^{AB}:=& m_{3}\beta f^{AB}{}_{C}  \!\! \int_{-\sigma_{1}}^{+\sigma_{2}}\!\!\!\!\!\!\!\!\!\!\mathrm{d}\sigma
\!\! \int_{-\sigma_{5}}^{+\sigma_{6}}\!\!\!\!\!\!\!\!\!\!\mathrm{d}\sigma' \,sA_{\sigma}{}^{C}(\sigma)\delta(\sigma \!-\! \sigma') \, ,
\\
Y_{2}^{AB}:=&\tfrac{m_{1}\beta s}{2}f_{CD}{}^{B}f^{AC}{}_{E}\!\!\!\int_{-\sigma_{1}}^{+\sigma_{2}}\!\!\!\!\!\!\!\!\!\!\mathrm{d}\sigma \!\! \int_{-\sigma_{3}}^{+\sigma_{4}}\!\!\!\!\!\!\!\!\!\!\mathrm{d}\sigma'\mathrm{d}\sigma''  \epsilon(\sigma'\!-\!\sigma'') B_{\tau}{}^{D}(\sigma'')B_{\tau}{}^{E}(\sigma)\delta(\sigma-\sigma') \, ,
\\
X_{1}^{AB}:=&m_{4}\beta s \int_{-\sigma_{1}}^{+\sigma_{2}}\!\!\!\!\!\!\!\!\!\!\mathrm{d}\sigma
\int_{-\sigma_{5}}^{+\sigma_{6}}\!\!\!\!\!\!\!\!\!\!\mathrm{d}\sigma' \,C^{AB}(\sigma')\partial_{\sigma}\delta(\sigma-\sigma') \,\, ,
\\
X_{2}^{AB}:=&\tfrac{m_{2}\beta s}{2}f_{CD}{}^{B}\!\!\!\int_{-\sigma_{1}}^{+\sigma_{2}}\!\!\!\!\!\!\!\!\!\!\mathrm{d}\sigma \!\! \int_{-\sigma_{3}}^{+\sigma_{4}}\!\!\!\!\!\!\!\!\!\!\mathrm{d}\sigma'\mathrm{d}\sigma''  \epsilon(\sigma'\!-\!\sigma'') B_{\tau}{}^{D}(\sigma'')\gamma^{AC}\partial_{\sigma}\delta(\sigma-\sigma') \, .
\end{aligned}
\end{equation}
Aiming at recombining the four contributions so as to recover the second relation in \eqref{Yangian-general-level0-and-level1-brackets}, with $Q^{(1)C}$ given in \eqref{charges-with-boundaries}, one can notice that $X_{1}^{AB}$ and $X_{2}^{AB}$ are boundary terms which cannot contribute to $Q^{(1)C}$ and should hence cancel each other or vanish separately, while $Y_{1}^{AB},Y_{2}^{AB}$ should be rearranged so as to match the expression in \eqref{charges-with-boundaries}.

The term $Y_{1}^{AB}$, the only one involving $A_{\sigma}$, can be immediately brought to the desired form by requiring the boundaries to satisfy $\sigma_{6}<\sigma_{2}$ and $\sigma_{5}<\sigma_{1}$, hence reproducing the first contribution to $Q^{(1)C}$ in \eqref{charges-with-boundaries} up to the prefactor $m_{3}\beta$:
\begin{equation}
Y_{1}^{AB}=m_{3} \beta f^{AB}{}_{C} \int_{-\sigma_{5}}^{+\sigma_{6}}\mathrm{d}\sigma'
\, sA_{\sigma}{}^{C}(\sigma') \,\, .
\end{equation}
Similarly, after a little manipulation described in \eqref{manipulation-Y2AB} and requiring the boundaries to satisfy $\sigma_{4}<\sigma_{2}$ and $\sigma_{3}<\sigma_{1}$, one arrives at the second contribution to $Q^{(1)C}$ in \eqref{charges-with-boundaries},
\begin{equation}
Y_{2}^{AB}= m_{1}\beta f^{AB}{}_{C} \int_{-\sigma_{3}}^{+\sigma_{4}}\mathrm{d}\sigma \,\, \tfrac{1}{2\beta}[B_{\tau}(\sigma),\chi^{(0)}(\sigma)]^{C} \,\, .
\end{equation}
The two terms can hence be recombined, such that
\begin{equation}
\{Q^{(0)A},Q^{(1)B}\}=m_{1}\beta f^{AB}{}_{C} Q^{(1)C}+X_{1}^{AB}+X_{2}^{AB}
\qquad \text{iff} \qquad m_{3}\equiv m_{1} \,\, ,
\end{equation}
and the boundaries are sent to infinity in the order
\begin{equation}\label{boundary-order-intermediate}
\sigma_{4},\sigma_{6} <\sigma_{2}
\qquad \text{and} \qquad 
\sigma_{3},\sigma_{5}<\sigma_{1} \,\, .
\end{equation}
These conditions should be kept in mind and remain consistent with further ones arising at later stages. Reproducing the result \eqref{Yangian-general-level0-and-level1-brackets} relies at this point on showing that $X_{1}^{AB}+X_{2}^{AB}=0$, and in fact it turns out that these terms separately vanish in light of the conditions \eqref{boundary-order-intermediate}, as discussed in more detail in \eqref{vanishing-X1AB} and \eqref{vanishing-X2AB}. Altogether, the second relation in \eqref{Yangian-general-level0-and-level1-brackets} is recovered, with the prescription \eqref{boundary-order-intermediate}, after imposing that
\begin{equation}\label{all-conditions-from-Q0-Q1}
m_{3}\equiv m_{1}
\qquad \text{and} \qquad 
\beta\equiv m_{1}^{-1} \,\, .
\end{equation}

\textbf{Calculation of $\{Q^{(1)A}(\sigma),Q^{(1)B}(\sigma')\}$.} \quad 
This bracket is instrumental to the recovery of the Serre relations \eqref{Serre-relations}, which should follow after contraction with the structure constants and antisymmetrisation. Using \eqref{charges-with-boundaries} and $\chi^{(0)}$ from \eqref{Yangian-B-chi0-relation}, one obtains
\begin{equation}\label{Q1-Q1-bracket-contributions}
\{Q^{(1)A}(\sigma),Q^{(1)B}(\sigma')\}=U^{AB}+V^{AB}+Z^{AB}
\qquad \text{with} \qquad V^{AB}:=V_{1}^{AB}+V_{2}^{AB} \, ,
\end{equation}
and each contribution explicitly defined as
\begin{align}\label{Q1-Q1-bracket-explicit-contributions-list}
U^{AB}\!\!:=&\int_{-\sigma_{5}}^{+\sigma_{6}}\!\!\!\!\!\!\!\!\mathrm{d}\sigma\mathrm{d}\sigma' \, \{ A_{\sigma}{}^{A}(\sigma),A_{\sigma}{}^{B}(\sigma') \} \, ,
\\
V_{1}^{AB}\!\!:=&\tfrac{1}{2}f_{CD}{}^{B}\int_{-\sigma_{5}}^{+\sigma_{6}}\!\!\!\!\!\!\!\!\mathrm{d}\sigma \int_{-\sigma_{3}}^{+\sigma_{4}}\!\!\!\!\!\!\!\!\mathrm{d}\sigma'\mathrm{d}\sigma'' \, \theta(\sigma'-\sigma'')\{A_{\sigma}{}^{A}(\sigma),B_{\tau}{}^{C}(\sigma')B_{\tau}{}^{D}(\sigma'')\} \, ,
\notag\\
V_{2}^{AB}\!\!:=&\tfrac{1}{2}f_{CD}{}^{A}\int_{-\sigma_{5}}^{+\sigma_{6}}\!\!\!\!\!\!\!\!\mathrm{d}\sigma' \int_{-\sigma_{3}}^{+\sigma_{4}}\!\!\!\!\!\!\!\!\mathrm{d}\sigma\mathrm{d}\sigma'' \, \theta(\sigma-\sigma'')\{B_{\tau}{}^{C}(\sigma)B_{\tau}{}^{D}(\sigma''),A_{\sigma}{}^{B}(\sigma')\} \, ,
\notag\\
Z^{AB}\!\!:=&\tfrac{1}{4}f_{CD}{}^{A}f_{EF}{}^{B}\!\!\!\int_{-\sigma_{3}}^{+\sigma_{4}}\!\!\!\!\!\!\!\!\!\!\mathrm{d}\sigma\mathrm{d}\sigma' \mathrm{d}\sigma'' \mathrm{d}\sigma'''  \theta(\sigma\!-\!\sigma'')\theta(\sigma'\!-\!\sigma''')  \{B_{\tau}{}^{C}(\sigma)B_{\tau}{}^{D}(\sigma''),B_{\tau}{}^{E}(\sigma')B_{\tau}{}^{F}(\sigma''')\} .
\notag
\end{align}
The term $U^{AB}$, computed in \eqref{UAB}, does not contribute to the Serre relations \eqref{Serre-relations} due to its central structure \eqref{level1-bracket-tensorial-structures-def}. We thus proceed with $Z^{AB}$, which after some manipulations of the definition in \eqref{Q1-Q1-bracket-explicit-contributions-list} becomes
\begin{align}
Z^{AB}\!\!=&\tfrac{1}{4}f_{CD}{}^{A}f_{EF}{}^{B}\!\!\!\int_{-\sigma_{3}}^{+\sigma_{4}}\!\!\!\!\!\!\!\!\!\!\mathrm{d}\sigma\mathrm{d}\sigma' \mathrm{d}\sigma'' \mathrm{d}\sigma'''  \theta(\sigma\!-\!\sigma'')\theta(\sigma'\!-\!\sigma''')  \{B_{\tau}{}^{C}(\sigma)B_{\tau}{}^{D}(\sigma''),B_{\tau}{}^{E}(\sigma')B_{\tau}{}^{F}(\sigma''')\} 
\notag \\
=&\tfrac{1}{4}f_{CD}{}^{A}f_{EF}{}^{B}\!\!\!\int_{-\sigma_{3}}^{+\sigma_{4}}\!\!\!\!\!\!\!\!\!\!\mathrm{d}\sigma\mathrm{d}\sigma' \mathrm{d}\sigma'' \mathrm{d}\sigma''' \, \epsilon(\sigma\!-\!\sigma'')\epsilon(\sigma'\!-\!\sigma''')B_{\tau}{}^{C}(\sigma)B_{\tau}{}^{E}(\sigma')\{B_{\tau}{}^{D}(\sigma''),B_{\tau}{}^{F}(\sigma''')\}
\notag\\
=&Z_{1}^{AB}+Z_{2}^{AB} \,\, ,
\end{align}
where we first used the Leibniz rule \eqref{def-PB-Leibniz} to break down the Poisson bracket into four building blocks and recombined them, introducing $\epsilon(\sigma-\sigma')$ defined in \eqref{epsilon-def}, after using antisymmetry of the structure constants and relabelling of the integration variables. In the second step we exploited the bracket \eqref{Yangian-general-brackets} and for simplicity defined 
\begin{align}\label{Z1-Z2-def}
Z_{1}^{AB}\!\!\!:=&\tfrac{m_{1}}{4}\!f_{CD}{}^{A}\!f_{EF}{}^{B}\!f^{DF}{}_{G}\!\!\!\int_{-\sigma_{3}}^{+\sigma_{4}}\!\!\!\!\!\!\!\!\!\!\!\mathrm{d}\sigma\mathrm{d}\sigma' \mathrm{d}\sigma'' \mathrm{d}\sigma'''  \!\epsilon(\sigma\!-\!\sigma'')\epsilon(\sigma'\!\!-\!\sigma''')\!B_{\tau}{}^{C}(\sigma)\!B_{\tau}{}^{E}(\sigma')\!B_{\tau}{}^{G}(\sigma'')\delta(\sigma''\!\!-\!\sigma''') ,
\notag \\
\\
Z_{2}^{AB}\!\!\!:=&\tfrac{m_{2}}{4}f_{CD}{}^{A}\!f_{EF}{}^{B}\gamma^{DF}\!\!\!\int_{-\sigma_{3}}^{+\sigma_{4}}\!\!\!\!\!\!\!\!\!\!\!\mathrm{d}\sigma\mathrm{d}\sigma' \mathrm{d}\sigma''\mathrm{d}\sigma'''   \epsilon(\sigma\!-\!\sigma'')\epsilon(\sigma'\!\!-\!\sigma''')B_{\tau}{}^{C}(\sigma)B_{\tau}{}^{E}(\sigma')\partial_{\sigma''}\delta(\sigma''\!\!-\!\sigma''') \,\, .
\notag 
\end{align}
As shown below \eqref{Z1-first-rearrangement}, the first contribution $Z_{1}^{AB}$ can be rearranged as 
\begin{equation}\label{ZAB-final-result}
Z_{1}^{AB}=-\tfrac{m_{1}s^3}{12\beta^3}T_{CFG}{}^{AB} \, \chi^{(0)C}(+\sigma_{4}) \chi^{(0)F}(+\sigma_{4}) \chi^{(0)G}(+\sigma_{4}) +\Bigl(\text{terms}\propto\chi^{(0)}(-\sigma_{3})\Bigr) \,\, ,
\end{equation}
and once all boundaries are sent to infinity, the relation \eqref{relation-charges-potential} ensures that expression \eqref{ZAB-final-result} exhibits the same structure as the right-hand side of the Serre relations \eqref{Serre-relations}:
\begin{equation}
\lim_{\sigma_{3},\sigma_{4}\rightarrow\infty}Z_{1}^{AB}=-\tfrac{m_{1}}{12\beta^3}T_{CFG}{}^{AB} \, Q^{(0)C}Q^{(0)F}Q^{(0)G} \,\, .
\end{equation}
Combining this result with the fact that $U^{AB}$ is central, that the $T$-tensor enjoys property \eqref{Serre-structure-T-tensor-contraction}, and the previously found requirements \eqref{all-conditions-from-Q0-Q1}, it is then immediate to realise that
\begin{equation}\label{Q1-Q1-contraction-almost-Serre}
\begin{aligned}
f_{D}{}^{[AB}\{Q^{(1)C]},Q^{(1)D}\} =&+\tfrac{m_{1}^4}{12}f^{A}{}_{PI}f^{B}{}_{QJ}f^{C}{}_{RK}f^{IJK} \, Q^{(0)P}Q^{(0)Q}Q^{(0)R}
\\
&+f_{D}{}^{[AB}Z_{2}^{C]D}+f_{D}{}^{[AB}V^{C]D} \,\, .
\end{aligned}
\end{equation}
The Serre relations \eqref{Serre-relations} are hence precisely recovered up to extra contributions depending on the details of $Z_{2}^{AB}$ and $V^{AB}$, which should either cancel each other or separately vanish or be central, not to spoil the result. As shown below equation \eqref{Z2AB-intermediate}, it turns out that $Z_{2}^{AB}$ contains both an identically vanishing contribution and a central one and hence disappears altogether from \eqref{Q1-Q1-contraction-almost-Serre}. We are thus left to examine the structure of $V^{AB}$.

As shown around equation \eqref{V2AB-manipulations}, one finds that
\begin{align}\label{VAB-contributions}
V^{AB}=&\mathcal{V}_{1}^{AB}+\mathcal{V}_{2}^{AB}+\mathcal{V}_{3}^{AB} \,\,,
\qquad \text{with} \qquad 
\\
\mathcal{V}_{1}^{AB}:=&-m_{4}\int_{-\sigma_{5}}^{+\sigma_{6}}\!\!\!\!\!\!\!\!\!\!\mathrm{d}\sigma' \!\!\! \int_{-\sigma_{3}}^{+\sigma_{4}}\!\!\!\!\!\!\!\!\!\!\mathrm{d}\sigma\mathrm{d}\sigma''  \delta(\sigma\!-\!\sigma')\delta(\sigma\!-\!\sigma'')B_{\tau}{}^{D}\!(\sigma'')\Bigl(\!f_{D}{}^{AC}C_{C}{}^{B}\!(\sigma')\!-\!f_{D}{}^{BC}C_{C}{}^{A}\!(\sigma')\!\Bigr) \, ,
\notag \\
\mathcal{V}_{2}^{AB}:=&+\tfrac{m_{4}}{2} \!\! \int_{-\sigma_{5}}^{+\sigma_{6}}\!\!\!\!\!\!\!\!\!\!\mathrm{d}\sigma' \!\!\! \int_{-\sigma_{3}}^{+\sigma_{4}}\!\!\!\!\!\!\!\!\!\!\mathrm{d}\sigma\mathrm{d}\sigma'' \,\,  \partial_{\sigma}\Bigl(\!\epsilon(\sigma\!-\!\sigma'')\delta(\sigma\!-\!\sigma')\!\Bigr)B_{\tau}{}^{D}\!(\sigma'')\Bigl(\!f_{D}{}^{AC}C_{C}{}^{B}\!(\sigma')\!-\!f_{D}{}^{BC}C_{C}{}^{A}\!(\sigma')\!\Bigr) \, ,
\notag \\
\mathcal{V}_{3}^{AB}:=&\!-\!\tfrac{m_{3}}{2}(f_{D}{}^{AC}f^{B}{}_{CE}\!-\!f_{D}{}^{BC}f^{A}{}_{CE})\!\!\!\int_{-\sigma_{5}}^{+\sigma_{6}}\!\!\!\!\!\!\!\!\!\!\mathrm{d}\sigma' \!\!\! \int_{-\sigma_{3}}^{+\sigma_{4}}\!\!\!\!\!\!\!\!\!\!\mathrm{d}\sigma\mathrm{d}\sigma''  \,\, \epsilon(\sigma\!-\!\sigma'')B_{\tau}{}^{D}\!(\sigma'')A_{\sigma}{}^{E}\!(\sigma)\delta(\sigma\!-\!\sigma') \,\, .
\notag 
\end{align}
The last contribution $\mathcal{V}_{3}$ is easily seen to be central due to the Jacobi identity \eqref{def-Jacobi-components},
\begin{equation}\label{VAB-first-line-Jacobi}
(f_{D}{}^{AC}f^{B}{}_{CE}\!-\!f_{D}{}^{BC}f^{A}{}_{CE})=-f^{AB}{}_{C}f^{C}{}_{ED} \,\, ,
\end{equation}
and hence does not contribute to the Serre relations. $\mathcal{V}_{2}$ and $\mathcal{V}_{1}$ originate from the same non-ultralocal term, the second line in \eqref{VAB-result-intermediate}, and should thus be separately central and/or vanishing. Examining $\mathcal{V}_{2}$, as detailed around equation \eqref{calV2-appendix}, reveals that this term vanishes provided the boundaries are sent to infinity in the order $\sigma_{6}<\sigma_{4}$ and $\sigma_{5}<\sigma_{3}$. Comparing these conditions with the ones previously found in \eqref{boundary-order-intermediate} completely removes the remaining ambiguity, establishing the ordering prescription
\begin{equation}\label{conditions-on-boundaries-summary}
\sigma_{6}<\sigma_{4}<\sigma_{2} 
\qquad \qquad \text{and} \qquad \qquad 
\sigma_{5}<\sigma_{3}<\sigma_{1} \,\, .
\end{equation}
Altogether, the bracket \eqref{Q1-Q1-contraction-almost-Serre} becomes
\begin{equation}\label{Q1-Q1-contraction-really-almost-Serre}
f_{D}{}^{[AB}\{Q^{(1)C]},Q^{(1)D}\} =\tfrac{m_{1}^4}{12}f^{A}{}_{PI}f^{B}{}_{QJ}f^{C}{}_{RK}f^{IJK} \, Q^{(0)P}Q^{(0)Q}Q^{(0)R}+f_{D}{}^{[AB}\mathcal{V}_{1}^{C]D} \,\, ,
\end{equation}
so that, after simplifying $\mathcal{V}_{1}^{AB}$ in \eqref{VAB-contributions} in light of the conditions \eqref{conditions-on-boundaries-summary} as
\begin{equation}
\mathcal{V}_{1}^{AB}=-m_{4}\int_{-\sigma_{5}}^{+\sigma_{6}}\!\!\!\!\!\!\!\!\!\!\mathrm{d}\sigma' \,   B_{\tau}{}^{D}(\sigma')\Bigl(\!f_{D}{}^{AC}C_{C}{}^{B}(\sigma')\!-\!f_{D}{}^{BC}C_{C}{}^{A}(\sigma')\!\Bigr) \,\, ,
\end{equation}
and sending all the boundaries to infinity, recovers the Serre relation \eqref{Serre-relations} subject to the requirement \eqref{Yangian-theorem-condition}, which concludes the proof.
\end{proof}
We close this section by noting that while \eqref{Yangian-theorem-condition} could in principle be satisfied by an identically vanishing $\mathcal{V}_{1}$, it is in practice more reasonable to regard it as a condition on its integrand, which should exhibit the central structure \eqref{level1-bracket-tensorial-structures-def}, namely
\begin{equation}\label{B-C-relation}
B_{\tau}{}^{D}(\sigma')\Bigl(\!f_{D}{}^{AC}C_{C}{}^{B}(\sigma')\!-\!f_{D}{}^{BC}C_{C}{}^{A}(\sigma')\!\Bigr)\equiv X^{D}(\sigma')f_{D}{}^{AB} \, ,
\end{equation}
so as not to spoil the desired result. 

The condition \eqref{B-C-relation}, derived here as a requirement for a general set of Poisson brackets \eqref{Yangian-general-brackets} to reproduce the Serre relations \eqref{Serre-relations}, is automatically satisfied with $X^{D}(\sigma')\equiv 2 B_{\tau}{}^{D}(\sigma')$ when $C^{CB}(\sigma')\equiv \gamma^{CB}$. However, more general choices of $C^{CB}(\sigma')$ can interestingly exhibit the form of identity \eqref{J-K-algebraic-relation-components}, which in the context of symmetric-space sigma models is satisfied by construction with $C^{AB}(\sigma')$ being the $K$-map \eqref{K-map-definition} and $B=X$ the flat and conserved current $L$ of \eqref{SSSM-J-definition}. This ensures that all classes of sigma models studied in Section \ref{section:application-to-AFSM} \emph{do} exhibit a central term $\mathcal{V}_{1}^{AB}$, whose contribution drops out of \eqref{Q1-Q1-contraction-really-almost-Serre} hence reproducing the Serre relations \eqref{Serre-relations}. However, this analysis additionally shows that if any other model enjoyed the Poisson bracket structure \eqref{Yangian-general-brackets} for some other choice of $C^{CB}(\sigma')\neq K^{CB}(\sigma')\neq\gamma^{CB}$, then this model would necessarily have to either respect \eqref{B-C-relation} or make $\mathcal{V}_{1}^{AB}$ vanish in order for the correct Yangian structure to arise.

\section{Application to Auxiliary Field Sigma Models}\label{section:application-to-AFSM}

Here we apply the results of Sections \ref{section:BIZZ-charges} and \ref{section:Yangians} to auxiliary field sigma models (AFSMs), first introduced as deformations of principal chiral models \cite{Ferko:2024ali} then successively generalised to more complicated classes of seed theories, and in fact representing the main source of inspiration in deriving the results of the previous sections.\\
\indent
To begin, we will briefly review AFSMs and highlight relevant common features using the framework introduced in Section 4.1 of \cite{Bielli:2025uiv}, recognising how (at least a subset of) such theories correctly satisfy the assumption of Theorem \ref{gBIZZ-theorem}, hence allowing for the infinite tower of generalised BIZZ conserved currents \eqref{generalised-BIZZ-tower} and a Lax of the form \eqref{gBIZZ-Lax} written in terms of a flat current $A$ and a conserved current $B$.\\
\indent
The second step will consist of showing, on a case-by-case basis, how each of the analysed models exhibits Poisson brackets within the structure \eqref{Yangian-theorem} and satisfies the requirements \eqref{Yangian-theorem-condition}, hence ensuring the existence of an underlying classical Yangian symmetry \eqref{Yangian-general-level0-and-level1-brackets} and \eqref{Serre-relations}. To achieve this goal it will be very instructive, and in fact basically sufficient, to first go through the construction of the relevant Poisson brackets in the respective undeformed theories -- whose details are given in Appendix \ref{appendix:C}. Interestingly, thanks to a recurrent pattern by which the relevant currents can always be written in a form which only depends explicitly on the physical fields, this preliminary analysis will also turn out to hold true in the presence of the deformation, which will exhibit (at least) two different mechanisms by which a flat and conserved current $L$ in the seed theory is transformed into a flat current $A$ and a conserved current $B$ in the deformed one.\\
\indent
As a byproduct of the above analysis, it will then be natural to study the Maillet bracket structure of the Lax arising from the generalised BIZZ construction \eqref{gBIZZ-Lax} in combination with the Poisson brackets \eqref{Yangian-general-brackets}, finally showing explicitly the Hamiltonian integrability of the examined auxiliary field sigma models.

\subsection{AFSM and generalised BIZZ}
Inspired by the $4d$ Ivanov-Zupnik formalism \cite{Ivanov:2002ab,Ivanov:2003uj}, auxiliary field deformations of $2d$ sigma models have been constructed by coupling auxiliary fields $v$ to a given seed theory, in such a way that the undeformed model is recovered for vanishing $v$-self-interactions \cite{Ferko:2024ali}. The AFSM theories thus considered exhibit the common structure \cite{Bielli:2025uiv}
\begin{equation}\label{AF-lagrangian-common-structure}
S^{\text{E}}[\phi,v] := \int_{\Sigma}\mathrm{d}\sigma^{+}\mathrm{d}\sigma^{-} \, \mathcal{L}^{\text{E}}(\phi,v) 
\qquad \text{with} \qquad
\mathcal{L}^{\text{E}}(\phi,v) := \mathcal{L}(\phi,v) + E(v) \, ,
\end{equation}
where $\Sigma$ denotes a flat 2d Lorentzian worldsheet with lightcone coordinates $\sigma^{\pm}$ and
\begin{itemize}
\item $\phi$ schematically denotes the fundamental fields of the theory -- typically interpreted as coordinates on some background $\mathcal{M}$ and regarded as maps $\phi: \Sigma \rightarrow \mathcal{M}$.
\item $v$ denotes the auxiliary fields -- $\mathfrak{g}$-valued as in \cite{Ferko:2024ali,Bielli:2024khq,Bielli:2024ach,Bielli:2024fnp,Bielli:2024oif}, naturally defined on $\Sigma$, and transforming with the fundamental fields under the isometries of $\mathcal{M}$ \cite{Bielli:2024khq}.
\item $E(v)$ is an unspecified self-interaction function for the auxiliary fields, only required to depend on them through the combinations \cite{Bielli:2025uiv}
\begin{equation}\label{newE}
E(v):=E(\nu_{+2},...,\nu_{+N},\nu_{-2},...,\nu_{-N}) 
\quad \text{with} \quad
\nu_{\pm n}:=\mathrm{tr}(v_{\pm}^n) 
\quad \forall \, n\in \{2,...N\} \, .
\end{equation}
This ensures classical integrability of the deformed theory via the relation
\begin{equation}\label{deltapm_in_review_sec}
[\Delta_{\pm}, v_{\mp}] \deq 0
\qquad \text{with} \qquad 
\Delta_{\pm}:=\delta_{v_{\mp}}E(\nu_{2},...,\nu_{N}) \, ,
\end{equation}
with $\deq$ denoting an equality which holds when the auxiliary field equations of motion are satisfied, and $N$ being a large enough integer to guarantee having a complete set of algebraic structures describing completely symmetric invariant tensors -- for instance, $N$ could be the rank of the Lie algebra. 
\end{itemize}
\noindent
The Lagrangian $\mathcal{L}^{\text{E}}(\phi,v)$ is engineered so as to reduce to the original undeformed theory $\mathcal{L}(\phi)$ when integrating out the auxiliary fields for $E=0$, while non-trivial interaction functions generically induce complicated deformations of the seed theory by operators $\mathcal{O}$, whose structure depends on the explicit form of $E$ and on the variables $\nu_{\pm n}$ it depends on. Specific choices of $E$ have been shown to induce $T\bar{T}$, root-$T\bar{T}$ and deformations by operators built out of higher-spin conserved currents, and various classes of seed theories have been deformed within the above framework, exhibiting interesting common features. Among these is the following set of EOM for Lie-algebra-valued one-forms $\mathcal{A},\mathcal{B},\mathcal{C}$ \cite{Bielli:2025uiv}:
\begin{align}
D_{+}\mathcal{B}_{-}\!+\!D_{-}\mathcal{B}_{+} \deq & \, 0 
\quad\,\,\,\, \text{with} \quad\,\,\,\,
\mathcal{B}_{\pm} := -(\mathcal{A}_{\pm}+2v_{\pm}) \, , \label{AF-general-EOM-for-identity}
\\
\mathcal{A}_{\pm}+v_{\pm}+\Delta_{\pm} \deq & \, 0 
\quad\,\,\,\, \text{with} \quad\,\,\,\, 
\Delta_{\pm}:= \delta_{v_{\mp}}E(v) 
\quad \text{and} \quad [\Delta_{\pm},v_{\mp}]\deq \, 0 \, , \notag
\\
D_{+}\mathcal{A}_{-} \!-\!D_{-}\mathcal{A}_{+} \!+\! a \,[\mathcal{A}_{+},\mathcal{A}_{-}] \deq & \, 0 
\quad\,\,\,\, \text{with} \quad\,\,\,\, 
 D_{\pm}:=\partial_{\pm}\!+\![\mathcal{C}_{\pm},-]
\quad \text{and} \quad a \! \in \! \mathbb{R} \, \text{const.} \, , \notag
\end{align}
whose detailed structure depends on the specific model under consideration.
\begin{itemize}
\item AF-PCMs are characterised by $a=1$ and
\begin{equation}
\mathcal{A}_{\pm}:= j_{\pm} 
\qquad , \qquad 
\mathcal{C}_{\pm}:=0
\qquad \text{with} \qquad 
j_{\pm}:=g^{-1}\partial_{\pm}g \, .
\end{equation}
\item AF-non-Abelian T-dual models are characterised by $a=1$ and
\begin{equation}
\mathcal{A}_{\pm} := \tilde{j}_{\pm}
\qquad , \qquad \mathcal{C}_{\pm}:=0
\qquad \text{with} \qquad 
\tilde{j}_{\pm}:= {\pm}\frac{1}{1\pm\mathrm{ad}_{X}}(\partial_{\pm}X\mp 2v_{\pm}) \, .
\end{equation}
\item AF-(bi)-Yang-Baxter models are characterised by $a=(1-c^2\eta^2+\tilde{c}^2\zeta^2)$ and
\begin{equation}\label{AF-summary-AFYB}
\begin{aligned}
\mathcal{A}_{\pm}:=  J_{\pm}^{\zeta}
\qquad &, \qquad 
\mathcal{C}_{\pm}:= \pm\zeta\tilde{\mathcal{R}}(\mathcal{B}_{\pm})
\qquad \text{with} \qquad 
\\
J_{\pm}^{\zeta}:= -(\mathfrak{J}_{\pm}^{\zeta}\!+\!2v_{\pm})
\qquad \text{a}&\text{nd} \qquad 
\mathfrak{J}_{\pm}^{\zeta}\!:=  \!-\frac{1}{1\!\mp\! \eta \mathcal{R}_{g} \!\mp\! \zeta \tilde{\mathcal{R}}}(j_{\pm}\!+\!2v_{\pm}) \, .
\end{aligned}
\end{equation}
\item AF PCM+WZ models\footnote{\label{footnote-WZ-rescaling}This model was imprecisely grouped with the AF-PCM in \cite{Bielli:2025uiv}: while the last two equations in \eqref{AF-general-EOM-for-identity} can in fact be characterised by $a\equiv 1$, $\mathcal{A}_{\pm}\equiv j_{\pm}$ and $\mathcal{C}_{\pm}\equiv 0$, as for the PCM, the relation $\mathcal{B}_{\pm}=-(\mathcal{A}_{\pm}+2v_{\pm})$ is then not respected by these identifications, since the correct conserved current for this theory is $\mathcal{B}_{\pm}$ in \eqref{AF-PCM+WZ-currents}, which differs from $\mathfrak{J}_{\pm}$. The identifications \eqref{AF-PCM+WZ-currents} arise from the Lax in \cite{Bielli:2024ach} and the flatness of $\mathcal{A}_{\pm}$ in \eqref{AF-PCM+WZ-currents} reduces indeed to the flatness of $j_{\pm}$ when $\mathcal{B}_{\pm}$ is conserved. Even though in the rest of this work we will use \eqref{AF-PCM+WZ-currents}, it should be noted that such currents still do break some of the relations in \eqref{AF-general-EOM-for-identity}. To begin, $\mathcal{B}_{\pm}=-(\mathcal{A}_{\pm}+2v_{\pm})$ is not respected unless the currents in \eqref{AF-PCM+WZ-currents} are rescaled as $\tilde{\mathcal{A}}_{\pm}=(1\mp\tfrac{\kay}{\hay})^{-1}\mathcal{A}_{\pm}$ and $\tilde{\mathcal{B}}_{\pm}=(1\mp\tfrac{\kay}{\hay})^{-1}\mathcal{B}_{\pm}$, which do satisfy $\tilde{\mathcal{B}}_{\pm}=-(\tilde{\mathcal{A}}_{\pm}+2v_{\pm})$. Moreover, even if one were to perform such a rescaling, the $v$-EOM, namely the second equation in \eqref{AF-general-EOM-for-identity}, would still read 
\begin{equation}
j_{\pm}+v_{\pm}+\Delta_{\pm}\deq 0 \,\, , 
\qquad \text{such that} \qquad j_{\pm}\neq \mathcal{A}_{\pm}\neq\tilde{\mathcal{A}}_{\pm}  \,\, .
\end{equation}
This difference comes from the fact that the WZ term is not modified by the auxiliary fields and hence the $v$-EOM must by construction agree with the one of the simpler AF-PCM. In the spirit of looking for a more general set of EOM, able to possibly encompass even more examples of AF-deformed sigma models, it is interesting to notice how the current $j_{\pm}$ appearing in the $v$-EOM could in fact be re-expressed as
\begin{equation}
j_{\pm}=(1\pm\tfrac{\kay}{\hay})^{-1}\Bigl(\tilde{\mathcal{A}}_{\pm}\pm \tfrac{\kay}{\hay} \tilde{\mathcal{B}}_{\pm} \Bigr)
\qquad \xrightarrow[\kay \rightarrow 0]
\qquad \qquad 
j_{\pm}=\tilde{\mathcal{A}}_{\pm}=\mathcal{A}_{\pm}\,\, ,
\end{equation}
a feature that one might interpret as more generally arising in the presence of a WZ term.
} are characterised by $a=1$, $\mathcal{C}_{\pm}:=0$ and
\begin{equation}\label{AF-PCM+WZ-currents}
\mathcal{A}_{\pm}:= (j_{\pm} \mp\tfrac{\kay}{\hay}\mathfrak{J}_{\pm})
\quad , \quad 
\mathcal{B}_{\pm}:= (\mathfrak{J}_{\pm} \mp\tfrac{\kay}{\hay}j_{\pm})
\qquad \text{with} \qquad 
\mathfrak{J}_{\pm}:=-(j_{\pm}+2v_{\pm}) \, .
\end{equation}
\item AF-Symmetric Space models, are characterised by $a=0$ and 
\begin{equation}
\mathcal{A}_{\pm}:=j_{\pm}^{(2)}
\quad , \quad 
\mathcal{C}_{\pm}:=j_{\pm}^{(0)} 
\qquad \text{with} \qquad
v_{\pm} \rightarrow v_{\pm}^{(2)} \, .
\end{equation}
\end{itemize}
A crucial feature of the above theories, which naturally relates to (and in fact inspires) the generalised BIZZ construction of Section \ref{section:BIZZ-charges}, is that the auxiliary field EOM 
\begin{equation}
\mathcal{A}_{\pm}+v_{\pm}+\Delta_{\pm}\deq 0 \,\, , 
\quad \text{together with}
\quad 
\mathcal{B}_{\pm}:=-(\mathcal{A}_{\pm}+2v_{\pm})
\quad \text{and} \quad
[\Delta_{\pm},v_{\mp}]\deq 0 \,\, ,
\end{equation}
implies the fundamental commutator identities\footnote{Note that for the AF-PCM-WZ model \eqref{AF-PCM+WZ-currents} these relations are ensured by the dependence on $\mathfrak{J}_{\pm}$.}
\begin{equation}
[\mathcal{B}_{+},\mathcal{B}_{-}]\deq [\mathcal{A}_{+},\mathcal{A}_{-}]
\qquad \text{and} \qquad 
[\mathcal{A}_{+},\mathcal{B}_{-}]\deq [\mathcal{B}_{+},\mathcal{A}_{-}] \,\, .
\end{equation}
These are used to show the Lax integrability of the above models and represent the lightcone components version of the generalised BIZZ requirements \eqref{A,B-commutators}. At the same time, in various classes of explicit models, the first and third EOM in \eqref{AF-general-EOM-for-identity} correspond to the remaining two conditions \eqref{A,B-definition}, and have indeed been encoded in Lax connections 
\begin{equation}\label{Lax-lightcone}
\mathfrak{L}_{\pm}=\frac{A_{\pm}\pm zB_{\pm}}{1-z^2} \,\, ,
\end{equation}
which represent the lightcone components version of the generalised BIZZ Lax \eqref{gBIZZ-Lax}.
More specifically, the expression \eqref{Lax-lightcone} has been obtained for AF-PCMs \cite{Ferko:2024ali}, AF-non-Abelian-T-dual models \cite{Bielli:2024khq,Bielli:2024ach}, AF-YB models\footnote{In this work we will not consider the case of bi-YB models, but for the sake of comparison, we briefly recall the expressions for the Lax connections written in terms of $\mathcal{A,B,C}$ in \cite{Bielli:2025uiv}, 
\begin{equation}\label{common-Lax}
\mathfrak{L}_{\pm}=l_{1}^{\pm}(z)\frac{\mathcal{A}_{\pm}\pm z\mathcal{B}_{\pm}}{1\pm z} + l_{2}^{\pm}(z)\Theta(v_{\pm})+l_{3}^{\pm}(z) \mathcal{C}_{\pm} \,\, ,
\end{equation}
with suitable choices of operators $\Theta:\mathfrak{g}\rightarrow\mathfrak{g}$ and functions $l^\pm_1(z)$, $l^\pm_2(z)$, $l^\pm_3(z)$.} \cite{Bielli:2024fnp} and AF-PCM-WZ models \cite{Bielli:2024ach} with 
\begin{equation}\label{A,B-identifications-PCM-TD-YB-PCMWZ}
A_{\pm}:=a\mathcal{A}_{\pm}
\qquad \text{and} \qquad 
B_{\pm}:=a\mathcal{B}_{\pm}
\,\, ,
\end{equation}
where $a\in \mathbb{R}$ is the constant characterising the AFSM EOM \eqref{AF-general-EOM-for-identity}. A similar relation also exists for AF-SSSMs, even though the expression proposed in \cite{Bielli:2024oif} does not technically take the form \eqref{Lax-lightcone}: it is indeed known, already for undeformed SSSM, that a gauge invariant Lax connection $\mathfrak{L}'(z)$ can be obtained from the non-invariant Lax $\mathfrak{L}(z)$ by taking the combination $\mathfrak{L}'(z)=\mathrm{g}(\mathfrak{L}(z)-j)\mathrm{g}^{-1}$. Doing this for the AF-SSSM Lax in \cite{Bielli:2024oif} leads to: \begin{equation}
\begin{aligned}\label{Lax-gauge-transformation-AF-SSSM}
\mathfrak{L}_{\pm}'(z)=\mathrm{Ad}_{\mathrm{g}}\Bigl( j_{\pm}^{(0)}+\frac{(z^2+1)j_{\pm}^{(2)}\mp2z\, \mathfrak{J}_{\pm}^{(2)}}{z^2-1} -j_{\pm}\Bigr) =\frac{-2\mathrm{Ad}_{\mathrm{g}}(j_{\pm}^{(2)})\pm 2z \mathrm{Ad}_{\mathrm{g}}(\mathfrak{J}_{\pm}^{(2)})}{1-z^2} \,\, ,
\end{aligned}
\end{equation} 
which matches \eqref{Lax-lightcone} after redefining $z\rightarrow-z$ and introducing
\begin{equation}\label{A,B-AF-SSSM}
A_{\pm}:=-2\mathrm{Ad}_{\mathrm{g}}(j_{\pm}^{(2)})=-2\mathrm{Ad}_{\mathrm{g}}(\mathcal{A}_{\pm})
\qquad , \qquad 
B_{\pm}:=-2\mathrm{Ad}_{\mathrm{g}}(\mathfrak{J}_{\pm}^{(2)})=-2\mathrm{Ad}_{\mathrm{g}}(\mathcal{B}_{\pm}) \,\, .
\end{equation}
The above identifications highlight the relation between the AFSM currents $\mathcal{A},\mathcal{B}$ in \eqref{AF-general-EOM-for-identity} and the generalised BIZZ currents $A,B$ in \eqref{A,B-definition}, subject to the constraints \eqref{A,B-commutators}, establishing the existence of the tower of conserved currents \eqref{generalised-BIZZ-tower} in these deformed theories.

\subsection{Poisson brackets}
In this subsection our task is to explicitly compute the Poisson brackets of the currents $A$ and $B$ for the classes of AFSMs discussed above, showing how they exhibit the structure \eqref{Yangian-general-brackets} and satisfy the requirements \eqref{Yangian-theorem-condition}, ensuring the existence of an underlying classical Yangian algebra obeying the relations \eqref{Yangian-general-level0-and-level1-brackets} and \eqref{Serre-relations}. This result in turn implies that such an underlying symmetry is also present for all the respective undeformed theories, which are recovered by integrating out the auxiliary fields when $E=0$ and are characterised by the single current $L=A=B$ and a Lax connection of the form 
\begin{equation}\label{Lax-undeformed-models}
\mathfrak{L}_{\pm}=\frac{L_{\pm}}{1\mp z} \,\, ,
\end{equation}
which is indeed the lightcone component version of \eqref{BIZZ-Lax} with the flat and conserved current $L$ in \eqref{j-properties-definition} depending on the specific theory.
In fact, as mentioned above, to determine the Poisson brackets of the currents $A,B$ it will be extremely useful, and basically sufficient, to first go through the analogous construction for the flat and conserved current $L$. As will become clear by the end of this section, the deformation introduced via the auxiliary fields acts indeed by splitting the current $L$ into $A$ and $B$ with a specific mechanism which depends on the theory under consideration but follows a systematic pattern, already observed in \cite{Ferko:2024ali} for PCMs and in \cite{Cesaro:2024ipq} for SSSM. This splitting mechanism does not alter the relation between the physical momenta of the theory and the currents, which can be entirely expressed in terms of physical fields and only implicitly depend on the auxiliaries. In turn, together with the feature discussed in the next paragraph, this will imply that the deformed Poisson brackets can be obtained from the undeformed ones by replacing each component of $L_{\tau}$ with $B_{\tau}$ and each component of $L_{\sigma}$ with $A_{\sigma}$. 

\paragraph{On the canonical structure of AFSM.} Before diving into a long survey of models and their respective auxiliary field deformations, we would like to briefly pause for an aside about another common feature underlying such theories, which are examples of constrained systems \cite{Dirac:1964:LQM}.   See also \cite{henneaux1992quantization} for a detailed discussion and \cite{mann2018lagrangianandhamiltoniandynamics} for a concise and effective introduction. We refer to section 4 of \cite{Ferko:2024ali}, section 3.3 of \cite{Bielli:2024ach}, section 3.2 of \cite{Cesaro:2024ipq} and section 3.4 of \cite{Cesaro:2025msv} for the case of AFSM, which we also summarise below. \\
\indent
Coupling a Lagrangian to a set of auxiliary fields generally has the side effect, when switching to the Hamiltonian description, of introducing constraints on the phase space, which gets effectively reduced by the fact that the new fields are by definition non-dynamical and hence the deformed Lagrangian is independent of their time derivatives. One should hence take into account the time evolution of such constraints, check whether this generates new conditions and in turn whether or not these are preserved in time, leading to further constraints in a recursive process which ultimately provides a list of all the conditions that should be taken into account. When some of these constraints, named second-class, do not commute, one should replace the Poisson brackets by the notion of Dirac brackets, which carry this information via extra contributions proportional to a matrix whose entries encode the result of such non-commutativity.\\
\indent
In AFSMs, the Lagrangian $\mathcal{L}_{\text{AFSM}}=\mathcal{L}_{\text{AFSM}}(Z^{\mu},\partial_{\sigma}Z^{\mu},\partial_{\tau}Z^{\mu},v_{\tau}{}^{A},v_{\sigma}{}^{A})$ is a function of the background coordinates $Z^{\mu}$, their derivatives with respect to the worldsheet coordinates and the $\tau$ and $\sigma$ components of the Lie-algebra valued auxiliary fields. The phase space is hence spanned by $(Z^{\mu},\pi_{\mu},v_{\alpha}{}^{A},\mathfrak{p}_{A}{}^{\alpha})$, with the momenta defined as
\begin{equation}
\pi_{\mu}:=\frac{\mathcal{L}_{\text{AFSM}}}{\partial(\partial_{\tau}Z^{\mu})}
\qquad \text{and} \qquad 
\mathfrak{p}_{A}{}^{\alpha}:=\frac{\mathcal{L}_{\text{AFSM}}}{\partial (\partial_{\tau}v_{\alpha}{}^{A})} \,\, .
\end{equation}
While the first relation can be inverted to write the so-called canonical Hamiltonian $H_{\text{AFSM}}$ in terms of $\{Z^{\mu},\pi_{\mu},v_{\alpha}{}^{A}\}$, the second is identically vanishing  $\mathfrak{p}_{A}{}^{\alpha}=0$ since $\partial_{\tau}v_{\alpha}{}^{A}$ is never introduced in the Lagrangian.
It hence leads to a set of constraints on the phase space
\begin{equation}
\phi_{i}(Z^{\mu},\pi_{\mu},v_{\alpha}{}^{A},\mathfrak{p}_{A}{}^{\alpha}):=\mathfrak{p}_{A}{}^{\alpha}  \approx 0
\qquad \text{with} \qquad i=1,...,2\text{dim}(\mathfrak{g}) \,\, ,
\end{equation}
which are named primary, as they come from the definition of momenta and do not involve EOM. The canonical Hamiltonian is uniquely defined on the constrained phase space where $\phi_{i}\approx 0$, but there are many continuations to the full phase space, so one needs to consider a modified Hamiltonian where the constraints are imposed via Lagrange multipliers $u^{i}=u^{i}(Z^{\mu},\pi_{\mu},v_{\alpha}{}^{A},\mathfrak{p}_{A}{}^{\alpha})$ that need to be determined
\begin{equation}
\mathcal{H}_{\text{AFSM}}=H_{\text{AFSM}}+u^{i}\phi_{i} \,\, .
\end{equation}
On the constrained phase space, the total Hamiltonian evolves quantities in time via
\begin{equation}
\dot{F}=\{F,\mathcal{H}_{\text{AFSM}}\}
=\{F,H_{\text{AFSM}}\}+u^{i}\{F,\phi_{i}\} 
\end{equation}
and one should make sure that the desired constraints are preserved under time evolution. Here starts a recursive procedure where new constraints are generated whenever the time evolution of a previous constraint cannot be set to zero by appropriately choosing the Lagrange multipliers. In the case of AFSM one finds that
\begin{equation}\label{phi-dot-eq}
\dot\phi_{i}=\{H_{\text{AFSM}},\phi_{i}\}+u^{j}\{\phi_{i},\phi_{j}\}=\{H_{\text{AFSM}},\mathfrak{p}_{A}{}^{\alpha}\}=\frac{\partial H_{\text{AFSM}}}{\partial v_{\alpha}{}^{A}}
=:\psi_{i} \,\, ,
\end{equation}
where in the second equality we used that $\{\phi_{i},\phi_{j}\}=0 \,\, \forall i,j$, since these are by construction canonical momenta obeying $\{\mathfrak{p}_{A}{}^{\alpha},\mathfrak{p}_{B}{}^{\beta}\}=0$, and in the third one we used that canonical momenta generate translations in the conjugate variables. Time evolution of the primary constraints corresponds to the auxiliary field EOM and since we cannot solve \eqref{phi-dot-eq} for the multipliers $u^j$, we end up with new constraints $\psi_{i}\approx 0$, called secondary. We then need to proceed further and check whether these are preserved under time evolution, namely
\begin{equation}\label{secondary-constraints}
\begin{aligned}
\dot\psi_{i} =\{H_{\text{AFSM}},\psi_{i}\}+u^{j}\{\psi_{i},\phi_{j}\} \approx 0 \,\, .
\end{aligned}
\end{equation}
This step is now different from the previous one, since the matrix
\begin{equation}
M_{ij}:=\{\psi_{i},\phi_{j}\}=\frac{\partial^2H_{\text{AFSM}}}{\partial v_{\alpha}{}^{A}\partial v_{\beta}{}^{B}}
\end{equation}
is generally invertible when the auxiliary field EOM, i.e. $\psi_{i}$, can be solved for the auxiliary fields\footnote{\label{footnote:M-invertibiliy}This represents an assumption that we will always make, since one would eventually like to integrate out the auxiliary fields and obtain the explicit deformed theory in terms of the physical fields.}. Hence at this stage the Dirac algorithm terminates and one can preserve the secondary constraints along time evolution by solving \eqref{secondary-constraints} for the multipliers
\begin{equation}
u^{i}=-M^{ij}\{H_{\text{AFSM}},\psi_{j}\}
\qquad \text{with} \qquad
M_{ij}M^{jk}=\delta_{i}{}^{k} \,\, .
\end{equation}
This fixes the modified Hamiltonian and all the resulting constraints can be enumerated
\begin{equation}
\Omega_{I}=(\phi_{i},\psi_{j})
\qquad \text{with} \qquad 
I=1,...,4\text{dim}(\mathfrak{g}) \,\, .
\end{equation}
The procedure then requires computing the Poisson brackets of all the constraints and building up a matrix out of the non-commuting ones, which in this case has the structure
\begin{equation}
D_{IJ}=
\begin{pmatrix}
\{\phi_{i},\phi_{k}\} & \{\phi_{i},\psi_{l}\} \\
\{ \psi_{j},\phi_{k} \} &  \{ \psi_{j},\psi_{l} \} 
\end{pmatrix}
=
\begin{pmatrix}
0 & -M_{li}^{T} \\
M_{jk} &  N_{jl} 
\end{pmatrix}
\,\, .
\end{equation}
The top left diagonal block vanishes again due to $\{\mathfrak{p}_{A}{}^{\alpha},\mathfrak{p}_{B}{}^{\beta}\}=0$ and in principle at this point one needs to invert $D$. There is a closed-form expression for this when $M$ is invertible -- which is the case for situations of physical interest where one is able to integrate out the auxiliary fields, as mentioned in footnote \ref{footnote:M-invertibiliy} -- but it turns out that for AFSM it is in fact not necessary, since all one needs is the structure of the inverse
\begin{equation}\label{D-inverse}
D^{IJ}= 
\begin{pmatrix}
* & * \\
* & 0
\end{pmatrix}
\qquad \text{where} \qquad 
D_{IJ}D^{JK}=\delta_{I}{}^{K} \,\, ,
\end{equation}
which immediately allows to compute the Dirac brackets of any two quantities
\begin{align}
\{f,g\}_{D}=&\{f,g\}-\{f,\Omega_{I}\} D^{IJ}\{\Omega_{J},g\}
\\
=&\{f,g\}-\{f,\phi_{i}\}D^{\phi_{i}\phi_{j}}\{\phi_{j},g\}-\{f,\phi_{i}\}D^{\phi_{i}\psi_{j}}\{\psi_{j},g\}-\{f,\psi_{i}\}D^{\psi_{i}\phi_{j}}\{\phi_{j},g\} \,\, .
\notag 
\end{align}
Importantly, since the bottom-right diagonal block in \eqref{D-inverse} vanishes, one can observe that the Dirac bracket only differs from the Poisson bracket by terms which involve Poisson brackets of the functions $f,g$ with at least one of the primary constraints $\phi_{i}:=\mathfrak{p}_{A}{}^{\alpha}$. Hence, since by construction in Hamiltonian mechanics one has that
\begin{equation}
\{Z^{\mu},\pi_{\nu}\}=\delta_{\nu}{}^{\mu} \qquad \text{and} \qquad \{v_{\alpha}{}^{A},\mathfrak{p}_{B}{}^{\beta}\}=\delta_{\alpha}{}^{\beta}\delta_{B}{}^{A} \,\, ,
\end{equation}
while all the other brackets are vanishing, and in particular 
\begin{equation}
\{\mathfrak{p}_{A}{}^{\alpha},\pi_{\mu}\}=0=\{\mathfrak{p}_{A}{}^{\alpha},Z^{\mu}\} \,\, ,
\end{equation}
one finds that for $f$ and $g$ only functions of the physical variables $Z^{\mu}$ and $\pi_{\mu}$, all the extra terms vanish and the Dirac and Poisson brackets coincide:
\begin{equation}
\{f(Z^{\mu},\pi_{\mu}),g(Z^{\nu},\pi_{\nu})\}_{D}=\{f(Z^{\mu},\pi_{\mu}),g(Z^{\nu},\pi_{\nu})\} \,\, ,
\end{equation}
ultimately implying that Poisson brackets are sufficient to study quantities of the above form. As it will become clear throughout the rest of this section, for all the AFSM under consideration it will be possible to completely hide the dependence on the auxiliary fields in the definition of the currents and to re-expressed them in terms of physical variables.

\subsubsection{Principal Chiral models}\label{subsubsec:PCM}

\paragraph{Generalities.} These theories represent one of the simplest examples of $2d$ sigma models, the only relevant field being the metric of the background geometry on a Lie group $G$.\footnote{We refer to Appendix \ref{appendix:A} for related conventions and notation.} This special structure has allowed to study them in great detail \cite{Luscher:1977rq,Luscher:1977uq,Wiegmann:1984pw,Polyakov:1984et,Faddeev:1985qu,Ogievetsky:1987vv} and to exploit them as powerful seed theories in the construction of various types of integrable deformations, e.g. by a Wess-Zumino term \cite{Wess:1971yu,Novikov:1982ei,Abdalla:1982yd,Witten:1983ar}, Yang-Baxter deformations \cite{Klimcik:2002zj,Klimcik:2008eq}, $\lambda$-deformations \cite{Sfetsos:2013wia}, and we refer to \cite{Zarembo:2017muf,Lacroix:2018njs,Seibold:2020ouf,Hoare:2021dix,Borsato:2023dis} for detailed discussions and references. Information about the background metric is encoded in the Lie-algebra-valued 1-form $j$ defined in \eqref{MC-form-def}, the Maurer-Cartan form, which is by construction flat and plays the role of a vielbeine on $G$ with the tangent space metric being the Cartan-Killing form \eqref{Cartan-Killing-def} on the Lie algebra $\mathfrak{g}$. This property plays a very important role, since it implies invertibility of $j$ and, as shown below, a special form of the canonical momenta. In fact, as will be clear by the end of Section \ref{section:application-to-AFSM}, this is also crucial for various other classes of sigma models built out of an underlying Lie group structure, as well as for their auxiliary field deformations.

\paragraph{The $L$-current.} PCMs can be entirely defined through the Maurer-Cartan form, which is flat by construction and conserved as a result of the EOM, hence
\begin{equation}
L:=j \,\, .
\end{equation}

\paragraph{Lagrangian and momentum.} The Lagrangian for these theories takes the form\footnote{Conventions for switching between lightcone and $(\tau,\sigma)$ coordinates are given in \eqref{def-lightcone-to-tau-sigma}.}
\begin{equation}
\mathcal{L}_{\textrm{PCM}}=-\tfrac{1}{2}\mathrm{tr}[j_{+}j_{-}]=-\tfrac{1}{2}\mathrm{tr}[j_{\tau}^2-j_{\sigma}^2]=-\tfrac{1}{2}\gamma_{AB}(\partial_{\tau}Z^{\mu}\partial_{\tau}Z^{\nu}-\partial_{\sigma}Z^{\mu}\partial_{\sigma}Z^{\nu})j_{\mu}{}^{A}j_{\nu}{}^{B} \, ,
\end{equation}
and consequently the associated canonical momentum \eqref{def-canonical-momentum} reads, in terms of $L$,
\begin{equation}\label{PCM-momentum-rel}
\pi_{\mu}=-\gamma_{AB}L_{\mu}{}^{A}L_{\tau}{}^{B}
\qquad \longleftrightarrow \qquad
L_{\tau}{}^{A}=-\gamma^{AB}L_{B}{}^{\mu}\pi_{\mu} \,\, .
\end{equation}
As highlighted above, the second relation is a result of the fact that $L$ plays the role of a vielbeine on the background, whose metric reads $G_{\mu\nu}=L_{\mu}{}^{A}L_{\nu}{}^{B}\gamma_{AB}$, and hence has an inverse $L_{A}{}^{\mu}$ such that $L_{\mu}{}^{A}L_{A}{}^{\nu}=\delta_{\mu}{}^{\nu}$ and $L_{\mu}{}^{A}L_{B}{}^{\mu}=\delta_{B}{}^{A}$.

\paragraph{Poisson brackets.} The components of $L$ satisfy the brackets
\begin{equation}\label{PCM-brackets-summary}
\begin{aligned}
\{ L_{\tau}{}^{A}(\sigma),L_{\tau}{}^{B}(\sigma') \} &= f^{AB}{}_{C} \, L_{\tau}{}^{C}(\sigma)\delta(\sigma-\sigma') \, ,
\\
\{ L_{\tau}{}^{A}(\sigma),L_{\sigma}{}^{B}(\sigma') \} &= f^{AB}{}_{C} \, L_{\sigma}{}^{C}(\sigma)\delta(\sigma-\sigma')-\gamma^{AB} \, \partial_{\sigma}\delta(\sigma-\sigma') \, ,
\\
\{ L_{\sigma}{}^{A}(\sigma),L_{\sigma}{}^{B}(\sigma') \} &= 0 \, .
\end{aligned}
\end{equation}
The details of this computation will not be reported in this work. It closely follows the procedure in Appendix C of \cite{Bielli:2024ach}, to which we refer for brevity.

\subsubsection{AF-Principal Chiral models}\label{subsubsec:AFPCM}

\paragraph{Generalities.} These models were introduced in \cite{Ferko:2024ali} and represent the first example of $2d$ integrable seed theory deformed à la Ivanov-Zupnik \cite{Ivanov:2002ab,Ivanov:2003uj}, also inspired by \cite{Borsato:2022tmu,Ferko:2023wyi}. They were furthermore derived from a 4D Chern-Simons\footnote{We refer to the seminal papers \cite{Costello:2019tri,Vicedo:2019dej,Delduc:2019whp}. See also the interesting recent works \cite{Boujakhrout:2023qlz,Lacroix:2024wrd,Cole:2024sje,Lacroix:2025ias,Sakamoto:2025hwi,Fukushima:2025tlj,Cole:2025zmq,Bittleston:2026tdr,Stedman:2026awg,Ashwinkumar:2026dwd} in relation to this topic and \cite{Lacroix:2021iit,Cole:2024hyt,Liniado:2025sin} for pedagogical treatments and a more complete list of references.} perspective in \cite{Fukushima:2024nxm}.

\paragraph{The $A$ and $B$ currents.} The Lax \eqref{gBIZZ-Lax} for these models is characterised by 
\begin{equation}\label{A-B-currents-AF-PCM}
A:=j
\qquad \text{and} \qquad 
B:=-(j+2v) \,\, . 
\end{equation}

\paragraph{Lagrangian and momentum.} Using \eqref{def-lightcone-to-tau-sigma}, the Lagrangian can be written as
\begin{equation}
\begin{aligned}
\label{eq:lag_af_pcm}
\mathcal{L}_{\textrm{AF-PCM}}=&+\tfrac{1}{2}\mathrm{tr}[j_{+}j_{-}]+\mathrm{tr}[j_{+}v_{-}+j_{-}v_{+}]+\mathrm{tr}[v_{+}v_{-}]+E(v)
\\
=&+\tfrac{1}{2}\mathrm{tr}[j_{\tau}^2-j_{\sigma}^2]+2\mathrm{tr}[j_{\tau}v_{\tau}-j_{\sigma}v_{\sigma}]+\mathrm{tr}[v_{\tau}^2-v_{\sigma}^2]+E(v)
\\
=&+\tfrac{1}{2}\gamma_{AB}(\partial_{\tau}Z^{\mu}\partial_{\tau}Z^{\nu}-\partial_{\sigma}Z^{\mu}\partial_{\sigma}Z^{\nu})j_{\mu}{}^{A}j_{\nu}{}^{B}+2\gamma_{AB}\partial_{\tau}Z^{\mu}j_{\mu}{}^{A}v_{\tau}{}^{B}
\\
&-2\gamma_{AB}\partial_{\sigma}Z^{\mu}j_{\mu}{}^{A}v_{\sigma}{}^{B}+\gamma_{AB}(v_{\tau}{}^{A}v_{\tau}{}^{B}-v_{\sigma}{}^{A}v_{\sigma}{}^{B})+E(v) \, ,
\end{aligned}
\end{equation}
so that the momentum \eqref{def-canonical-momentum} takes the form
\begin{equation}
\pi_{\mu}=\gamma_{AB}\partial_{\tau}Z^{\nu}j_{\mu}{}^{A}j_{\nu}{}^{B}+2\gamma_{AB}j_{\mu}{}^{A}v_{\tau}{}^{B} \,\, .
\end{equation}
Comparing with \eqref{A-B-currents-AF-PCM} and using invertibility of $j_{\mu}{}^{A}$ one arrives at the expressions
\begin{equation}\label{AF-PCM-momentum-inversion}
\pi_{\mu}=-\gamma_{AB}A_{\mu}{}^{A}B_{\tau}{}^{B}
\qquad 
\longleftrightarrow \qquad 
B_{\tau}{}^{A}=-\gamma^{AB}A_{B}{}^{\mu}\pi_{\mu} \,\, ,
\end{equation}
which represent the deformed analogue of \eqref{PCM-momentum-rel}. In particular this implies that both $B_{\tau}^{A}$ and $A_{\sigma}^{A}$ only depend on the physical fields, the latter being in fact undeformed.

\paragraph{Poisson brackets.} The components of $A$ and $B$ satisfy the following relations:
\begin{equation}\label{AFPCM-brackets-summary}
\begin{aligned}
\{ B_{\tau}{}^{A}(\sigma),B_{\tau}{}^{B}(\sigma') \} &= f^{AB}{}_{C} \, B_{\tau}{}^{C}(\sigma)\delta(\sigma-\sigma') \, ,
\\
\{ B_{\tau}{}^{A}(\sigma),A_{\sigma}{}^{B}(\sigma') \} &= f^{AB}{}_{C} \, A_{\sigma}{}^{C}(\sigma)\delta(\sigma-\sigma')-\gamma^{AB} \, \partial_{\sigma}\delta(\sigma-\sigma') \, ,
\\
\{ A_{\sigma}{}^{A}(\sigma),A_{\sigma}{}^{B}(\sigma') \} &= 0 \, .
\end{aligned}
\end{equation}
We refer again to appendix C of \cite{Bielli:2024ach} for the details of the derivation and highlight how \eqref{AFPCM-brackets-summary} naively differ from the undeformed brackets \eqref{PCM-brackets-summary} by the simple replacements $L_{\tau}\rightarrow B_{\tau}$ and $L_{\sigma}\rightarrow A_{\sigma}$, a pattern which will repeatedly appear throughout this section. It should also be stressed at this point how the above brackets fall within the class \eqref{Yangian-general-brackets} and respect the conditions \eqref{Yangian-theorem-condition} in light of the fact that $C^{AB}(\sigma'):=\gamma^{AB}$. This ensures that all members of this family of deformed theories exhibit an underlying classical Yangian symmetry and for the undeformed setting \eqref{PCM-brackets-summary} this had already been shown in \cite{MacKay:1992he}.

\subsubsection{Symmetric-space sigma models}\label{subsubsec:SSSM}

\paragraph{Generalities.} These theories represent another important class of sigma models built out of Lie groups and play an important role in many contexts of physical interest. Compared to PCMs, the underlying Lie algebra of these models is characterised by an extra ingredient, namely the existence of an order-2 automorphism which induces the splitting in \eqref{SSSM-Lie-algebra-commutators} and a coset structure, leading to a much richer picture \cite{Eichenherr:1979ci,Eichenherr:1979hz,Forger:1991cm,Forger:1991ty,Laartz:1992rw,Forger:1992kn}. Particularly relevant to the analysis carried out below is the reference \cite{Forger:1991cm}, from which various identities have been summarised in the second part of Appendix \ref{appendix:A}.

\paragraph{The $L$-current.} SSSM can be defined through the projections of the Maurer-Cartan form \eqref{j-projections}. The whole 1-form is still flat by construction, but decomposes as in \eqref{def-sssm-MC-eq-decomposition}, such that the relevant flat current which is also conserved on-shell reads  
\begin{equation}
L:=-2\mathrm{Ad}_{\mathrm{g}}(j^{(2)}) \,\, .
\end{equation}

\paragraph{Lagrangian and momentum.}
The Lagrangian for this class of models is defined as 
\begin{equation}
\begin{aligned}
\mathcal{L}_{\textrm{SSSM}}&=-\tfrac{1}{2} \mathrm{tr}[j^{(2)}_{+}j^{(2)}_{-}]= -\tfrac{1}{2}\mathrm{tr}\bigl[\bigl(j^{(2)}_{\tau}\bigr)^{2}-\bigl(j^{(2)}_{\sigma}\bigr)^{2}\bigr] 
\\
&=-\tfrac{1}{2}\gamma_{ab}(\partial_{\tau}Z^{\mu}\partial_{\tau}Z^{\nu}-\partial_{\sigma}Z^{\mu}\partial_{\sigma}Z^{\nu})j^{(2)a}_{\mu}j_{\nu}^{(2)b} \, ,
\end{aligned}
\end{equation}
and consequently the associated canonical momentum \eqref{def-canonical-momentum} reads
\begin{equation}\label{SSSM-canonical-momentum}
\pi_{\mu}= -\gamma_{ab} \, j_{\mu}^{(2)a}j_{\tau}^{(2)b}
\qquad \longleftrightarrow \qquad 
j_{\tau}^{(2)a}=-\gamma^{ab}j_{b}^{(2)\mu}\pi_{\mu} \,\, .
\end{equation}
The second relation is a result of the fact that $j^{(2)}$ is a vielbeine on the background, i.e. defines a metric $G_{\mu\nu}=j_{\mu}^{(2)a}j_{\nu}^{(2)b}\gamma_{ab}$ and has inverse $j_{a}^{(2)\mu}=\gamma_{ab}j_{\nu}^{(2)b}G^{\nu\mu}$, which is defined by requiring invertibility of the metric $G_{\mu\nu}G^{\nu\rho}=\delta_{\mu}{}^{\rho}$. For these reasons $j^{(2)}$ and its inverse satisfy $j_{\mu}^{(2)a}j_{a}^{(2)\nu}=\delta_{\mu}{}^{\nu}$ and $j_{\mu}^{(2)a}j_{b}^{(2)\mu}=\delta_{b}{}^{a}$. Multiplying the second relation in \eqref{SSSM-canonical-momentum} by $-2$ and contracting with $W_{a}{}^{A}$ one further arrives at 
\begin{equation}\label{SSSM-Jtau-momentum}
L_{\tau}{}^{A}=-\pi_{\mu}L_{B}{}^{\mu}\gamma^{BA} 
\qquad \text{with} \qquad 
L_{B}{}^{\mu}=\gamma_{BA}G^{\mu\nu}L_{\nu}{}^{A} \,\, ,
\end{equation}
after having exploited the definition of the inverse vielbeine and equation \eqref{SSSM-J-definition}. Finally, this also allows us to obtain the following important relations:
\begin{equation}\label{SSSM-J-vielbeine-and-momentum(J)}
L_{\mu}{}^{A}L_{A}{}^{\nu}=4\delta_{\mu}{}^{\nu} \qquad , \qquad L_{\mu}{}^{B}L_{A}{}^{\mu}=4K_{A}{}^{B} \qquad \text{and} \qquad \pi_{\mu}=-\tfrac{1}{4}\gamma_{AB}L_{\mu}{}^{A}L_{\tau}{}^{B}\,\, .
\end{equation}
Differentiating the second relation in \eqref{SSSM-J-vielbeine-and-momentum(J)} one further obtains the useful identity
\begin{equation}\label{SSSM-useful-relation}
L_{\mu}{}^{B}\partial_{\rho}L_{A}{}^{\mu}=4\partial_{\rho}K_{A}{}^{B}-L_{A}{}^{\mu}\partial_{\rho}L_{\mu}{}^{B} \,\, .
\end{equation}

\paragraph{Poisson brackets.}
For these models one obtains the following expressions:
\begin{equation}\label{SSSM-brackets-summary}
\begin{aligned}
\{ L_{\tau}{}^{A}(\sigma),L_{\tau}{}^{B}(\sigma') \} &= 2f^{AB}{}_{C} \, L_{\tau}{}^{C}(\sigma)\delta(\sigma-\sigma') \, ,
\\
\{ L_{\tau}{}^{A}(\sigma),L_{\sigma}{}^{B}(\sigma') \} &= 2f^{AB}{}_{C} \, L_{\sigma}{}^{C}(\sigma)\delta(\sigma-\sigma')-4K^{AB}(\sigma') \, \partial_{\sigma}\delta(\sigma-\sigma') \, ,
\\
\{ L_{\sigma}{}^{A}(\sigma),L_{\sigma}{}^{B}(\sigma') \} &= 0 \, ,
\\
\{ L_{\tau}{}^{A}(\sigma),K^{BC}(\sigma') \} &= 2\Bigl(f^{AB}{}_{D}K^{DC}(\sigma)+f^{AC}{}_{D}K^{DB}(\sigma)\Bigr)\delta(\sigma-\sigma') \,  .
\end{aligned}
\end{equation}
These were first found in \cite{Forger:1991cm} and we re-derive them explicitly in Appendix \ref{appendix:C-SSSM}.

\subsubsection{AF-Symmetric-space sigma models}\label{subsubsec:AFSSSM}
\paragraph{Generalities.}
AF deformations of SSSMs were constructed in \cite{Bielli:2024oif,Cesaro:2024ipq} by introducing auxiliary fields respecting the decomposition \eqref{SSSM-Lie-algebra-commutators} of the underlying Lie algebra,
\begin{equation}
v=\mathrm{P}^{(0)}(v)+\mathrm{P}^{(2)}(v)=v^{(0)}+v^{(2)} \,\, .
\end{equation} 
We refer to Appendix \ref{appendix:A} for relevant notation and conventions on this matter.

\paragraph{The $A$ and $B$ currents.} The Lax \eqref{gBIZZ-Lax} for these models is characterised by
\begin{equation}\label{AF-SSSM-J-Jtilde-def}
A=-2\mathrm{Ad}_{\mathrm{g}}(j^{(2)})
\quad 
\text{and}
\quad 
B=-2\mathrm{Ad}_{\mathrm{g}}(\mathfrak{J}^{(2)})
\qquad \text{with} \qquad \mathfrak{J}^{(2)}=-(j^{(2)}+2v^{(2)}) \,\, .
\end{equation}
The first current is the same flat current (off-shell) discussed in the undeformed setting, namely $A=L$, while the second one represents a deformed version of the first and is conserved on-shell. Notice its similarity with the deformed current \eqref{A-B-currents-AF-PCM} appearing in AF-PCMs, resulting from the appearance of the (appropriately projected) $\mathfrak{J}$ combination.

\paragraph{Lagrangian and momentum.} Using \eqref{def-lightcone-to-tau-sigma} the AF-SSSM Lagrangian reads
\begin{align}
\mathcal{L}_{\mathrm{AF-SSSM}}=&+\tfrac{1}{2}\mathrm{tr}[j^{(2)}_{+}j^{(2)}_{-}]+\mathrm{tr}[j^{(2)}_{+}v_{-}^{(2)}+j^{(2)}_{-}v_{+}^{(2)}]+\mathrm{tr}[v_{+}^{(2)}v_{-}^{(2)}]+E(v)
\notag \\
=&+\tfrac{1}{2}\mathrm{tr}\bigl[\bigl(j^{(2)}_{\tau}\bigr)^2-\bigl(j^{(2)}_{\sigma}\bigr)^2\bigr]+2\mathrm{tr}[j^{(2)}_{\tau}v_{\tau}^{(2)}-j^{(2)}_{\sigma}v_{\sigma}^{(2)}]+\mathrm{tr}\bigl[\bigl(v_{\tau}^{(2)}\bigr)^2-\bigl(v_{\sigma}^{(2)}\bigr)^2\bigr]+E(v)
\notag \\
=&+\tfrac{1}{2}\gamma_{ab}(\partial_{\tau}Z^{\mu}\partial_{\tau}Z^{\nu}-\partial_{\sigma}Z^{\mu}\partial_{\sigma}Z^{\nu})j_{\mu}^{(2)a}j_{\nu}^{(2)b}+2\gamma_{ab}\partial_{\tau}Z^{\mu}j_{\mu}^{(2)a}v_{\tau}^{(2)b}
\notag \\
&-2\gamma_{ab}\partial_{\sigma}Z^{\mu}j_{\mu}^{(2)a}v_{\sigma}^{(2)b}+\gamma_{ab}(v_{\tau}^{(2)a}v_{\tau}^{(2)b}-v_{\sigma}^{(2)a}v_{\sigma}^{(2)b})+E(v) \,\, ,
\end{align}
such that the momentum \eqref{def-canonical-momentum} takes the form
\begin{equation}
\pi_{\mu}=\gamma_{ab}\partial_{\tau}Z^{\nu}j_{\mu}^{(2)a}j_{\nu}^{(2)b}+2\gamma_{ab}j_{\mu}^{(2)a}v_{\tau}^{(2)b} \,\, .
\end{equation}
A quick look at \eqref{AF-SSSM-J-Jtilde-def}, and invertibility of the vielbeine $j_{\mu}^{(2)a}$, then leads to
\begin{equation}\label{AF-SSSM-canonical-momentum}
\pi_{\mu}=-\gamma_{ab}j_{\mu}^{(2)a}\mathfrak{J}_{\tau}^{(2)b}
\qquad \longleftrightarrow \qquad
\mathfrak{J}_{\tau}^{(2)a}=-\gamma^{ab}j_{b}^{(2)\mu}\pi_{\mu} \,\, ,
\end{equation}
which immediately reveals a close resemblance with \eqref{SSSM-canonical-momentum}. It is then natural to look for a generalisation of \eqref{SSSM-Jtau-momentum} in terms of the new current $B$, and indeed one finds
\begin{equation}
\begin{aligned}
B_{\tau}=&-2\mathrm{Ad}_{\mathfrak{g}}(\mathfrak{J}_{\tau}^{(2)})
\\
=&+2\mathrm{Ad}_{\mathfrak{g}}(\gamma^{ab}j_{b}^{(2)\mu}\pi_{\mu}T_{a})
\\
=&+2\mathrm{Ad}_{\mathfrak{g}}(j_{\nu}^{(2)a}G^{\nu\mu}\pi_{\mu}T_{a})
\\
=&-\pi_{\mu}G^{\mu\nu}A_{\nu} \,\, ,
\end{aligned}
\end{equation}
where in the first step we exploited \eqref{AF-SSSM-canonical-momentum}, in the second we substituted the definition of the inverse vielbeine $j_{a}^{(2)\mu}=\gamma_{ab}j_{\nu}^{(2)b}G^{\nu\mu}$, and in the last step we recombined the remaining quantities in terms of $A$ in \eqref{AF-SSSM-J-Jtilde-def}. This ultimately shows that
\begin{equation}\label{AF-SSSM-Jtau-momentum}
B_{\tau}{}^{A}=-\pi_{\mu}A_{B}{}^{\mu}\gamma^{BA} 
\qquad \text{with} \qquad 
A_{B}{}^{\mu}=\gamma_{BA}G^{\mu\nu}A_{\nu}{}^{A} \,\, ,
\end{equation}
which is exactly the same as \eqref{SSSM-Jtau-momentum}, after the replacement $L_{\tau}\rightarrow B_{\tau}$, and implies that the relevant components of the currents can be entirely expressed in terms of physical fields. The algebraic off-shell relations satisfied by $A=L$ remain obviously still true and hence one also arrives at the analogue of equation \eqref{SSSM-J-vielbeine-and-momentum(J)}, which in terms of $A$ and $B$ becomes
\begin{equation}\label{AF-SSSM-J-vielbeine-and-momentum(J)}
A_{\mu}{}^{A}A_{A}{}^{\nu}=4\delta_{\mu}{}^{\nu} \qquad , \qquad A_{\mu}{}^{B}A_{A}{}^{\mu}=4K_{A}{}^{B} \qquad \text{and} \qquad \pi_{\mu}=-\tfrac{1}{4}\gamma_{AB}A_{\mu}{}^{A}B_{\tau}{}^{B}\,\, .
\end{equation}

\paragraph{Poisson brackets.}
For these models one finds the structure
\begin{align}\label{AF-SSSM-brackets-summary}
\{ B_{\tau}{}^{A}(\sigma),B_{\tau}{}^{B}(\sigma') \} &= 2f^{AB}{}_{C} \, B_{\tau}{}^{C}(\sigma)\delta(\sigma-\sigma') \, ,
\notag \\
\{ B_{\tau}{}^{A}(\sigma),A_{\sigma}{}^{B}(\sigma') \} &= 2f^{AB}{}_{C} \, A_{\sigma}{}^{C}(\sigma)\delta(\sigma-\sigma')-4K^{AB}(\sigma') \, \partial_{\sigma}\delta(\sigma-\sigma') \, ,
\\
\{ A_{\sigma}{}^{A}(\sigma),A_{\sigma}{}^{B}(\sigma') \} &= 0 \, ,
\notag \\
\{ B_{\tau}{}^{A}(\sigma),K^{BC}(\sigma') \} &= 2\Bigl(f^{AB}{}_{D}K^{DC}(\sigma)+f^{AC}{}_{D}K^{DB}(\sigma)\Bigr)\delta(\sigma-\sigma') \, ,
\notag 
\end{align}
which again naively differs from the undeformed \eqref{SSSM-brackets-summary} by the replacements $L_{\tau}\rightarrow B_{\tau}$ and $L_{\sigma}\rightarrow A_{\sigma}$, keeping in mind that for this model $A\equiv L$ and hence the vanishing bracket is in practice unaffected. Almost the same is also true for the remaining brackets, whose derivation is briefly commented in Appendix \ref{appendix:C-AFSSSM}. We remark that \eqref{AF-SSSM-brackets-summary} falls into the class \eqref{Yangian-general-brackets} and satisfy the conditions \eqref{Yangian-theorem-condition} in light of the identification $C^{AB}(\sigma')=K^{AB}(\sigma')$ and the algebraic identity \eqref{J-K-algebraic-relation-components}. This ensures the existence of an underlying Yangian symmetry for each member of this infinite family of deformed theories. For the undeformed setting \eqref{SSSM-brackets-summary}, this had previously been shown in \cite{Klose:2016qfv}.

\subsubsection{Yang-Baxter sigma models}\label{subsubsec:YB}
\paragraph{Generalities.}
These models were introduced in \cite{Klimcik:2002zj}\footnote{See also \cite{Klimcik:2008eq,Klimcik:2014bta} for related studies and \cite{Yoshida:2021qfl} for a pedagogical treatment.} by deforming the PCM through the inclusion of an operator on the underlying Lie algebra $\mathcal{R}:\mathfrak{g} \longrightarrow\mathfrak{g}$, acting as $\mathcal{R}(X)=X^{A}\mathcal{R}_{A}{}^{B}T_{B} \,\, \forall\, X\in \mathfrak{g}$. This is known as an R-matrix and is required to be antisymmetric with respect to the Cartan-Killing form, which implies that
\begin{equation}\label{R-matrix-antisymmetry}
\mathcal{R}_{AB}=\mathcal{R}_{A}{}^{C}\gamma_{CB}=\mathrm{tr}[\mathcal{R}(T_{A})T_{B}]=-\mathrm{tr}[T_{A}\mathcal{R}(T_{B})]=-\mathcal{R}_{B}{}^{C}\gamma_{AC}=-\mathcal{R}_{BA} \,\, ,
\end{equation}
and is assumed to satisfy the modified Classical Yang-Baxter equation (mCYBE):
\begin{equation}
[\mathcal{R}(X),\mathcal{R}(Y)]-\mathcal{R}[X,\mathcal{R}(Y)]-\mathcal{R}[\mathcal{R}(X),Y]+c^2[X,Y]=0 \qquad \forall \, X,Y \in \mathfrak{g} \,\, .
\end{equation}
In fact, the field-independent R-matrix appears in the model in a dressed fashion
\begin{equation}
\mathcal{R}_{\mathrm{g}}=\mathrm{Ad}_{\mathrm{g}}^{-1} \circ\mathcal{R}\circ \mathrm{Ad}_{\mathrm{g}} \,\, ,
\end{equation}
which preserves the above two properties, and this dressed R-matrix appears in the model multiplied by a deformation parameter $\eta$ as
\begin{equation}
M\equiv \eta \mathcal{R}_{\mathrm{g}} \,\, .
\end{equation}
Exploiting antisymmetry of $M$, namely $M_{AB}=-M_{BA}$, the mCYBE 
\begin{equation}
[M(X),M(Y)]-M[X,M(Y)]-M[M(X),Y]+\eta^2c^2[X,Y]=0 
\end{equation}
can be written in components as
\begin{equation}\label{mCYBE-components}
M_{C}{}^{B}M_{D}{}^{A}f^{CDE}+M_{C}{}^{A}M_{D}{}^{E}f^{CDB}-M_{C}{}^{B}M_{D}{}^{E}f^{CDA}=\eta^2 c^2 f^{ABE} \,\, .
\end{equation}
By direct calculation it is then also easy to verify that the dressed R-matrix satisfies
\begin{equation}
\partial_{\mu} \bigl(\mathcal{R}_{\mathrm{g}}(X) \bigr) = \mathcal{R}_{\mathrm{g}}(\partial_{\mu}X) + [\mathcal{R}_{\mathrm{g}}(X),j_{\mu}]+\mathcal{R}_{\mathrm{g}}[j_{\mu},X] \,\, ,
\end{equation}
which after multiplying by $\eta$ and choosing $X=T_{A}$ becomes
\begin{equation}\label{M-derivative-components}
\partial_{\mu}M_{A}{}^{B}=j_{\mu}{}^{D}(M_{A}{}^{C}f_{CD}{}^{B}+f_{DA}{}^{C}M_{C}{}^{B}) \,\, .
\end{equation}
In terms of $M$ one can then introduce the operators
\begin{equation}
\mathcal{O}_{\pm}=\frac{1}{1\pm M}=\sum_{k=0}^{\infty}(\mp)^k
M^k \,\, , 
\quad \text{acting as} \quad \mathcal{O}_{\pm}(X)=X^{A}\mathcal{O}_{\pm A}{}^{B}T_{B} \quad \forall \, X\in \mathfrak{g} \,\, .
\end{equation}
The definition ensures that $\mathcal{O}_{\pm}$ and $M$ are all commuting and satisfy the relations
\begin{equation}\label{sum-difference-Opm}
\frac{1}{1-M}-\frac{1}{1+M}=\frac{2M}{1-M^2} \, ,
\qquad 
\qquad 
\frac{1}{1-M}+\frac{1}{1+M}=\frac{2}{1-M^2} \,\, .
\end{equation}
Due to the antisymmetry of $M$ they are furthermore related by transposition, namely
\begin{equation}
\mathrm{tr}[X\mathcal{O}_{\pm}(Y)]=\mathrm{tr}[\mathcal{O}_{\mp}(X)Y] 
\qquad \Longrightarrow \qquad
\mathcal{O}_{\pm AB}=\mathcal{O}_{\mp BA} \,\, .
\end{equation}

\paragraph{The $L$-current.} While the Maurer-Cartan form $j$ still plays an important role in these models, the on-shell conserved current is naturally a deformation of the latter via $M$. This is flat as well, when conserved and after using the flatness of $j$, and reads
\begin{equation}
L_{\pm}=(1-c^2\eta^2)J_{\pm}
\qquad \text{with} \qquad 
J_{\pm}=\frac{1}{1\pm M} \, j_{\pm} =\mathcal{O}_{\pm}(j_{\pm}) \,\, .
\end{equation}
For simplicity it will be convenient to work with $J$, keeping in mind that whenever one needs to translate to $L$ it is sufficient to multiply by $(1-c^2\eta^2)$. Converting from lightcone to $(\tau,\sigma)$ coordinates using \eqref{def-lightcone-to-tau-sigma} and \eqref{sum-difference-Opm} one finds that $J$ has components
\begin{equation}\label{YB-Jtausigma}
J_{\tau}=\frac{1}{1-M^2}\Bigl( j_{\tau}-M(j_{\sigma})\Bigr)
\qquad \text{and} \qquad 
J_{\sigma}=\frac{1}{1-M^2}\Bigl( j_{\sigma}-M(j_{\tau})\Bigr) \,\, ,
\end{equation}
which also satisfy the important relations
\begin{equation}\label{Jtausigma-relations}
J_{\sigma}=j_{\sigma}-M(J_{\tau}) 
\qquad \text{and} \qquad 
J_{\tau}=j_{\tau}-M(J_{\sigma}) \,\, .
\end{equation}

\paragraph{Lagrangian and momentum.} Using again \eqref{def-lightcone-to-tau-sigma} and \eqref{sum-difference-Opm} the Lagrangian reads
\begin{equation}\label{YB-Lagrangian}
\begin{aligned}
\mathcal{L}_{\text{YB}}=&-\tfrac{1}{2}\mathrm{tr}[j_{+}\frac{1}{1-M} j_{-}]
\\
=&-\tfrac{1}{2}\mathrm{tr}[j_{\tau}\frac{1}{1-M} j_{\tau}-j_{\tau}\frac{2M}{1-M^2} j_{\sigma}-j_{\sigma}\frac{1}{1-M} j_{\sigma}]
\\
=&-\tfrac{1}{2}\gamma_{AB}\partial_{\tau}Z^{\mu}\partial_{\tau}Z^{\nu}j_{\mu}{}^{A}j_{\nu}{}^{C}\mathcal{O}_{-C}{}^{B}-\tfrac{1}{2}\gamma_{AB}\partial_{\sigma}Z^{\mu}\partial_{\sigma}Z^{\nu}j_{\mu}{}^{A}j_{\nu}{}^{C}\mathcal{O}_{-C}{}^{B}
\\
&+\gamma_{AB}\partial_{\tau}Z^{\mu}\partial_{\sigma}Z^{\nu}j_{\mu}{}^{A}j_{\nu}{}^{C}\mathcal{O}_{-C}{}^{D}\mathcal{O}_{+D}{}^{E}M_{E}{}^{B} \,\, ,
\end{aligned} 
\end{equation}
such that the canonical momentum \eqref{def-canonical-momentum} takes the form
\begin{equation}
\pi_{\mu}=-\gamma_{AB}j_{\mu}{}^{A}j_{\tau}{}^{C}\mathcal{O}_{-C}{}^{D}\mathcal{O}_{+D}{}^{B}+\gamma_{AB}j_{\mu}{}^{A}j_{\sigma}{}^{C}\mathcal{O}_{-C}{}^{D}\mathcal{O}_{+D}{}^{E}M_{E}{^{B}} \,\, ,
\end{equation}
and after a quick look at the definitions \eqref{YB-Jtausigma} can be written as
\begin{equation}\label{YB-pi-J-relation}
\pi_{\mu}=-\gamma_{AB}j_{\mu}{}^{A}J_{\tau}{}^{B}
\qquad \longleftrightarrow \qquad
J_{\tau}{}^{A}=-\pi_{\mu}j_{B}{}^{\mu}\gamma^{BA} \,\, ,
\end{equation}
exploiting the fact that $j$ is a vielbeine. This relation has exactly the structure observed in PCMs and SSSMs, but in this setting also has the important implication that due to \eqref{Jtausigma-relations} both $J_{\tau}$ and $J_{\sigma}$ now exhibit dependence on the canonical momentum:
\begin{equation}\label{Jtau-and-Jsigma-functions-of-momentum}
J_{\tau}{}^{A}=-\pi_{\mu}j_{B}{}^{\mu}\gamma^{BA} 
\qquad \text{and} \qquad 
J_{\sigma}{}^{B}=j_{\sigma}{}^{B}-J_{\tau}{}^{C}M_{C}{^{B}} \,\, .
\end{equation}

\paragraph{Poisson brackets.}
For these models one finds the brackets
\begin{equation}\label{YB-brackets-summary}
\begin{aligned}
\{ J_{\tau}{}^{A}(\sigma),J_{\tau}{}^{B}(\sigma') \} &= f^{AB}{}_{C} \, J_{\tau}{}^{C}(\sigma)\delta(\sigma-\sigma') \, ,
\\
\{ J_{\tau}{}^{A}(\sigma),J_{\sigma}{}^{B}(\sigma') \} &= f^{AB}{}_{C} \, J_{\sigma}{}^{C}(\sigma)\delta(\sigma-\sigma')-\gamma^{AB} \, \partial_{\sigma}\delta(\sigma-\sigma') \, ,
\\
\{ J_{\sigma}{}^{A}(\sigma),J_{\sigma}{}^{B}(\sigma') \} &= \eta^2 c^2 f^{AB}{}_{C}J_{\tau}{}^{C}\delta(\sigma-\sigma') \, ,
\\
\{ J_{\tau}{}^{A}(\sigma),j_{\sigma}{}^{B}(\sigma') \} &= f^{AB}{}_{C}j_{\sigma}{}^{C}\delta(\sigma-\sigma')-\gamma^{AB}\partial_{\sigma}\delta(\sigma-\sigma') \, ,
\\
\{ J_{\tau}{}^{A}(\sigma),M_{C}{}^{B}(\sigma') \} &= \bigl(f^{AB}{}_{D}M_{C}{}^{D}-f_{C}{}^{AD}M_{D}{}^{B}\bigr)\delta(\sigma-\sigma') \, ,
\end{aligned}
\end{equation}
which are explicitly constructed in Appendix \ref{appendix:C-YB}, by direct computation and using the relations highlighted above, but were first derived in \cite{Delduc:2013fga} via different methods (see also \cite{Delduc:2013qra,Delduc:2014kha} for related works and applications). These brackets can be rewritten in terms of $L$ by multiplying the two sides with the appropriate powers of $(1-c^2\eta^2)$. Notice that these correctly reduce to the PCM \eqref{PCM-brackets-summary} in the limit $\eta \rightarrow 0$, in which $M\rightarrow 0$ and $J\rightarrow j $.

\subsubsection{AF-Yang-Baxter sigma models}\label{subsubsec:AFYB}
\paragraph{Generalities.} This class of models was constructed in \cite{Bielli:2024fnp} by coupling the Yang-Baxter models of the previous subsection to Lie-algebra-valued auxiliary fields $v$. Compared to the AF deformation of PCMs and SSSM discussed above in subsections \ref{subsubsec:AFPCM} and \ref{subsubsec:AFSSSM}, where we observed how the auxiliary fields induce a modification of the conserved current $B$ while leaving the flat ones $A$ unaltered, this model turns out to be characterised by a deformation of both currents. In particular, unlike previous examples, this implies that the whole set of Poisson brackets \eqref{YB-brackets-summary} has to be modified and one might in principle be worried about things getting more complicated. In fact it will turn out that the special relation \eqref{Jtausigma-relations} also has a deformed avatar, which together with the deformed version of \eqref{Jtau-and-Jsigma-functions-of-momentum} plays the crucial role of keeping things under control, preserving the previously observed pattern $L_{\tau}\rightarrow B_{\tau}$, $L_{\sigma}\rightarrow A_{\sigma}$ and in turn the whole structure.

\paragraph{The $A$ and $B$ currents.} The Lax \eqref{Lax-lightcone} of AF-YB models is characterised by
\begin{equation}
A=(1-c^2\eta^2)\tilde{J}
\qquad \text{and} \qquad 
B=(1-c^2\eta^2)\mathfrak{J} \,\, ,
\end{equation} 
defined in terms of 
\begin{equation}
\tilde{J}_{\pm}=\frac{1}{1\mp M}\Bigl( j_{\pm}\pm M(2v_{\pm}) \Bigr)
\qquad \text{and} \qquad 
\mathfrak{J}_{\pm}=-\frac{1}{1\mp M}\Bigl( j_{\pm}+2v_{\pm}  \Bigr) \,\, ,
\end{equation}
with $M=\eta \mathcal{R}_{\mathrm{g}}$ as in the undeformed setting, and satisfy the relation
\begin{equation}\label{AFYB-frakJ-tildeJ-relation}
\mathfrak{J}=-(\tilde{J}+2v) \,\, .
\end{equation}
As for undeformed YB models, it will be convenient to work in terms of $\tilde{J}$ and $\mathfrak{J}$, keeping in mind that translating to $A$ and $B$ simply requires multiplying by $(1-c^2\eta^2)$.
Converting from lightcone to $(\tau,\sigma)$ coordinates using \eqref{def-lightcone-to-tau-sigma} and \eqref{sum-difference-Opm} one finds
\begin{equation}\label{AF-YB-frakJ-tausigma-components}
\begin{aligned}
\mathfrak{J}_{\tau}&=-\frac{1}{1-M^2}\Bigl( j_{\tau}+2v_{\tau}+M(j_{\sigma}+2v_{\sigma}) \Bigr) \, ,
\\
\mathfrak{J}_{\sigma}&=-\frac{1}{1-M^2}\Bigl( j_{\sigma}+2v_{\sigma}+M(j_{\tau}+2v_{\tau}) \Bigr) \,\, ,
\end{aligned}
\end{equation}
and exploiting \eqref{AFYB-frakJ-tildeJ-relation} these imply that
\begin{equation}\label{AF-YB-tildeJ-frakJ-relation}
\tilde{J}_{\sigma}=j_{\sigma}-M(\mathfrak{J}_{\tau})
\qquad \text{and} \qquad 
\tilde{J}_{\tau}=j_{\tau}-M(\mathfrak{J}_{\sigma}) \,\, ,
\end{equation}
which clearly represent the deformed analogue of \eqref{Jtausigma-relations} and play the crucial role of re-encoding the relation \eqref{AFYB-frakJ-tildeJ-relation} with no explicit reference to the auxiliary fields.

\paragraph{Lagrangian and momentum.} Using again \eqref{def-lightcone-to-tau-sigma} and \eqref{sum-difference-Opm} the Lagrangian reads
\begin{align}\label{AF-YB-Lagrangian}
\mathcal{L}_{\mathrm{AFYB}}=&+\tfrac{1}{2}\mathrm{tr}[j_{-}\frac{1}{1-M} j_{+}]+\mathrm{tr}[j_{-}\frac{1}{1-M} v_{+}+v_{-}\frac{1}{1-M} j_{+}]+\mathrm{tr}[v_{-}\frac{1+M}{1-M} v_{+}]+E(v)
\notag \\
=&+\tfrac{1}{2}\mathrm{tr}[j_{\tau}\frac{1}{1-M} j_{\tau}+j_{\tau}\frac{2M}{1-M^2} j_{\sigma}-j_{\sigma}\frac{1}{1-M} j_{\sigma}]+
\notag \\
&+\mathrm{tr}[j_{\tau}\frac{2}{1-M^2} v_{\tau}+j_{\tau}\frac{2M}{1-M^2} v_{\sigma}-j_{\sigma}\frac{2M}{1-M^2} v_{\tau}-j_{\sigma}\frac{2}{1-M^2} v_{\sigma}]+
\notag \\
&+\mathrm{tr}[v_{\tau}\frac{1+M}{1-M} v_{\tau}+v_{\tau}\frac{4M}{1-M^2} v_{\sigma}-v_{\sigma}\frac{1+M}{1-M} v_{\sigma}]+ E(v) \,\, .
\end{align} 
Since we are interested in computing the momentum \eqref{def-canonical-momentum}, we only need to keep track of terms containing $j_{\tau}$. In components this means we should take into account
\begin{equation}
\begin{aligned}
\!\!\!\!\mathcal{L}_{\mathrm{\!AFYB}}\!|_{j_{\tau}}\!\!\!=\!&+\!\tfrac{1}{2}\gamma_{AB}\partial_{\tau}Z^{\mu}\partial_{\tau}Z^{\nu}j_{\mu}{}^{A}j_{\nu}{}^{C}\mathcal{O}_{-C}{}^{B} \!\!+\!\gamma_{AB}\partial_{\tau}Z^{\mu}\partial_{\sigma}Z^{\nu}j_{\mu}{}^{A}j_{\nu}{}^{C}\mathcal{O}_{-C}{}^{D}\mathcal{O}_{+D}{}^{E}M_{E}{}^{B}
\\
&+\!\gamma_{AB}\partial_{\tau}Z^{\mu}j_{\mu}{}^{A}(2v_{\tau}{}^{C})\mathcal{O}_{-C}{}^{D}\mathcal{O}_{+D}{}^{B}\!\!+\!\gamma_{AB}\partial_{\tau}Z^{\mu}j_{\mu}{}^{A}(2v_{\sigma}{}^{C})\mathcal{O}_{-C}{}^{D}\mathcal{O}_{+D}{}^{E}M_{E}{}^{B} \, ,
\end{aligned}
\end{equation}
hence leading to
\begin{equation}
\pi_{\mu}=\gamma_{AB}j_{\mu}{}^{A}(j_{\tau}{}^{C}+2v_{\tau}{}^{C})\mathcal{O}_{-C}{}^{D}\mathcal{O}_{+D}{}^{B}+\gamma_{AB}j_{\mu}{}^{A}(j_{\sigma}{}^{C}+2v_{\sigma}{}^{C})\mathcal{O}_{-C}{}^{D}\mathcal{O}_{+D}{}^{E}M_{E}{}^{B} \,\, .
\end{equation}
A quick look at the components \eqref{AF-YB-frakJ-tausigma-components} of $\mathfrak{J}$ then reveals that
\begin{equation}\label{AF-YB-pi-frakJ-relation}
\pi_{\mu}=-\gamma_{AB}j_{\mu}{}^{A}\mathfrak{J}_{\tau}{}^{B}
\qquad \longleftrightarrow \qquad
\mathfrak{J}_{\tau}{}^{A}=-\pi_{\mu}j_{B}{}^{\mu}\gamma^{BA} \,\, ,
\end{equation}
which once again is in close parallel with the undeformed identifications \eqref{YB-pi-J-relation} and importantly relies on the presence of $j_{\mu}{}^{A}$, which is unaffected by the deformation and being a vielbeine allows us to invert the relation for $\frak{J}_{\tau}$. This pattern is the same as in previously discussed models and using the relation \eqref{AF-YB-tildeJ-frakJ-relation} leads to the deformed version of \eqref{Jtau-and-Jsigma-functions-of-momentum},
\begin{equation}\label{AF-YB-Jtau-and-Jsigma-functions-of-momentum}
\mathfrak{J}_{\tau}{}^{A}=-\pi_{\mu}j_{B}{}^{\mu}\gamma^{BA} 
\qquad \text{and} \qquad 
\tilde{J}_{\sigma}{}^{B}=j_{\sigma}{}^{B}-\mathfrak{J}_{\tau}{}^{C}M_{C}{^{B}} \,\, .
\end{equation}
This highlights that, as in the undeformed YB setting, both the $\tau$-component of $\mathfrak{J}$ and the $\sigma$-component of $\tilde{J}$ depend on the momenta non-trivially, and at the same time these can be expressed in terms of the physical fields with no explicit reference to the auxiliaries.

\paragraph{Poisson brackets.}
For these models one finds the structure
\begin{align}\label{AF-YB-brackets-summary}
\{ \mathfrak{J}_{\tau}{}^{A}(\sigma),\mathfrak{J}_{\tau}{}^{B}(\sigma') \} &= f^{AB}{}_{C} \, \mathfrak{J}_{\tau}{}^{C}(\sigma)\delta(\sigma-\sigma') \, ,
\notag \\
\{ \mathfrak{J}_{\tau}{}^{A}(\sigma),\tilde{J}_{\sigma}{}^{B}(\sigma') \} &= f^{AB}{}_{C} \, \tilde{J}_{\sigma}{}^{C}(\sigma)\delta(\sigma-\sigma')-\gamma^{AB} \, \partial_{\sigma}\delta(\sigma-\sigma') \, ,
\notag \\
\{ \tilde{J}_{\sigma}{}^{A}(\sigma),\tilde{J}_{\sigma}{}^{B}(\sigma') \} &= \eta^2 c^2 f^{AB}{}_{C}\mathfrak{J}_{\tau}{}^{C}\delta(\sigma-\sigma') \, ,
\\
\{ \mathfrak{J}_{\tau}{}^{A}(\sigma),j_{\sigma}{}^{B}(\sigma') \} &= f^{AB}{}_{C}j_{\sigma}{}^{C}\delta(\sigma-\sigma')-\gamma^{AB}\partial_{\sigma}\delta(\sigma-\sigma') \, ,
\notag \\
\{ \mathfrak{J}_{\tau}{}^{A}(\sigma),M_{C}{}^{B}(\sigma') \} &= \bigl(f^{AB}{}_{D}M_{C}{}^{D}-f_{C}{}^{AD}M_{D}{}^{B}\bigr)\delta(\sigma-\sigma') \, .
\notag
\end{align}
It is an important sanity check to notice that
\begin{itemize}
\item For $E(v)=0$ the auxiliary field EOM of AF-YB models reads
\begin{equation}
v_{\pm}=-\frac{1}{1\pm M}j_{\pm} = -J_{\pm} \,\, ,
\end{equation}
where $J_{\pm}$ is the current characterising the undeformed YB models. This implies that both $\frak{J}_{\pm}$ and $\tilde{J}_{\pm}$ reduce to $J_{\pm}$ in this limit, ensuring that the above brackets correctly reduce to the ones found for the undeformed YB models \eqref{YB-brackets-summary}.
\item Taking the limit $\eta\rightarrow0$, one finds that 
\begin{equation}
\begin{aligned}
\tilde{J}_{\pm}=\frac{1}{1\mp M}\Bigl( j_{\pm}\pm M(2v_{\pm}) \Bigr) &\qquad \longrightarrow \qquad +j_{\pm} \, ,
\\
\mathfrak{J}_{\pm}=-\frac{1}{1\mp M}\Bigl( j_{\pm}+2v_{\pm}\Bigr)\quad&\qquad\longrightarrow\qquad-(j_{\pm}+2v_{\pm}) \,\, ,
\end{aligned}
\end{equation}
which ensures the correct reduction to the brackets of the AF-PCM.
\item Combining these limits leads back to the undeformed PCM, as one should expect.
\end{itemize} 
From naive expectations, as highlighted at the beginning of the subsection, one might be worried that the detailed derivation of the above brackets would exhibit complications due to the double action of the deformation on the two currents. As anticipated this is in fact practically not true, since the relations \eqref{AF-YB-Jtau-and-Jsigma-functions-of-momentum} automatically ensure that each one of the brackets in \eqref{AF-YB-brackets-summary} has exactly the same structure as the analogue bracket in \eqref{YB-brackets-summary} and can be evaluated by following the exact same sequence of steps as for the undeformed case in Appendix \ref{appendix:C-YB}, by simply replacing $J_{\tau} \rightarrow\mathfrak{J}_{\tau}$ and $J_{\sigma}\rightarrow \tilde{J}_{\sigma}$. We conclude by remarking, once again, how the brackets \eqref{AF-YB-brackets-summary} fall within the class \eqref{Yangian-general-brackets} and obey the conditions \eqref{Yangian-theorem-condition}, ensuring the existence of an underlying classical Yangian symmetry for each model in this family. This clearly includes the undeformed sector \eqref{YB-brackets-summary}, whose Yangian structure had not been previously highlighted, to the best of our knowledge.

\subsubsection{Non-Abelian T-dual sigma models}\label{subsubsec:TD}
\paragraph{Generalities.} These models can be obtained from the PCM on a Lie group G by performing non-Abelian T-duality \cite{Fridling:1983ha,Fradkin:1984ai,delaOssa:1992vci,Alvarez:1993qi}\footnote{Non-Abelian T-duality generalizes the Abelian case \cite{Buscher:1987sk,Buscher:1987qj} via the gauging procedure \cite{Rocek:1991ps}.} on a subsector of the isometry group G$_{L}\times$G$_{R}$, and here we will be concerned with dual models obtained by dualising the full G$_{L}$ sector. The T-dualisation procedure needed to obtain the T-dual model is characterised by an operator $M:\mathfrak{g}\rightarrow\mathfrak{g}$ on the underlying Lie algebraic structure
\begin{equation}
M\equiv \mathrm{ad}_{X} 
\qquad \text{such that} \qquad
\mathrm{ad}_{X}(Y)=[X,Y]=X^{B}Y^{C}f_{BC}{}^{A}T_{A} \qquad \forall \, Y\in \mathfrak{g} \,\, ,
\end{equation}
where $X=X^{A}T_{A}$ are Lie-algebra-valued coordinates on the T-dual background. In components, the above operator can be expressed as
\begin{equation}\label{TD-M-in-components}
M(T_{A})=M_{A}{}^{B}T_{B}=[X,T_{A}]=X^{C}f_{CA}{}^{B}T_{B}  
\qquad \Rightarrow \qquad 
M_{A}{}^{B}=X^{C}f_{CA}{^{B}} \,\, , 
\end{equation}
and is antisymmetric with respect to the Cartan-Killing form
\begin{equation}
\begin{aligned}
M_{AB}&=M_{A}{}^{D}\gamma_{DB}=X^{C}f_{CA}{}^{D}\gamma_{DB}=X^{C}f_{CAB}
\\
&=-X^{C}f_{CBA}=-X^{C}f_{CB}{}^{D}\gamma_{DA}=-M_{B}{}^{D}\gamma_{DA}=-M_{BA} \,\, .
\end{aligned}
\end{equation} 
Given the dependence of $M$ on the background coordinates $X^{A}$ one can recognise that
\begin{equation}\label{TD-M-derivative-components}
\partial_{C}M_{A}{}^{B}=f_{CA}{}^{B} \,\, .
\end{equation}
Finally, as in the case of YB, in terms of $M$ one can introduce the operators
\begin{equation}\label{TD-operators-definition}
\mathcal{O}_{\pm}=\frac{1}{1\pm M}=\sum_{k=0}^{\infty}(\mp)^k
M^k \,\, , 
\quad \text{acting as} \quad \mathcal{O}_{\pm}(X)=X^{A}\mathcal{O}_{\pm A}{}^{B}T_{B} \quad \forall \, X\in \mathfrak{g} \,\, .
\end{equation}
The definition ensures that $\mathcal{O}_{\pm}$ and $M$ are effectively commuting and satisfy the relations
\begin{equation}\label{TD-sum-difference-Opm}
\frac{1}{1-M}-\frac{1}{1+M}=\frac{2M}{1-M^2} \, , 
\qquad 
\qquad 
\frac{1}{1-M}+\frac{1}{1+M}=\frac{2}{1-M^2} \,\, .
\end{equation}
Due to the antisymmetry of $M$ they are furthermore related by transposition, namely
\begin{equation}
\mathrm{tr}[X\mathcal{O}_{\pm}(Y)]=\mathrm{tr}[\mathcal{O}_{\mp}(X)Y] 
\qquad \Longrightarrow \qquad
\mathcal{O}_{\pm AB}=\mathcal{O}_{\mp BA} \,\, .
\end{equation}

\paragraph{The $L$-current.} A known feature of T-duality is the exchange in role of EOM and Bianchi identities, which causes the T-dual EOM to coincide with the flatness of the current $L$. On-shell, this is also flat and hence the T-dual Lax \eqref{Lax-undeformed-models} is characterised by
\begin{equation}
L_{\pm}=\pm\frac{1}{1\pm M} \partial_{\pm}X = \pm \mathcal{O}_{\pm}(\partial_{\pm} X) \,\, ,
\end{equation}
which after converting from lightcone to $(\tau,\sigma)$ coordinates using \eqref{def-lightcone-to-tau-sigma} and \eqref{TD-sum-difference-Opm} reads
\begin{equation}\label{TD-jtilde-tau-sigma-coordinates}
L_{\tau}=\frac{1}{1-M^2}\Bigl( \partial_{\sigma}X - M(\partial_{\tau}X) \Bigr)
\qquad \text{and} \qquad
L_{\sigma}=\frac{1}{1-M^2}\Bigl( \partial_{\tau}X - M(\partial_{\sigma}X) \Bigr) \,\, .
\end{equation}
Notice that the structure is very similar to the one encountered for YB models in \eqref{YB-Jtausigma}\footnote{This is in fact no coincidence, given the deep connections between YB and T-dual models \cite{Osten:2016dvf,Borsato:2016pas,Borsato:2017qsx,Borsato:2018idb}.}, except that the role of $\tau$ and $\sigma$ is exchanged, as one would expect from T-duality. 
The components of the current then also satisfy the important relations
\begin{equation}\label{TD-tildejtau-tildejsigma-relation}
L_{\sigma}=\partial_{\tau}X-M(L_{\tau})
\qquad \text{and} \qquad
L_{\tau}=\partial_{\sigma}X-M(L_{\sigma}) \,\, .
\end{equation}

\paragraph{Lagrangian and momentum.} Using again \eqref{def-lightcone-to-tau-sigma} and \eqref{TD-sum-difference-Opm} the Lagrangian reads
\begin{equation}\label{TD-Lagrangian}
\begin{aligned}
\mathcal{L}_{\mathrm{TD}}=&-\tfrac{1}{2}\mathrm{tr}[\partial_{+}X\frac{1}{1-M} \partial_{-}X]
\\
=&-\tfrac{1}{2}\mathrm{tr}[\partial_{\tau}X\frac{1}{1-M} \partial_{\tau}X-\partial_{\tau}X\frac{2M}{1-M^2} \partial_{\sigma}X-\partial_{\sigma}X\frac{1}{1-M} \partial_{\sigma}X]
\\
=&-\tfrac{1}{2}\gamma_{AB}\partial_{\tau}X^{A}\partial_{\tau}X^{C}\mathcal{O}_{-C}{}^{B}-\tfrac{1}{2}\gamma_{AB}\partial_{\sigma}X^{A}\partial_{\sigma}X^{C}\mathcal{O}_{-C}{}^{B}
\\
&+\gamma_{AB}\partial_{\tau}X^{A}\partial_{\sigma}X^{C}\mathcal{O}_{-C}{}^{D}\mathcal{O}_{+D}{}^{E}M_{E}{}^{B} \,\, ,
\end{aligned}
\end{equation}
such that the canonical momentum \eqref{def-canonical-momentum} takes the form
\begin{equation}
\pi_{M}=-\gamma_{MB}\partial_{\tau}X^{C}\mathcal{O}_{-C}{}^{D}\mathcal{O}_{+D}{}^{B}+\gamma_{MB}\partial_{\sigma}X^{C}\mathcal{O}_{-C}{}^{D}\mathcal{O}_{+D}{}^{E}M_{E}{^{B}} \,\, ,
\end{equation}
and after a quick look at the definitions \eqref{TD-jtilde-tau-sigma-coordinates} can be written as
\begin{equation}\label{TD-momentum-jtilde-relation}
\pi_{M}=-\gamma_{MB}L_{\sigma}{}^{B}
\qquad \longleftrightarrow \qquad
L_{\sigma}{}^{A}=-\pi_{M}\gamma^{MA} \,\, .
\end{equation}
This relation has a structure very similar to the one observed in previous models, except for the fact that the usual quadratic dependence of the momenta on the vielbeine is now secretly hidden into the definition of $L$, which for the T-dual model is \emph{not} a vielbeine. One can indeed verify that the T-dual Lagrangian allows for four possible choices of vielbeine,
\begin{equation}\label{TD-vielbeine}
E_{1\pm}=\frac{s_{1}}{1-M}\partial_{\pm}X=s_{1}\mathcal{O}_{-}(\partial_{\pm}X) 
\quad \text{and} \quad
E_{2\pm}=\frac{s_{2}}{1+M}\partial_{\pm}X =s_{2}\mathcal{O}_{+}(\partial_{\pm}X) \,\, ,
\end{equation}
with $s_{1}=\pm$ and $s_{2}=\pm$ independent signs such that $s_{1}^2=s_{2}^2=1$. These choices are related by a Lorentz transformation $\Lambda$, such that $\Lambda\Lambda^{T}=\Lambda^{T}\Lambda=1$, taking the form
\begin{equation}
E_{1\pm}=\Lambda(E_{2\pm}) \, , 
\quad E_{2\pm}=\Lambda^{T}(E_{1\pm}) \, , \quad \text{with} \quad
\Lambda= \frac{s_{1}}{s_{2}}\frac{1+M}{1-M} \, , 
\quad 
\Lambda^{T}= \frac{s_{1}}{s_{2}}\frac{1-M}{1+M}\,\, ,
\end{equation}
which allows us to write the T-dual model as
\begin{equation}
\mathcal{L}_{\mathrm{TD}}=-\frac{1}{2}\mathrm{tr}[E_{1+}E_{1-}]-\frac{1}{2}\mathrm{tr}[E_{1+}M(E_{1-})]=-\frac{1}{2}\mathrm{tr}[E_{2+}E_{2-}]-\frac{1}{2}\mathrm{tr}[E_{2+}M(E_{2-})] \,\, .
\end{equation}
In terms of the dual vielbeine, $L_{\sigma}$ in \eqref{TD-jtilde-tau-sigma-coordinates} can be recast in the form 
\begin{equation}
L_{\sigma}=s_{1}\mathcal{O}_{+}\bigl( E_{1\tau}- M(E_{1\sigma})\bigr)=s_{2}\mathcal{O}_{-}\bigl( E_{2\tau}- M(E_{2\sigma})\bigr) \,\, ,
\end{equation}
such that $\pi_{M}$ in \eqref{TD-momentum-jtilde-relation} can be written as
\begin{equation}
\begin{aligned}
\pi_{M}&=-s_{1}\gamma_{MB}(E_{1\tau}{}^{C}\!-\!E_{1\sigma}{}^{D}M_{D}{}^{C})\mathcal{O}_{+C}{}^{B}
=-s_{2}\gamma_{MB}(E_{2\tau}{}^{C}\!-\!E_{2\sigma}{}^{D}M_{D}{}^{C})\mathcal{O}_{-C}{}^{B} \, .
\end{aligned}
\end{equation}
Comparison with \eqref{SSSM-canonical-momentum} and \eqref{YB-pi-J-relation} highlights how the components $\mathcal{O}_{\pm C}{}^{B}$ of the operators \eqref{TD-operators-definition} play the role of components of the T-dual vielbeine, as indeed confirmed by the definition \eqref{TD-vielbeine}, i.e. they play the role of $j_{\mu}^{(2)a}$ or $j_{\mu}{}^{A}$ in the SSSM or PCM/YB models. The role of $j_{\tau}^{(2)a}$ or $j_{\tau}{}^A$/$J_{\tau}{}^A$, on the other hand, is  now played by the combinations $E_{1\tau}{}^{C}-E_{1\sigma}{}^{D}M_{D}{}^{C}$ and $E_{2\tau}{}^{C}-E_{2\sigma}{}^{D}M_{D}{}^{C}$, which non-trivially mix the $\tau$ and $\sigma$ components of the dual vielbeine.
Finally, as observed in YB models but unlike PCMs and SSSMs, relations \eqref{TD-tildejtau-tildejsigma-relation} and \eqref{TD-momentum-jtilde-relation} imply that both $L_{\tau}$ and $L_{\sigma}$ depend on the canonical momentum:
\begin{equation}\label{TD-Jtau-and-Jsigma-functions-of-momentum}
L_{\sigma}{}^{A}=-\pi_{M}\gamma^{MA} 
\qquad \text{and} \qquad 
L_{\tau}{}^{B}=\partial_{\sigma}X^{B}-L_{\sigma}{}^{C}M_{C}{^{B}} \,\, .
\end{equation}

\paragraph{Poisson brackets.} As shown in detail in Appendix \ref{appendix:C-TD}, for these models one obtains
\begin{equation}\label{TD-brackets-summary}
\begin{aligned}
\{ L_{\tau}{}^{A}(\sigma),L_{\tau}{}^{B}(\sigma') \} &= f^{AB}{}_{C} \, L_{\tau}{}^{C}(\sigma)\delta(\sigma-\sigma') \, ,
\\
\{ L_{\tau}{}^{A}(\sigma),L_{\sigma}{}^{B}(\sigma') \} &= f^{AB}{}_{C} \, L_{\sigma}{}^{C}(\sigma)\delta(\sigma-\sigma')-\gamma^{AB} \, \partial_{\sigma}\delta(\sigma-\sigma') \, ,
\\
\{ L_{\sigma}{}^{A}(\sigma),L_{\sigma}{}^{B}(\sigma') \} &= 0 \, ,
\\
\{ \partial_{\sigma}X^{A}(\sigma),\pi_{B}(\sigma') \} &= -\delta_{B}{}^{A}\partial_{\sigma'}\delta(\sigma'-\sigma) \, ,
\\
\{ M^{AB}(\sigma),\pi_{C}(\sigma') \} &= f^{AB}{}_{C}\delta(\sigma-\sigma') \, ,
\end{aligned}
\end{equation}
and the first three brackets exhibit the same structure as the PCM, a feature that one might reasonably expect by viewing T-duality as a canonical transformation \cite{Alvarez:1994wj,Lozano:1995jx}.

\subsubsection{AF-non-Abelian T-dual sigma models}\label{subsubsec:AFTD}
\paragraph{Generalities.} These models were first constructed in \cite{Bielli:2024khq} and \cite{Bielli:2024ach} by applying non-Abelian T-duality to AF-PCMs in \cite{Ferko:2024ali}. As observed in the undeformed setting, the structure of T-dual models closely resembles that of YB and for this reason they will also be characterised by the splitting of the flat and conserved current $L$ into two new currents $A$ and $B$ which are both deformed by the presence of the auxiliary fields.

\paragraph{The $A$ and $B$ currents.}
The AF-TD model Lax \eqref{Lax-lightcone} is characterised by 
\begin{equation}
A_{\pm}=\mp\frac{1}{1\mp M}\Bigl( \partial_{\pm}X \pm 2v_{\pm}\Bigr) 
\qquad \text{and} \qquad 
B_{\pm}=\pm\frac{1}{1\mp M}\Bigl( \partial_{\pm}X+M(2v_{\pm}) \Bigr) \,\, ,
\end{equation}
with $M=\mathrm{ad}_{X}$ as in the undeformed setting, and satisfy the relation
\begin{equation}\label{AFTD-tildefrakJ-tildej-relation}
B=-(A+2v) \,\, .
\end{equation}
Converting from lightcone to $(\tau,\sigma)$ coordinates using \eqref{def-lightcone-to-tau-sigma} and \eqref{sum-difference-Opm} one finds
\begin{equation}\label{AF-TD-tildej-tausigma-components}
\begin{aligned}
A_{\tau}&=-\frac{1}{1-M^2}\Bigl( M(\partial_{\tau}X)+2v_{\tau}+\partial_{\sigma}X+M(2v_{\sigma}) \Bigr) \, ,
\\
A_{\sigma}&=-\frac{1}{1-M^2}\Bigl( \partial_{\tau}X+M(2v_{\tau})+M(\partial_{\sigma}X)+2v_{\sigma} \Bigr) \,\, ,
\end{aligned}
\end{equation}
and exploiting \eqref{AFTD-tildefrakJ-tildej-relation} these imply that
\begin{equation}\label{AF-TD-tildej-tildefrakJ-relation}
B_{\sigma}=\partial_{\tau}X-M(A_{\tau})
\qquad \text{and} \qquad 
B_{\tau}=\partial_{\sigma}X-M(A_{\sigma}) \,\, ,
\end{equation}
which clearly represent the deformed analogues of \eqref{TD-tildejtau-tildejsigma-relation} and for AF-YB models play the crucial role of re-encoding \eqref{AFTD-tildefrakJ-tildej-relation} with no explicit reference to the auxiliary fields.

\paragraph{Lagrangian and momentum.} Using again \eqref{def-lightcone-to-tau-sigma} and \eqref{sum-difference-Opm} the Lagrangian reads
\begin{align}\label{AF-TD-Lagrangian}
\mathcal{L}_{\mathrm{AFTD}}\!=\!& +\!\tfrac{1}{2}\mathrm{tr}[\partial_{-}X\frac{1}{1\!-\!M} \partial_{+}X]\!+\!\mathrm{tr}[\partial_{-}X\frac{1}{1\!-\!M} v_{+} \!-\! \partial_{+}X\frac{1}{1\!+\!M} v_{-}] \!-\! \mathrm{tr}[v_{-}\frac{1\!+\!M}{1\!-\!M} v_{+}]\!+\!E(v)
\notag \\
\!=\!&+\!\tfrac{1}{2}\mathrm{tr}[\partial_{\tau}X\frac{1}{1-M} \partial_{\tau}X\!+\!\partial_{\tau}X\frac{2M}{1\!-\!M^2} \partial_{\sigma}X \!-\! \partial_{\sigma}X\frac{1}{1\!-\!M} \partial_{\sigma}X]
\notag\\
&+\!\mathrm{tr}[\partial_{\tau}X\frac{M}{1\!-\!M^2}(2v_{\tau})\!+\!\partial_{\tau}X\frac{1}{1\!-\!M^2}(2v_{\sigma})\!-\!\partial_{\sigma}X\frac{1}{1\!-\!M^2}(2v_{\tau})\!-\!\partial_{\sigma}X\frac{M}{1\!-\!M^2}(2v_{\sigma})]
\notag \\
&-\!\mathrm{tr}[v_{\tau}\frac{1\!-\!M}{1\!+\!M}v_{\tau}\!+\!v_{\tau}\frac{4M}{1\!-\!M^2}v_{\sigma}\!-\!v_{\sigma}\frac{1\!-\!M}{1\!+\!M}v_{\sigma}] \!+\! E(v) \,\, .
\end{align}
Since we are interested in computing the momentum \eqref{def-canonical-momentum}, we only need to keep track of terms containing $\partial_{\tau}X$. In components we should hence look at
\begin{equation}
\begin{aligned}
\!\!\!\!\!\!\mathcal{L}_{\mathrm{AFTD}}|_{\partial_{\tau}X}\!\!=\!&+\!\tfrac{1}{2}\gamma_{AB}\partial_{\tau}X^{A}\partial_{\tau}X^{C}\mathcal{O}_{-C}{}^{B} \!+\!\gamma_{AB}\partial_{\tau}X^{A}\partial_{\sigma}X^{C}\mathcal{O}_{-C}{}^{D}\mathcal{O}_{+D}{}^{E}M_{E}{}^{B}
\\
&+\!\gamma_{AB}\partial_{\tau}X^{A}(2v_{\tau}{}^{C})\mathcal{O}_{-C}{}^{D}\mathcal{O}_{+D}{}^{E}M_{E}{}^{B}\!+\!\gamma_{AB}\partial_{\tau}X^{A}(2v_{\sigma}{}^{C})\mathcal{O}_{-C}{}^{D}\mathcal{O}_{+D}{}^{B} \, ,
\end{aligned}
\end{equation}
which leads to
\begin{equation}
\begin{aligned}
\pi_{M}=&+\gamma_{MB}\partial_{\tau}X^{C}\mathcal{O}_{-C}{}^{D}\mathcal{O}_{+D}{}^{B}+\gamma_{MB}\partial_{\sigma}X^{C}\mathcal{O}_{-C}{}^{D}\mathcal{O}_{+D}{}^{E}M_{E}{}^{B}
\\
&+\gamma_{MB}(2v_{\tau}{}^{C})\mathcal{O}_{-C}{}^{D}\mathcal{O}_{+D}{}^{E}M_{E}{}^{B}+\gamma_{MB}(2v_{\sigma}{}^{C})\mathcal{O}_{-C}{}^{D}\mathcal{O}_{+D}{}^{B} \,\, .
\end{aligned}
\end{equation}
A quick comparison with \eqref{AF-TD-tildej-tausigma-components} then reveals that
\begin{equation}\label{AFTD-momentum-jtilde-relation}
\pi_{M}=-\gamma_{MA}A_{\sigma}{}^{A}
\qquad \longleftrightarrow \qquad 
A_{\sigma}{}^{A}=-\pi_{M}\gamma^{MA} \,\, ,
\end{equation}
which is the deformed analogue of \eqref{TD-momentum-jtilde-relation}. Finally, combining this result with \eqref{AF-TD-tildej-tildefrakJ-relation} one can also recover the expected deformed version of \eqref{TD-Jtau-and-Jsigma-functions-of-momentum}\footnote{It should be noted that in this subsection we have used a slightly different convention as compared to the original derivation \cite{Bielli:2024khq,Bielli:2024ach} of the AF-T-dual model \eqref{AF-TD-Lagrangian}, replacing the Lagrange multipliers $X$ with $-X$. This is of course a mere choice of notation, which does not affect the final result, namely the Poisson brackets \eqref{AFTD-brackets-summary} satisfied by the deformed currents, but has the nice effect of bringing all the identities relevant to AF-TD models on the same footing as the other classes of sigma models. For example, it is not hard to check that in the original conventions \cite{Bielli:2024khq,Bielli:2024ach} the slightly different signs in the definitions of the Lagrangian, and of $A$ and $B$, would have led to a different definition of momenta, with $A_{\sigma}{}^{A}=+\pi_{M}\gamma^{MA}$ and $B_{\tau}{}^{B}=-\partial_{\sigma}X^{B}+A_{\sigma}{}^{C}M_{C}{^{B}}$ carrying different signs with respect to the other classes of sigma models, but in fact still leading to the same Poisson brackets \eqref{AFTD-brackets-summary}. This rescaling of the Lagrange multipliers has a simple interpretation: it brings them to have the same sign as the kinetic term of the pre-dualisation undeformed PCM, which is not true in the conventions of the original derivation, where the multipliers appear with the same sign as the kinetic term of the pre-dualisation deformed PCM, which is flipped with respect to the undeformed setting.
}
\begin{equation}\label{AFTD-Jtau-and-Jsigma-functions-of-momentum}
A_{\sigma}{}^{A}=-\pi_{M}\gamma^{MA} 
\qquad \text{and} \qquad 
B_{\tau}{}^{B}=\partial_{\sigma}X^{B}-A_{\sigma}{}^{C}M_{C}{^{B}} \,\, ,
\end{equation}
which once again shows that the relevant components of the currents $A$ and $B$ can be entirely expressed in terms of physical fields with no explicit reference to the auxiliaries.

\paragraph{Poisson brackets.} For this class of models one obtains
\begin{equation}\label{AFTD-brackets-summary}
\begin{aligned}
\{ B_{\tau}{}^{A}(\sigma),B_{\tau}{}^{B}(\sigma') \} &= f^{AB}{}_{C} \, B_{\tau}{}^{C}(\sigma)\delta(\sigma-\sigma') \, ,
\\
\{ B_{\tau}{}^{A}(\sigma),A_{\sigma}{}^{B}(\sigma') \} &= f^{AB}{}_{C} \, A_{\sigma}{}^{C}(\sigma)\delta(\sigma-\sigma')-\gamma^{AB} \, \partial_{\sigma}\delta(\sigma-\sigma') \, ,
\\
\{ A_{\sigma}{}^{A}(\sigma),A_{\sigma}{}^{B}(\sigma') \} &= 0 \, ,
\\
\{ \partial_{\sigma}X^{A}(\sigma),\pi_{B}(\sigma') \} &= -\delta_{B}{}^{A}\partial_{\sigma'}\delta(\sigma'-\sigma) \, ,
\\
\{ M^{AB}(\sigma),\pi_{C}(\sigma') \} &= f^{AB}{}_{C}\delta(\sigma-\sigma') \, .
\end{aligned}
\end{equation}
As noted for AF-YB models, it is straightforward to see that using \eqref{AFTD-Jtau-and-Jsigma-functions-of-momentum} all the brackets have in fact the same structure as in the undeformed T-dual case and the calculation of each of them goes through as highlighted in Appendix \ref{appendix:C-TD}. As one might again expect from the canonical interpretation of T-duality, the above brackets also have the same structure as the one of AF-PCMs \eqref{AFPCM-brackets-summary}, and indeed feature the replacements $L_{\tau}\rightarrow B_{\tau}$, $L_{\sigma}\rightarrow A_{\sigma}$ and respect Theorem \ref{Yangian-theorem}, which ensures the existence of a Yangian.

\subsubsection{Principal Chiral models with Wess-Zumino term}\label{subsubsec:PCMWZ}
\paragraph{Generalities.} These models \cite{Wess:1971yu,Novikov:1982ei} represent an important class of deformations of the PCM, preserving integrability for any value of the deformation parameter \cite{Abdalla:1982yd} and exhibiting quantum conformal symmetry for specific choices of it, leading to Wess-Zumino-Witten models \cite{Witten:1983ar}. The deformation is performed by adding to the PCM action a contribution living on a 3-manifold $\mathcal{M}_{3}$ whose boundary coincides with $\Sigma$, namely $\partial \mathcal{M}_{3}=\Sigma$. We will follow the notation in section 5 of \cite{Bielli:2024ach}, hence writing the action as
\begin{equation}\label{PCM_WZ}
S_{\text{PCM-WZ}} = - \frac{\hay}{2} \int_{\Sigma} \mathrm{d}^2 \sigma \, \tr ( j_+ j_- ) + \frac{\kay}{6} \int_{\mathcal{M}_3} \mathrm{d}^3 x \, \epsilon^{ijk} \tr \left( j_i [ j_j, j_k ] \right) \, .
\end{equation}
As reviewed in Appendix A.3 of \cite{Bielli:2024ach}, the variations of the WZ term correctly localise to the boundary of $\mathcal{M}_{3}$, namely to $\Sigma$, hence giving rise to a well-defined deformation of the PCM, which can be recovered by choosing $\hay = 1$ and $\kay=0$.

\paragraph{The $L$-current.} The presence of the WZ term induces a deformation of the flat and conserved current $L=j$ of the standard PCM, leading to 
\begin{equation}\label{PCM-WZ-BA-currents}
L_{\tau}{}^{A}(\sigma)=j_{\tau}{}^{A}(\sigma)-\tfrac{\kay}{\hay}j_{\sigma}{}^{A}(\sigma)
\qquad \text{and} \qquad 
L_{\sigma}{}^{A}(\sigma)=j_{\sigma}{}^{A}(\sigma)-\tfrac{\kay}{\hay}j_{\tau}{}^{A}(\sigma) \,\, .
\end{equation}

\paragraph{Lagrangian and momentum.} The new contribution can be interpreted as coupling the PCM to a background 3-form $H$ which is closed but not exact, and hence cannot be encoded in a globally-defined 2-form $B$. For our purposes it is however sufficient to consider a local description of the worldsheet theory, since the BIZZ construction of Theorem \ref{BIZZ-theorem} also relies on this assumption. We thus proceed by expressing the WZ contribution in terms of a locally-defined $B$-field on the worldsheet, arriving at the Lagrangian
\begin{equation}
\begin{aligned}
\mathcal{L}_{\text{PCM-WZ}}=&-\frac{\hay}{2} \mathrm{tr}(j_{+}j_{-})+\frac{\kay}{2}B_{\mu\nu} \epsilon^{\alpha\beta}\partial_{\alpha}Z^{\mu}\partial_{\beta}Z^{\nu}
\\
=&-\frac{\hay}{2}\gamma_{AB}(\partial_{\tau}Z^{\mu}\partial_{\tau}Z^{\nu}-\partial_{\sigma}Z^{\mu}\partial_{\sigma}Z^{\nu})j_{\mu}{}^{A}j_{\nu}{}^{B}+\kay B_{\mu\nu}\partial_{\tau}Z^{\mu}\partial_{\sigma}Z^{\nu} \,\, .
\end{aligned}
\end{equation}
The associated canonical momentum \eqref{def-canonical-momentum} then takes the form
\begin{equation}\label{PCM-WZ-canonical-momentum}
\!\!\!\pi_{\mu}\!\!=\! -\hay \gamma_{AB} \, j_{\mu}{}^{A}j_{\tau}{}^{B}\!+\!\kay B_{\mu\nu}\partial_{\sigma}Z^{\nu}
\,\,\, \longleftrightarrow \,\,\, 
j_{\tau}{}^{A}\!\!=\!-\tfrac{1}{\hay}\gamma^{AB}j_{B}{}^{\mu}\pi_{\mu}\!+\tfrac{\kay}{\hay}\gamma^{AB}j_{B}{}^{\mu}B_{\mu\nu}\partial_{\sigma}Z^{\nu}  ,
\end{equation}
where the second relation relies on invertibility of $j_{\mu}{}^{A}$.

\paragraph{Poisson brackets.} For these models one finds
\begin{align}\label{PCM-WZ-brackets-summary}
\{ L_{\tau}{}^{A}(\sigma),L_{\tau}{}^{B}(\sigma') \} =& \tfrac{1}{\hay}f^{AB}{}_{C} \, L_{\tau}{}^{C}(\sigma)\delta(\sigma-\sigma')+\tfrac{2\kay}{\hay^2}\gamma^{AB}\partial_{\sigma}\delta(\sigma-\sigma') \, ,
\notag \\
\{ L_{\tau}{}^{A}(\sigma),L_{\sigma}{}^{B}(\sigma') \} =& \tfrac{1}{\hay}f^{AB}{}_{C} \, L_{\sigma}{}^{C}(\sigma)\delta(\sigma-\sigma')-\tfrac{1}{\hay}(1+\tfrac{\kay^2}{\hay^2})\gamma^{AB} \, \partial_{\sigma}\delta(\sigma-\sigma') \, ,
\\
\{ L_{\sigma}{}^{A}(\sigma),L_{\sigma}{}^{B}(\sigma') \} =& f^{AB}{}_{C}\Bigl(-\tfrac{2\kay}{\hay^2}L_{\sigma}{}^{C}-\tfrac{\kay^2}{\hay^3}L_{\tau}{}^{C}\Bigr)\delta(\sigma-\sigma')+\tfrac{2\kay}{\hay^2}\gamma^{AB}\partial_{\sigma}\delta(\sigma-\sigma') \,\, .
\notag
\end{align}
To determine this structure, discussed for example in \cite{Abdalla:1993sn,Evans:2000hx}, one needs to make use of
\begin{align}\label{PCM-WZ-building-block-brackets}
\{ j_{A}{}^{\mu}\pi_{\mu}(\sigma),j_{B}{}^{\nu}\pi_{\nu}(\sigma') \} =& -f_{ABC}\gamma^{CD}j_{D}{}^{\mu}\pi_{\mu}\delta(\sigma-\sigma') \, ,
\notag \\
\{ j_{A}{}^{\mu}\pi_{\mu}(\sigma),j_{\sigma}{}^{B}(\sigma') \} =& -f_{A}{}^{B}{}_{C}j_{\sigma}{}^{C}\delta(\sigma-\sigma')+\delta_{A}{}^{B}\partial_{\sigma}\delta(\sigma-\sigma') \, ,
\\
\{ j_{A}{}^{\mu}\pi_{\mu}(\sigma),j_{B}{}^{\nu}B_{\nu\rho}\partial_{\sigma}Z^{\rho}(\sigma') \} =&+ \partial_{\sigma}Z^{\rho} j_{A}{}^{\mu}j_{B}{}^{\nu}(\partial_{\rho}B_{\nu\mu}-\partial_{\mu}B_{\nu\rho})\delta(\sigma-\sigma') 
\notag \\
&+\partial_{\sigma}Z^{\rho} j_{A}{}^{\mu}j_{B}{}^{\nu}j_{E}{}^{\lambda}(B_{\lambda\rho}\partial_{\mu}j_{\nu}{}^{E}-B_{\lambda\mu}\partial_{\rho}j_{\nu}{}^{E})\delta(\sigma-\sigma') 
\notag \\
&+ j_{A}{}^{\mu}(\sigma)j_{B}{}^{\nu}(\sigma)B_{\nu\mu}(\sigma)\partial_{\sigma}\delta(\sigma-\sigma')\,\, .
\notag
\end{align}
The first and second brackets in \eqref{PCM-WZ-building-block-brackets} are the same ones encountered in the case of undeformed PCMs and AF-PCMs and for this reason have simply been borrowed from Appendix C of \cite{Bielli:2024ach}, to which we refer for the details of their computation. The last bracket in \eqref{PCM-WZ-building-block-brackets} represents the main new intermediate ingredient and in fact, for our purposes, will always appear in a specific combination $I^{AB}$, which reads
\begin{equation}\label{IAB-combination-definition}
\!\!\!I^{AB}\!=\!\gamma^{AC}\gamma^{BD}\{j_{C}{}^{\mu}\pi_{\mu}(\sigma),j_{D}{}^{\nu}B_{\nu\rho}\partial_{\sigma}Z^{\rho}(\sigma')\}\!-\!\gamma^{AD}\gamma^{BC}\{j_{C}{}^{\mu}\pi_{\mu}(\sigma'),j_{D}{}^{\nu}B_{\nu\rho}\partial_{\sigma}Z^{\rho}(\sigma)\}.
\end{equation}
This can be recast in the form
\begin{equation}\label{IAB-combination-rearranged}
I^{AB}=-f^{AB}{}_{C}\gamma^{CD}j_{D}{}^{\mu}B_{\mu\nu}\partial_{\sigma}Z^{\nu}\delta(\sigma-\sigma')-f^{AB}{}_{C}j_{\sigma}{}^{C}\delta(\sigma-\sigma') \, ,
\end{equation}
and together with the first two brackets in \eqref{PCM-WZ-building-block-brackets} plays the role of a building block for the main three brackets in \eqref{PCM-WZ-brackets-summary}. In Appendix \ref{appendix:C-PCMWZ} we evaluate the last bracket in \eqref{PCM-WZ-building-block-brackets}, then show how $I^{AB}$ takes the form \eqref{IAB-combination-rearranged}, and finally compute \eqref{PCM-WZ-brackets-summary}.

\subsubsection{AF-Principal Chiral models with Wess-Zumino term}\label{subsubsec:AFPCMWZ}
\paragraph{Generalities.} These models were constructed in \cite{Fukushima:2024nxm} from 4D Chern-Simons and in \cite{Bielli:2024ach} by adding a Wess-Zumino term to the PCM deformed by auxiliary fields. Following the notation of \cite{Bielli:2024ach}, adding a WZ term to the AFSM action leads to
\begin{equation}\label{AF-PCM-WZ}
\!\!\!\! S_{\!\text{AFSM}}^{\text{WZ}} \!=\! \hay \!\! \int_{\Sigma} \!\!\! \mathrm{d}^2 \sigma  \Bigl( \! \tfrac{1}{2} \mathrm{tr} ( j_+ j_- ) \!+\! \mathrm{tr} ( v_+ v_- ) \!+\! \mathrm{tr} ( j_+ v_- \!+\! j_- v_+ ) \!+\! E  \! \Bigr)  
\!+ \frac{\kay}{6}\!\! \int_{\mathcal{M}_3} \!\!\!\!\!\!\! \mathrm{d}^3 x \, \epsilon^{ijk} \tr \left( j_i [ j_j, j_k ] \right) \,.
\end{equation}
When $\kay=0$ and $\hay=1$ one recovers the original AFSM, while for $E=0$ integrating out the auxiliary fields leads to the PCM-WZ model \eqref{PCM_WZ} of the previous subsection. For this latter step, it is also crucial that the WZ term exhibits no coupling to the auxiliaries.

\paragraph{The $A$ and $B$ currents.} Comparing with PCMs, the already modified current \eqref{PCM-WZ-BA-currents} is further deformed by the presence of auxiliary fields. Since the latter do not however couple to the WZ term, one should not expect all terms in \eqref{PCM-WZ-BA-currents} to receive auxiliary contributions, since otherwise letting $\kay\rightarrow0$ would not recover the simpler AF-PCM case. This leads to a Lax \eqref{Lax-lightcone} characterised by $\mathfrak{J}=-(j+2v)$, as for $\kay=0$, and with
\begin{equation}\label{AFPCM-WZ-BA-currents}
A_{\sigma}{}^{A}(\sigma)=j_{\sigma}{}^{A}(\sigma)-\tfrac{\kay}{\hay}\mathfrak{J}_{\tau}{}^{A}(\sigma)
\qquad \text{and} \qquad 
B_{\tau}{}^{A}(\sigma)=\mathfrak{J}_{\tau}{}^{A}(\sigma)-\tfrac{\kay}{\hay}j_{\sigma}{}^{A}(\sigma)  \,\, .
\end{equation}

\paragraph{Lagrangian and momentum.} Since, as mentioned above, the WZ term is decoupled from the auxiliary fields, it is natural to proceed as in the previous subsection, considering a Lagrangian which includes a $B$-field locally-defined on the worldsheet:
\begin{align}
\mathcal{L}_{\mathrm{AFSM}}^{\text{WZ}}\!\!=\!&+\!\hay\Bigl(\tfrac{1}{2}\mathrm{tr}[j_{+}j_{-}]+\mathrm{tr}[j_{+}v_{-}+j_{-}v_{+}]+\mathrm{tr}[v_{+}v_{-}]+E(v) \!\Bigr)+\frac{\kay}{2}B_{\mu\nu} \epsilon^{\alpha\beta}\partial_{\alpha}Z^{\mu}\partial_{\beta}Z^{\nu}
\notag \\
=\!&+\!\hay\Bigl(\tfrac{1}{2}\mathrm{tr}[j_{\tau}^2-j_{\sigma}^2]+2\mathrm{tr}[j_{\tau}v_{\tau}-j_{\sigma}v_{\sigma}]+\mathrm{tr}[v_{\tau}^{2}-v_{\sigma}^{2}]+E(v)\!\Bigr)+\kay B_{\mu\nu}\partial_{\tau}Z^{\mu}\partial_{\sigma}Z^{\nu}
\notag \\
=\!&+\!\tfrac{\hay}{2}\gamma_{AB}(\partial_{\tau}Z^{\mu}\partial_{\tau}Z^{\nu}\!-\!\partial_{\sigma}Z^{\mu}\partial_{\sigma}Z^{\nu})j_{\mu}{}^{A}j_{\nu}{}^{B}+2\hay\gamma_{AB}\partial_{\tau}Z^{\mu}j_{\mu}{}^{A}v_{\tau}^{B}
\\
\!&-\!2\hay\gamma_{AB}\partial_{\sigma}Z^{\mu}j_{\mu}{}^{A}v_{\sigma}^{B}+\hay\gamma_{AB}(v_{\tau}^{A}v_{\tau}^{B}\!-\!v_{\sigma}^{A}v_{\sigma}^{B})+\hay E(v) +\kay B_{\mu\nu}\partial_{\tau}Z^{\mu}\partial_{\sigma}Z^{\nu}  .
\notag
\end{align}
The associated canonical momentum \eqref{def-canonical-momentum} then takes the form
\begin{equation}\label{AFPCM-WZ-canonical-momentum}
\!\!\!\pi_{\mu}\!\!=\! -\hay \gamma_{AB} \, j_{\mu}{}^{A}\mathfrak{J}_{\tau}{}^{B}\!+\!\kay B_{\mu\nu}\partial_{\sigma}Z^{\nu}
\,\,\, \longleftrightarrow \,\,\,  
\mathfrak{J}_{\tau}{}^{A}\!\!=\!-\tfrac{1}{\hay}\gamma^{AB}j_{B}{}^{\mu}\pi_{\mu}\!+\tfrac{\kay}{\hay}\gamma^{AB}j_{B}{}^{\mu}B_{\mu\nu}\partial_{\sigma}Z^{\nu}  ,
\end{equation}
where we re-introduced the auxiliary-field-deformed current $\mathfrak{J}=-(j+2v)$ and the second relation results again from the fact that $j_{\mu}{}^{A}$ plays the role of a vielbeine. The relation \eqref{AFPCM-WZ-canonical-momentum} ensures that the relevant components of $A$ and $B$ in \eqref{AFPCM-WZ-BA-currents} exhibit no explicit reference to the auxiliary fields and can be expressed in terms of physical fields only.

\paragraph{Poisson brackets.}
For these models one then obtains
\begin{align}\label{AFPCM-WZ-brackets-summary}
\{ B_{\tau}{}^{A}(\sigma),B_{\tau}{}^{B}(\sigma') \} =& \tfrac{1}{\hay}f^{AB}{}_{C} \, B_{\tau}{}^{C}(\sigma)\delta(\sigma-\sigma')+\tfrac{2\kay}{\hay^2}\gamma^{AB}\partial_{\sigma}\delta(\sigma-\sigma') \, ,
\notag \\
\{ B_{\tau}{}^{A}(\sigma),A_{\sigma}{}^{B}(\sigma') \} =& \tfrac{1}{\hay}f^{AB}{}_{C} \, A_{\sigma}{}^{C}(\sigma)\delta(\sigma-\sigma')-\tfrac{1}{\hay}(1+\tfrac{\kay^2}{\hay^2})\gamma^{AB} \, \partial_{\sigma}\delta(\sigma-\sigma') \, ,
\\
\{ A_{\sigma}{}^{A}(\sigma),A_{\sigma}{}^{B}(\sigma') \} =& f^{AB}{}_{C}\Bigl(-\tfrac{2\kay}{\hay^2}A_{\sigma}{}^{C}-\tfrac{\kay^2}{\hay^3}B_{\tau}{}^{C}\Bigr)\delta(\sigma-\sigma')+\tfrac{2\kay}{\hay^2}\gamma^{AB}\partial_{\sigma}\delta(\sigma-\sigma') \,\, .
\notag
\end{align}
Thanks to the relation \eqref{AFPCM-WZ-canonical-momentum}, when expressed in terms of canonical momentum the currents $A$ and $B$ exhibit the same form as the current \eqref{PCM-WZ-BA-currents} from the undeformed setting, which are rearranged by making use of \eqref{PCM-WZ-canonical-momentum}. Hence the brackets exhibit the same structure, after once again using the intermediate brackets \eqref{PCM-WZ-building-block-brackets} and the $I^{AB}$ combination \eqref{IAB-combination-definition}, which can be explicitly evaluated to \eqref{IAB-combination-rearranged}. The details of the calculation are as in the undeformed case and we thus refer to Appendix \ref{appendix:C-PCMWZ}. We finally stress one last time how the above brackets only differ from the undeformed ones \eqref{PCM-WZ-brackets-summary} by the replacements $L_{\tau}\rightarrow B_{\tau}$, $L_{\sigma}\rightarrow A_{\sigma}$, respecting all previously observed patterns as well as falling into the class \eqref{Yangian-general-brackets}, ensuring the presence of an underlying Yangian.

\subsection{Patterns and observations}
In this subsection we would like to take a moment to summarise a few observations resulting from the above long survey of models and their respective auxiliary field deformations.

\paragraph{Current splitting mechanisms.} Each of the above families of integrable sigma models is characterised by a single current $L$, which is on-shell both flat and conserved, while their auxiliary field deformations feature two currents $A$ and $B$, which are on-shell respectively flat and conserved. These properties can be encoded in the flatness of Lax connections
\begin{equation}\label{current-splitting-mechanism}
\mathfrak{L}=\frac{L+z\star L}{1-z^2}
\qquad 
\xlongrightarrow[\text{coupling to auxiliary fields}]
\qquad \qquad 
\mathfrak{L}=\frac{A+z\star B}{1-z^2} \, ,
\end{equation}
where, for simplicity in lightcone coordinates, one finds that
\begin{itemize}
\item  \textbf{PCM:} The current $L_{\pm}:=j_{\pm}$ separates into
\begin{equation}
A_{\pm}:=j_{\pm} \qquad \text{and} \qquad B_{\pm}:=-(j_{\pm}+2v_{\pm}) \,\, .
\end{equation}
\item \textbf{SSSM:} The current $L_{\pm}:=-2\mathrm{Ad}_{\mathrm{g}}(j_{\pm}^{(2)})$ separates into
\begin{equation}
A_{\pm}:=-2\mathrm{Ad}_{\mathrm{g}}(j_{\pm}^{(2)}) \qquad \text{and} \qquad B_{\pm}:=+2\mathrm{Ad}_{\mathrm{g}}(j_{\pm}^{(2)}+2v_{\pm}^{(2)}) \,\, .
\end{equation}
\item \textbf{YB:} The current $L_{\pm}:=(1-c^2\eta^2)(1\pm \eta \mathcal{R}_{\mathrm{g}})^{-1}j_{\pm}$ separates into
\begin{equation}
A:=\frac{1-c^2\eta^2}{1\mp \eta \mathcal{R}_{\mathrm{g}}}\Bigl(j_{\pm}\pm\eta\mathcal{R}_{\mathrm{g}}(2v_{\pm})\Bigr) \qquad \text{and} \qquad B:=-\frac{1-c^2\eta^2}{1\mp \eta \mathcal{R}_{\mathrm{g}}}\Bigl(j_{\pm}+2v_{\pm}\Bigr) \,\, .
\end{equation}
\item \textbf{TD:} The current $L_{\pm}:=\pm(1\pm\mathrm{ad}_{X})^{-1}\partial_{\pm}X$ separates into
\begin{equation}
A_{\pm}:=\frac{\mp1}{1\mp \mathrm{ad}_{X}}\Bigl(\partial_{\pm}X\pm 2v_{\pm}\Bigr) \quad \text{and} \quad B_{\pm}:=\frac{\pm1}{1\mp \mathrm{ad}_{X}}\Bigl(\partial_{\pm}X+\mathrm{ad}_{X}(2v_{\pm})\Bigr) \,\, .
\end{equation}
\item \textbf{PCM-WZ:} The current $L_{\pm}:=(1\mp\tfrac{\kay}{\hay})j_{\pm} $ separates into 
\begin{equation}
A_{\pm}:=j_{\pm}\pm\tfrac{\kay}{\hay}(j_{\pm}+2v_{\pm}) \qquad \text{and} \qquad B_{\pm}:=-(j_{\pm}+2v_{\pm}) \mp \tfrac{\kay}{\hay}j_{\pm} \,\, .
\end{equation}
\end{itemize}
Looking at the above explicit expressions, one can observe the existence of two splitting mechanisms by which the auxiliary fields act on the flat and conserved current $L$ of the seed theory in order to generate the two currents $A$ and $B$:
\begin{enumerate}
\item In the first and simpler mechanism, characterising PCMs and SSSMs, the flat and conserved current $L$ of the undeformed theory still plays a role after introducing the deformation. This however only retains one of its defining properties -- presumably either flatness or conservation, in general, even though in both models taken into account the flatness is preserved -- while the second property is transferred to a newly appearing current which explicitly depends on the auxiliary fields $v$.
\item In the second and more involved mechanism, characterising YB, T-dual and PCM-WZ models, the flat and conserved current $L$ of the undeformed theory plays no more role after introducing the deformation. Its defining properties are instead transferred to two newly defined currents, one of which is flat while the other is conserved, both of which exhibit a non-trivial dependence on the auxiliary fields $v$. 
\end{enumerate} 
While substantially different, in the above models these two splitting mechanisms are characterised by the important common relation
\begin{equation}\label{to-be-recalled}
\mathcal{B}=-(\mathcal{A}+2v) \,\, ,
\end{equation}
where for SSSM
\begin{equation}
A=-2\mathrm{Ad}_{\mathrm{g}}(\mathcal{A})
\qquad \text{and} \qquad 
B=-2\mathrm{Ad}_{\mathrm{g}}(\mathcal{B}) \,\, ,
\end{equation}
while for the remaining models
\begin{equation}
A=a\mathcal{A}
\qquad \text{and} \qquad B=a\mathcal{B} \,\, ,
\end{equation}
with $\mathcal{A},\mathcal{B}$ and $a$ defined in \eqref{AF-general-EOM-for-identity} and case-by-case details are provided below that equation. It should be stressed that while in models characterised by the first mechanism the current $A$ is by construction independent of the auxiliary fields, and hence the relation \eqref{to-be-recalled} only plays a role when defining the physical momenta in terms of $B$, for models characterised by the second mechanism, the relation \eqref{to-be-recalled} plays a much more subtle role. This not only ensures the same dependence of the physical momenta on $B$ as in the first mechanism, but also that the non-trivially deformed $A$ can be entirely expressed in terms of physical fields as much as $B$. If this were not the case, the construction would be different, since the Poisson brackets involving $A$ would have to be replaced by Dirac brackets.

\paragraph{Canonical momenta.} Another interesting feature that one can notice from the above survey is that, independently of the specific splitting mechanism taking place in a given model, the auxiliary fields always couple to the physical ones in such a way that the physical canonical momentum of the deformed theory can be expressed in terms of the new currents without explicit reference to the auxiliary fields -- see equations \eqref{AF-PCM-momentum-inversion}, \eqref{AF-SSSM-canonical-momentum}, \eqref{AF-YB-pi-frakJ-relation}, \eqref{AFTD-momentum-jtilde-relation} and \eqref{AFPCM-WZ-canonical-momentum}. In particular, the canonical momentum can always be rewritten as a contraction, via the Cartan-Killing form, of the $\tau$-component of the conserved current $B$ with an invertible current\footnote{An exception is AF-TD models, where the $\sigma$-component of the flat current $A$ appears, rather than the $\tau$-component of the conserved current $B$. This difference is not only in agreement with the intuitive picture that T-duality should swap the role of $\tau$ and $\sigma$, which is already visible in the undeformed TD model, but also with the observation that even after the auxiliary field deformation the only relevant ingredients are encoded in $B_{\tau}$ and $A_{\sigma}$, while their complementary components $B_{\sigma}$ and $A_{\tau}$ never appear.} -- typically the Maurer-Cartan form, replaced by its coset projection in the case of SSSMs -- such that one can always invert the relation and obtain the deformed current as an explicit function of the momentum. \\
\indent
This phenomenon is absolutely crucial in being able to compute the Poisson brackets of $A$ and $B$, since, together with \eqref{to-be-recalled}, it completely hides the deformation introduced by the auxiliary fields and allows to reduce many computations which would in principle be new, and potentially very hard to deal with, to ones which are already known, ultimately providing a set of simple replacement rules $\{L_{\sigma}\rightarrow A_{\sigma} \, , \, L_{\tau}\rightarrow B_{\tau}\}$ which allows us to recover the Poisson brackets of the deformed theory given those of the undeformed one.

\paragraph{Relation to field-dependent metrics.}
We would like to close this subsection with a more speculative comment about the nature of the generalised BIZZ currents. To point out this structure, it may be helpful to compare the explicit expressions for the first several currents in the two towers. From the original BIZZ tower \eqref{original-BIZZ-tower} one has
\begin{equation}
\begin{aligned}
\star J^{(0)}&=\beta \, \star L \, ,
\\
\star J^{(1)}&= L+\tfrac{1}{2\beta}[\star L,\chi^{(0)}] \, ,
\\
\star J^{(2)}&=\star L+\tfrac{1}{2\beta}[ L,\chi^{(0)}]+[\star L,\chi^{(1)}] \, ,
\\
\star J^{(3)}&= L+\tfrac{1}{2\beta}[ \star L,\chi^{(0)}]+[ L,\chi^{(1)}]+[\star L,\chi^{(2)}] \, ,
\end{aligned}
\end{equation}
and it is easy to notice that for $n\geq 2$ the currents satisfy the relation
\begin{equation}\label{relation-original-BIZZ-currents}
\star J^{(n)}= J^{(n-1)}+[\star L,\chi^{(n-1)}] \,\, .
\end{equation}
On the other hand, from the generalised BIZZ tower \eqref{generalised-BIZZ-tower} one has
\begin{equation}
\begin{aligned}
\star J^{(0)}&=\beta \, \star B \, ,
\\
\star J^{(1)}&=  A + \tfrac{1}{2\beta}[\star B,\chi^{(0)}] \, ,
\\
\star J^{(2)}&=\star B+\tfrac{1}{2\beta}[ A,\chi^{(0)}]+[\star B,\chi^{(1)}] \, ,
\\
\star J^{(3)}&= A + \tfrac{1}{2\beta}[\star B,\chi^{(0)}]+[ A,\chi^{(1)}]+[\star B,\chi^{(2)}] \, ,
\end{aligned}
\end{equation}
As opposed to the original tower, it seems now that the relation \eqref{relation-original-BIZZ-currents} has been lost, but suppose that one could define a modified Hodge star operation $\hat{\star}$ such that
\begin{equation}\label{modified-Hodge-def}
\star B = \hat{\star} A
\qquad \qquad \text{and} \qquad \qquad  \hat{\star}^2=1
\qquad , \qquad 
\alpha \wedge \hat{\star}\beta = -\hat{\star}\alpha \wedge \beta\,\, .
\end{equation}
When $B=A$ this reduces to the usual Hodge star. However, in a more general setting, such as for AFSMs discussed above, this is a different involution which could perhaps be interpreted as Hodge duality with respect to a field-dependent metric and in terms of this quantity the starting point \eqref{A,B-definition} of the generalised BIZZ construction becomes
\begin{equation}
\mathrm{d}(\hat{\star}A)=0
\qquad \text{and} \qquad 
\mathrm{d}A+\tfrac{1}{2}[A,A]=0 \,\, ,
\end{equation}
leading to a tower
\begin{equation}
\begin{aligned}
\hat{\star} J^{(0)}&=\beta \, \hat{\star} A \, ,
\\
\hat{\star} J^{(1)}&=  A + \tfrac{1}{2\beta}[\hat{\star} A,\chi^{(0)}] \, ,
\\
\hat{\star} J^{(2)}&=\hat{\star} A+\tfrac{1}{2\beta}[ A,\chi^{(0)}]+[\hat{\star} A,\chi^{(1)}] \, ,
\\
\hat{\star} J^{(3)}&= A + \tfrac{1}{2\beta}[\hat{\star} A,\chi^{(0)}]+[ A,\chi^{(1)}]+[\hat{\star} A ,\chi^{(2)}] \, ,
\end{aligned}
\end{equation}
where the recursion relation \eqref{relation-original-BIZZ-currents} is recovered as 
\begin{equation}\label{modified-Hodge-BIZZ-currents-rule}
\hat{\star} J^{(n)}=J^{(n-1)}+[\hat{\star} A,\chi^{(n-1)}] \qquad \forall \, n\geq 2 
\end{equation}
and the Lax takes the form
\begin{equation}
\mathfrak{L}=\frac{A+z\, \hat{\star}A}{1-z^2} \,\, .
\end{equation}
Having just recalled, around \eqref{to-be-recalled}, how for AFSMs the $A$ and $B$ currents relate to $\mathcal{A}$ and $\mathcal{B}$, one can then write the Lax, in lightcone coordinates, as
\begin{equation}
\mathfrak{L}_{\pm}=\frac{A_{\pm}\pm z \, B_{\pm}}{1-z^2}=\mathcal{O}\Bigl( \frac{\mathcal{A}_{\pm}\pm z \, \mathcal{B}_{\pm}}{1-z^2} \Bigr) \,\, ,
\end{equation}
with $\mathcal{O}$ denoting either multiplication by a constant $a$ or the action of $-2\mathrm{Ad}_{\mathrm{g}}$. This now directly connects to an observation made in \cite{Ferko:2025bhv}, which can be extended to all the above classes of AFSM: using the relations in \eqref{AF-general-EOM-for-identity},
\begin{equation}
\mathcal{A}_{\pm}+v_{\pm}+\Delta_{\pm}=0
\qquad \text{and} \qquad 
\mathcal{B}_{\pm}=-(\mathcal{A}_{\pm}+2v_{\pm}) \, ,
\end{equation}
one can rewrite\footnote{As noted in \cite{Ferko:2025bhv}, this only works for $E(\nu_{-N},...\nu_{-2},\nu_{+2},...,\nu_{+N})\equiv E(\nu_{-2},\nu_{+2})$ with $\nu_{\pm 2}:=\mathrm{tr}(v_{\pm}^2)$.}
\begin{equation}
\mathcal{B}_{+}=h_{++}\mathcal{A}_{-}+h_{+-}\mathcal{A}_{+}
\qquad \text{and} \qquad 
\mathcal{B}_{-}=h_{--}\mathcal{A}_{+}+h_{+-}\mathcal{A}_{-} \, ,
\end{equation}
where 
\begin{equation}
h_{+-}\!=\!\frac{1\!+\!4\partial_{\nu_{-2}}E\partial_{\nu_{+2}}E}{1\!-\!4\partial_{\nu_{-2}}E\partial_{\nu_{+2}}E} \, ,
\quad 
h_{++}\!=\!\frac{-4\partial_{\nu_{-2}}E}{1\!-\!4\partial_{\nu_{-2}
}E\partial_{\nu_{+2}}E} \, ,
\quad 
h_{--}\!=\!\frac{-4\partial_{\nu_{+2}}E}{1\!-\!4\partial_{\nu_{-2}}E\partial_{\nu_{+2}}E} \, .
\end{equation}
are the entries of a Toeplitz matrix
\begin{equation}
h=
\begin{pmatrix}
h_{+-} & h_{++}
\\
h_{--} & h_{+-}
\end{pmatrix}
\qquad \text{such that} \qquad 
\mathrm{det}(h)=h_{+-}^2-h_{++}h_{--}=1 \,\,  .
\end{equation}
Comparison then leads to the expression
\begin{equation}
(\hat{\star}\mathcal{A})_{\pm}=\pm \mathcal{B}_{\pm}=\pm (h_{\pm\pm}\mathcal{A}_{\mp}+h_{+-}\mathcal{A}_{\pm}) \,\, ,
\end{equation}
which allows us to check the last two properties in \eqref{modified-Hodge-def}, hence providing the sought modified Hodge star operation, and in turn means that the Lax can be written as
\begin{equation}\label{Lax-field-dependent-metric}
\mathfrak{L}_{\pm}=\mathcal{O}\Bigl( \frac{\mathcal{A}_{\pm}\pm z (h_{\pm\pm}\mathcal{A}_{\mp}+h_{+-}\mathcal{A}_{\pm})}{1-z^2}\Bigr)=\frac{A_{\pm}\pm z(h_{\pm\pm}A_{\mp}+h_{+-}A_{\pm})}{1-z^2} \,\, .
\end{equation}
This has the nice result of encoding the information about the auxiliary fields contained in $B$ into the field dependent metric $h$, such that for those cases in which the splitting does not affect $A$, one can rewrite the Lax in terms of a current which only depends on physical fields. This observation seems however to break in those cases where also $A$ is affected by the deformation, since then the currents in \eqref{Lax-field-dependent-metric} still explicitly depend on the auxiliaries. AF-PCM-WZ models represent an exception also in this case, since $A$ and $B$ are deformed via $\mathfrak{J}=-(j+2v)$ and using the rescalings in footnote \ref{footnote-WZ-rescaling} one may get rid of the auxiliary field dependence in both of them by first using the above procedure and then applying to $\mathfrak{J}$ the same field dependent transformation as for AF-PCMs.

\subsection{Maillet brackets}
In previous subsections we have computed the Poisson brackets of the $A$ and $B$ currents characterising the Lax connections of various classes of AFSMs, showing their relation to the generalised BIZZ construction of Theorem \ref{gBIZZ-theorem} and how they fit the classical Yangian structure in Theorem \ref{Yangian-theorem}. It is then also natural to exploit the broad class of Poisson brackets \eqref{Yangian-general-brackets} to study when the $\sigma$-components of the generalised BIZZ Lax \eqref{gBIZZ-Lax},
\begin{equation}\label{gBIZZ-Lax-tau-sigma-coordinates}
\mathfrak{L}_{\sigma}(\sigma,z)=\frac{A_{\sigma}+z \, B_{\tau}}{1-z^2} \, ,
\end{equation}
exhibit the Maillet bracket structure which ensures Hamiltonian integrability. Notice that for the AFSM deformations studied in the previous subsection the Maillet bracket structure is ensured, when the respective undeformed theory also has such structure, by the fact that the deformation only alters the Poisson brackets via $\{L_{\sigma}\rightarrow A_{\sigma} \, , \, L_{\tau}\rightarrow B_{\tau}\}$, as had already been noticed in \cite{Cesaro:2024ipq,Cesaro:2025msv} in relation to AF-SSSM. 

\paragraph{Generalities.} Maillet brackets were first proposed in \cite{MAILLET198654,Maillet:1985ec} as a non-ultralocal extension of Sklyanin's sufficient condition for the involution of conserved charges \cite{Sklyanin:1980ij},
\begin{align}\label{general-Maillet-brackets}
\{ & \mathfrak{L}_{\sigma,1}(\sigma,z),\mathfrak{L}_{\sigma,2}(\sigma',z') \}
\notag \\
=&\!+\![r_{12}(\sigma,z,z'),\mathfrak{L}_{\sigma,1}(\sigma,z)]\delta(\sigma-\sigma')-[r_{21}(\sigma,z',z),\mathfrak{L}_{\sigma,2}(\sigma,z')]\delta(\sigma\!-\!\sigma')
\\
&\!-\!\tfrac{1}{2}\Bigl(\! s_{12}(\sigma,z,z')+s_{12}(\sigma',z,z') \!\Bigr)\partial_{\sigma}\delta(\sigma\!-\!\sigma')-\tfrac{1}{2}\partial_{\sigma}\Bigl( \!r_{12}(\sigma,z,z')-r_{21}(\sigma,z',z) \! \Bigr)\delta(\sigma\!-\!\sigma') \,\, , 
\notag 
\end{align}
with the standard notation
\begin{equation}
X_{1}=X\otimes \mathrm{1}
\qquad \text{and} \qquad 
X_{2}=\mathrm{1}\otimes X \qquad \forall \, X \in \mathfrak{g}\,\, ,
\end{equation}
and the definition
\begin{equation}\label{s12-definition}
s_{12}(\sigma,z,z')=r_{12}(\sigma,z,z')+r_{21}(\sigma,z',z)=s_{21}(\sigma,z',z)\,\, .
\end{equation}
In cases where the matrices $r_{12}$ and $r_{21}$ do not depend on $\sigma$, \eqref{general-Maillet-brackets} reduces to 
\begin{align}\label{sigma-independent-Maillet-bracket}
\{ & \mathfrak{L}_{\sigma,1}(\sigma,z),\mathfrak{L}_{\sigma,2}(\sigma',z') \}
\\
=&[r_{12}(z,z'),\mathfrak{L}_{\sigma,1}(\sigma,z)]\delta(\sigma-\sigma')-[r_{21}(z',z),\mathfrak{L}_{\sigma,2}(\sigma,z')]\delta(\sigma\!-\!\sigma')-s_{12}(z,z')\partial_{\sigma}\delta(\sigma-\sigma') \,\, ,
\notag
\end{align}
which agrees with conventions in \cite{Bielli:2024ach}\footnote{Notice that, for consistency with Section 3.3 and Appendix C of \cite{Bielli:2024ach}, from which we will also borrow some identities, in equation \eqref{general-Maillet-brackets} we use a notation slightly different from the original papers \cite{MAILLET198654,Maillet:1985ec}. Equation 1.5 in \cite{MAILLET198654} can be recovered from \eqref{general-Maillet-brackets} by letting $\mathfrak{L}\rightarrow-\mathfrak{L}$, using \eqref{s12-definition} and identifying
\begin{equation}\label{our-choice-of-r12}
\begin{aligned}
r_{12}(\sigma,z,z'):=s(\sigma,z,z')-r(\sigma,z,z')
\quad , \quad 
r_{21}(\sigma,z',z):=s(\sigma,z,z')+r(\sigma,z,z') \,\, ,
\end{aligned}
\end{equation}
finally using \eqref{def-delta-func-identity-3} on the derivatives of $r_{12}$ and $r_{21}$ to recombine them with the non-ultralocal contribution.} and will be useful in the next subsection.

\subsubsection{The case $C^{AB}(\sigma)\equiv \gamma^{AB}$} 
The brackets \eqref{Yangian-general-brackets} specialised to $C^{AB}(\sigma)\equiv \gamma^{AB}$ encompass various models in previous subsections, which exhibit an underlying Yangian structure encoded in Theorem \ref{Yangian-theorem}. To study compatibility with the Maillet structure \eqref{sigma-independent-Maillet-bracket}, it is easy to combine them with the Lax \eqref{gBIZZ-Lax-tau-sigma-coordinates} to compute
\begin{align}
\{&\mathfrak{L}_{\sigma}{}^{A}(\sigma,z),\mathfrak{L}_{\sigma}{}^{B}(\sigma',z')\}=
\notag\\
=&+\frac{1}{(1-z^2)(1-z'{}^{2})}f^{AB}{}_{C}\Bigl( \bigl(m_{3}(z+z')+m_{5}\bigr)A_{\sigma}{}^{C}(\sigma)+\bigl(m_{1}zz'+m_{6}\bigr)B_{\tau}{}^{C}(\sigma)  \Bigr)\delta(\sigma-\sigma')
\notag \\
&+\frac{1}{(1-z^2)(1-z'{}^{2})}\Bigl(  m_{4}(z+z')+m_{2}zz'+m_{7} \Bigr)\gamma^{AB}\partial_{\sigma}\delta(\sigma-\sigma') \,\, ,
\end{align}
which after contracting with the Lie algebra generators $T_{A}\otimes T_{B}$, defining the quadratic Casimir $C_{12}=\gamma^{AB}T_{A}\otimes T_{B}$, and using the identities
\begin{equation}\label{contraction-identities-1}
\begin{aligned}
\{X^{A},X^{B}\}T_{A}\otimes T_{B}=\{X_{1},X_{2}\} \, ,
\qquad & \qquad 
f^{AB}{}_{C}X^{C}T_{A}\otimes T_{B}=[X_{2},C_{12}] \, ,
\\
[C_{12},X_{1}] &=-[C_{12},X_{2}] \,\, ,
\end{aligned}
\end{equation}
directly taken from Appendix C in \cite{Bielli:2024ach}, leads to
\begin{equation}\label{to-compare-1}
\begin{aligned}
\{\mathfrak{L}_{\sigma,1}(\sigma,z),\mathfrak{L}_{\sigma,2}(\sigma',z')\}=
&+\frac{m_{3}(z+z')+m_{5}}{(1-z^2)(1-z'{}^{2})} \, [ A_{\sigma,2}(\sigma),C_{12}]\delta(\sigma-\sigma')
\\
& +\frac{m_{1}zz'+m_{6}}{(1-z^2)(1-z'{}^{2})} \, [ B_{\tau,2}(\sigma),C_{12}]\delta(\sigma-\sigma')
\\
&+\frac{m_{4}(z+z')+m_{2}zz'+m_{7}}{(1-z^2)(1-z'{}^{2})}\,C_{12}\partial_{\sigma}\delta(\sigma-\sigma') \,\, .
\end{aligned}
\end{equation}
This expression should then be compared with the $\sigma$-independent Maillet bracket \eqref{sigma-independent-Maillet-bracket}, which after substituting the Lax \eqref{gBIZZ-Lax-tau-sigma-coordinates}, using \eqref{contraction-identities-1} and the ansatz
\begin{equation}
r_{12}(z,z')=f(z,z') \,C_{12} \, ,
\qquad \qquad
r_{21}(z',z)=f(z',z) \,C_{12} \,\, ,
\end{equation}
takes the form
\begin{equation}\label{compare-with-1}
\begin{aligned}
\{\mathfrak{L}_{\sigma,1}(\sigma,z),\mathfrak{L}_{\sigma,2}(\sigma',z')\}=&+\Bigl( \frac{f(z,z')}{1-z^2} + \frac{f(z',z)}{1-z'{}^2}\Bigr)[A_{\sigma,2}(\sigma),C_{12}]\delta(\sigma-\sigma')
\\
&+\Bigl( \frac{zf(z,z')}{1-z^2} + \frac{z'f(z',z)}{1-z'{}^2}\Bigr)[B_{\tau,2}(\sigma),C_{12}]\delta(\sigma-\sigma')
\\
&-\Bigl(f(z,z')+f(z',z) \Bigr)C_{12}\partial_{\sigma}\delta(\sigma-\sigma') \,\, .
\end{aligned}
\end{equation}
Comparing \eqref{to-compare-1} and \eqref{compare-with-1} gives three conditions on $f(z,z')$. Two are solved by
\begin{equation}
f(z,z')=\frac{m_{6}-z'(m_{5}+z(m_{3}-m_{1})+z'm_{3})}{(1-z'{}^{2})(z-z')} \,\, .
\end{equation}
The third, coming from the non-ultralocal term, is generically not satisfied by the above choice, but it can be solved by separately setting to zero the coefficients of the resulting polynomial in $z$ and $z'$, which leads to the conditions
\begin{equation}\label{m-coefficients-conditions}
m_{3}-m_{1}=0
\quad , \quad 
m_{5}+m_{7}=0
\quad , \quad 
m_{2}+m_{5}=0
\quad, \quad 
m_{3}+m_{4}-m_{6}=0 \,\, ,
\end{equation}
the first of which had in fact already been found as a requirement for Theorem \ref{Yangian-theorem}. 
\\
\indent
Altogether one concludes that theories characterised by a flat current $A$ and a conserved current $B$ respecting the assumptions of Theorem \ref{gBIZZ-theorem}, and having Poisson brackets \eqref{Yangian-general-brackets} with $C^{AB}(\sigma)\equiv \gamma^{AB}$ and coefficients solving \eqref{m-coefficients-conditions}, exhibit an underlying Yangian symmetry as in Theorem \ref{Yangian-theorem} and a Lax of the form \eqref{gBIZZ-Lax} obeying the Maillet structure \eqref{sigma-independent-Maillet-bracket},
with $r$ and $s$ matrices which can be written as
\begin{equation}
\begin{aligned}
&r_{12}(z,z')=\frac{C_{12}}{z-z'} \, \varphi^{-1}(z') 
\qquad \text{with} \qquad 
\varphi(z')=\frac{(1-z'{}^2)}{m_{6}-m_{5}z'-m_{1}z'{}^2} \,\, ,
\\
&\text{and} \quad \qquad
s_{12}(z,z')=\frac{m_{5}(1+zz')-m_{4}(z+z')}{(1-z^2)(1-z'{}^{2})} C_{12}\,\, .
\end{aligned}
\end{equation}


\paragraph{AF-\{PCM, TD, YB, PCM-WZ\}.}
These models can be addressed at once by writing the coefficients in \eqref{Yangian-general-brackets} in terms of $a,\hay,\kay$, where $a\in \mathbb{R}$ was introduced in \eqref{AF-general-EOM-for-identity} and relates these models to the generalised BIZZ construction via \eqref{A,B-identifications-PCM-TD-YB-PCMWZ}, while $\hay,\kay$ are the relative coefficients between the PCM and WZ terms in \eqref{AF-PCM-WZ}. The brackets take the form 
\begin{align}\label{PCM-TD-YB-Poisson-brackets}
\{ B_{\tau}{}^{A}(\sigma),B_{\tau}{}^{B}(\sigma') \} \!\!=& \tfrac{a}{\hay}f^{AB}{}_{C} \, B_{\tau}{}^{C}(\sigma)\delta(\sigma\!-\!\sigma')\!+\!\tfrac{2\kay a^2}{\hay^2}\gamma^{AB}\partial_{\sigma}\delta(\sigma\!-\!\sigma') \, ,
\notag \\
\{ B_{\tau}{}^{A}(\sigma),A_{\sigma}{}^{B}(\sigma') \} \!\!=& \tfrac{a}{\hay}f^{AB}{}_{C} \, A_{\sigma}{}^{C}(\sigma)\delta(\sigma\!-\!\sigma')\!-\!\tfrac{a^2}{\hay}(1\!+\!\tfrac{\kay^2}{\hay^2})\gamma^{AB} \, \partial_{\sigma}\delta(\sigma\!-\!\sigma') \, ,
\\
\{ A_{\sigma}{}^{A}(\sigma),A_{\sigma}{}^{B}(\sigma') \} \!\!=& f^{AB}{}_{C}\Bigl(\!-\tfrac{2\kay a}{\hay^2}A_{\sigma}{}^{C}\!-\! a(\tfrac{\kay^2}{\hay^3}\!-\!\bigl(1-a)\bigr)B_{\tau}{}^{C}\!\Bigr)\delta(\sigma\!-\!\sigma')+\tfrac{2\kay a^2}{\hay^2}\gamma^{AB}\partial_{\sigma}\delta(\sigma\!-\!\sigma') \, ,
\notag
\end{align}
while the $r$-matrix reads
\begin{equation}
r_{12}(z,z')\!=\!\frac{C_{12}}{z-z'} \, \varphi^{-1}(z') 
\quad \text{with} \quad 
\varphi(z')=\frac{-\hay^3(1-z'{}^2)}{a\Bigl(\!\kay^2-\hay^3(1\!-\!a)-2\kay\hay z'+\hay^2 z'{}^{2}\! \Bigr)} \,\, .
\end{equation}
Notice that the relations \eqref{m-coefficients-conditions} are not satisfied for generic choices of $a,\hay,\kay$, but they are for those which recover the theories under consideration: AF-PCMs and AF-T-dual models correspond to $a=1,\hay=1,\kay=0$, AF-YB models to $a=1-c^2\eta^2,\hay=1,\kay=0$ \cite{Delduc:2013fga,Lacroix:2018njs}, and AF-PCM-WZ models to $a=1,\hay \neq0, \kay \neq 0$.

\paragraph{Relation to another deformation.} As a further sanity check of the above results, here we simply point out how the 2-parameter deformation of the PCM introduced in \cite{Rajeev:1988hq,Balog:1993es}, and whose Yangian structure has been studied in \cite{Itsios:2014vfa}, agrees with the results from the previous subsection. In particular, the Poisson brackets (3.5) in \cite{Itsios:2014vfa} are of the form \eqref{Yangian-general-brackets} with $C^{AB}(\sigma)\equiv \gamma^{AB}$ and the coefficients
\begin{equation}
\begin{aligned}
m_{1}&=-2e^2(1+\rho^2+(1-\rho^2)x)=m_{3}
\quad , \quad 
m_{2}=-8e^2\rho=m_{7}=-m_{5} \, ,
\\
m_{4}&=4e^2(1+\rho^2)
\quad , \quad 
m_{6}=2e^2(1+\rho^2-(1-\rho^2)x) \, ,
\end{aligned}
\end{equation}
meet the conditions in \eqref{m-coefficients-conditions}, as one would expect.

\subsubsection{The case $C^{AB}(\sigma)\neq \gamma^{AB}$} 
Recovering the Maillet structure \eqref{general-Maillet-brackets} for the generalised BIZZ Lax \eqref{gBIZZ-Lax-tau-sigma-coordinates} assuming Poisson brackets \eqref{Yangian-general-brackets} with general $C^{AB}(\sigma)$ is not an easy task, and, as observed in Section \ref{section:Yangians}, one should expect that further conditions would arise. For simplicity, in this subsection we will not carry out a general study, but focus on a symmetric-space-like restriction of  $C^{AB}(\sigma)$ which, as noticed around \eqref{B-C-relation}, already partly emerges by demanding compatibility with an underlying Yangian symmetry.

More specifically, we will take inspiration from the derivation in \cite{Forger:1991ty} and consider the bracket \eqref{Yangian-general-brackets} with the coefficients $m_{1},...,m_{7}$ restricted to
\begin{equation}
\begin{aligned}\label{SSSM-brackets-Lax-section}
\{ B_{\tau}{}^{A}(\sigma),B_{\tau}{}^{B}(\sigma') \} &= m_{1}f^{AB}{}_{C} \, B_{\tau}{}^{C}(\sigma)\delta(\sigma-\sigma') \, ,
\\
\{ B_{\tau}{}^{A}(\sigma),A_{\sigma}{}^{B}(\sigma') \} &= m_{1}f^{AB}{}_{C} \, A_{\sigma}{}^{C}(\sigma)\delta(\sigma-\sigma')-m_{1}^2C^{AB}(\sigma') \, \partial_{\sigma}\delta(\sigma-\sigma') \, ,
\\
\{ A_{\sigma}{}^{A}(\sigma),A_{\sigma}{}^{B}(\sigma') \} &= m_{1}\Lambda f^{AB}{}_{C}B_{\tau}{}^{C}\delta(\sigma-\sigma') \,\, .
\end{aligned}
\end{equation}
We will furthermore assume that $C^{AB}(\sigma')=C^{BA}(\sigma')$ satisfies
\begin{equation}
B_{\tau}{}^{D}(\sigma')\Bigl(\!f_{D}{}^{AC}C_{C}{}^{B}(\sigma')\!-\!f_{D}{}^{BC}C_{C}{}^{A}(\sigma')\!\Bigr)=B_{\tau}{}^{D}(\sigma')f^{AB}{}_{D} \,\,  ,
\end{equation}
which guarantees the assumptions of Theorem \ref{Yangian-theorem},
and demand other conditions needed for the argument in \cite{Forger:1991ty} to work. These are
\begin{equation}\label{C-differential-identity}
\partial_{\mu} C^{AB}(\sigma)=\tfrac{1}{2}A_{\mu}{}^{C}\Bigl(\!f_{C}{}^{AD}C_{D}{}^{B}(\sigma)+f_{C}{}^{BD}C_{D}{}^{A}(\sigma)\!\Bigr) \,\, ,
\end{equation}
and 
\begin{equation}\label{C-structure-constant-identities}
C_{B}{}^{D}f_{DA}{}^{E}=-C_{A}{}^{C}C_{B}{}^{D}f_{CD}{}^{E}-C_{B}{}^{D}f_{AD}{}^{F}C_{F}{}^{E} 
\quad , \quad 
X^{B}(\sigma)C_{B}{}^{A}(\sigma)=X{}^{A}(\sigma) \,\, ,
\end{equation}
for $X$ either the current $A$ or $B$. All of the above are satisfied by SSSM and AF-SSSM, where $C^{AB}(\sigma')=K^{AB}(\sigma')$, the constants take values $m_{1}=2$, $\Lambda=0$ and \eqref{C-differential-identity} is the analogue of \eqref{J-K-differential-relation} while \eqref{C-structure-constant-identities} are the analogues of \eqref{important-K-relation} and \eqref{J-eigenvector-components}.
Let us then start by computing the Poisson brackets of the Lax \eqref{gBIZZ-Lax-tau-sigma-coordinates} using \eqref{SSSM-brackets-Lax-section},
\begin{equation}\label{intermediate-SSSM-Lax}
\begin{aligned}
\{& \mathfrak{L}_{\sigma}{}^{A}(\sigma,z),\mathfrak{L}_{\sigma}{}^{B}(\sigma',z') \}=
\\
=&+\frac{m_{1}}{(1-z^2)(1-z'{}^2)}f^{AB}{}_{C}\Bigl(\! (z+z')A_{\sigma}{}^{C}(\sigma)+(zz'+\Lambda)B_{\tau}{}^{C}(\sigma) \! \Bigr)\delta(\sigma-\sigma')
\\
&-
\frac{m_{1}^2}{(1-z^2)(1-z'{}^2)}\Bigl(\!z'C^{AB}(\sigma)+zC^{AB}(\sigma') \! \Bigr)\partial_{\sigma}\delta(\sigma-\sigma') \,\, ,
\end{aligned}
\end{equation}
where we used $C^{BA}=C^{AB}$ and the identity \eqref{def-delta-func-identity-2} to arrange the second line. There are two useful rewritings that can be applied to \eqref{intermediate-SSSM-Lax}: on the first line one has
\begin{equation}
\begin{aligned}
&\frac{1}{(1-z^2)(1-z'{}^2)}\Bigl(\! (z+z')A_{\sigma}{}^{C}(\sigma)+(zz'+\Lambda)B_{\tau}{}^{C}(\sigma) \! \Bigr)=
\\
=&-\frac{(z'{}^2-\Lambda)}{(z-z')(1-z'{}^2)}\mathfrak{L}_{\sigma}{}^{C}(\sigma,z)+\frac{(z^2-\Lambda)}{(z-z')(1-z^2)}\mathfrak{L}_{\sigma}{}^{C}(\sigma,z') \,\, ,
\end{aligned}
\end{equation}
while on the second 
\begin{align}\label{int-1}
&\Bigl(\!z'C^{AB}(\sigma)+zC^{AB}(\sigma') \! \Bigr)\partial_{\sigma}\delta(\sigma-\sigma')
\\
=&\tfrac{1}{2}(z\!+\!z')\Bigl( \! C^{AB}(\sigma)\!+\!C^{AB}(\sigma') \! \Bigr)\partial_{\sigma}\delta(\sigma-\sigma')-\tfrac{1}{2}(z\!-\!z')\Bigl( \! C^{AB}(\sigma)\!-\!C^{AB}(\sigma') \! \Bigr)\partial_{\sigma}\delta(\sigma-\sigma') 
\notag \\
=& \tfrac{1}{2}(z\!+\!z')\Bigl( \! C^{AB}(\sigma)\!+\!C^{AB}(\sigma') \! \Bigr)\partial_{\sigma}\delta(\sigma-\sigma')+\tfrac{1}{2}\partial_{\sigma}C^{AB}(\sigma)\delta(\sigma-\sigma')
\notag \\
=&\tfrac{1}{2}(z\!+\!z')\Bigl( \! C^{AB}(\sigma)\!+\!C^{AB}(\sigma') \! \Bigr)\partial_{\sigma}\delta(\sigma-\sigma')+\tfrac{1}{4}A_{\sigma}{}^{C}\Bigl(\!f_{C}{}^{AD}C_{D}{}^{B}(\sigma)\!+\!f_{C}{}^{BD}C_{D}{}^{A}(\sigma)\!\Bigr)\delta(\sigma-\sigma') \, ,
\notag 
\end{align}
where in the last two steps we first used the identity \eqref{def-delta-func-identity-3} and then the assumption \eqref{C-differential-identity}. At this point one can proceed by contracting the resulting expression with the Lie algebra generators by again using \eqref{contraction-identities-1} and noting that
\begin{equation}\label{int-2}
[C(\sigma),T_{C}\otimes \mathrm{1}]=f_{C}{}^{AD}C_{D}{}^{B}(\sigma)T_{A}\otimes T_{B}
\quad , \quad 
[C(\sigma),\mathrm{1}\otimes T_{C}]=f_{C}{}^{BD}C_{D}{}^{A}(\sigma)T_{A}\otimes T_{B} \, ,
\end{equation}
which leads to
\begin{equation}\label{SSSM-Maillet-bracket-intermediate}
\begin{aligned}
\{& \mathfrak{L}_{\sigma,1}(\sigma,z),\mathfrak{L}_{\sigma,2}(\sigma',z') \}
\\
=&-m_{1}\Bigl[C_{12}, \frac{(z'{}^2-\Lambda)}{(z-z')(1-z'{}^2)}\mathfrak{L}_{\sigma,1}(\sigma,z)+\frac{(z^2-\Lambda)}{(z-z')(1-z^2)}\mathfrak{L}_{\sigma,2}(\sigma,z') \Bigr]\delta(\sigma-\sigma')
\\
&-\frac{m_{1}^2(z+z')}{2(1-z^2)(1-z'{}^2)}\Bigl(\! C(\sigma)+C(\sigma') \! \Bigr)\partial_{\sigma}\delta(\sigma-\sigma')
\\
&-\frac{m_{1}^2(z-z')}{4(1-z^2)(1-z'{}^2)}[C(\sigma),A_{\sigma,1}(\sigma)+A_{\sigma,2}(\sigma)]\delta(\sigma-\sigma') \,\, .
\end{aligned}
\end{equation}
The next step is combining the first and third lines in the above expression, which first requires establishing a relation between $[C(\sigma),X_{i}]$ and $[C_{12},X_{i}]$ for $i=1,2$,
\begin{equation}\label{relation-C-C12}
[C_{12},X_{1}(\sigma)]=[C(\sigma),X_{1}(\sigma)-X_{2}(\sigma)] \,\, ,
\end{equation}
which can be easily obtained using the assumptions \eqref{C-structure-constant-identities}:
\begin{equation}
\begin{aligned}
[C_{12},X_{1}(\sigma)]=&-X{}^{C}f_{C}{}^{BD}T_{D}\otimes T_{B}
\\
=&X^{C}\gamma^{BE}\Bigl(C_{E}{}^{H}C_{C}{}^{F}f_{HF}{}^{D}+C_{C}{}^{F}f_{EF}{}^{H}C_{H}{}^{D}\Bigr)T_{D}\otimes T_{B}
\\
=&X^{C}\Bigl( f_{C}{}^{AD}C_{D}{}^{B}-f_{C}{}^{BD}C_{D}{}^{A} \Bigr)T_{A}\otimes T_{B}
\\
=&[C(\sigma),X_{1}(\sigma)-X_{2}(\sigma)] \,\, .
\end{aligned}
\end{equation}
Combined with the third relation in \eqref{contraction-identities-1}, equation \eqref{relation-C-C12} implies 
\begin{equation}
[C(\sigma),X_{1}(\sigma)]=[C(\sigma)-C_{12},X_{2}(\sigma)]
\quad \text{and} \quad 
[C(\sigma),X_{2}(\sigma)]=[C(\sigma)-C_{12},X_{1}(\sigma)] \, ,
\end{equation}
which can be readily exploited to obtain 
\begin{equation}\label{C-A-to-substitute-in-brackets}
\begin{aligned}
[C(\sigma)&,A_{\sigma,1}(\sigma)+A_{\sigma,2}(\sigma)]
\\
&=-\frac{(1\!-\!z^2)(1\!-\!z'{}^2)}{(z\!-\!z')}\Bigl[ 2C(\sigma)-C_{12} , \frac{z'}{1-z'{}^2}\mathfrak{L}_{\sigma,1}(\sigma,z) -\frac{z}{1-z^2}\mathfrak{L}_{\sigma,2}(\sigma,z')\Bigr] \,\, .
\end{aligned}
\end{equation}
Together with \eqref{int-1} and \eqref{int-2} this shows that
\begin{equation}\label{C-derivative}
\partial_{\sigma}C(\sigma)=-\frac{(1-z^2)(1-z'{}^2)}{2(z-z')}[2C(\sigma)-C_{12},\frac{z'}{1-z'{}^2}\mathfrak{L}_{\sigma,1}(\sigma,z)-\frac{z}{1-z^2}\mathfrak{L}_{\sigma,2}(\sigma,z')] \,\, .
\end{equation}
Substituting then \eqref{C-A-to-substitute-in-brackets} into \eqref{SSSM-Maillet-bracket-intermediate} and collecting terms one obtains
\begin{equation}\label{SSSM-Maillet-to-compare}
\begin{aligned}
\{ \mathfrak{L}_{\sigma,1}(\sigma,z),\mathfrak{L}_{\sigma,2}(\sigma',z') \}=&
-\frac{m_{1}}{1-z'{}^2}\Bigl( \frac{z'{}^2-\Lambda}{z-z'}+\frac{m_{1}z'}{4} \Bigr)[C_{12},\mathfrak{L}_{\sigma,1}(\sigma,z)]\delta(\sigma-\sigma')
\\
&-\frac{m_{1}}{1-z^2}\Bigl( \frac{z^2-\Lambda}{z-z'}-\frac{m_{1}z}{4} \Bigr)[C_{12},\mathfrak{L}_{\sigma,2}(\sigma,z')]\delta(\sigma-\sigma')
\\
&+\frac{m_{1}^2z'}{2(1-z'{}^2)}[C(\sigma),\mathfrak{L}_{\sigma,1}(\sigma,z)]\delta(\sigma-\sigma')
\\
&-\frac{m_{1}^2z}{2(1-z^2)}[C(\sigma),\mathfrak{L}_{\sigma,2}(\sigma,z')]\delta(\sigma-\sigma')
\\
&-\frac{m_{1}^2(z+z')}{2(1-z^2)(1-z'{}^2)}\Bigl(\! C(\sigma)+C(\sigma') \! \Bigr)\partial_{\sigma}\delta(\sigma-\sigma') \,\, ,
\end{aligned}
\end{equation}
which should now be directly compared with the general Maillet brackets \eqref{general-Maillet-brackets}. To this end one can consider the ansatz for the $r$-matrix
\begin{equation}\label{SSSM-r-matrix-ansatz}
\begin{aligned}
r_{12}(\sigma,z,z')&=f(z,z')C_{12}+g(z,z')C(\sigma) \, ,
\\
r_{21}(\sigma,z',z)&=f(z',z)C_{12}+g(z',z)C(\sigma) \, , 
\\
s_{12}(\sigma,z,z')&=\frac{m_{1}^2(z+z')}{(1-z^2)(1-z'{}^2)}C(\sigma)=s_{21}(\sigma,z',z) \,\, ,
\end{aligned}
\end{equation}
which substituted into \eqref{general-Maillet-brackets} and combined with identity \eqref{def-delta-func-identity-3} on the derivatives of $r$, and again using the relation \eqref{C-derivative}, brings the Maillet brackets to the form
\begin{align}\label{SSSM-Maillet-compare-with}
\{& \mathfrak{L}_{\sigma,1}(\sigma,z),\mathfrak{L}_{\sigma,2}(\sigma',z') \}
\\
=&+\Bigl(f(z,z')-\bigl(g(z,z')-g(z',z) \bigr)\frac{z'(1-z^2)}{4(z-z')}  \Bigr)[C_{12},\mathfrak{L}_{\sigma,1}(\sigma,z)]\delta(\sigma-\sigma')
\notag \\
&+\Bigl(-f(z,z')+\bigl(g(z,z')-g(z',z) \bigr)\frac{z(1-z'{}^2)}{4(z-z')}  \Bigr)[C_{12},\mathfrak{L}_{\sigma,2}(\sigma,z')]\delta(\sigma-\sigma')
\notag \\
&+\Bigl(g(z,z')+\bigl(g(z,z')-g(z',z) \bigr)\frac{z'(1-z^2)}{2(z-z')}  \Bigr)[C(\sigma),\mathfrak{L}_{\sigma,1}(\sigma,z)]\delta(\sigma-\sigma')
\notag \\
&+\Bigl(-g(z,z')-\bigl(g(z,z')-g(z',z) \bigr)\frac{z(1-z'{}^2)}{2(z-z')}  \Bigr)[C(\sigma),\mathfrak{L}_{\sigma,2}(\sigma,z')]\delta(\sigma-\sigma')
\notag \\
&-\frac{m_{1}^2(z+z')}{2(1-z^2)(1-z'{}^2)}\Bigl( \!C(\sigma)+C(\sigma')\!\Bigr)\partial_{\sigma}\delta(\sigma-\sigma') \,\, .
\notag
\end{align}
Comparing \eqref{SSSM-Maillet-to-compare} and \eqref{SSSM-Maillet-compare-with} leads to four conditions on $f$ and $g$, which are solved by
\begin{equation}
f(z,z')\!=\! \frac{2m_{1}(\Lambda\!-\!z'{}^2)(1\!-\!zz')-m_{1}^2z'(z\!-\!z')}{2(1\!-\!z'{}^2)(z\!-\!z')(1\!-\!zz')}
\quad , \quad 
g(z,z')\!=\!\frac{m_{1}^2z'}{(1\!-\!z'{}^2)(1\!-\!zz')} \,\, .
\end{equation}
We conclude by noting that setting $m_{1}=2$, $\Lambda=0$ and choosing $C^{AB}(\sigma)=K^{AB}(\sigma)$, we recover the case of SSSMs and AF-SSSMs leading to
\begin{equation}
f(z,z')=-\frac{2zz'}{(z-z')(1-zz')} 
\qquad \text{and} \qquad 
g(z,z')=\frac{4z'}{(1-z'{}^2)(1-zz')} \,\, ,
\end{equation}
and recalling the ansatz \eqref{SSSM-r-matrix-ansatz} and the relations \eqref{s12-definition} and \eqref{our-choice-of-r12} between our $r$-matrix and the variables in \cite{Forger:1991ty}, allows us to recover the expressions
\begin{equation}
\begin{aligned}
r(\sigma,z,z')&=\frac{2zz'}{(z-z')(1-zz')}C_{12}+\frac{2(z-z')(1+zz')}{(1-z^2)(1-z'{}^2)(1-zz')}C(\sigma) \, ,
\\
s(\sigma,z,z')&=\frac{2(z+z')}{(1-z^2)(1-z'{}^2)}C(\sigma) \,\, .
\end{aligned}
\end{equation}

\section{Conclusions}\label{section:conclusion}

In this article, we have proposed a generalization of one standard construction of the Yangian charges which underlie the integrability of several $2d$ sigma models. Building upon the seminal work of Brezin, Itzykson, Zinn-Justin, and Zuber (BIZZ) \cite{BREZIN1979442}, we have shown in Theorem \ref{gBIZZ-theorem} that an infinite tower of conserved charges can be recursively constructed from a pair of $1$-forms, one which is flat and one which is conserved, and which satisfy appropriate commutator identities. However, the mere existence of a tower of conserved charges does not immediately imply the existence of a Yangian structure. In Theorem \ref{Yangian-theorem}, we have also given sufficient conditions for this tower of charges to generate a classical Yangian algebra, excluding the edge case $\mathfrak{g} = \mathfrak{sl}(2)$.

In addition to giving proofs of these theorems, we have also attempted to illustrate their usefulness. Our view is that this generalized construction provides a unifying framework for studying the Hamiltonian integrability and Yangian structure in a large class of $2d$ integrable sigma models. In particular, our discussion applies rather uniformly both to many well-studied undeformed (``seed'') integrable models (the PCM with or without WZ term, symmetric space sigma models, Yang-Baxter models, and non-Abelian T-dual sigma models) and to the infinite family of deformed models that can be generated by coupling each such seed to auxiliary fields in a prescribed way. Although the precise mechanism by which each of these examples achieves the assumptions of our general theorem varies, in each case we conclude that the corresponding example indeed exhibits an underlying classical Yangian symmetry, and that each model's Lax connection enjoys the Maillet structure which guarantees involution of its charges. Establishing these results for the family of auxiliary field deformed models represents one of the key applications of our general formalism and the primary motivation for carrying out this work.

There remain several interesting directions for future research. First, our discussion has been entirely classical, but an important long-term goal for this research program is to study the integrability of auxiliary field sigma models at the quantum level. Imitating the analogous analysis for the PCM \cite{Luscher:1977uq}, one potential starting point for this direction is to study the nonlocal charges in auxiliary field sigma models in the quantum theory using point-splitting to handle coincident points. The condition that the Yangian charges remain conserved at the quantum level might introduce additional constraints on the interaction function, which would signal that only a subset of the classically-integrable auxiliary field sigma models are also quantum-integrable.

Second, we briefly commented around equation (\ref{Lax-field-dependent-metric}) that part of the deformation of the currents in our construction can be interpreted via a modified Hodge star operation, or field-dependent metric.\footnote{More precisely, this interpretation is available for a certain restricted class of interaction functions.} There has been considerable recent work on studying various related deformations such as $\TT$ from such a geometric perspective \cite{Conti_2019,Conti:2022egv,Morone:2024ffm,Ran:2024vgl,Ran:2025xas,Conti:2019dxg,Tsolakidis:2024wut,Babaei-Aghbolagh:2024hti,Caputa:2020lpa,Brizio:2024arr,Morone:2024sdg}, including in higher dimensions. In settings such as that of \cite{Ferko:2025bhv}, this field-dependent metric seems sufficient to capture the \emph{entire} effect of the auxiliary field deformation, but in some of the more complicated examples studied here, the geometric perspective only seems to suffice for capturing part but not all of the deformation. It would be intriguing to see whether some extension of this geometrical interpretation can encode auxiliary field deformations for broader classes of models, or for more general interaction functions.

Finally, we have restricted attention to $2d$ sigma models, which are only a subclass of integrable models in two dimensions. It would be useful to extend the method of auxiliary fields to other types of $2d$ IQFTs, especially those which are not classically conformally invariant, such as the sine-Gordon model or more general affine Toda theories. One way to approach auxiliary field deformations of the sine-Gordon model may be to realize this theory via Pohlmeyer reduction of the $S^2$ coset sigma model \cite{osti_5772113}, since as we have reviewed, auxiliary field deformations of coset models are comparatively well-understood.

We hope to return to some of these intriguing directions in future work.

\section*{Acknowledgements}

We thank Nicola Baglioni, Mattia Cesàro, Lewis Cole, Ben Hoare, Parita Shah and Alessandro Torrielli for stimulating interactions and collaboration on various topics related to this project. D.B. is grateful to Mattia Cesàro and Alessandro Torrielli for particularly helpful discussions and feedbacks.
D.B. has been supported by a Young Scientist Training (YST) Fellowship from Asia Pacific Center for Theoretical Physics (APCTP) and the Thailand NSRF via PMU-B, grant number B13F680083. C.\,F. is supported by the National Science Foundation under Cooperative Agreement PHY-2019786 (the NSF AI Institute for Artificial Intelligence and Fundamental Interactions).
G.\,T.-M., and M.\,G. have been supported by the Australian Research Council (ARC) Discovery Project DP240101409, and a faculty start-up funding of UQ’s School of Mathematics and Physics.

\newpage

\appendix
\addtocontents{toc}{\protect\setcounter{tocdepth}{1}}

\section{Conventions and useful identities}\label{appendix:A}
In this appendix we collect a series of conventions, useful definitions, and identities.

\subsection{Lie groups and algebras}

\paragraph{Generalities.} The models considered in this work are based on some Lie group $G$ with associated Lie algebra $\mathfrak{g}$, and are characterised by the left-invariant Maurer-Cartan form 
\begin{equation}\label{MC-form-def}
j:=\mathrm{g^{-1}dg}=\mathrm{d}Z^{\mu}j_{\mu}{}^{A}T_{A}
\quad \in \quad  \Omega^{1}(G,\mathfrak{g}) \,\, ,
\end{equation}
with $\{Z^{\mu}\}$ denoting coordinates on $G$, $\{T_{A}\}$ the generators of $\mathfrak{g}$ and $\Omega^{p}(\mathcal{M},\mathfrak{g})$ the set of $\mathfrak{g}$-valued $p$-forms on a manifold $\mathcal{M}$. Einstein summation convention over repeated indices is understood, unless otherwise specified, and $j$ is flat by construction,
\begin{equation}\label{j-flatness}
\mathrm{d}j+\tfrac{1}{2}[j,j]=0
\qquad \longleftrightarrow \qquad 
\partial_{\mu}j_{\nu}{}^{A}-\partial_{\nu}j_{\mu}{}^{A}+j_{\mu}{}^{B}j_{\nu}{}^{C}f_{BC}{}^{A}=0 \,\, .
\end{equation}
The component notation makes explicit use of the structure constants of $\mathfrak{g}$, defined via
\begin{equation}
[T_{A},T_{B}]=f_{AB}{}^{C}T_{C} \,\, .
\end{equation}
The Cartan-Killing form on $\mathfrak{g}$, and its inverse, are respectively denoted by $\gamma_{AB}$ and $\gamma^{AB}$. These are both symmetric in the exchange of their indices, defined through
\begin{equation}\label{Cartan-Killing-def}
\gamma_{AB}:=\mathrm{tr}[T_{A}T_{B}]
\qquad \text{and} \qquad \gamma_{AC}\gamma^{CB}:=\delta_{A}{}^{B} \,\, ,
\end{equation}
and used to raise and lower indices, with $\delta_{A}{}^{B}$ the Kronecker-symbol.
Cyclicity of the trace ensures invariance under the adjoint action of $\mathrm{g}\in G$, meaning that
\begin{equation}\label{Ad-invariance-Killing-form}
\begin{aligned}
\mathrm{tr}[T_{A}T_{B}]=&\mathrm{tr}[\mathrm{Ad}_{\mathrm{g}}(T_{A})\mathrm{Ad}_{\mathrm{g}}(T_{B})] =\mathrm{tr}[\mathrm{Ad}_{\mathrm{g}}^{-1}(T_{A})\mathrm{Ad}_{\mathrm{g}}^{-1}(T_{B})] \,\, ,
\\
\text{with} \quad
\mathrm{Ad}_{\mathrm{g}}(X):=&\mathrm{g} \, X \, \mathrm{g}^{-1} 
\quad  \text{and}   \quad
\mathrm{Ad}_{\mathrm{g}}^{-1}(X):=\mathrm{g}^{-1} \, X \, \mathrm{g}=\mathrm{Ad}_{\mathrm{g}^{-1}}(X) 
\quad \forall \, X\in \mathfrak{g} \,\, .
\end{aligned}
\end{equation}
The components of the Ad-map, and its inverse, will be denoted by
\begin{equation}
\mathrm{Ad}_{\mathrm{g}}(X)=X^{A}W_{A}{}^{B}T_{B}
\qquad \text{and} \qquad
\mathrm{Ad}_{\mathrm{g}}^{-1}(X)=X^{A}(W^{-1})_{A}{}^{B}T_{B}
\qquad 
\forall \, X\in \mathfrak{g} \,\, ,
\end{equation}
such that
\begin{equation}
\mathrm{Ad}_{\mathrm{g}}\circ \mathrm{Ad}_{\mathrm{g}}^{-1}=\mathrm{Ad}_{\mathrm{g}}^{-1}\circ \mathrm{Ad}_{\mathrm{g}}= \mathbb{1} 
\qquad \Rightarrow \qquad
W_{A}{}^{C}(W^{-1})_{C}{}^{B}=(W^{-1})_{A}{}^{C}W_{C}{}^{B}=\delta_{A}{}^{B} \, ,
\end{equation}
and invariance \eqref{Ad-invariance-Killing-form} of the Cartan-Killing form translates into
\begin{equation}
\gamma_{AB} = W_{A}{}^{C}W_{B}{}^{D}\gamma_{CD} = (W^{-1})_{A}{}^{C}(W^{-1})_{B}{}^{D}\gamma_{CD} \,\, ,
\end{equation}
leading to the relations
\begin{equation}
W_{A}{}^{B}=\gamma^{BC}(W^{-1})_{C}{}^{D}\gamma_{DA}
\qquad \text{and} \qquad 
(W^{-1})_{A}{}^{B}=\gamma^{BC}W_{C}{}^{D}\gamma_{DA} \,\, .
\end{equation}
The generators of $\mathfrak{g}$ satisfy the Jacobi identity
\begin{equation}
[T_{A},[T_{B},T_{C}]]+[T_{B},[T_{C},T_{A}]]+[T_{C},[T_{A},T_{B}]]=0 \,\, ,
\end{equation}
which in terms of the structure constants takes the form
\begin{equation}\label{def-Jacobi-components}
f_{BC}{}^{D}f_{AD}{}^{E}+f_{CA}{}^{D}f_{BD}{}^{E}+f_{AB}{}^{D}f_{CD}{}^{E}=0 \,\, .
\end{equation}

\paragraph{$\mathfrak{g}$-valued forms.}
Here we briefly collect the main properties used throughout Section \ref{section:BIZZ-charges}.
To begin, the commutator of two Lie-algebra valued forms on a manifold $\mathcal{M}$, $\alpha\in \Omega^{p}(\mathcal{M},\mathfrak{g})$ and $\beta\in \Omega^{q}(\mathcal{M},\mathfrak{g})$, with $p$ and $q$ their respective degrees, can be defined as
\begin{equation}\label{forms-commutator-def}
\begin{aligned}
[\alpha,\beta] &=\alpha^{A}\wedge\beta^{B}[T_{A},T_{B}]
\\
&=(-1)^{pq}\beta^{B}\wedge \alpha^{A}[T_{A},T_{B}]
\\
&=-(-1)^{pq}\beta^{B}\wedge \alpha^{A}[T_{B},T_{A}]
\\
&=-(-1)^{pq}[\beta,\alpha] \,\, ,
\end{aligned}
\end{equation}
and the Jacobi identity, considering $\gamma\in \Omega^{l}(\mathcal{M},\mathfrak{g})$, reads
\begin{equation}\label{Jacobi-forms}
(-1)^{pl}[\alpha,[\beta,\gamma]]+(-1)^{qp}[\beta,[\gamma,\alpha]]+(-1)^{lq}[\gamma,[\alpha,\beta]]=0 \,\, .
\end{equation}
The definition \eqref{forms-commutator-def}, in particular, implies that
\begin{equation}\label{form-commutator-relations}
\begin{aligned}
[\alpha,\beta]&=[\beta,\alpha] \quad\,\,\,\, \text{for} \quad \alpha,\beta \in \Omega^{1}(\mathcal{M},\mathfrak{g}) \, ,
\\
[\alpha,f]&=-[f,\alpha] 
\quad \text{for} \quad \,\,\,\,\,\,
f\in\Omega^{0}(\mathcal{M},\mathfrak{g}) 
\quad \text{and} \quad 
\alpha\in \Omega^{p}(\mathcal{M},\mathfrak{g}) \,\,\,\, \forall p\in\mathbb{N} \,\, .
\end{aligned}
\end{equation}
In this work we only consider the case where $\mathcal{M} = \Sigma$ is a $2d$ Lorentzian manifold, on which the Hodge star $\star$ is an involution antisymmetric with respect to the wedge product:
\begin{equation}\label{properties-Hodge-star}
\star^2=1 
\qquad \text{and} \qquad 
\alpha \wedge \star \beta =-\star\alpha \wedge \beta \qquad \text{for} \qquad \alpha,\beta \in \Omega^{1}(\Sigma,\mathfrak{g}) \,\, .
\end{equation} 
Combining \eqref{form-commutator-relations} with \eqref{properties-Hodge-star} for $\alpha,\beta\in\Omega^{1}(\Sigma,\mathfrak{g})$ leads to the useful relations
\begin{equation}\label{form-commutator-star-identities}
\begin{aligned}
[\alpha,\star\beta]&=-[\star\alpha,\beta]=-[\beta,\star\alpha] 
\qquad , \qquad 
[\alpha,\star\alpha]=-[\star\alpha,\alpha]=-[\alpha,\star \alpha]=0 \,\, ,
\\
&\qquad \text{and} \quad\qquad [\star \alpha,\star \beta]=-[\star^2\alpha,\beta]=-[\alpha,\beta] \,\, .
\end{aligned}
\end{equation}
The Jacobi identity \eqref{Jacobi-forms} then also implies
\begin{equation}\label{useful-Jacobi-1}
[\alpha,[\alpha,f]]=-\tfrac{1}{2}[f,[\alpha,\alpha]] 
\quad \text{for} \quad \beta=\alpha\in \Omega^{1}(\Sigma,\mathfrak{g})
\quad \text{and} \quad 
f\in \Omega^{0}(\Sigma,\mathfrak{g}) \,\, ,
\end{equation}
as well as
\begin{equation}\label{useful-Jacobi-2}
[\alpha,[\star \beta,f]]=-[\star \beta,[\alpha,f]] \,\, ,
\qquad \text{for} \qquad \alpha,\beta\in \Omega^{1}(\Sigma,\mathfrak{g}) \qquad \text{and} \qquad
[\alpha,\star\beta]\equiv0 \,\, .
\end{equation}

\subsection{Symmetric-spaces $G/H$}

\paragraph{Generalities.} 
For these models the Lie algebra enjoys an orthogonal decomposition
\begin{equation}\label{SSSM-Lie-algebra-commutators}
\begin{aligned}
\mathfrak{g}=\mathfrak{g}^{(0)}\oplus& \mathfrak{g}^{(2)} \,\, , 
\qquad   \text{with commutation relations} 
\\
[\mathfrak{g}^{(0)},\mathfrak{g}^{(0)}] \subseteq \mathfrak{g}^{(0)} \qquad &, \qquad [\mathfrak{g}^{(0)},\mathfrak{g}^{(2)}] \subseteq \mathfrak{g}^{(2)}
\qquad , \qquad [\mathfrak{g}^{(2)},\mathfrak{g}^{(2)}] \subseteq \mathfrak{g}^{(0)} \,\, ,
\end{aligned}
\end{equation}
and the generators $\{T_{A}\}$ associated to each subspace are denoted by
\begin{equation}
\mathfrak{g}^{(0)} =\text{span\{}T_{\dot{a}}\}
\qquad \qquad
\mathfrak{g}^{(2)} =\text{span\{}T_{a}\} \,\, .
\end{equation}
Introducing projectors on each subspace
\begin{equation}\label{projectors}
\mathrm{P}^{(0)}
\qquad \text{and} \qquad \mathrm{P}^{(2)} \qquad \text{such that }\qquad  1=\mathrm{P}^{(0)}+\mathrm{P}^{(2)} \,\, ,  
\end{equation}
the Maurer-Cartan form decomposes as
\begin{equation}\label{j-projections}
j= j^{(0)}+j^{(2)} \,\, , \quad \text{with} \quad  j^{(0)}=j^{(0)\dot{a}}T_{\dot{a}}=\mathrm{P}^{(0)}(j) \quad , \quad j^{(2)}=j^{(2)a}T_{a}=\mathrm{P}^{(2)}(j) \,\, ,
\end{equation}
while the flatness condition \eqref{j-flatness} becomes
\begin{equation} \label{def-sssm-MC-eq-decomposition}
\begin{aligned}
\mathfrak{g}^{(0)}:& \quad 0=  \partial_{\mu}j^{(0)}_{\nu}-\partial_{\nu}j^{(0)}_{\mu}+[j^{(0)}_{\mu},j^{(0)}_{\nu}]+[j^{(2)}_{\mu},j^{(2)}_{\nu}] \, ,
\\
\mathfrak{g}^{(2)}:& \quad 0=\partial_{\mu}j^{(2)}_{\nu}-\partial_{\nu}j^{(2)}_{\mu}+[j^{(0)}_{\mu},j^{(2)}_{\nu}]-[j^{(0)}_{\nu},j^{(2)}_{\mu}] \,\, .
\end{aligned}
\end{equation}
The Cartan-Killing form is assumed to be compatible with the decomposition, meaning that it takes a block-diagonal form $\gamma_{a\dot{a}}=0$, which implies that
\begin{equation}
\delta_{A}{}^{B}=\gamma_{AC}\gamma^{CB} \qquad \text{leads to} \qquad
\delta_{a}{}^{b}=\gamma_{ac}\gamma^{cb}
\qquad \text{and} \qquad
\delta_{\dot{a}}{}^{\dot{b}}=\gamma_{\dot{a}\dot{c}}\gamma^{\dot{c}\dot{b}} \,\, .
\end{equation}
We describe below a set of quantities relevant to the computation of the Poisson brackets reported in Appendices \ref{appendix:C-SSSM} and \ref{appendix:C-AFSSSM}. These were first introduced in \cite{Forger:1991cm} and we will use a notation mixing the original one with the more recent \cite{Klose:2016qfv}.

\paragraph{The $K$-map.}
It is an automorphism of the Lie algebra $K:\mathfrak{g}\rightarrow \mathfrak{g}$ defined as
\begin{equation}\label{K-map-definition}
K=\mathrm{Ad}_{\mathrm{g}} \circ\mathrm{P}^{(2)}\circ \mathrm{Ad}_{\mathrm{g}}^{-1}
\qquad \text{such that} \qquad
K(X)=\mathrm{g} \, \mathrm{P}^{(2)}(\mathrm{g}^{-1}X \mathrm{g}) \, \mathrm{g}^{-1}
\qquad \forall \, X \in \mathfrak{g} \,\, .
\end{equation}
It has components $K(T_{A})=K_{A}{}^{B}T_{B}$, which are symmetric when both indices are raised or lowered using the Cartan-Killing form or its inverse:
\begin{equation}
\begin{aligned}
K_{AB}&=K_{A}{}^{C}\gamma_{BC}=\mathrm{tr}[T_{B}K(T_{A})]=\mathrm{tr}[\mathrm{P}^{(2)}(\mathrm{g}^{-1}T_{B}\mathrm{g}) \, \mathrm{P}^{(2)}(\mathrm{g}^{-1}T_{A}\mathrm{g})] 
\\
& = (W^{-1})_{B}{}^{a}(W^{-1})_{A}{}^{b}\gamma_{ab} = \gamma_{BC}W_{b}{}^{C}(W^{-1})_{A}{}^{b} \,\, .
\end{aligned}
\end{equation}
Raising back the second index using the Cartan-Killing form leads to
\begin{equation}\label{K-in-terms-of-W}
K_{A}{}^{B}=(W^{-1})_{A}{}^{a}W_{a}{}^{B} 
\qquad \text{and} \qquad 
K_{A}{}^{C}K_{C}{}^{B}=K_{A}{}^{B} \,\, ,
\end{equation}
the latter following from the relation $W_{a}{}^{C}(W^{-1})_{C}{}^{b}=\delta_{a}{}^{b}$. An important identity follows from the definition of $K$ and the general Lie-algebraic properties described above:
\begin{equation}\label{important-K-relation}
K_{A}{}^{C}K_{B}{}^{D}f_{CD}{}^{E}=K_{B}{}^{D}f_{AD}{}^{E}-K_{B}{}^{D}f_{AD}{}^{F}K_{F}{}^{E} \,\, .
\end{equation}
This can be derived by noting that
\begin{equation}
K_{A}{}^{C}K_{B}{}^{D}f_{CD}{}^{E} = \bigl( [K(T_{A}),K(T_{B})] \bigr)^{E} \, ,
\end{equation}
and computing
\begin{equation}
\begin{aligned}
[K(T_{A}),K(T_{B})]&=[\mathrm{Ad}_{\mathrm{g}} \, \mathrm{P}^{(2)} \, \mathrm{Ad}_{\mathrm{g}}^{-1}(T_{A}),\mathrm{Ad}_{\mathrm{g}} \, \mathrm{P}^{(2)} \, \mathrm{Ad}_{\mathrm{g}}^{-1}(T_{B})]
\\
&=\mathrm{Ad}_{\mathrm{g}}\Bigl([\mathrm{P}^{(2)} \, \mathrm{Ad}_{\mathrm{g}}^{-1}(T_{A}),\mathrm{P}^{(2)} \, \mathrm{Ad}_{\mathrm{g}}^{-1}(T_{B})] \Bigr)
\\
&= \mathrm{Ad}_{\mathrm{g}} \circ \mathrm{P}^{(0)} \Bigl( [\mathrm{Ad}_{\mathrm{g}}^{-1}(T_{A}),\mathrm{P}^{(2)}\mathrm{Ad}_{\mathrm{g}}^{-1}(T_{B})]\Bigr)
\\
&= \mathrm{Ad}_{\mathrm{g}} \circ \mathrm{P}^{(0)} \circ \mathrm{Ad}_{\mathrm{g}}^{-1} \Bigl( [T_{A},\mathrm{Ad}_{\mathrm{g}} \mathrm{P}^{(2)}\mathrm{Ad}_{\mathrm{g}}^{-1}(T_{B})]\Bigr)
\\
&=[T_{A},K(T_{B})]-K\Bigl([T_{A},K(T_{B})] \Bigr) \,\, .
\end{aligned}
\end{equation}
In the second line we used the commutators \eqref{SSSM-Lie-algebra-commutators} and the projectors \eqref{projectors} to write $[\mathrm{P}^{(2)}(X),\mathrm{P}^{(2)}(Y)]=\mathrm{P}^{(0)}[X,\mathrm{P}^{(2)}(Y)]$ for any $X,Y \in \mathfrak{g}$ and in the fourth line we used again \eqref{projectors} to get rid of $\mathrm{P}^{(0)} $ and obtain two terms involving $K$ only.

\paragraph{The $L$-current.} The Noether current associated to the invariance of symmetric-space sigma models under the global left action of the group takes the form
\begin{equation}\label{SSSM-J-definition}
L = -2 \, \mathrm{Ad}_{\mathrm{g}}(j^{(2)}) = -2 \, j^{(2)a}W_{a}{}^{A}T_{A} \,\, .
\end{equation}
As a result of the second equation in \eqref{def-sssm-MC-eq-decomposition}, it is flat off-shell:
\begin{equation}\label{SSSM-J-flat}
\partial_{\mu}L_{\nu}-\partial_{\nu}L_{\mu}+[L_{\mu},L_{\nu}]=0  \qquad \Longrightarrow \qquad \partial_{\mu}L_{\nu}{}^{A}-\partial_{\nu}L_{\mu}{}^{A}=-L_{\mu}{}^{B}L_{\nu}{}^{C}f_{BC}{}^{A} .
\end{equation}

\paragraph{Relations between $K$ and $L$.} There are three main relations that these objects satisfy:
\begin{itemize}
\item $L$ is an eigenvector of $K$ with eigenvalue 1, namely $K(L)=L$, resulting from the definition of $L$ and $\mathrm{P}^{(2)}(j^{(2)})=j^{(2)}$. In components this reads
\begin{equation}\label{J-eigenvector-components}
K(L)=L^{A}K_{A}{}^{B}T_{B}=L^{B}T_{B}
\qquad \Rightarrow \qquad
L^{A}K_{A}{}^{B}=L^{B} \,\, .
\end{equation}
\item An algebraic relation
\begin{equation}
\mathrm{ad}_{L} = K \circ \mathrm{ad}_{L}+\mathrm{ad}_{L}\circ K \,\, ,
\end{equation}
which can be checked explicitly by acting on any $X\in \mathfrak{g}$ and in components reads
\begin{equation}\label{J-K-algebraic-relation-components}
L^{C}f_{C}{}^{AB}=L^{C}(f_{C}{}^{AD}K_{D}{}^{B}-f_{C}{}^{BD}K_{D}{}^{A}) \,\, .
\end{equation}
\item A differential relation
\begin{equation}\label{J-K-differential-relation}
\partial_{\mu} \bigl( K(X)\bigr)=K(\partial_{\mu}X)+\tfrac{1}{2}K[L_{\mu},X]-\tfrac{1}{2}[L_{\mu},K(X)] \,\, ,
\end{equation}
which can again be checked explicitly upon noting that
\begin{equation}
\partial_{\mu}\mathrm{g} \, \mathrm{g}^{-1} = \mathrm{Ad}_{\mathrm{g}}(j^{(0)}_{\mu})-\tfrac{1}{2}L_{\mu} \quad\,\, , \quad\,\, [\mathrm{Ad}_{\mathrm{g}}(j^{(0)}_{\mu}),K(X)]-K[\mathrm{Ad}_{\mathrm{g}}(j^{(0)}_{\mu}),X]=0 \,\, ,
\end{equation}
and in components leads to
\begin{equation}\label{J-K-differential-relation-components}
\partial_{\mu}K^{AB}=\tfrac{1}{2}L_{\mu}{}^{C}(f_{C}{}^{AD}K_{D}{}^{B}+f_{C}{}^{BD}K_{D}{}^{A}) \,\, .
\end{equation}
\end{itemize}

\newpage

\section{Details about the Yangian}\label{appendix:B}
In this appendix we collect some details of the computation related to Section \ref{section:Yangians}.

\subsection{About Jacobi identities}
Here we briefly comment on the compatibility of the brackets \eqref{Yangian-general-brackets} with Jacobi identities:
\begin{equation}
\begin{aligned}
1)\qquad 0=&\bigl\{B_{\tau}{}^{A}(\sigma),\{B_{\tau}{}^{B}(\sigma'),B_{\tau}{}^{C}(\sigma'') \}\bigr\}+\bigl\{B_{\tau}{}^{B}(\sigma'),\{B_{\tau}{}^{C}(\sigma''),B_{\tau}{}^{A}(\sigma) \}\bigr\}
\\
& \qquad \qquad \qquad 
+\bigl\{B_{\tau}{}^{C}(\sigma''),\{B_{\tau}{}^{A}(\sigma),B_{\tau}{}^{B}(\sigma') \}\bigr\}
\\
2)\qquad 0=&\bigl\{B_{\tau}{}^{A}(\sigma),\{B_{\tau}{}^{B}(\sigma'),A_{\sigma}{}^{C}(\sigma'') \}\bigr\}+\bigl\{B_{\tau}{}^{B}(\sigma'),\{A_{\sigma}{}^{C}(\sigma''),B_{\tau}{}^{A}(\sigma) \}\bigr\}
\\
& \qquad \qquad \qquad 
+\bigl\{A_{\sigma}{}^{C}(\sigma''),\{B_{\tau}{}^{A}(\sigma),B_{\tau}{}^{B}(\sigma') \}\bigr\}
\\
3) \qquad 0=&\bigl\{B_{\tau}{}^{A}(\sigma),\{A_{\sigma}{}^{B}(\sigma'),A_{\sigma}{}^{C}(\sigma'') \}\bigr\}+\bigl\{A_{\sigma}{}^{B}(\sigma'),\{A_{\sigma}{}^{C}(\sigma''),B_{\tau}{}^{A}(\sigma) \}\bigr\}
\\
& \qquad \qquad \qquad 
+\bigl\{A_{\sigma}{}^{C}(\sigma''),\{B_{\tau}{}^{A}(\sigma),A_{\sigma}{}^{B}(\sigma') \}\bigr\}
\\
4)\qquad 0=&\bigl\{A_{\sigma}{}^{A}(\sigma),\{A_{\sigma}{}^{B}(\sigma'),A_{\sigma}{}^{C}(\sigma'') \}\bigr\}+\bigl\{A_{\sigma}{}^{B}(\sigma'),\{A_{\sigma}{}^{C}(\sigma''),A_{\sigma}{}^{A}(\sigma) \}\bigr\}
\\
& \qquad \qquad \qquad 
+\bigl\{A_{\sigma}{}^{C}(\sigma''),\{A_{\sigma}{}^{A}(\sigma),A_{\sigma}{}^{B}(\sigma') \}\bigr\}
\end{aligned}
\end{equation}
The first of these relations is satisfied due to the Jacobi identity \eqref{def-Jacobi-components} of $\mathfrak{g}$, while the remaining three impose, after simplifying with \eqref{def-Jacobi-components} and using the requirement $m_{3}=m_{1}$ found from the Yangian in \eqref{Yangian-theorem-condition}, the following conditions on $C^{AB}(\sigma)$:
\begin{align}
2) \qquad 0=&+m_{4}\{B_{\tau}{}^{A}(\sigma),C^{BC}(\sigma'')\}\partial_{\sigma'}\delta(\sigma'-\sigma'')-m_{4}\{B_{\tau}{}^{B}(\sigma'),C^{AC}(\sigma'')\}\partial_{\sigma}\delta(\sigma-\sigma'')
\notag \\
&+m_{3}m_{4}\Bigl(f^{BC}{}_{D}C^{AD}(\sigma')\delta(\sigma'-\sigma'')\partial_{\sigma}\delta(\sigma-\sigma')
\notag \\
& \qquad \qquad \qquad +f^{CA}{}_{D}C^{BD}(\sigma)\delta(\sigma-\sigma'')\partial_{\sigma'}\delta(\sigma'-\sigma) 
\notag \\
& \qquad \qquad \qquad \qquad \qquad -f^{AB}{}_{D}C^{CD}(\sigma'')\delta(\sigma-\sigma')\partial_{\sigma}\delta(\sigma-\sigma'') \Bigr) \, ,
\\
\notag \\
3) \qquad 0=& +m_{4}\{A_{\sigma}{}^{C}(\sigma''),C^{AB}(\sigma')\}\partial_{\sigma}\delta(\sigma-\sigma')-m_{4}\{A_{\sigma}{}^{B}(\sigma'),C^{AC}(\sigma'')\}\partial_{\sigma}\delta(\sigma-\sigma'')
\notag \\
&+m_{4}m_{5}f^{BC}{}_{D}C^{AD}(\sigma')\delta(\sigma'-\sigma'')\partial_{\sigma}\delta(\sigma-\sigma') \, ,
\\
\notag \\
4) \qquad 0=& m_{4}m_{6}\Bigl(f^{BC}{}_{D}C^{AD}(\sigma)\delta(\sigma'-\sigma'')\partial_{\sigma'}\delta(\sigma'-\sigma)
\notag \\
& \qquad \qquad \qquad +f^{CA}{}_{D}C^{BD}(\sigma')\delta(\sigma''-\sigma)\partial_{\sigma''}\delta(\sigma''-\sigma') 
\notag \\
& \qquad \qquad \qquad \qquad \qquad +f^{AB}{}_{D}C^{CD}(\sigma'')\delta(\sigma-\sigma')\partial_{\sigma}\delta(\sigma-\sigma'') \Bigr) \, .
\end{align}
While these conditions -- when understood in the distributional sense and hence integrated against a test function $\varphi(\sigma,\sigma'\sigma'')$ which vanishes at infinity -- are satisfied for $C^{AB}(\sigma)\equiv \gamma^{AB}$, they should be verified on a case-by-case basis for non-trivial choices. In Section \ref{section:application-to-AFSM} all the models under consideration fall in the trivial case, except for SSSMs and their auxiliary field deformations. For this class of models $C^{AB}(\sigma)\equiv K^{AB}(\sigma)$, defined in \eqref{K-map-definition}, and the conditions arising from 3) and 4) are trivially satisfied due to $m_{5}=m_{6}=0$ and Poisson commutativity of the current $A_{\sigma}^{A}$ with $K^{AB}$. The condition resulting from 2) can be verified, again in the distributional sense, by using the non-vanishing Poisson brackets between $B_{\tau}{}^{A}$ and $K^{AB}$ in \eqref{AF-SSSM-brackets-summary}, together with $m_{3}=m_{4}=2$.

\subsection{About relevant tensorial structures}
As discussed around equation \eqref{level1-bracket-tensorial-structures-def}, recovering the Serre relations \eqref{Serre-relations} highlights the existence of two main tensorial structures which exhibit important properties:
\begin{equation}\label{tensorial-structrures-appendix}
C_{E_{1}..E_{n}}{}^{AB}:=X_{E_{1}...E_{n}}{}^{P}f_{P}{}^{AB}
\qquad \text{and} \qquad
T_{CFG}{}^{AB}:=f_{CD}{}^{[A}f_{EF}{}^{B]}f^{DE}{}_{G}
 \,\, .
\end{equation}
\paragraph{$C$-structures are central.} The centrality condition \eqref{central-tensors-def} follows in a few steps,
\begin{align}\label{vanishing-central-terms}
&f_{B}{}^{[CD}C_{E_{1}...E_{n}}{}^{A]B}=
\notag \\
=&f_{B}{}^{[CD}f_{P}{}^{A]B}X_{E_{1}...E_{n}}{}^{P}
\notag \\
=& \tfrac{1}{3}\bigl( f_{B}{}^{CD}f_{P}{}^{AB}+f_{B}{}^{AC}f_{P}{}^{DB}+f_{B}{}^{DA}f_{P}{}^{CB} \bigr)X_{E_{1}...E_{n}}{}^{P}
\notag \\
=&\tfrac{1}{3}\bigl( -\gamma^{AH}\gamma^{CL}f_{PH}{}^{B}f_{LB}{}^{D}+f_{B}{}^{AC}f_{P}{}^{DB}+f_{B}{}^{DA}f_{P}{}^{CB} \bigr)X_{E_{1}...E_{n}}{}^{P}
\notag \\
=&\tfrac{1}{3}\bigl( +\gamma^{AH}\gamma^{CL}(f_{HL}{}^{B}f_{PB}{}^{D}+f_{LP}{}^{B}f_{HB}{}^{D})+f_{B}{}^{AC}f_{P}{}^{DB}+f_{B}{}^{DA}f_{P}{}^{CB} \bigr)X_{E_{1}...E_{n}}{}^{P}
\notag \\
=&0 \,\, ,
\end{align}
where in the third line we explicitly wrote the 6 antisymmetrised contributions, which in our conventions include an overall normalisation factor that in this case reads $\tfrac{1}{3!}$, and combined them exploiting antisymmetry of the structure constants. In the fourth line we simply lowered some of the indices on the first term so as to more easily exploit the Jacobi identity \eqref{def-Jacobi-components} in the successive step. In the penultimate line the four terms cancel pairwise, hence giving the desired result.

\paragraph{Symmetry of $T$-structures.}
The $T_{CFG}{}^{AB}$ tensor \eqref{tensorial-structrures-appendix} is fully symmetric, up to central contributions, in the three lower indices. This can also be checked explicitly, starting from the exchange in the first two lower indices:
\begin{equation}\label{T-sym-first-two-indices}
\begin{aligned}
T_{FCG}{}^{AB}=&f_{FD}{}^{[A}f_{EC}{}^{B]}f^{DE}{}_{G}
\\
=&-f_{FE}{}^{[A}f_{DC}{}^{B]}f^{DE}{}_{G}
\\
=&f_{CD}{}^{[A}f_{EF}{}^{B]}f^{DE}{}_{G}
\\
=&T_{CFG}{}^{AB} \,\, .
\end{aligned}
\end{equation}
The exchange in the last two lower indices is only slightly more involved,
\begin{equation}\label{T-sym-last-two-indices}
\begin{aligned}
T_{CGF}{}^{AB}=&f_{CD}{}^{[A}f_{EG}{}^{B]}f^{DE}{}_{F}
\\
=&f_{DF}{}^{E}f_{GE}{}^{[B}f_{C}{}^{D|A]}
\\
=&-\bigl( f_{FG}{}^{E}f_{DE}{}^{[B}+f_{GD}{}^{E}f_{FE}{}^{[B} \bigr)f_{C}{}^{D|A]}
\\
=&T_{CFG}{}^{AB}-\tfrac{1}{2}f_{FG}{}^{E}\bigl( \gamma^{AH}f_{CH}{}^{D}f_{ED}{}^{B}-f_{C}{}^{BD}f_{ED}{}^{A} \bigr)
\\
=&T_{CFG}{}^{AB}+\tfrac{1}{2}f_{FG}{}^{E}\bigl( f^{A}{}_{E}{}^{D}f_{CD}{}^{B}+f_{EC}{}^{D}f^{A}{}_{D}{}^{B}+f_{C}{}^{BD}f_{ED}{}^{A}\bigr)
\\
=&T_{CFG}{}^{AB}+X_{CFG}{}^{D}f_{D}{}^{AB} \,\, ,
\end{aligned}
\end{equation}
where we defined $X_{CFG}{}^{D}=-\tfrac{1}{2}f_{FG}{}^{E}f_{EC}{}^{D}$ for the central contribution and used the Jacobi identity in going from the second to the third line and from the fourth to fifth. The exchange in the first and third lower index then follows from the above results as
\begin{equation}\label{T-sym-first-and-last-indices}
\begin{aligned}
T_{GFC}{}^{AB}=&T_{GCF}{}^{AB}-\tfrac{1}{2}f_{CF}{}^{E}f_{EG}{}^{D}f_{D}{}^{AB}
\\
=& T_{CGF}{}^{AB}-\tfrac{1}{2}f_{CF}{}^{E}f_{EG}{}^{D}f_{D}{}^{AB}
\\
=& T_{CFG}{}^{AB}-\tfrac{1}{2}f_{FG}{}^{E}f_{EC}{}^{D}f_{D}{}^{AB}-\tfrac{1}{2}f_{CF}{}^{E}f_{EG}{}^{D}f_{D}{}^{AB}
\\
=& T_{CFG}{}^{AB}+Y_{CFG}{}^{D}f_{D}{}^{AB} \,\, ,
\end{aligned}
\end{equation}
where we defined $Y_{CFG}{}^{D}=f_{F[G}{}^{E}f_{C]E}{}^{D}$ for the central contribution.

\paragraph{Contraction of $T$-structures.} As noted around equation \eqref{T-tensor-contraction}, when its lower indices are contracted with three copies of the same object, $T$ becomes automatically antisymmetric in the upper indices and the emergence of \eqref{Serre-structure-T-tensor-contraction} can be seen in two steps. The first is contraction with a further copy of the structure constants as
\begin{equation}\label{T-contraction-intermediate}
\begin{aligned}
f_{D}&{}^{AB}T_{PQR}{}^{CD}O^{P}O^{Q}O^{R}=
\\
=&f_{D}{}^{AB}f^{C}{}_{PE}f^{D}{}_{FQ}f^{EF}{}_{R}O^{P}O^{Q}O^{R}
\\
=&-f^{ABD}f_{FDQ}f^{C}{}_{PE}f^{EF}{}_{R}O^{P}O^{Q}O^{R}
\\
=& \bigl( f^{B}{}_{F}{}^{D}f^{A}{}_{DQ}+f_{F}{}^{AD}f^{B}{}_{DQ} \bigr) f^{C}{}_{PE}f^{EF}{}_{R}O^{P}O^{Q}O^{R}
\\
=&\bigl( -f^{BDF}f^{E}{}_{FR}f^{A}{}_{DQ}+f^{ADF}f^{E}{}_{FR}f^{B}{}_{DQ} \bigr) f^{C}{}_{PE}O^{P}O^{Q}O^{R}
\\
=&\bigl( f^{DEF}f^{B}{}_{FR}f^{A}{}_{DQ}+f^{EBF}f^{D}{}_{FR}f^{A}{}_{DQ}+f^{ADF}f^{E}{}_{FR}f^{B}{}_{DQ} \bigr) f^{C}{}_{PE}O^{P}O^{Q}O^{R}
\\
=& -f^{A}{}_{PI}f^{B}{}_{QJ}f^{C}{}_{RK}f^{IJK}O^{P}O^{Q}O^{R} - S_{PQR}{}^{ABC}O^{P}O^{Q}O^{R} \,\, ,
\end{aligned}
\end{equation}
where we used the Jacobi identity \eqref{def-Jacobi-components} in the third and fifth line, defining in the end
\begin{equation}
S_{PQR}{}^{ABC}:=\bigl(f^{A}{}_{DQ}f^{B}{}_{EF}f^{C}{}_{P}{}^{E}f^{DF}{}_{R}-f^{B}{}_{DQ}f^{C}{}_{PE}f^{ADF}f^{E}{}_{FR} \bigr) \,\, .
\end{equation}
The second step consists of antisymmetrising the upper $ABC$ indices in \eqref{T-contraction-intermediate}; this does not affect the first term, which is already antisymmetric by construction, but causes the $S$-contribution to vanish, ensuring the recovery of \eqref{Serre-structure-T-tensor-contraction}:
\begin{equation}\label{T-contraction-antisymmetrisation}
\begin{aligned}
S_{PQR}&{}^{[ABC]}O^{P}O^{Q}O^{R}=
\\
=&\bigl(f^{[A}{}_{DQ}f^{B}{}_{EF}f^{C]}{}_{P}{}^{E}f^{DF}{}_{R}-f^{[B}{}_{DQ}f^{C}{}_{PE}f^{A]DF}f^{E}{}_{FR} \bigr)O^{P}O^{Q}O^{R}
\\
=&\bigl(f^{[A}{}_{DQ}f^{B}{}_{EF}f^{C]}{}_{P}{}^{E}f^{DF}{}_{R}+f^{[C}{}_{DQ}f^{B}{}_{PE}f^{A]DF}f^{E}{}_{FR} \bigr)O^{P}O^{Q}O^{R}
\\
=&\bigl(f^{[A}{}_{DQ}f^{B}{}_{EF}f^{C]}{}_{P}{}^{E}f^{DF}{}_{R}-f^{[C}{}_{DQ}f^{A}{}_{PE}f^{B]DF}f^{E}{}_{FR} \bigr)O^{P}O^{Q}O^{R}
\\
=&\bigl(f^{[A}{}_{DQ}f^{B}{}_{EF}f^{C]}{}_{P}{}^{E}f^{DF}{}_{R}-f^{[C}{}_{EP}f^{A}{}_{QD}f^{B]EF}f^{D}{}_{FR} \bigr)O^{P}O^{Q}O^{R}
\\
=&\bigl(f^{[A}{}_{DQ}f^{B}{}_{EF}f^{C]}{}_{P}{}^{E}f^{DF}{}_{R}-f^{[C}{}_{PE}f^{A}{}_{DQ}f^{B]EF}f^{D}{}_{FR} \bigr)O^{P}O^{Q}O^{R}
\\
=&\bigl(f^{[A}{}_{DQ}f^{B}{}_{EF}f^{C]}{}_{P}{}^{E}f^{DF}{}_{R}-f^{[C|}{}_{P}{}^{E}f^{A}{}_{DQ}f^{B]}{}_{EF}f^{DF}{}_{R} \bigr)O^{P}O^{Q}O^{R}
\\
=& 0 \,\, .
\end{aligned}
\end{equation}
Here we have exclusively worked on the second set of contracted structure constants. In the second and third line we used antisymmetrisation over $ABC$ to swap them, in the fourth line we relabelled a few indices and used symmetry of the $O$'s, in the fifth line we swapped the lower indices $PE$ and $DQ$, finally rising and lowering $E$ and $F$ to obtain cancellation between the two terms.

\subsection{About the proof of Theorem \ref{Yangian-theorem}}
Here we collect technical steps underlying the proof of Theorem \ref{Yangian-theorem}.

\paragraph{Contributions in $\{Q^{(0)A}(\sigma),Q^{(1)B}(\sigma')\}$.} \quad
Using \eqref{charges-with-boundaries} and \eqref{Yangian-general-brackets} one finds
\begin{align}
\{&Q^{(0)A}(\sigma),Q^{(1)B}(\sigma')\}=
\notag \\
=&+\beta s \!\!\int_{-\sigma_{1}}^{+\sigma_{2}}\!\!\!\!\!\!\!\!\!\!\mathrm{d}\sigma
\!\!\int_{-\sigma_{5}}^{+\sigma_{6}}\!\!\!\!\!\!\!\!\!\!\mathrm{d}\sigma' \,\{B_{\tau}{}^{A}(\sigma),A_{\sigma}{}^{B}(\sigma')\}
+ \tfrac{1}{2} \!\! \int_{-\sigma_{1}}^{+\sigma_{2}}\!\!\!\!\!\!\!\!\!\!\mathrm{d}\sigma \!\! \int_{-\sigma_{3}}^{+\sigma_{4}}\!\!\!\!\!\!\!\!\!\!\mathrm{d}\sigma' \, \{B_{\tau}{}^{A}(\sigma),[B_{\tau}{}(\sigma'),\chi^{(0)}(\sigma')]^{B}\}
\notag \\
=&+\beta s \!\! \int_{-\sigma_{1}}^{+\sigma_{2}}\!\!\!\!\!\!\!\!\!\!\mathrm{d}\sigma
\!\! \int_{-\sigma_{5}}^{+\sigma_{6}}\!\!\!\!\!\!\!\!\!\!\mathrm{d}\sigma' \,\{B_{\tau}{}^{A}(\sigma),A_{\sigma}{}^{B}(\sigma')\}
+ \tfrac{1}{2}f_{CD}{}^{B}\!\!\! \int_{-\sigma_{1}}^{+\sigma_{2}}\!\!\!\!\!\!\!\!\!\!\mathrm{d}\sigma \!\! \int_{-\sigma_{3}}^{+\sigma_{4}}\!\!\!\!\!\!\!\!\!\!\mathrm{d}\sigma' \, \{B_{\tau}{}^{A}(\sigma),B_{\tau}{}^{C}(\sigma')\chi^{(0)D}(\sigma')\}
\notag \\
=&+\beta s \!\! \int_{-\sigma_{1}}^{+\sigma_{2}}\!\!\!\!\!\!\!\!\!\!\mathrm{d}\sigma
\!\! \int_{-\sigma_{5}}^{+\sigma_{6}}\!\!\!\!\!\!\!\!\!\!\mathrm{d}\sigma' \,\{B_{\tau}{}^{A}(\sigma),A_{\sigma}{}^{B}(\sigma')\}
\notag \\
& \qquad \qquad \!+\!\tfrac{\beta s}{2}f_{CD}{}^{B}\!\!\!\int_{-\sigma_{1}}^{+\sigma_{2}}\!\!\!\!\!\!\!\!\!\!\mathrm{d}\sigma \!\! \int_{-\sigma_{3}}^{+\sigma_{4}}\!\!\!\!\!\!\!\!\!\!\mathrm{d}\sigma'\mathrm{d}\sigma''  \theta(\sigma'\!-\!\sigma'') \{B_{\tau}{}^{A}(\sigma),B_{\tau}{}^{C}(\sigma')B_{\tau}{}^{D}(\sigma'')\}
\notag \\
=&+\beta s \!\! \int_{-\sigma_{1}}^{+\sigma_{2}}\!\!\!\!\!\!\!\!\!\!\mathrm{d}\sigma
\!\! \int_{-\sigma_{5}}^{+\sigma_{6}}\!\!\!\!\!\!\!\!\!\!\mathrm{d}\sigma' \,\{B_{\tau}{}^{A}(\sigma),A_{\sigma}{}^{B}(\sigma')\}
\notag \\
&\qquad \qquad \!+\!\tfrac{\beta s}{2}f_{CD}{}^{B}\!\!\!\int_{-\sigma_{1}}^{+\sigma_{2}}\!\!\!\!\!\!\!\!\!\!\mathrm{d}\sigma \!\! \int_{-\sigma_{3}}^{+\sigma_{4}}\!\!\!\!\!\!\!\!\!\!\mathrm{d}\sigma'\mathrm{d}\sigma''  \epsilon(\sigma'\!-\!\sigma'') B_{\tau}{}^{D}(\sigma'') \{B_{\tau}{}^{A}(\sigma),B_{\tau}{}^{C}(\sigma')\}
\notag \\
=&+m_{3}\beta s \!\! \int_{-\sigma_{1}}^{+\sigma_{2}}\!\!\!\!\!\!\!\!\!\!\mathrm{d}\sigma
\!\! \int_{-\sigma_{5}}^{+\sigma_{6}}\!\!\!\!\!\!\!\!\!\!\mathrm{d}\sigma' \,f^{AB}{}_{C}A_{\sigma}{}^{C}(\sigma)\delta(\sigma \!-\! \sigma')+m_{4} \beta s\!\!\! \int_{-\sigma_{1}}^{+\sigma_{2}}\!\!\!\!\!\!\!\!\!\!\mathrm{d}\sigma
\!\! \int_{-\sigma_{5}}^{+\sigma_{6}}\!\!\!\!\!\!\!\!\!\!\mathrm{d}\sigma' \,C^{AB}(\sigma')\partial_{\sigma}\delta(\sigma-\sigma')
\notag \\
&\qquad \qquad +\tfrac{m_{1}\beta s}{2}f_{CD}{}^{B}f^{AC}{}_{E}\!\!\!\int_{-\sigma_{1}}^{+\sigma_{2}}\!\!\!\!\!\!\!\!\!\!\mathrm{d}\sigma \!\! \int_{-\sigma_{3}}^{+\sigma_{4}}\!\!\!\!\!\!\!\!\!\!\mathrm{d}\sigma'\mathrm{d}\sigma''  \epsilon(\sigma'\!-\!\sigma'') B_{\tau}{}^{D}(\sigma'')B_{\tau}{}^{E}(\sigma)\delta(\sigma-\sigma')
\notag \\
&\qquad \qquad +\tfrac{m_{2}\beta s}{2}f_{CD}{}^{B}\!\!\!\int_{-\sigma_{1}}^{+\sigma_{2}}\!\!\!\!\!\!\!\!\!\!\mathrm{d}\sigma \!\! \int_{-\sigma_{3}}^{+\sigma_{4}}\!\!\!\!\!\!\!\!\!\!\mathrm{d}\sigma'\mathrm{d}\sigma''  \epsilon(\sigma'\!-\!\sigma'') B_{\tau}{}^{D}(\sigma'')\gamma^{AC}\partial_{\sigma}\delta(\sigma-\sigma')
\notag \\
=& Y_{1}^{AB}+Y_{2}^{AB}+X_{1}^{AB}+X_{2}^{AB} \, ,
\label{Q0-Q1-bracket-intermediate}
\end{align}
where in the first step we expanded the commutator in the second integral and in the second step we substituted $\chi^{(0)}$ using the relation \eqref{Yangian-B-chi0-relation}. In the third one we used Leibniz rule inside the second integral and reassembled the resulting two contributions by relabelling indices on the structure constants, the integration variable $\sigma'\leftrightarrow\sigma''$, and defining
\begin{equation}\label{epsilon-def}
\epsilon(\sigma-\sigma'):=\theta(\sigma-\sigma')-\theta(\sigma'-\sigma) \,\, .
\end{equation}
In the fourth step we used again the brackets \eqref{Yangian-general-brackets} and finally defined the quantities $Y_{1}^{AB},Y_{2}^{AB}$ and $X_{1}^{AB},X_{2}^{AB}$ given in \eqref{Q0Q1-first-step}.

\paragraph{Rearrangement of $Y_{2}^{AB}$.} Starting from the definition \eqref{Q0Q1-first-step} one finds
\begin{equation}\label{manipulation-Y2AB}
\begin{aligned}
Y_{2}^{AB}=&\tfrac{m_{1}\beta s}{2}f_{CD}{}^{B}f^{AC}{}_{E}\!\!\!\int_{-\sigma_{1}}^{+\sigma_{2}}\!\!\!\!\!\!\!\!\!\!\mathrm{d}\sigma \!\! \int_{-\sigma_{3}}^{+\sigma_{4}}\!\!\!\!\!\!\!\!\!\!\mathrm{d}\sigma'\mathrm{d}\sigma''  \epsilon(\sigma'-\sigma'')B_{\tau}{}^{D}(\sigma'')B_{\tau}{}^{E}(\sigma)\delta(\sigma-\sigma')
\\
=&\tfrac{m_{1}\beta s}{2}f_{CD}{}^{B}f^{AC}{}_{E}\!\!\! \int_{-\sigma_{3}}^{+\sigma_{4}}\!\!\!\!\!\!\!\!\!\!\mathrm{d}\sigma'\mathrm{d}\sigma''  \Bigl( \! \theta(\sigma'\!-\!\sigma'')\!-\!\theta(\sigma'' \!-\!\sigma') \!\Bigr)B_{\tau}{}^{D}(\sigma'')B_{\tau}{}^{E}(\sigma')
\\
=&\tfrac{m_{1}\beta s}{2}\gamma^{AF} \Bigl( f_{EF}{}^{D}f_{CD}{}^{B}-f_{CF}{}^{D}f_{ED}{}^{B} \Bigr) \!\!\! \int_{-\sigma_{3}}^{+\sigma_{4}}\!\!\!\!\!\!\!\!\!\!\mathrm{d}\sigma'\mathrm{d}\sigma''\theta(\sigma'-\sigma'') B_{\tau}{}^{E}(\sigma'')B_{\tau}{}^{C}(\sigma')
\\
=&-\tfrac{m_{1}}{2}\gamma^{AF} f_{CE}{}^{D}f_{FD}{}^{B} \int_{-\sigma_{3}}^{+\sigma_{4}}\!\!\!\!\!\!\!\!\!\!\mathrm{d}\sigma B_{\tau}{}^{C}(\sigma)\chi^{(0)E}(\sigma)
\\
=& + m_{1}\beta f^{AB}{}_{C} \int_{-\sigma_{3}}^{+\sigma_{4}}\!\!\!\!\!\!\!\!\!\!\mathrm{d}\sigma \, \tfrac{1}{2\beta}[B_{\tau}(\sigma),\chi^{(0)}(\sigma)]^{C} \,\, .
\end{aligned}
\end{equation}
In the first line we used \eqref{epsilon-def} to write $\epsilon(\sigma'-\sigma'')$ explicitly and chose the boundaries to satisfy $\sigma_{4}<\sigma_{2}$ and $\sigma_{3}<\sigma_{1}$, so as to get rid of the integration from $-\sigma_{1}$ to $+\sigma_{2}$ using the $\delta$-function, as needed to reproduce the second term of $Q^{(1)}$ in \eqref{charges-with-boundaries}. In the second line we relabelled the integration variables and the indices on the structure constants so as to combine terms inside the integral and in the third step we exploited the relation \eqref{Yangian-B-chi0-relation} and Jacobi identity \eqref{Jacobi-forms}. In the final line we then expressed the integrand as the desired commutator in \eqref{charges-with-boundaries}.

\paragraph{Vanishing of $X_{1}^{AB}$ and $X_{2}^{AB}$.}
Starting from $X_{1}^{AB}$, one finds
\begin{equation}\label{vanishing-X1AB}
\begin{aligned}
X_{1}^{AB}=&m_{4}\beta s \int_{-\sigma_{1}}^{+\sigma_{2}}\!\!\!\!\!\!\mathrm{d}\sigma
\int_{-\sigma_{5}}^{+\sigma_{6}}\!\!\!\!\!\!\mathrm{d}\sigma' \,C^{AB}(\sigma')\partial_{\sigma}\delta(\sigma-\sigma') 
\\
=&m_{4}\beta s \int_{-\sigma_{5}}^{+\sigma_{6}}\!\!\!\!\!\!\mathrm{d}\sigma' \,C^{AB}(\sigma') \Bigl( \delta(+\sigma_{2}-\sigma')-\delta(-\sigma_{1}-\sigma') \Bigr)
\\
=&m_{4}\beta s \Bigl( C^{AB}(\sigma_{2})\theta(\sigma_{6}-\sigma_{2})-C^{AB}(-\sigma_{1})\theta(\sigma_{5}-\sigma_{1}) \Bigr) \,\, ,
\end{aligned}
\end{equation}
where in the first line we explicitly wrote down the boundary terms arising from the derivative of the $\delta$-function and in the second line we used the fact that the resulting contributions $\delta(+\sigma_{2}-\sigma')$ and $\delta(-\sigma_{1}-\sigma')$ can only be non-vanishing if the domain $\sigma' \in [-\sigma_{5},+\sigma_{6}]$ respectively contains $\sigma_{2}$ and $-\sigma_{1}$, namely only if $\sigma_{6}\geq \sigma_{2}$ and $\sigma_{5}\geq \sigma_{1}$. These conditions can be rewritten in terms of the $\theta$-functions on the last line, which show that the boundary contribution $X_{1}^{AB}$ is generically non-vanishing unless one chooses $\sigma_{6}<\sigma_{2}$ and $\sigma_{5}<\sigma_{1}$. Such a choice is coherent with the requirements \eqref{boundary-order-intermediate} found from $Y_{1}^{AB},Y_{2}^{AB}$ and in agreement with \cite{Klose:2016qfv}, where the above term was first discussed. This ensures the correct vanishing of $X_{1}^{AB}$, such that recovery of \eqref{Yangian-general-level0-and-level1-brackets} is now subject to the vanishing of $X_{2}^{AB}$. Looking explicitly at this term from \eqref{Q0Q1-first-step} one then finds
\begin{equation}\label{vanishing-X2AB}
\begin{aligned}
\!X_{2}^{AB}=&\tfrac{m_{2}\beta s}{2}f_{CD}{}^{B}\!\!\!\int_{-\sigma_{1}}^{+\sigma_{2}}\!\!\!\!\!\!\!\!\!\!\mathrm{d}\sigma \!\! \int_{-\sigma_{3}}^{+\sigma_{4}}\!\!\!\!\!\!\!\!\!\!\mathrm{d}\sigma'\mathrm{d}\sigma''  \epsilon(\sigma'\!-\!\sigma'') B_{\tau}{}^{D}(\sigma'')\gamma^{AC}\partial_{\sigma}\delta(\sigma-\sigma')
\\
=&-\tfrac{m_{2}\beta s}{2}f^{AB}{}_{D}\!\!\!\int_{-\sigma_{3}}^{+\sigma_{4}}\!\!\!\!\!\!\!\!\!\!\mathrm{d}\sigma'\mathrm{d}\sigma''  \epsilon(\sigma'\!-\!\sigma'') B_{\tau}{}^{D}(\sigma'')\Bigl(\delta(+\sigma_{2}-\sigma')-\delta(-\sigma_{1}-\sigma')\Bigr)
\\
=&-\tfrac{m_{2}\beta s}{2}f^{AB}{}_{D}\!\!\!\int_{-\sigma_{3}}^{+\sigma_{4}}\!\!\!\!\!\!\!\!\!\!\mathrm{d}\sigma'' B_{\tau}{}^{D}(\sigma'')\Bigl(\epsilon(\sigma_{2}\!-\!\sigma'')\theta(\sigma_{4}\!-\!\sigma_{2})-\epsilon(-\sigma_{1}\!-\!\sigma'')\theta(\sigma_{3}\!-\!\sigma_{1})\Bigr) \,\, ,
\end{aligned}
\end{equation}
after having followed exactly the same reasoning as for $X_{1}^{AB}$ above, with the different domain $\sigma'' \in [-\sigma_{3},+\sigma_{4}]$. Now the boundary terms can only contribute if $\sigma_{4}\geq \sigma_{2}$ and $\sigma_{3}\geq \sigma_{1}$, such that making the opposite choice, namely $\sigma_{4}<\sigma_{2}$ and $\sigma_{3}<\sigma_{1}$, ensures the vanishing of the last line in agreement with the conditions \eqref{boundary-order-intermediate} arising from $Y_{1}^{AB},Y_{2}^{AB}$ .

\paragraph{Evaluation of $U^{AB}$.} This computation can be done explicitly by using \eqref{Yangian-general-brackets}:
\begin{align}\label{UAB}
U^{AB}=&\int_{-\sigma_{5}}^{+\sigma_{6}}\!\!\!\!\!\!\!\!\!\!\!\mathrm{d}\sigma\mathrm{d}\sigma' \, \{ A_{\sigma}{}^{A}(\sigma),A_{\sigma}{}^{B}(\sigma') \}
\\
=&  f^{AB}{}_{C}\int_{-\sigma_{5}}^{+\sigma_{6}}\!\!\!\!\!\!\!\!\!\!\!\mathrm{d}\sigma\mathrm{d}\sigma' \Bigl(\!m_{5}A_{\sigma}{}^{C}+m_{6}B_{\tau}{}^{C}\!\Bigr)\delta(\sigma\!-\!\sigma')
+m_{7}\gamma^{AB}\int_{-\sigma_{5}}^{+\sigma_{6}}\!\!\!\!\!\!\!\!\!\!\!\mathrm{d}\sigma\mathrm{d}\sigma' \, \partial_{\sigma}\delta(\sigma\!-\!\sigma') \,\, .
\notag 
\end{align}
The second integral contains boundary terms which cancel each other, while the first one exhibits a tensorial structure of the type \eqref{level1-bracket-tensorial-structures-def}, with $X_{E_{1}...E_{N}}{}^{P}$ carrying no lower indices, and satisfies \eqref{central-tensors-def}. $U^{AB}$ hence does not contribute to the Serre relations \eqref{Serre-relations}, as it vanishes after contraction with the structure constants and antisymmetrisation.

\paragraph{Evaluation of $Z_{1}^{AB}$.}
Starting from the definition \eqref{Z1-Z2-def} one can exploit the $\delta$-function to get rid of the integration over $\sigma'''$ and then perform the rearrangement
\begin{align}\label{Z1-first-rearrangement}
Z_{1}^{AB}\!\!=&\tfrac{m_{1}}{4}f_{CD}{}^{A}f_{EF}{}^{B}f^{DF}{}_{G}\!\!\!\int_{-\sigma_{3}}^{+\sigma_{4}}\!\!\!\!\!\!\!\!\!\!\mathrm{d}\sigma\mathrm{d}\sigma' \mathrm{d}\sigma''  \, \epsilon(\sigma\!-\!\sigma'')\epsilon(\sigma'\!-\!\sigma'')B_{\tau}{}^{C}(\sigma)B_{\tau}{}^{E}(\sigma')B_{\tau}{}^{G}(\sigma'')
\notag \\
=& -\tfrac{m_{1}}{8}\Bigl( f_{CD}{}^{A}f_{EF}{}^{B}f^{DE}{}_{G} +f_{FD}{}^{A}f_{EC}{}^{B}f^{DE}{}_{G}\Bigr)\times
\notag \\
& \qquad \qquad \quad \times \!\!\!\int_{-\sigma_{3}}^{+\sigma_{4}}\!\!\!\!\!\!\!\!\!\!\mathrm{d}\sigma\mathrm{d}\sigma' \mathrm{d}\sigma''  \, \epsilon(\sigma-\sigma'')\epsilon(\sigma'-\sigma'')B_{\tau}{}^{C}(\sigma)B_{\tau}{}^{F}(\sigma')B_{\tau}{}^{G}(\sigma'')
\notag \\
=& -\tfrac{m_{1}}{4}T_{CFG}{}^{AB}\!\!\!\int_{-\sigma_{3}}^{+\sigma_{4}}\!\!\!\!\!\!\!\!\!\!\mathrm{d}\sigma\mathrm{d}\sigma' \mathrm{d}\sigma''  \, \epsilon(\sigma-\sigma'')\epsilon(\sigma'-\sigma'')B_{\tau}{}^{C}(\sigma)B_{\tau}{}^{F}(\sigma')B_{\tau}{}^{G}(\sigma'') \,\, ,
\end{align}
where in the first line we split the integral into two equal pieces and relabelled the indices on the structure constants of the second piece by using symmetry of the integral under the simultaneous exchange of $\sigma \leftrightarrow \sigma'$ and $C\leftrightarrow E$. In the second line we recombined the two pieces into the $T$-tensor introduced in \eqref{level1-bracket-tensorial-structures-def}. At this point it is convenient to explicitly write down the product of $\epsilon$-functions recalling the definition \eqref{epsilon-def},
\begin{equation}
\begin{aligned}
\epsilon(\sigma-\sigma'')\epsilon(\sigma'-\sigma'')=&+1-2\theta(\sigma''-\sigma)-2\theta(\sigma''-\sigma')+4\theta(\sigma''-\sigma)\theta(\sigma''-\sigma')
\\
=& -1+2\theta(\sigma-\sigma'')-2\theta(\sigma''-\sigma')+4\theta(\sigma''-\sigma)\theta(\sigma''-\sigma') \,\, ,
\end{aligned}
\end{equation}
having used the identity in \eqref{Heaviside-definition} to go from the first to the second line. Substituting this expression into the integral and recalling that the $T$-tensor is fully antisymmetric in the lower indices up to central terms \eqref{T-symmetry}, which we will ignore since due to \eqref{central-tensors-def} they do not contribute to the Serre relations, one can realise that the contributions linear in $\theta$ cancel each other after a simple relabelling of integration variables, leaving us with
\begin{equation}
Z_{1}^{AB}=\tfrac{m_{1}}{4}T_{CFG}{}^{AB}\!\!\!\int_{-\sigma_{3}}^{+\sigma_{4}}\!\!\!\!\!\!\!\!\!\!\mathrm{d}\sigma\mathrm{d}\sigma' \mathrm{d}\sigma''  \, \Bigl( 1-4\theta(\sigma''-\sigma)\theta(\sigma''-\sigma') \Bigr)B_{\tau}{}^{C}(\sigma)B_{\tau}{}^{F}(\sigma')B_{\tau}{}^{G}(\sigma'') \,\, .
\end{equation}
At this point, recalling the relations \eqref{Yangian-B-chi0-relation} we explicitly write down the two contributions
\begin{equation}
\begin{aligned}
Z_{1}^{AB}=&+\tfrac{m_{1}s^3}{4\beta^3}T_{CFG}{}^{AB} \int_{-\sigma_{3}}^{+\sigma_{4}}\!\!\!\!\!\!\!\!\mathrm{d}\sigma  \, \partial_{\sigma} \chi^{(0)C}(\sigma)\int_{-\sigma_{3}}^{+\sigma_{4}}\!\!\!\!\!\!\!\!\mathrm{d}\sigma'  \, \partial_{\sigma'} \chi^{(0)F}(\sigma')\int_{-\sigma_{3}}^{+\sigma_{4}}\!\!\!\!\!\!\!\!\mathrm{d}\sigma''  \, \partial_{\sigma''} \chi^{(0)G}(\sigma'') 
\\
&-\tfrac{m_{1}s^3}{\beta^3}T_{CFG}{}^{AB}\int_{-\sigma_{3}}^{+\sigma_{4}}\!\!\!\!\!\!\!\!\mathrm{d}\sigma'' \, \chi^{(0)C}(\sigma'')\chi^{(0)F}(\sigma'') \,\partial_{\sigma''}\chi^{(0)G}(\sigma'') \,\, ,
\end{aligned}
\end{equation}
and use again the symmetry, up to central terms, of $T$ in the lower indices to rearrange the integrand on the second line as a total derivative
\begin{equation}
\begin{aligned}
Z_{1}^{AB}=&+\tfrac{m_{1}s^3}{4\beta^3}T_{CFG}{}^{AB} \int_{-\sigma_{3}}^{+\sigma_{4}}\!\!\!\!\!\!\!\!\mathrm{d}\sigma  \, \partial_{\sigma} \chi^{(0)C}(\sigma)\int_{-\sigma_{3}}^{+\sigma_{4}}\!\!\!\!\!\!\!\!\mathrm{d}\sigma'  \, \partial_{\sigma'} \chi^{(0)F}(\sigma')\int_{-\sigma_{3}}^{+\sigma_{4}}\!\!\!\!\!\!\!\!\mathrm{d}\sigma''  \, \partial_{\sigma''} \chi^{(0)G}(\sigma'')
\\
&-\tfrac{m_{1}s^3}{3\beta^3}T_{CFG}{}^{AB}\int_{-\sigma_{3}}^{+\sigma_{4}}\!\!\!\!\!\!\!\!\mathrm{d}\sigma'' \, \partial_{\sigma''}\Bigl( \chi^{(0)C}(\sigma'')\chi^{(0)F}(\sigma'') \chi^{(0)G}(\sigma'') \Bigr) \,\, .
\end{aligned}
\end{equation}
We can now directly evaluate all the integrals, arriving at \eqref{ZAB-final-result}:
\begin{equation}
Z_{1}^{AB}=-\tfrac{m_{1}s^3}{12\beta^3}T_{CFG}{}^{AB} \, \chi^{(0)C}(+\sigma_{4}) \chi^{(0)F}(+\sigma_{4}) \chi^{(0)G}(+\sigma_{4}) +\Bigl(\text{terms}\propto\chi^{(0)}(-\sigma_{3})\Bigr) \,\, .
\end{equation}

\paragraph{Evaluation of $Z_{2}^{AB}$.}
From the definition of $Z_{2}^{AB}$ in \eqref{Z1-Z2-def} one finds that
\begin{align}\label{Z2AB-intermediate}
Z_{2}^{AB}\!\!\!=&\tfrac{m_{2}}{4}f_{CD}{}^{A}\!f_{EF}{}^{B}\gamma^{DF}\!\!\!\int_{-\sigma_{3}}^{+\sigma_{4}}\!\!\!\!\!\!\!\!\!\!\!\mathrm{d}\sigma\mathrm{d}\sigma' \mathrm{d}\sigma''\mathrm{d}\sigma'''   \epsilon(\sigma\!-\!\sigma'')\epsilon(\sigma'\!\!-\!\sigma''')B_{\tau}{}^{C}(\sigma)B_{\tau}{}^{E}(\sigma')\partial_{\sigma''}\delta(\sigma''\!\!-\!\sigma''') \,\,
\notag \\
=&\tfrac{m_{2}}{4}f_{C}{}^{DA}\!f_{ED}{}^{B}\!\!\!\int_{-\sigma_{3}}^{+\sigma_{4}}\!\!\!\!\!\!\!\!\!\!\!\mathrm{d}\sigma\mathrm{d}\sigma' \mathrm{d}\sigma''\mathrm{d}\sigma'''   \epsilon(\sigma'\!\!-\!\sigma''')B_{\tau}{}^{C}(\sigma)B_{\tau}{}^{E}(\sigma')\times
\notag \\
&\qquad\qquad\qquad\qquad \times\Bigl[\partial_{\sigma''}\Bigl(\epsilon(\sigma\!-\!\sigma'')\delta(\sigma''\!\!-\!\sigma''')\Bigr)-\delta(\sigma''-\sigma''')\partial_{\sigma''}\epsilon(\sigma-\sigma'')\Bigr]
\notag \\
=&\tfrac{m_{2}}{4}f_{C}{}^{DA}\!f_{ED}{}^{B}\!\!\!\int_{-\sigma_{3}}^{+\sigma_{4}}\!\!\!\!\!\!\!\!\!\!\!\mathrm{d}\sigma\mathrm{d}\sigma' \mathrm{d}\sigma'''   \epsilon(\sigma'\!\!-\!\sigma''')B_{\tau}{}^{C}(\sigma)B_{\tau}{}^{E}(\sigma')\times
\notag \\
&\qquad\qquad\qquad\qquad \times 
\Bigl[ \epsilon(\sigma-\sigma_{4})\delta(\sigma_{4}-\sigma''')-\epsilon(\sigma+\sigma_{3})\delta(-\sigma_{3}-\sigma''')+2\delta(\sigma'''-\sigma) \Bigr]
\notag \\
=&\tfrac{m_{2}}{4}f_{C}{}^{DA}\!f_{ED}{}^{B}\!\!\!\int_{-\sigma_{3}}^{+\sigma_{4}}\!\!\!\!\!\!\!\!\!\!\!\mathrm{d}\sigma\mathrm{d}\sigma' \! \Bigl[ \! \epsilon(\sigma\!-\!\sigma_{4})\epsilon(\sigma'\!-\!\sigma_{4})\!-\!\epsilon(\sigma\!+\!\sigma_{3})\epsilon(\sigma'\!+\!\sigma_{3})\!+\!2\epsilon(\sigma'\!-\!\sigma)\!\Bigr]\!B_{\tau}{}^{C}\!(\sigma)B_{\tau}{}^{E}\!(\sigma')
\notag \\
=& Z_{2-\text{vanish}}^{AB}+Z_{2-\text{central}}^{AB} \,\, ,
\end{align}
where in the first step we integrated by parts and in the second step we performed the integration over $\sigma''$, using the definition of $\epsilon$ in \eqref{epsilon-def} and \eqref{Heaviside-definition} in the last term to find
\begin{equation}
\partial_{\sigma''}\epsilon(\sigma-\sigma'')=\partial_{\sigma''}\bigl(\theta(\sigma-\sigma'')-\theta(\sigma''-\sigma)\bigr)=\partial_{\sigma''}\bigl(1-2\theta(\sigma''-\sigma)\bigr)=-2\delta(\sigma''-\sigma) \, .
\end{equation}
In the third step in \eqref{Z2AB-intermediate} we performed the integration over $\sigma'''$, finally introducing
\begin{align}\label{Z2vanish-Z2central-definition}
Z_{2-\text{vanish}}^{AB}\!\!=&\tfrac{m_{2}}{4}f_{C}{}^{DA}\!f_{ED}{}^{B}\!\!\!\int_{-\sigma_{3}}^{+\sigma_{4}}\!\!\!\!\!\!\!\!\!\!\!\mathrm{d}\sigma\mathrm{d}\sigma' \! \Bigl[ \! \epsilon(\sigma\!-\!\sigma_{4})\epsilon(\sigma'\!-\!\sigma_{4})\!-\!\epsilon(\sigma\!+\!\sigma_{3})\epsilon(\sigma'\!+\!\sigma_{3})\!\Bigr]\!B_{\tau}{}^{C}\!(\sigma)B_{\tau}{}^{E}\!(\sigma')
\notag\\
Z_{2-\text{central}}^{AB}\!\!=&\tfrac{m_{2}}{2}f_{C}{}^{DA}\!f_{ED}{}^{B}\!\!\!\int_{-\sigma_{3}}^{+\sigma_{4}}\!\!\!\!\!\!\!\!\!\!\!\mathrm{d}\sigma\mathrm{d}\sigma' \, \epsilon(\sigma'-\sigma)B_{\tau}{}^{C}(\sigma)B_{\tau}{}^{E}(\sigma') \,\, .
\end{align}
As suggested by their labels, these two contributions respectively vanish identically and exhibit central structure, ensuring that $Z_{2}^{AB}$ does not contribute to \eqref{Q1-Q1-contraction-almost-Serre}. To see how the first one vanishes, one can use \eqref{Heaviside-definition} to notice that
\begin{equation}\label{epsilon-useful-rel}
\begin{aligned}
\epsilon(\sigma-\sigma_{4})&:=\theta(\sigma-\sigma_{4})-\theta(\sigma_{4}-\sigma)=2\theta(\sigma-\sigma_{4})-1 \, ,
\\
\epsilon(\sigma+\sigma_{3})&:=\theta(\sigma+\sigma_{3})-\theta(-\sigma_{3}-\sigma)=1-2\theta(-\sigma_{3}-\sigma) \,\, ,
\end{aligned}
\end{equation}
with similar expressions holding for $\epsilon(\sigma'-\sigma_{4})$ and $\epsilon(\sigma'+\sigma_{3})$. The key point is that each term exhibiting one of the $\theta$-functions $\theta(\sigma-\sigma_{4}),\theta(\sigma'-\sigma_{4})$ or $\theta(-\sigma_{3}-\sigma),\theta(-\sigma_{3}-\sigma')$ will vanish by construction, since the conditions $\sigma>\sigma_{4},\sigma'>\sigma_{4}$ or $\sigma<-\sigma_{3},\sigma'<-\sigma_{3}$ are never satisfied on the domain of integration. Only the identity parts in the relations \eqref{epsilon-useful-rel} can contribute to the integral, but in fact these end up cancelling each other:
\begin{equation}
Z_{2-\text{vanish}}^{AB}\!\!=\tfrac{m_{2}}{4}f_{C}{}^{DA}\!f_{ED}{}^{B}\!\!\!\int_{-\sigma_{3}}^{+\sigma_{4}}\!\!\!\!\!\!\!\!\!\!\!\mathrm{d}\sigma\mathrm{d}\sigma'  \bigl[ 1-1\bigr]\!B_{\tau}{}^{C}\!(\sigma)B_{\tau}{}^{E}\!(\sigma') = 0 \,\, .
\end{equation}
We are thus left to show that the second line in \eqref{Z2vanish-Z2central-definition} is indeed central:
\begin{equation}
\begin{aligned}
Z_{2-\text{central}}^{AB}=&\tfrac{m_{2}}{2}f_{C}{}^{DA}f_{ED}{}^{B}\int_{-\sigma_{3}}^{+\sigma_{4}}\!\!\!\!\!\!\!\!\!\!\!\mathrm{d}\sigma\mathrm{d}\sigma' \, \epsilon(\sigma'-\sigma)B_{\tau}{}^{C}(\sigma)B_{\tau}{}^{E}(\sigma')
\\
=&\tfrac{m_{2}}{4}\Bigl(f_{C}{}^{DA}f_{ED}{}^{B}-f_{E}{}^{DA}f_{CD}{}^{B}\Bigr)\int_{-\sigma_{3}}^{+\sigma_{4}}\!\!\!\!\!\!\!\!\!\!\!\mathrm{d}\sigma\mathrm{d}\sigma' \, \epsilon(\sigma'-\sigma)B_{\tau}{}^{C}(\sigma)B_{\tau}{}^{E}(\sigma')
\\
=&\tfrac{m_{2}}{4}f^{AB}{}_{D}f_{CE}{}^{D}\int_{-\sigma_{3}}^{+\sigma_{4}}\!\!\!\!\!\!\!\!\!\!\!\mathrm{d}\sigma\mathrm{d}\sigma' \, \epsilon(\sigma'-\sigma)B_{\tau}{}^{C}(\sigma)B_{\tau}{}^{E}(\sigma') \,\, .
\end{aligned}
\end{equation}
In the first step we split the integral into two equal pieces and relabelled the integration variables and the indices on the structure constants of the second piece. This produces a minus sign, since by definition $\epsilon(\sigma'-\sigma):=\theta(\sigma'-\sigma)-\theta(\sigma-\sigma')=-\epsilon(\sigma-\sigma')$. In the second step we used the Jacobi identity \eqref{def-Jacobi-components}, arriving at a central structure \eqref{level1-bracket-tensorial-structures-def}.

\paragraph{Contributions to $V^{AB}$.}  This term has been defined as $V^{AB}:=V_{1}^{AB}+V_{2}^{AB}$, but $V_{1}^{AB}$ and $V_{2}^{AB}$ in \eqref{Q1-Q1-bracket-explicit-contributions-list} are related by $V_{2}^{AB}=-V_{1}^{BA}$ via antisymmetry of the Poisson brackets \eqref{def-PB-antisymmetry} and relabelling of the integration variables $\sigma \leftrightarrow \sigma'$, such that one only needs to explicitly compute either of them. Proceeding with the second, one finds
\begin{align}\label{V2AB-manipulations}
V_{2}^{AB}=&\tfrac{1}{2}f_{CD}{}^{A}\int_{-\sigma_{5}}^{+\sigma_{6}}\!\!\!\!\!\!\!\!\mathrm{d}\sigma' \int_{-\sigma_{3}}^{+\sigma_{4}}\!\!\!\!\!\!\!\!\mathrm{d}\sigma\mathrm{d}\sigma'' \, \theta(\sigma\!-\!\sigma'')\{B_{\tau}{}^{C}(\sigma)B_{\tau}{}^{D}(\sigma''),A_{\sigma}{}^{B}(\sigma')\}
\\
=& \tfrac{1}{2}f_{CD}{}^{A}\int_{-\sigma_{5}}^{+\sigma_{6}}\!\!\!\!\!\!\!\!\mathrm{d}\sigma' \int_{-\sigma_{3}}^{+\sigma_{4}}\!\!\!\!\!\!\!\!\mathrm{d}\sigma\mathrm{d}\sigma'' \, \Bigl(\theta(\sigma\!-\!\sigma'')\!-\!\theta(\sigma''\!-\!\sigma)\Bigr)B_{\tau}{}^{D}(\sigma'')\{B_{\tau}{}^{C}(\sigma),A_{\sigma}{}^{B}(\sigma')\}
\notag \\
=& \!-\!\tfrac{m_{3}}{2}f_{D}{}^{AC}f^{B}{}_{CE}\!\!\!\int_{-\sigma_{5}}^{+\sigma_{6}}\!\!\!\!\!\!\!\!\!\!\!\mathrm{d}\sigma' \!\!\int_{-\sigma_{3}}^{+\sigma_{4}}\!\!\!\!\!\!\!\!\!\!\!\mathrm{d}\sigma\mathrm{d}\sigma'' \, \Bigl(\!\theta(\sigma\!-\!\sigma'')\!-\!\theta(\sigma''\!-\!\sigma)\!\Bigr)B_{\tau}{}^{D}(\sigma'')A_{\sigma}{}^{E}(\sigma)\delta(\sigma\!-\!\sigma')
\notag \\
& \!+\!\tfrac{m_{4}}{2}\!\!\!\int_{-\sigma_{5}}^{+\sigma_{6}}\!\!\!\!\!\!\!\!\!\!\!\mathrm{d}\sigma' \!\! \int_{-\sigma_{3}}^{+\sigma_{4}}\!\!\!\!\!\!\!\!\!\!\!\mathrm{d}\sigma\mathrm{d}\sigma'' \! \Bigl(\!\theta(\sigma\!-\!\sigma'')\!-\!\theta(\sigma''\!-\!\sigma)\!\Bigr)\!B_{\tau}{}^{D}(\sigma'')f_{D}{}^{AE}\gamma_{EC}C^{CB}\!(\sigma')\partial_{\sigma}\delta(\sigma\!-\!\sigma') \,  ,
\notag
\end{align}
where in the first step we used the Leibniz rule \eqref{def-PB-Leibniz} on the Poisson bracket and recombined the two resulting terms via antisymmetry of the structure constants and relabelling of the integration variables $\sigma\leftrightarrow\sigma''$. In the second step we used the brackets \eqref{Yangian-general-brackets} and slightly rearranged the order of the indices. At this point it is convenient to re-introduce $\epsilon(\sigma-\sigma'')$, defined in \eqref{epsilon-def}, such that after using $V_{2}{}^{AB}=-V_{1}{}^{BA}$ one obtains
\begin{align}\label{VAB-result-intermediate}
V^{AB}\!\!=& \!-\!\tfrac{m_{3}}{2}(f_{D}{}^{AC}f^{B}{}_{CE}\!-\!f_{D}{}^{BC}f^{A}{}_{CE})\!\!\!\int_{-\sigma_{5}}^{+\sigma_{6}}\!\!\!\!\!\!\!\!\!\!\mathrm{d}\sigma' \!\!\! \int_{-\sigma_{3}}^{+\sigma_{4}}\!\!\!\!\!\!\!\!\!\!\mathrm{d}\sigma\mathrm{d}\sigma''  \epsilon(\sigma\!-\!\sigma'')B_{\tau}{}^{D}\!(\sigma'')A_{\sigma}{}^{E}\!(\sigma)\delta(\sigma\!-\!\sigma')
\notag \\
& \!+\!\tfrac{m_{4}}{2}\!\!\!\int_{-\sigma_{5}}^{+\sigma_{6}}\!\!\!\!\!\!\!\!\!\!\mathrm{d}\sigma' \!\!\! \int_{-\sigma_{3}}^{+\sigma_{4}}\!\!\!\!\!\!\!\!\!\!\mathrm{d}\sigma\mathrm{d}\sigma''   \epsilon(\sigma\!-\!\sigma'')B_{\tau}{}^{D}(\sigma'')\partial_{\sigma}\delta(\sigma\!-\!\sigma') \times
\notag \\
& \qquad \qquad \qquad \qquad\qquad \qquad \times \Bigl(\!f_{D}{}^{AE}\gamma_{EC}C^{CB}(\sigma')\!-\!f_{D}{}^{BE}\gamma_{EC}C^{CA}(\sigma')\!\Bigr) \, .
\end{align}
Integrating by parts in the second line of \eqref{VAB-result-intermediate}, one finally recovers the contributions in \eqref{VAB-contributions}, where $\mathcal{V}_{1}^{AB}$ is obtained after using the properties in \eqref{Heaviside-definition}, which imply that
\begin{equation}
\partial_{\sigma}\epsilon(\sigma-\sigma'')=\partial_{\sigma}\bigl( \theta(\sigma-\sigma'')-\theta(\sigma''-\sigma)\bigr)=\partial_{\sigma}\bigl(2\theta(\sigma-\sigma'')-1 \bigr)=2\delta(\sigma-\sigma'') \,\, .
\end{equation}

\paragraph{Vanishing of $\mathcal{V}_{2}^{AB}$.} 
We first repeat the definition \eqref{VAB-contributions} for convenience:
\begin{equation}\label{calV2-appendix}
\begin{aligned}
\mathcal{V}_{2}^{AB}:=&+\tfrac{m_{4}}{2} \!\! \int_{-\sigma_{5}}^{+\sigma_{6}}\!\!\!\!\!\!\!\!\!\!\mathrm{d}\sigma' \!\!\! \int_{-\sigma_{3}}^{+\sigma_{4}}\!\!\!\!\!\!\!\!\!\!\mathrm{d}\sigma\mathrm{d}\sigma'' \,\,  \partial_{\sigma}\Bigl(\!\epsilon(\sigma\!-\!\sigma'')\delta(\sigma\!-\!\sigma')\!\Bigr) \times
\\
& \qquad \qquad \qquad \qquad \qquad \times B_{\tau}{}^{D}\!(\sigma'')\Bigl(\!f_{D}{}^{AC}C_{C}{}^{B}\!(\sigma')\!-\!f_{D}{}^{BC}C_{C}{}^{A}\!(\sigma')\!\Bigr) \, .
\end{aligned}
\end{equation}
This contribution contains the boundary integral
\begin{equation}
\begin{aligned}
\int_{-\sigma_{3}}^{+\sigma_{4}}\!\!\!\!\!\!\!\!\!\!\mathrm{d}\sigma \partial_{\sigma}\Bigl(\!\epsilon(\sigma\!-\!\sigma'')\delta(\sigma\!-\!\sigma')\!\Bigr)\!=\!
 \Bigl(\!\epsilon(\sigma_{4}\!-\!\sigma'')\delta(\sigma_{4}\!-\!\sigma')-\epsilon(-\sigma_{3}\!-\!\sigma'')\delta(-\sigma_{3}\!-\!\sigma')\!\Bigr) \, ,
\end{aligned}
\end{equation}
such that the two resulting terms are respectively non-vanishing for $\sigma'=\sigma_{4}$ and $\sigma'=-\sigma_{3}$. Since $\sigma'\in[-\sigma_{5},\sigma_{6}]$, these terms can only be non-vanishing if we impose some restriction on the boundaries, namely $\sigma_{4}<\sigma_{6}$ and $\sigma_{3}<\sigma_{5}$. Hence the two $\delta$-functions are respectively replaced by Heaviside functions $\theta(\sigma_{6}-\sigma_{4})$ and $\theta(\sigma_{5}-\sigma_{3})$, similarly to what was observed in the evaluation of $X_{1}^{AB}$ in \eqref{vanishing-X1AB}. At this point one is left with a single integral over $\sigma''$, which after using the relations in \eqref{Heaviside-definition} to rewrite
\begin{equation}
\epsilon(\sigma_{4}-\sigma'')=1-2\theta(\sigma''-\sigma_{4}) 
\qquad \text{and} \qquad 
\epsilon(-\sigma_{3}-\sigma'')=2\theta(-\sigma_{3}-\sigma'')-1 \,\, ,
\end{equation}
takes the following expression: 
\begin{align}
\mathcal{V}_{2}^{AB}\!\!\!=\!\!\!\int_{-\sigma_{3}}^{+\sigma_{4}}\!\!\!\!\!\!\!\!\!\!\mathrm{d}\sigma'' \Bigl\{&\!\!\!+\theta(\sigma_{6}\!-\!\sigma_{4})\Bigl(1\!-\!2\theta(\sigma''\!-\!\sigma_{4}) \Bigr) B_{\tau}{}^{D}(\sigma'')\Bigl(\!f_{D}{}^{AC}C_{C}{}^{B}(\sigma_{4})\!-\!f_{D}{}^{BC}C_{C}{}^{A}(\sigma_{4})\!\Bigr)
\\
&\!\!\!-\theta(\sigma_{5}\!-\!\sigma_{3})\Bigl(2\theta(-\sigma_{3}\!-\!\sigma'')\!-\!1 \Bigr) B_{\tau}{}^{D}(\sigma'')\Bigl(\!f_{D}{}^{AC}C_{C}{}^{B}(-\sigma_{3})\!-\!f_{D}{}^{BC}C_{C}{}^{A}(-\sigma_{3})\!\Bigr) \!\! \Bigr\} \, .
\notag
\end{align}
Clearly, the two Heaviside functions $\theta(\sigma''-\sigma_{4})$ and $\theta(-\sigma_{3}-\sigma'')$ cannot contribute to the above integral, since by construction $\sigma''\in [-\sigma_{3},+\sigma_{4}]$ and hence neither $\sigma''>\sigma_{4}$ nor $\sigma''<-\sigma_{3}$ can ever hold. This finally leads to the expression
\begin{align}\label{VAB-first-line-final-result}
\mathcal{V}_{2}^{AB}=\!\!\int_{-\sigma_{3}}^{+\sigma_{4}}\!\!\!\!\!\!\!\!\!\!\mathrm{d}\sigma'' \Bigl\{&\!\!+\theta(\sigma_{6}\!-\!\sigma_{4}) B_{\tau}{}^{D}(\sigma'')\Bigl(\!f_{D}{}^{AC}C_{C}{}^{B}(\sigma_{4})\!-\!f_{D}{}^{BC}C_{C}{}^{A}(\sigma_{4})\!\Bigr)
\\
&\!\!\!+\theta(\sigma_{5}\!-\!\sigma_{3}) B_{\tau}{}^{D}(\sigma'')\Bigl(\!f_{D}{}^{AC}C_{C}{}^{B}(-\sigma_{3})\!-\!f_{D}{}^{BC}C_{C}{}^{A}(-\sigma_{3})\!\Bigr) \! \Bigr\} \, ,
\notag
\end{align}
which is of the form (5.31) in \cite{Klose:2016qfv} and generically non-vanishing even after contraction with the structure constants and antisymmetrisation. Getting rid of \eqref{VAB-first-line-final-result} thus requires imposing that the boundaries should satisfy $\sigma_{6}<\sigma_{4}$ and $\sigma_{5}<\sigma_{3}$, recovering \eqref{conditions-on-boundaries-summary}.

\section{Details about the Poisson brackets}\label{appendix:C}
In this appendix we explicitly compute the Poisson brackets, for various classes of sigma models, which are necessary to highlight their classical Yangian structure. We will use the same notation as in \cite{Bielli:2024ach}, defining the Poisson bracket as
\begin{equation}\label{def-PB}
\{ A(\sigma),B(\sigma') \} := \int\mathrm{d}\sigma'' \, \Bigl( \frac{\partial A(\sigma)}{\partial Z^{\rho}(\sigma'')} \frac{\partial B(\sigma')}{\partial \pi_{\rho}(\sigma'')} - \frac{\partial B(\sigma')}{\partial Z^{\rho}(\sigma'')} \frac{\partial A(\sigma)}{\partial \pi_{\rho}(\sigma'')}\Bigr) \,\, ,
\end{equation}
with $Z^{\rho}(\tau,\sigma)$ coordinates on the background and conjugate momentum defined as
\begin{equation}\label{def-canonical-momentum}
\pi_{\rho}(\tau,\sigma):=\frac{\partial \mathcal{L}}{\partial (\partial_{\tau}Z^{\rho})} \,\, .
\end{equation}
By construction, the Poisson bracket is antisymmetric,
\begin{equation}\label{def-PB-antisymmetry}
\{ A(\sigma),B(\sigma') \} = - \{ B(\sigma'),A(\sigma) \} \, ,
\end{equation}
and satisfies the Leibniz rule
\begin{equation}\label{def-PB-Leibniz}
\begin{aligned}
\{ A(\sigma)C(\sigma),B(\sigma')D(\sigma') \}&=A(\sigma)B(\sigma') \{ C(\sigma),D(\sigma') \} + A(\sigma)D(\sigma') \{ C(\sigma),B(\sigma') \}+
\\
& +C(\sigma)B(\sigma') \{ A(\sigma),D(\sigma') \}+C(\sigma)D(\sigma') \{ A(\sigma),B(\sigma') \} \,\, .
\end{aligned}
\end{equation}
We will also need to use the following identities involving $\delta$-functions:
\begin{equation}\label{def-delta-func-identity-1}
\frac{\delta f(x)}{\delta f(y)}=\delta(x-y) \, ,
\qquad \qquad
\frac{\delta(\partial_{x}f(x))}{\delta f(y)}=-\partial_{y}\delta(y-x) \, ,
\end{equation}
\begin{equation}\label{def-delta-func-identity-2}
\delta(x-y)=\delta(y-x) \, ,\qquad \qquad
\partial_{x}\delta(x-y)=-\partial_{y}\delta(y-x) \, ,
\end{equation}
\begin{equation}\label{def-delta-func-identity-3}
f(y)\partial_{x}\delta(x-y)=f(x)\partial_{x}\delta(x-y)+f'(x)\delta(x-y) \,\, ,
\end{equation}
and the following conventions to convert between lightcone and $(\tau,\sigma)$ coordinates:
\begin{equation}\label{def-lightcone-to-tau-sigma}
V_{+}=V_{\tau}+V_{\sigma} \qquad V_{-}=V_{\tau}-V_{\sigma} \qquad \text{and} \qquad V_{\tau}=\tfrac{1}{2}(V_{+}+V_{-}) \qquad V_{\sigma}=\tfrac{1}{2}(V_{+}-V_{-}) \,\, .
\end{equation}

\subsection{Symmetric-space sigma models}\label{appendix:C-SSSM}
\paragraph{Calculation of $\{ L_{\tau}{}^{A}(\sigma),L_{\tau}^{B}(\sigma') \}$.}
Using the relation \eqref{SSSM-Jtau-momentum} to write $L_{\tau}^{A}$ in terms of momenta and the definition \eqref{def-PB} of Poisson brackets one obtains
\begin{align}\label{SSSM-Jtau-Jtau-bracket-starting-point}
&\{ L_{\tau}{}^{A}(\sigma) ,L_{\tau}^{B}(\sigma') \}=\gamma^{AC}\gamma^{BD}\{ \pi_{\mu}(\sigma)L_{C}{}^{\mu}(\sigma),\pi_{\nu}(\sigma')L_{D}{}^{\nu}(\sigma') \} =
\notag \\
& = \gamma^{AC}\gamma^{BD} \int \mathrm{d}\sigma'' \, \Bigl( \pi_{\mu} L_{D}{}^{\rho}\, \partial_{\rho}L_{C}{}^{\mu} \delta(\sigma-\sigma'')\delta(\sigma'-\sigma'')- \pi_{\mu} L_{C}{}^{\rho}\, \partial_{\rho}L_{D}{}^{\mu} \delta(\sigma'-\sigma'')\delta(\sigma-\sigma'')\Bigr)
\notag \\
& =-\tfrac{1}{4}\gamma^{AC}\gamma^{BD} \gamma_{FE}L_{\tau}{}^{F}L_{\mu}{}^{E}\bigl(L_{D}{}^{\rho}\, \partial_{\rho}L_{C}{}^{\mu}-L_{C}{}^{\rho}\, \partial_{\rho}L_{D}{}^{\mu}\bigr)\delta(\sigma-\sigma') \,\, ,
\end{align}
where in the last step we used the third relation in \eqref{SSSM-J-vielbeine-and-momentum(J)} to express the momenta in terms of $L$. Bringing $L_{\mu}{}^{E}$ into the parentheses and exploiting \eqref{SSSM-useful-relation} one then finds
\begin{equation}\label{intermediate-1}
\begin{aligned}
\{ L_{\tau}{}^{A}(\sigma) ,L_{\tau}^{B}(\sigma') \} =&+  \tfrac{1}{4}\gamma^{AC}\gamma^{BD} \gamma_{FE}L_{\tau}{}^{F}L_{D}{}^{\mu}L_{C}{}^{\nu}(\partial_{\mu}L_{\nu}{}^{E}-\partial_{\nu}L_{\mu}{}^{E})\delta(\sigma-\sigma')+
\\
& -\gamma^{AC}\gamma^{BD} \gamma_{FE}L_{\tau}{}^{F}(L_{D}{}^{\rho}\partial_{\rho}K_{C}{}^{E}-L_{C}{}^{\rho}\partial_{\rho}K_{D}{}^{E})\delta(\sigma-\sigma') \,\, .
\end{aligned}
\end{equation}
At this point, on the first line of \eqref{intermediate-1} one can use the off-shell flatness of $J$ in \eqref{SSSM-J-flat} and the second relation in \eqref{SSSM-J-vielbeine-and-momentum(J)} twice, while on the second line both terms involving derivatives of $K$ can be exchanged for contractions between $K,L$ and the structure constants by exploiting \eqref{J-K-differential-relation-components}, thus arriving at 
\begin{align}\label{intermediate-2}
\{ L_{\tau}{}^{A}(\sigma) ,L_{\tau}^{B}(\sigma') \} =&-4\gamma^{AC}\gamma^{BD}\gamma_{FE}L_{\tau}{}^{F}K_{D}{}^{G}K_{C}{}^{H}f_{GH}{}^{E}\delta(\sigma-\sigma')+
\\
& -\tfrac{1}{2}\gamma^{AC}\gamma^{BD}\gamma_{FE}L_{\tau}{}^{F}L_{D}{}^{\rho}L_{\rho}{}^{G}(f_{GC}{}^{H}K_{H}{}^{E}-K_{C}{}^{H}f_{GH}{}^{E})\delta(\sigma-\sigma')+
\notag \\
& +\tfrac{1}{2}\gamma^{AC}\gamma^{BD}\gamma_{FE}L_{\tau}{}^{F}L_{C}{}^{\rho}L_{\rho}{}^{G}(f_{GD}{}^{H}K_{H}{}^{E}-K_{D}{}^{H}f_{GH}{}^{E})\delta(\sigma-\sigma') \,\, .
\notag
\end{align}
Using $L_{A}{}^{\rho}L_{\rho}^{B}=4K_{A}{}^{B}$ on the last two lines of \eqref{intermediate-2}, one can recognise that the second term on each of these two lines sum up, precisely cancelling the first line and leaving 
\begin{equation}\label{intermediate-3}
\{ L_{\tau}{}^{A}(\sigma) ,L_{\tau}^{B}(\sigma') \} = -2 \gamma^{AC}\gamma^{BD}\gamma_{FE}L_{\tau}{}^{F}(K_{D}{}^{G}f_{GC}{}^{H}-K_{C}{}^{G}f_{GD}{}^{H})K_{H}{}^{E}\delta(\sigma-\sigma') \,\, .
\end{equation}
Using now that $K_{AB}$ is symmetric and $L$ is an eigenvector of $K$ \eqref{J-eigenvector-components}, one finds
\begin{equation}
\gamma_{FE}L_{\tau}{}^{F}K_{H}{}^{E}=L_{\tau}{}^{F}K_{HF}=L_{\tau}{}^{F}K_{FH}=L_{\tau}{}^{F}K_{F}{}^{E}\gamma_{EH}=L_{\tau}{}^{E}\gamma_{EH} \,\, ,
\end{equation}
which upon substitution into \eqref{intermediate-3} leads to the desired relation
\begin{equation}\label{SSSM-Jtau-Jtau-bracket-last-steps}
\begin{aligned}
\{ L_{\tau}{}^{A}(\sigma) ,L_{\tau}^{B}(\sigma') \} =&  -2 \gamma^{AC}\gamma^{BD}L_{\tau}{}^{E}\gamma_{EH}(K_{D}{}^{G}f_{GC}{}^{H}-K_{C}{}^{G}f_{GD}{}^{H})\delta(\sigma-\sigma') 
\\
=& -2 \gamma^{AC}\gamma^{BD}L_{\tau}{}^{E}(f_{ED}{}^{G}K_{GC}-f_{EC}{}^{G}K_{GD})\delta(\sigma-\sigma')
\\
=& -2L_{\tau}{}^{E}(f_{E}{}^{BG}K_{G}{}^{A}-f_{E}{}^{AG}K_{G}{}^{B})\delta(\sigma-\sigma')
\\
=&+2f^{AB}{}_{C}L_{\tau}{}^{C}\delta(\sigma-\sigma') \,\, ,
\end{aligned}
\end{equation}
after rearranging indices in the first step, using symmetry of $K_{AB}$ and antisymmetry of $f_{ABC}$, and exploiting \eqref{J-K-algebraic-relation-components} in the last step.

\paragraph{Calculation of $\{ L_{\tau}{}^{A}(\sigma),L_{\sigma}^{B}(\sigma') \}$.} We begin by again using the relation \eqref{SSSM-Jtau-momentum} for $L_{\tau}{}^{A}$ and exploiting that $L_{\sigma}{}^{B}$ is independent of the momentum, so that only the second term in the general Poisson bracket \eqref{def-PB} survives:
\begin{align}\label{SSSM-Jtau-Jsigma-bracket-first-step}
&\!\!\!\!\!\{ L_{\tau}{}^{A}(\sigma),L_{\sigma}^{B}(\sigma') \} = -\gamma^{AC} \{ \pi_{\mu}(\sigma)L_{C}{}^{\mu}(\sigma),L_{\sigma}{}^{B}(\sigma') \}=\gamma^{AC}\!\!\!\!\int \mathrm{d}\sigma'' \, L_{C}{}^{\rho}(\sigma)\delta(\sigma-\sigma'') \frac{\partial L_{\sigma}{}^{B}(\sigma')}{\partial Z^{\rho}(\sigma'')}
\notag \\
&\hspace{-15pt} = \gamma^{AC}\!\!\!\! \int \!\! \mathrm{d}\sigma'' L_{C}{}^{\rho}(\sigma)\delta(\sigma \!-\!\sigma'') \Biggl( \!\frac{\partial (\partial_{\sigma'} Z^{\mu}(\sigma'))}{\partial Z^{\rho}(\sigma'')} \, L_{\mu}{}^{B}(\sigma')\!+\!\partial_{\sigma'}Z^{\mu}(\sigma')\frac{\partial L_{\mu}{}^{B}(\sigma')}{\partial Z^{\rho}(\sigma'')}\delta(\sigma'\!-\!\sigma'') \!\!\Biggr) \! .
\end{align}
The second line is obtained by using $L_{\sigma}{}^{B}(\sigma')=\partial_{\sigma'}Z^{\mu}(\sigma') L_{\mu}{}^{B}(\sigma')$. Using then the second relation in \eqref{def-delta-func-identity-1} on the first term in the bracket, and successively the identity \eqref{def-delta-func-identity-3} on the resulting term $L_{\rho}{}^{B}(\sigma')\partial_{\sigma''} \delta(\sigma''-\sigma')$ we arrive at
\begin{align}
&\{ L_{\tau}{}^{A}(\sigma),L_{\sigma}^{B}(\sigma') \} =\gamma^{AC}\int \mathrm{d}\sigma'' \, L_{C}{}^{\rho}(\sigma)\delta(\sigma-\sigma'')
\\
& \Biggl( -L_{\rho}{}^{B}(\sigma'')\partial_{\sigma''}\delta(\sigma''-\sigma')-\delta(\sigma''-\sigma')\partial_{\sigma''}L_{\rho}{}^{B}(\sigma'')+\partial_{\sigma'}Z^{\mu}(\sigma')\frac{\partial L_{\mu}{}^{B}(\sigma')}{\partial Z^{\rho}(\sigma'')}\delta(\sigma'-\sigma'') \Biggr) \,\, ,
\notag
\end{align}
which using $\partial_{\sigma''}L_{\rho}{}^{B}(\sigma'')=\partial_{\sigma''}Z^{\mu}(\sigma'')\partial_{\mu}L_{\rho}{}^{B}(\sigma'')$ and $\delta(\sigma''-\sigma')=\delta(\sigma'-\sigma'')$ leads to
\begin{align}\label{intermediate-4}
\{ L_{\tau}{}^{A}(\sigma),L_{\sigma}^{B}(\sigma') \} =& \gamma^{AC}L_{C}{}^{\rho}(\sigma) \partial_{\sigma}Z^{\mu}(\sigma)(\partial_{\rho}L_{\mu}{}^{B}\!-\!\partial_{\mu}L_{\rho}{}^{B})\delta(\sigma\!-\!\sigma')-\gamma^{AC}L_{C}{}^{\rho}L_{\rho}{}^{B}\partial_{\sigma}\delta(\sigma\!-\!\sigma')
\notag\\
=&-4\gamma^{AC} K_{C}{}^{G}\partial_{\sigma}Z^{\mu}L_{\mu}{}^{H}f_{GH}{}^{B}\delta(\sigma\!-\!\sigma')-4K^{AB}(\sigma)\partial_{\sigma}\delta(\sigma\!-\!\sigma') \,\, ,
\end{align}
where the second line is obtained by using the off-shell flatness of $L$ \eqref{SSSM-J-flat} in the first term and $L_{C}{}^{\rho}(\sigma)L_{\rho}{}^{B}(\sigma)=4K_{C}{}^{B}(\sigma)$ in the second term. To proceed we note that the first term in \eqref{intermediate-4} can be rearranged as
\begin{equation}
\begin{aligned}
4\gamma^{AC}K_{C}{}^{G}L_{\mu}{}^{H}f_{GH}{}^{B}=& 2\gamma^{AC}K_{C}{}^{G}L_{\mu}{}^{H}f_{GH}{}^{B}+2\gamma^{AC}K_{C}{}^{G}L_{\mu}{}^{H}f_{GH}{}^{B}
\\
=&2L_{\mu}{}^{H}f_{H}{}^{BG}K_{G}{}^{A}+2\gamma^{AC}L_{\mu}{}^{E}K_{C}{}^{G}K_{E}{}^{H}f_{GH}{}^{B}
\\
=& 2L_{\mu}{}^{H}f_{H}{}^{BG}K_{G}{}^{A}+2\gamma^{AC}L_{\mu}{}^{E}(K_{E}{}^{H}f_{CH}{}^{B}-K_{E}{}^{H}f_{CH}{}^{G}K_{G}{}^{B})
\\
=&2L_{\mu}{}^{H}(f_{H}{}^{BG}K_{G}{}^{A}+f_{H}{}^{AG}K_{G}{}^{B})-2f^{AB}{}_{C}L_{\mu}{}^{C}
\\
=& 4\partial_{\mu} K^{AB} -2f^{AB}{}_{C}L_{\mu}{}^{C} \,\, ,
\end{aligned}
\end{equation}
where in the first step we split the term into two equal contributions, in the second step we rearranged the first contribution and exploited \eqref{J-eigenvector-components} on the second one, in the third step we used the identity \eqref{important-K-relation}, in the fourth we re-used \eqref{J-eigenvector-components} and collected terms, while in the last we exploited the differential relation \eqref{J-K-differential-relation-components} between $K$ and $L$. This leads to
\begin{equation}
\begin{aligned}
\{ L_{\tau}{}^{A}(\sigma),L_{\sigma}^{B}(\sigma') \}=&+2f^{AB}{}_{C}\partial_{\sigma}Z^{\mu}(\sigma)L_{\mu}{}^{C}\delta(\sigma-\sigma')
\\
&-4\partial_{\sigma}Z^{\mu}(\sigma) \partial_{\mu}K^{AB}(\sigma)\delta(\sigma-\sigma')-4K^{AB}(\sigma)\partial_{\sigma}\delta(\sigma-\sigma')
\\
=& +2f^{AB}{}_{C}L_{\sigma}{}^{C}\delta(\sigma-\sigma')-4K^{AB}(\sigma')\partial_{\sigma}\delta(\sigma-\sigma') \,\, ,
\end{aligned}
\end{equation}
where in the last step we used $\partial_{\sigma}Z^{\mu}L_{\mu}{}^{C}=L_{\sigma}{}^{C}$, $\partial_{\sigma}Z^{\mu}\partial_{\mu}K^{AB}=\partial_{\sigma}K^{AB}$ and finally exploited the identity \eqref{def-delta-func-identity-3} to combine the two contributions involving $K$.

\paragraph{Calculation of $\{ L_{\tau}{}^{A}(\sigma),K^{BC}(\sigma') \}$.} We start by using \eqref{SSSM-Jtau-momentum} and since $K$ is independent of $\pi$, only the second term in the Poisson bracket definition \eqref{def-PB} survives,
\begin{equation}\label{SSSM-Jtau-K-bracket}
\begin{aligned}
\{ L_{\tau}{}^{A}(\sigma),K^{BC}(\sigma') \}=&-\gamma^{AD} \{ \pi_{\mu}(\sigma)L_{D}{}^{\mu}(\sigma),K^{BC}(\sigma') \}
\\
=&\gamma^{DA}\int \mathrm{d}\sigma'' \, \delta(\sigma'-\sigma'')\partial_{\rho}K^{BC}(\sigma') \, L_{D}{}^{\rho}(\sigma)\delta(\sigma-\sigma'')
\\
=&\gamma^{AD}L_{D}{}^{\rho}\partial_{\rho}K^{BC} \delta(\sigma-\sigma')
\\
=&\tfrac{1}{2}\gamma^{AD}L_{D}{}^{\rho}L_{\rho}{}^{E}(f_{E}{}^{BF}K_{F}{}^{C}+f_{E}{}^{CF}K_{F}{}^{B})\delta(\sigma-\sigma')
\\
=& 2\gamma^{AD}K_{D}{}^{E}(f_{E}{}^{BF}K_{F}{}^{C}+f_{E}{}^{CF}K_{F}{}^{B})\delta(\sigma-\sigma') \,\, ,
\end{aligned}
\end{equation}
where in the third step we exploited \eqref{J-K-differential-relation-components} and in the last one $L_{D}{}^{\rho}L_{\rho}{}^{E}=K_{D}{}^{E}$. At this point we can slightly rearrange the indices so as to use the relation \eqref{important-K-relation} on both terms,
\begin{equation}
\begin{aligned}
\{ &L_{\tau}{}^{A}(\sigma),K^{BC}(\sigma') \}=
\\
=&-2\gamma^{AD}(\gamma^{CG}K_{D}{}^{E}K_{G}{}^{F}f_{EF}{}^{B}+\gamma^{BG}K_{D}{}^{E}K_{G}{}^{F}f_{EF}{}^{C})\delta(\sigma-\sigma')
\\
=& -2\gamma^{AD}(\gamma^{CG}f_{DG}{}^{B}-\gamma^{CG}f_{DG}{}^{H}K_{H}{}^{B}+\gamma^{BG}f_{DG}{}^{C}-\gamma^{BG}f_{DG}{}^{H}K_{H}{}^{C})\delta(\sigma-\sigma')
\\
=&+2\gamma^{AD}(f_{D}{}^{CH}K_{H}{}^{B}+f_{D}{}^{BH}K_{H}{}^{C})\delta(\sigma-\sigma')
\\
=&+2(f^{AB}{}_{D}K^{DC}+f^{AC}{}_{D}K^{DB})\delta(\sigma-\sigma') \,\, ,
\end{aligned}
\end{equation}
obtaining the desired result.

\subsection{AF-Symmetric-space sigma models}\label{appendix:C-AFSSSM}
\paragraph{Calculation of $\{ B_{\tau}{}^{A}(\sigma),B_{\tau}{}^{B}(\sigma') \}$.} Using the relation \eqref{AF-SSSM-Jtau-momentum} one finds
\begin{align}
\{ B_{\tau}{}^{A}(\sigma),B_{\tau}{}^{B}(\sigma') \} =& \gamma^{AC}\gamma^{BD}\{ \pi_{\mu}(\sigma)A_{C}{}^{\mu}(\sigma),\pi_{\nu}(\sigma')A_{D}{}^{\nu}(\sigma') \} 
\\
=& -\tfrac{1}{4}\gamma^{AC}\gamma^{BD} \gamma_{FE}B_{\tau}{}^{F}A_{\mu}{}^{E}\bigl(A_{D}{}^{\rho}\, \partial_{\rho}A_{C}{}^{\mu}-A_{C}{}^{\rho}\, \partial_{\rho}A_{D}{}^{\mu}\bigr)\delta(\sigma-\sigma') \,\, ,
\notag 
\end{align}
where we used that the starting point is the same as in the undeformed setting \eqref{SSSM-Jtau-Jtau-bracket-starting-point}, and obtained a $B_{\tau}$ term in the second line using \eqref{AF-SSSM-J-vielbeine-and-momentum(J)} to rewrite $\pi_{\mu}$, after evaluating the integral. Successive manipulations solely rely on off-shell properties of $A\equiv L$, which remain unchanged, and hence we immediately end-up with the deformed version of \eqref{intermediate-3},
\begin{equation}
\{ B_{\tau}{}^{A}(\sigma),B_{\tau}{}^{B}(\sigma') \}=-2 \gamma^{AC}\gamma^{BD}\gamma_{FE}B_{\tau}{}^{F}(K_{D}{}^{G}f_{GC}{}^{H}-K_{C}{}^{G}f_{GD}{}^{H})K_{H}{}^{E}\delta(\sigma-\sigma') \,\, .
\end{equation}
The final manipulations needed in the last few steps of the undeformed calculation \eqref{SSSM-Jtau-Jtau-bracket-last-steps} rely on the fact that $A\equiv L$ is an eigenvector of $K$, simple rearrangements and the identity \eqref{J-K-algebraic-relation-components}; the same steps can once again be performed straightforwardly after recalling that
\begin{equation}
B_{\tau}{}^{A}=-\pi_{\mu}G^{\mu\nu}A_{\nu}{}^{A} \,\, ,
\end{equation}
since this implies that any identity involving the Lie-algebra index $A$ of $A_{\nu}{}^{A}$ automatically transfers to $B_{\tau}{}^{A}$. From the definition \eqref{AF-SSSM-J-Jtilde-def} it is indeed also easy to recognise that $B$ is also an eigenvector of $K$, as is any object defined as $X\propto \mathrm{Ad}_{\mathfrak{g}}\circ \mathrm{P}_{\mathfrak{m}}(Y)$. Repeating the steps after \eqref{intermediate-3} one finally reaches the expected result:
\begin{equation}
\{ B_{\tau}{}^{A}(\sigma),B_{\tau}{}^{B}(\sigma') \} = 2f^{AB}{}_{C} \, B_{\tau}{}^{C}(\sigma)\delta(\sigma-\sigma') \,\, .
\end{equation}

\paragraph{Calculation of $\{ B_{\tau}{}^{A}(\sigma),A_{\sigma}{}^{B}(\sigma') \}$.} Again using the relation \eqref{AF-SSSM-Jtau-momentum} one finds
\begin{equation}
\{ B_{\tau}{}^{A}(\sigma),A_{\sigma}{}^{B}(\sigma') \}=-\gamma^{AC} \{ \pi_{\mu}(\sigma)A_{C}{}^{\mu}(\sigma),A_{\sigma}{}^{B}(\sigma') \} \,\, ,
\end{equation}
and also in this case the bracket has exactly the same initial form \eqref{SSSM-Jtau-Jsigma-bracket-first-step} as in the undeformed setting. Now, however, all the steps after \eqref{SSSM-Jtau-Jsigma-bracket-first-step} can be repeated verbatim, since they only involve $A\equiv L$, which is not modified by the deformation. We hence obtain
\begin{equation}
\{ B_{\tau}{}^{A}(\sigma),A_{\sigma}{}^{B}(\sigma') \}=2f^{AB}{}_{C} \, A_{\sigma}{}^{C}(\sigma)\delta(\sigma-\sigma')-4K^{AB}(\sigma') \, \partial_{\sigma}\delta(\sigma-\sigma') \,\, .
\end{equation}

\paragraph{Calculation of $\{ B_{\tau}{}^{A}(\sigma),K^{BC}(\sigma') \}$.} Also in this last case the relation \eqref{AF-SSSM-Jtau-momentum} allows us to recast the bracket in the same form as undeformed SSSM \eqref{SSSM-Jtau-K-bracket}. All steps follow, since once again they only involve the unmodified current $A\equiv L$. We immediately obtain
\begin{equation}
\{ B_{\tau}{}^{A}(\sigma),K^{BC}(\sigma') \} = 2\Bigl(f^{AB}{}_{D}K^{DC}(\sigma)+f^{AC}{}_{D}K^{DB}(\sigma)\Bigr)\delta(\sigma-\sigma') \,\,  .
\end{equation}

\subsection{Yang-Baxter sigma models}\label{appendix:C-YB}
\paragraph{Calculation of $\{ J_{\tau}{}^{A}(\sigma),j_{\sigma}{}^{B}(\sigma') \}$.} Thanks to the first relation in \eqref{Jtau-and-Jsigma-functions-of-momentum}, this Poisson bracket becomes exactly the same as $\{ j_{\tau}{}^{A}(\sigma),j_{\sigma}{}^{B}(\sigma') \}$ for the PCM or $\{ \mathfrak{J}_{\tau}{}^{A}(\sigma),j_{\sigma}{}^{B}(\sigma') \}$ for the AF-PCM, as can be seen for instance around equation (C.33) in \cite{Bielli:2024ach}. For this reason we just import the result, with no need to actually do any computation:
\begin{equation}\label{YB-Jtau-jsigma-bracket}
\{ J_{\tau}{}^{A}(\sigma),j_{\sigma}{}^{B}(\sigma') \} = f^{AB}{}_{C}j_{\sigma}{}^{C}\delta(\sigma-\sigma')-\gamma^{AB}\partial_{\sigma}\delta(\sigma-\sigma') \,\, .
\end{equation}

\paragraph{Calculation of $\{ J_{\tau}{}^{A}(\sigma),M_{C}{}^{B}(\sigma') \}$.} Using \eqref{Jtau-and-Jsigma-functions-of-momentum} and noting that $M$ is independent of $\pi_{\mu}$, only the second contribution to the Poisson bracket \eqref{def-PB} remains, leading to
\begin{equation}
\begin{aligned}
\{ J_{\tau}{}^{A}(\sigma),M_{C}{}^{B}(\sigma') \}=&-\gamma^{AD}\{ \pi_{\mu}(\sigma)j_{D}{}^{\mu}(\sigma),M_{C}{}^{B}(\sigma') \}=
\\
=&\gamma^{AD}\int \mathrm{d}\sigma'' \, j_{D}{}^{\rho}(\sigma) \delta(\sigma-\sigma'') \partial_{\rho}M_{C}{}^{B}(\sigma')\delta(\sigma'-\sigma'')
\\
=&\gamma^{AD}j_{D}{}^{\rho}\partial_{\rho}M_{C}{}^{B} \delta(\sigma-\sigma')
\\
=&\gamma^{AD}j_{D}{}^{\rho}j_{\rho}{}^{E}(M_{C}{}^{F}f_{FE}{}^{B}+f_{EC}{}^{F}M_{F}{}^{B})\delta(\sigma-\sigma'),
\end{aligned}
\end{equation}
where in the last step we used the relation \eqref{M-derivative-components}. Using now that $j_{D}{}^{\rho}j_{\rho}{}^{E}=\delta_{D}{}^{E}$ and slightly rearranging the indices one obtains the desired result:
\begin{equation}\label{YB-Jtau-M-bracket}
\{ J_{\tau}{}^{A}(\sigma),M_{C}{}^{B}(\sigma') \} = \bigl(f^{AB}{}_{D}M_{C}{}^{D}-f_{C}{}^{AD}M_{D}{}^{B}\bigr)\delta(\sigma-\sigma') \,\, . 
\end{equation}

\paragraph{Calculation of $\{ J_{\tau}{}^{A}(\sigma),J_{\tau}^{B}(\sigma') \}$.} Also in this case the relation \eqref{Jtau-and-Jsigma-functions-of-momentum} is sufficient to bring this Poisson bracket to exactly the same form as $\{ j_{\tau}{}^{A}(\sigma),j_{\tau}^{B}(\sigma') \}$ for the PCM or $\{ \mathfrak{J}_{\tau}{}^{A}(\sigma),\mathfrak{J}_{\tau}^{B}(\sigma') \}$ for the AF-PCM, again by repeating arguments around equation (C.23) in \cite{Bielli:2024ach}. For this reason we again simply import the result:
\begin{equation}\label{YB-Jtau-Jtau-bracket}
\{ J_{\tau}{}^{A}(\sigma),J_{\tau}^{B}(\sigma') \} = f^{AB}{}_{C} \, J_{\tau}{}^{C}(\sigma)\delta(\sigma-\sigma') \,\, .
\end{equation}

\paragraph{Calculation of $\{ J_{\tau}{}^{A}(\sigma),J_{\sigma}^{B}(\sigma') \}$.} For this bracket we first exploit the second relation in \eqref{Jtau-and-Jsigma-functions-of-momentum} and successively use the Leibniz rule \eqref{def-PB-Leibniz} to recast all contributions in terms of brackets that have already been computed,
\begin{align}
\{ J_{\tau}{}^{A}(\sigma)&,J_{\sigma}^{B}(\sigma') \}=\{ J_{\tau}{}^{A}(\sigma), j_{\sigma}{}^{B}(\sigma')-J_{\tau}{}^{C}(\sigma')M_{C}{}^{B}(\sigma') \}
\\
=& \{ J_{\tau}{}^{A}(\sigma),j_{\sigma}{}^{B}(\sigma') \} - \{J_{\tau}{}^{A}(\sigma),J_{\tau}{}^{C}(\sigma')\}M_{C}{}^{B}(\sigma') - \{J_{\tau}{}^{A}(\sigma),M_{C}{}^{B}(\sigma')\}J_{\tau}{}^{C}(\sigma')
\notag\\
=& +f^{AB}{}_{C}j_{\sigma}{}^{C}\delta(\sigma-\sigma')-\gamma^{AB}\partial_{\sigma}\delta(\sigma-\sigma')-f^{AC}{}_{D} \, J_{\tau}{}^{D}(\sigma)\delta(\sigma-\sigma')M_{C}{}^{B}(\sigma')
\notag\\
&-\bigl(f^{AB}{}_{D}M_{C}{}^{D}-f_{C}{}^{AD}M_{D}{}^{B}\bigr)\delta(\sigma-\sigma')J_{\tau}{}^{C}(\sigma')
\notag\\
=&+f^{AB}{}_{C}(j_{\sigma}{}^{C}-J_{\tau}{}^{D}M_{D}{}^{C})\delta(\sigma-\sigma')-\gamma^{AB}\partial_{\sigma}\delta(\sigma-\sigma')
\notag\\
&+J_{\tau}{}^{D}M_{C}{}^{B}(f_{D}{}^{AC}-f^{AC}{}_{D})\delta(\sigma-\sigma') \,\, 
\notag
\end{align}
where in the last step we collected similar structures after some relabelling. The last line clearly vanishes, while the one above can be rewritten in terms of $J_{\sigma}$ using again the second relation in \eqref{Jtau-and-Jsigma-functions-of-momentum}, thus obtaining
\begin{equation}\label{YB-Jtau-Jsigma-bracket}
\{ J_{\tau}{}^{A}(\sigma),J_{\sigma}^{B}(\sigma') \}=f^{AB}{}_{C}J_{\sigma}{}^{C}\delta(\sigma-\sigma')-\gamma^{AB}\partial_{\sigma}\delta(\sigma-\sigma') \,\, .
\end{equation}

\paragraph{Calculation of $\{ J_{\sigma}{}^{A}(\sigma),J_{\sigma}^{B}(\sigma') \}$.} This final bracket is slightly trickier than the previous ones, since we will need to carefully keep track of the $\sigma$ and $\sigma'$ dependencies of each quantity in order to be able to combine them in the correct way. We start by exploiting the second relation in \eqref{Jtau-and-Jsigma-functions-of-momentum}, the Leibniz rule \eqref{def-PB-Leibniz}, antisymmetry \eqref{def-PB-antisymmetry} and the vanishing of the last three brackets in \eqref{YB-brackets-summary} to obtain the expression
\begin{align}
\{ J_{\sigma}{}^{A}(\sigma)&,J_{\sigma}^{B}(\sigma') \} = \{ j_{\sigma}{}^{A}(\sigma)-J_{\tau}{}^{C}(\sigma)M_{C}{}^{A}(\sigma) \, , \, j_{\sigma}{}^{B}(\sigma')-J_{\tau}{}^{D}(\sigma')M_{D}{}^{B}(\sigma') \}
\\
=& -M_{C}{}^{A}(\sigma)\{J_{\tau}{}^{C}(\sigma),j_{\sigma}{}^{B}(\sigma')\}+M_{C}{}^{B}(\sigma')\{J_{\tau}{}^{C}(\sigma'),j_{\sigma}{}^{A}(\sigma)\}+
\notag\\
&+M_{C}{}^{A}(\sigma)M_{D}{}^{B}(\sigma')\{ J_{\tau}{}^{C}(\sigma),J_{\tau}{}^{D}(\sigma') \}+
\notag\\
&+J_{\tau}{}^{D}(\sigma')M_{C}{}^{A}(\sigma)\{J_{\tau}{}^{C}(\sigma),M_{D}{}^{B}(\sigma')\}-J_{\tau}{}^{D}(\sigma)M_{C}{}^{B}(\sigma')\{J_{\tau}{}^{C}(\sigma'),M_{D}{}^{A}(\sigma)\} \,\, .
\notag
\end{align}
At this point we can proceed by simply substituting the results from the other Poisson brackets, namely \eqref{YB-Jtau-jsigma-bracket}, \eqref{YB-Jtau-Jtau-bracket} and \eqref{YB-Jtau-M-bracket}, thus obtaining
\begin{align}\label{YB-intermediate-step}
\{ J_{\sigma}{}^{A}(\sigma),J_{\sigma}^{B}(\sigma') \} =& -M_{C}{}^{A}(\sigma)\Bigl( f^{CB}{}_{D}j_{\sigma}{}^{D}(\sigma)\delta(\sigma-\sigma')-\gamma^{CB}\partial_{\sigma}\delta(\sigma-\sigma')\Bigr)
\\
& +M_{C}{}^{B}(\sigma')\Bigl( f^{CA}{}_{D}j_{\sigma}{}^{D}(\sigma')\delta(\sigma'-\sigma)-\gamma^{CA}\partial_{\sigma'}\delta(\sigma'-\sigma) \Bigr)
\notag \\
&+ M_{C}{}^{A}(\sigma)M_{D}{}^{B}(\sigma')f^{CD}{}_{E}J_{\tau}{}^{E}(\sigma)\delta(\sigma-\sigma')
\notag \\
&+J_{\tau}{}^{D}(\sigma')M_{C}{}^{A}(\sigma) \Bigl( f^{CB}{}_{E}M_{D}{}^{E}(\sigma)\delta(\sigma\!-\!\sigma')-f_{D}{}^{CE}M_{E}{}^{B}(\sigma)\delta(\sigma\!-\!\sigma') \Bigr)
\notag \\
&-J_{\tau}{}^{D}(\sigma)M_{C}{}^{B}(\sigma') \Bigl( f^{CA}{}_{E}M_{D}{}^{E}(\sigma')\delta(\sigma'\!-\!\sigma)-f_{D}{}^{CE}M_{E}{}^{A}(\sigma')\delta(\sigma'\!-\!\sigma) \Bigr)
\notag  .
\end{align}
Given that $\delta(\sigma-\sigma')=\delta(\sigma'-\sigma)$, in most of the terms we can at this point simply forget about $\sigma'$ and set all dependences to $\sigma$. This is however not true for the last two terms on the first two lines, which involve derivatives of the $\delta$ function. Starting from these terms,
\begin{equation}
\begin{aligned}
&M^{BA}(\sigma)\partial_{\sigma}\delta(\sigma-\sigma')-M^{AB}(\sigma')\partial_{\sigma'}\delta(\sigma'-\sigma)
\\
&=-M^{AB}(\sigma)\partial_{\sigma}\delta(\sigma-\sigma')+M^{AB}(\sigma')\partial_{\sigma}\delta(\sigma-\sigma')
\\
&=\partial_{\sigma}M^{AB}(\sigma)\delta(\sigma-\sigma')
\\
&=\gamma^{AC}\partial_{\sigma}Z^{\mu}(\sigma)\partial_{\mu}M_{C}{}^{B}(\sigma)\delta(\sigma-\sigma')
\\
&=\gamma^{AC}\partial_{\sigma}Z^{\mu}(\sigma)j_{\mu}{}^{E}(\sigma)\Bigl( M_{C}{}^{D}(\sigma)f_{DE}{}^{B}+f_{EC}{}^{D}M_{D}{}^{B}(\sigma) \Bigr)\delta(\sigma-\sigma')
\\
&=j_{\sigma}{}^{E}(\sigma)\Bigl( M^{AD}(\sigma)f_{DE}{}^{B}+f_{E}{}^{AD}M_{D}{}^{B}(\sigma) \Bigr) \delta(\sigma-\sigma')
\\
&= -j_{\sigma}{}^{E}(\sigma)\Bigl( M_{C}{}^{B}(\sigma)f^{CA}{}_{E}-M_{C}{}^{A}(\sigma)f^{CB}{}_{E} \Bigr) \delta(\sigma-\sigma')\,\,  .
\end{aligned}
\end{equation}
In the first step we exploited antisymmetry of $M^{AB}$ on the first term and the second identity in \eqref{def-delta-func-identity-2} on the second term. In the second step we combined the two terms via the identity \eqref{def-delta-func-identity-3}, taken with $x=\sigma$ and $y=\sigma'$. In the third and fourth steps we used the chain rule and \eqref{M-derivative-components}. In the fifth step we again used  the chain rule in the reversed direction and in the final step we simply relabelled some indices and rearranged using antisymmetry of $M^{AB}$ and of the structure constants. At this point all terms can be treated as functions of $\sigma$ only and combining the first two and the last two lines of \eqref{YB-intermediate-step} in structures which are antisymmetric in $A\leftrightarrow B$ leads to
\begin{equation}\label{YB-intermediate-step2}
\begin{aligned}
\{ J_{\sigma}{}^{A}(\sigma),J_{\sigma}^{B}(\sigma') \}=& -j_{\sigma}{}^{E}\Bigl( M_{C}{}^{B}f^{CA}{}_{E}-M_{C}{}^{A}f^{CB}{}_{E} \Bigr) \delta(\sigma-\sigma')
\\
&+
j_{\sigma}{}^{D}\Bigl( M_{C}{}^{B}f^{CA}{}_{D} -M_{C}{}^{A}f^{CB}{}_{D}\Bigr)\delta(\sigma-\sigma')
\\
&+M_{C}{}^{A}M_{D}{}^{B}f^{CD}{}_{E}J_{\tau}{}^{E}\delta(\sigma-\sigma')
\\
&+J_{\tau}{}^{D}M_{D}{}^{E}\Bigl( M_{C}{}^{A}f^{CB}{}_{E}-M_{C}{}^{B}f^{CA}{}_{E} \Bigr)\delta(\sigma-\sigma')
\\
&+J_{\tau}{}^{D}f_{D}{}^{CE}\Bigl( M_{C}{}^{B}M_{E}{}^{A}-M_{C}{}^{A}M_{E}{}^{B} \Bigr)\delta(\sigma-\sigma') \,\, .
\end{aligned}
\end{equation}
After a quick index relabelling it is clear that the first two lines cancel each other, while the third line cancels the last term on the last line, leaving us with three terms only,
\begin{align}
\{ J_{\sigma}{}^{A}(\sigma)&,J_{\sigma}^{B}(\sigma') \}= 
\notag\\
=&J_{\tau}{}^{D} \Bigl( M_{D}{}^{E}M_{C}{}^{A}f^{CB}{}_{E}-M_{D}{}^{E}M_{C}{}^{B}f^{CA}{}_{E}+f_{D}{}^{CE}M_{C}{}^{B}M_{E}{}^{A}\Bigr)\delta(\sigma-\sigma')
\notag\\
=&J_{\tau}{}^{D} \Bigl( M_{C}{}^{B}M_{E}{}^{A}f^{CE}{}_{D}+M_{C}{}^{A}M_{ED}f^{CEB}-M_{C}{}^{B}M_{ED}f^{CEA}\Bigr)\delta(\sigma-\sigma')
\notag\\
=&J_{\tau}{}^{F}\gamma_{FE} \Bigl( M_{C}{}^{B}M_{D}{}^{A}f^{CDE}+M_{C}{}^{A}M_{D}{}^{E}f^{CDB}-M_{C}{}^{B}M_{D}{}^{E}f^{CDA}\Bigr)\delta(\sigma-\sigma')
\notag\\
=&J_{\tau}{}^{F}\gamma_{FE} \,\, \eta^2 c^2 f^{ABE} \delta(\sigma-\sigma') \,\, .
\end{align}
In the first step we simply copied the fourth line and first term of the fifth line in \eqref{YB-intermediate-step2}. In the second step we slightly relabelled some indices and used antisymmetry of $M_{AB}$ and the structure constants. In the third step we performed some further relabelling and raising of indices so as to be able to directly exploit \eqref{mCYBE-components} in the last step, obtaining
\begin{equation}
\{ J_{\sigma}{}^{A}(\sigma),J_{\sigma}^{B}(\sigma') \}= \eta^2 c^2f^{AB}{}_{C}J_{\tau}{}^{C}\delta(\sigma-\sigma') \,\, .
\end{equation}

\subsection{Non-Abelian T-dual sigma models}\label{appendix:C-TD}
\paragraph{Calculation of $\{ \partial_{\sigma}X^{A}(\sigma),\pi_{B}(\sigma') \}$.} Only the first term survives in the definition of the Poisson bracket \eqref{def-PB} and using the second identity in \eqref{def-delta-func-identity-2} one finds
\begin{equation}
\begin{aligned}
\{ \partial_{\sigma}X^{A}(\sigma),\pi_{B}(\sigma') \}&=\int \mathrm{d}\sigma'' \, \frac{\partial(\partial_{\sigma}X^{A}(\sigma))}{\partial X^{C}(\sigma'')}\frac{\partial \pi_{B}(\sigma')}{\partial\pi_{C}(\sigma'')}
\\
&= -\int \mathrm{d}\sigma'' \, \delta_{C}{}^{A}\partial_{\sigma''}\delta(\sigma''-\sigma) \, \delta_{B}{}^{C}\delta(\sigma'-\sigma'') \,\, ,
\end{aligned}
\end{equation}
hence obtaining the desired result
\begin{equation}\label{TD-partialX-pi-bracket}
\{ \partial_{\sigma}X^{A}(\sigma),\pi_{B}(\sigma') \}=-\delta_{B}{}^{A}\partial_{\sigma'}\delta(\sigma'-\sigma) \,\, .
\end{equation}

\paragraph{Calculation of $\{ M^{AB}(\sigma),\pi_{C}(\sigma') \}$.} Since $M$ is independent of the momenta, also in this case it is clear from the definition of Poisson bracket \eqref{def-PB} that the first term vanishes. Using the relation \eqref{TD-M-derivative-components} one then finds
\begin{equation}
\begin{aligned}
\{ M^{AB}(\sigma),\pi_{C}(\sigma') \}&=\int \mathrm{d}\sigma'' \, \frac{\partial M^{AB}(\sigma)}{\partial X^{D}(\sigma'')}\frac{\partial \pi_{C}(\sigma')
}{\partial\pi_{D}(\sigma'')}
\\
&=\int \mathrm{d}\sigma'' \, f_{D}{}^{AB}\delta(\sigma-\sigma'')\delta_{C}{}^{D}\delta(\sigma'-\sigma'') \, ,
\end{aligned}
\end{equation}
hence giving
\begin{equation}\label{TD-M-pi-bracket}
\{ M^{AB}(\sigma),\pi_{C}(\sigma') \}=f^{AB}{}_{C}\delta(\sigma-\sigma') \,\, .
\end{equation}

\paragraph{Calculation of $\{ L_{\tau}{}^{A}(\sigma),L_{\tau}{}^{B}(\sigma') \}$.} Substituting for $L_{\tau}$ using \eqref{TD-Jtau-and-Jsigma-functions-of-momentum}, invoking the Leibniz rule \eqref{def-PB-Leibniz} and the trivial vanishing of the last four brackets in \eqref{TD-brackets-summary}, we find
\begin{equation}
\begin{aligned}
\!\!\{ L_{\tau}{}^{A}(\sigma)&,L_{\tau}{}^{B}(\sigma') \} = \{\partial_{\sigma}X^{A}(\sigma)+\pi_{C}(\sigma)M^{CA}(\sigma) , \partial_{\sigma}X^{B}(\sigma')+\pi_{D}(\sigma')M^{DB}(\sigma') \}
\\
=& +M^{DB}(\sigma')\{ \partial_{\sigma}X^{A}(\sigma), \pi_{D}(\sigma') \}-M^{DA}(\sigma)\{ \partial_{\sigma}X^{B}(\sigma'), \pi_{D}(\sigma) \}+
\\&+M^{DB}(\sigma')\pi_{C}(\sigma)\{M^{CA}(\sigma),\pi_{D}(\sigma')\}-M^{DA}(\sigma)\pi_{C}(\sigma')\{M^{CB}(\sigma'),\pi_{D}(\sigma)\}  .
\end{aligned}
\end{equation}
At this point one can exploit the results of the brackets \eqref{TD-partialX-pi-bracket} and \eqref{TD-M-pi-bracket}, being careful in keeping track of the $\sigma$ and $\sigma'$ dependences, to find
\begin{align}
\{ L_{\tau}{}^{A}(\sigma), L_{\tau}{}^{B}(\sigma') \} =&-M^{AB}(\sigma')\partial_{\sigma'}\delta(\sigma'-\sigma)+M^{BA}(\sigma)\partial_{\sigma}\delta(\sigma-\sigma')+
\notag \\
&+M^{DB}(\sigma')\pi_{C}(\sigma)f^{CA}{}_{D}\delta(\sigma-\sigma')-M^{DA}(\sigma)\pi_{C}(\sigma')f^{CB}{}_{D}\delta(\sigma'-\sigma)
\notag \\
=& +M^{AB}(\sigma')\partial_{\sigma}\delta(\sigma-\sigma')-M^{AB}(\sigma)\partial_{\sigma}\delta(\sigma-\sigma')+
\notag \\
&+(M^{BD}f_{D}{}^{AC}-M^{AD}f_{D}{}^{BC})\pi_{C}\delta(\sigma-\sigma') \,\, ,
\end{align}
where in the second step we used the second identity in \eqref{def-delta-func-identity-2} and antisymmetry of $M^{AB}$ and the structure constants to rearrange. Now we can use \eqref{def-delta-func-identity-3} to combine the two terms on the first line and the explicit component form of $M^{AB}$ on the second line \eqref{TD-M-in-components} to find
\begin{equation}
\begin{aligned}
\{ L_{\tau}{}^{A}(\sigma),L_{\tau}{}^{B}(\sigma') \} =&+\partial_{\sigma}M^{AB}(\sigma) \delta(\sigma-\sigma')+
\\
&+\gamma^{AE}\gamma^{BF}X^{G}(f_{GF}{}^{D}f_{DE}{}^{C}-f_{GE}{}^{D}f_{DF}{}^{C})\pi_{C}\delta(\sigma-\sigma') \,\, .
\end{aligned}
\end{equation}
Using \eqref{TD-M-derivative-components} on the first line and the Jacobi identity \eqref{def-Jacobi-components} on the second line we see
\begin{equation}
\begin{aligned}
\{ L_{\tau}{}^{A}(\sigma),L_{\tau}{}^{B}(\sigma') \} =&f^{AB}{}_{C}\partial_{\sigma}X^{C}\delta(\sigma-\sigma')-\gamma^{AE}\gamma^{BF}X^{G}f_{EF}{}^{D}f_{GD}{}^{C}\pi_{C}\delta(\sigma-\sigma')
\\
=&f^{AB}{}_{C}\partial_{\sigma}X^{C}\delta(\sigma-\sigma')-f^{AB}{}_{D}M^{DC}\pi_{C}\delta(\sigma-\sigma') \,\, ,
\end{aligned}
\end{equation}
which gives the desired result after again using \eqref{TD-momentum-jtilde-relation}:
\begin{equation}
\{ L_{\tau}{}^{A}(\sigma),L_{\tau}{}^{B}(\sigma') \} = f^{AB}{}_{C}L_{\tau}{}^{C}\delta(\sigma-\sigma') \,\, .
\end{equation}

\paragraph{Calculation of $\{ L_{\tau}{}^{A}(\sigma),L_{\sigma}{}^{B}(\sigma') \}$.} Exploiting the relations \eqref{TD-Jtau-and-Jsigma-functions-of-momentum}, the Leibniz rule \eqref{def-PB-Leibniz}, the vanishing of the bracket of two momenta and \eqref{TD-partialX-pi-bracket} and \eqref{TD-M-pi-bracket}, one finds
\begin{equation}
\begin{aligned}
\{ L_{\tau}{}^{A}(\sigma),L_{\sigma}{}^{B}(\sigma') \} &= \{\partial_{\sigma}X^{A}(\sigma)+\pi_{C}(\sigma)M^{CA}(\sigma) ,-\pi_{D}(\sigma')\gamma^{DB} \}
\\
&=-\gamma^{DB} \{\partial_{\sigma}X^{A}(\sigma),\pi_{D}(\sigma')\}-\gamma^{DB}\pi_{C}(\sigma) \{M^{CA}(\sigma),\pi_{D}(\sigma')\}
\\
&=\gamma^{AB}\partial_{\sigma'}\delta(\sigma'-\sigma)-\gamma^{DB}\pi_{C}(\sigma)f^{CA}{}_{D}\delta(\sigma-\sigma') 
\\
&=\gamma^{AB}\partial_{\sigma'}\delta(\sigma'-\sigma)-f^{AB}{}_{C}\pi_{D}(\sigma)\gamma^{DC}\delta(\sigma-\sigma') \,\, ,
\end{aligned}
\end{equation}
finally arriving at the desired expression after using the second relation in \eqref{def-delta-func-identity-2} and the first relation in \eqref{TD-Jtau-and-Jsigma-functions-of-momentum},
\begin{equation}\label{TD-jtildetau-jtildesigma-bracket}
\{ L_{\tau}{}^{A}(\sigma),L_{\sigma}{}^{B}(\sigma') \}=f^{AB}{}_{C}L_{\sigma}{}^{C}\delta(\sigma-\sigma')-\gamma^{AB}\partial_{\sigma}\delta(\sigma-\sigma') \,\, .
\end{equation}

\paragraph{Calculation of $\{ L_{\sigma}{}^{A}(\sigma),L_{\sigma}{}^{B}(\sigma') \}$.} Exploiting \eqref{TD-Jtau-and-Jsigma-functions-of-momentum} it is straightforward to see that
\begin{equation}
\{ L_{\sigma}{}^{A}(\sigma),L_{\sigma}{}^{B}(\sigma') \}=\gamma^{CA}\gamma^{DB}\{ \pi_{C}(\sigma),\pi_{D}(\sigma') \}=0 \,\, .
\end{equation}

\subsection{Principal Chiral models with Wess-Zumino term}\label{appendix:C-PCMWZ}
\paragraph{Calculation of $\{ j_{A}{}^{\mu}\pi_{\mu}(\sigma),j_{B}{}^{\nu}B_{\nu\rho}\partial_{\sigma}Z^{\rho}(\sigma') \}$.} Since the second entry in this bracket is independent of the momentum $\pi$, after using \eqref{def-PB-Leibniz} the definition \eqref{def-PB} becomes
\begin{align}
\{ &j_{A}{}^{\mu}\pi_{\mu}(\sigma),j_{B}{}^{\nu}B_{\nu\rho}\partial_{\sigma}Z^{\rho}(\sigma') \} =
\notag \\
=&+j_{A}{}^{\mu}(\sigma)\{ \pi_{\mu}(\sigma),j_{B}{}^{\nu}B_{\nu\rho}\partial_{\sigma}Z^{\rho}(\sigma') \}
\notag \\
=&-j_{A}{}^{\mu}(\sigma) \int \mathrm{d}\sigma'' \, \frac{\partial \pi_{\mu}(\sigma)}{\partial \pi_{\lambda}(\sigma'')} \, \frac{\partial\bigl(j_{B}{}^{\nu}B_{\nu\rho}\partial_{\sigma}Z^{\rho}(\sigma')\bigr)}{\partial Z^{\lambda}(\sigma'')}
\notag \\
=&-j_{A}{}^{\mu}(\sigma)\int \mathrm{d}\sigma'' \delta(\sigma-\sigma'') 
\notag \\
& \Bigl( B_{\nu\rho}(\sigma')\partial_{\sigma'}Z^{\rho}(\sigma')\frac{\partial j_{B}{}^{\nu}(\sigma')}{\partial Z^{\mu}(\sigma'')} +j_{B}{}^{\nu}(\sigma')\partial_{\sigma'}Z^{\rho}(\sigma')\frac{\partial B_{\nu\rho}(\sigma')}{\partial Z^{\mu}(\sigma'')}+j_{B}{}^{\nu}(\sigma')B_{\nu\rho}(\sigma')\frac{\partial \bigl(\partial_{\sigma'}Z^{\rho}(\sigma')\bigr)}{\partial Z^{\mu}(\sigma'')}\Bigr)
\notag \\
=&-j_{A}{}^{\mu}(\sigma)\int \mathrm{d}\sigma'' \delta(\sigma-\sigma'') 
\notag \\
& \Bigl(\!-\partial_{\sigma'}Z^{\rho}(\sigma')B_{\delta\rho}(\sigma')j_{E}{}^{\delta}(\sigma')j_{B}{}^{\nu}(\sigma')\partial_{\mu}j_{\nu}{}^{E}(\sigma')\delta(\sigma''\!\!\!-\!\sigma')\!+\!j_{B}{}^{\nu}(\sigma')\partial_{\sigma'}Z^{\rho}(\sigma')\partial_{\mu}B_{\nu\rho}(\sigma')\delta(\sigma''\!\!\!-\!\sigma') 
\notag \\
& \,\, - j_{B}{}^{\nu}(\sigma'')B_{\nu\mu}(\sigma'')\partial_{\sigma''}\delta(\sigma''-\sigma')-\partial_{\sigma''}\bigl(j_{B}{}^{\nu}(\sigma'')B_{\nu\mu}(\sigma'')\bigr)\delta(\sigma''-\sigma')\Bigr) \,\, .
\end{align}
To obtain the last line we performed two main manipulations. In the first term on the line above the last, we introduced an identity in the form $\delta_{\nu}{}^{\delta}=j_{\nu}{}^{E}(\sigma')j_{E}{}^{\delta}(\sigma')$ and then exploited $j_{\nu}{}^{E}(\sigma')\partial_{\lambda}j_{B}{}^{\nu}(\sigma')=-j_{B}{}^{\nu}(\sigma')\partial_{\lambda}j_{\nu}{}^{E}(\sigma')$, which follows by differentiating $j_{B}{}^{\nu}(\sigma')j_{\nu}{}^{E}(\sigma')=\delta_{B}{}^{E}$. In the third term on the line above the last, we first exploited the second $\delta$-function identity in equation \eqref{def-delta-func-identity-2} and successively the identity \eqref{def-delta-func-identity-3} with $y=\sigma'$ and $x=\sigma''$. At this point the integral collapses to $\sigma''=\sigma$ thanks to the $\delta$-function and the last term translates into two new terms, one containing $\partial_{\sigma}j_{B}{}^{\nu}$ and the other $\partial_{\sigma}B_{\nu\lambda}$. The former can be manipulated exactly as done above, namely introducing an identity in the form $\delta_{\nu}{}^{\delta}=j_{\nu}{}^{E}(\sigma')j_{E}{}^{\delta}(\sigma')$ and following subsequent steps. In the end all this leads to five contributions, two of which contain derivatives of $j_{\nu}{}^{E}$ and two of which involve derivatives of the $B$-field, leading to the desired result
\begin{align}\label{B-field-brackets-intermediate}
\{ j_{A}{}^{\mu}\pi_{\mu}(\sigma),j_{B}{}^{\nu}B_{\nu\rho}\partial_{\sigma}Z^{\rho}(\sigma') \} =&+ \partial_{\sigma}Z^{\rho} j_{A}{}^{\mu}j_{B}{}^{\nu}(\partial_{\rho}B_{\nu\mu}-\partial_{\mu}B_{\nu\rho})\delta(\sigma-\sigma') 
\notag \\
&+\partial_{\sigma}Z^{\rho} j_{A}{}^{\mu}j_{B}{}^{\nu}j_{E}{}^{\lambda}(B_{\lambda\rho}\partial_{\mu}j_{\nu}{}^{E}-B_{\lambda\mu}\partial_{\rho}j_{\nu}{}^{E})\delta(\sigma-\sigma') 
\notag \\
&+ j_{A}{}^{\mu}(\sigma)j_{B}{}^{\nu}(\sigma)B_{\nu\mu}(\sigma)\partial_{\sigma}\delta(\sigma-\sigma')\,\, .
\end{align}

\paragraph{Calculation of $I^{AB}$.} Substituting \eqref{B-field-brackets-intermediate} in the definition \eqref{IAB-combination-definition} of $I^{AB}$ one obtains
\begin{equation}
I^{AB}= I^{AB}_{1}+I^{AB}_{2}+I^{AB}_{3} \,\, .
\end{equation}
The first contribution is defined, and rearranged, as
\begin{equation}
\begin{aligned}
I^{AB}_{1}\equiv &+(\gamma^{AC}\gamma^{BD}-\gamma^{AD}\gamma^{BC})\partial_{\sigma}Z^{\rho}j_{C}{}^{\mu}j_{D}{}^{\nu}(\partial_{\rho}B_{\nu\mu}-\partial_{\mu}B_{\nu\rho})\delta(\sigma-\sigma')
\\
=&+\gamma^{AC}\gamma^{BD}\partial_{\sigma}Z^{\rho}j_{C}{}^{\mu}j_{D}{}^{\nu}(\partial_{\rho}B_{\nu\mu}-\partial_{\mu}B_{\nu\rho}-\partial_{\rho}B_{\mu\nu}+\partial_{\nu}B_{\mu\rho})\delta(\sigma-\sigma')
\\
=&+\gamma^{AC}\gamma^{BD}\partial_{\sigma}Z^{\rho}j_{C}{}^{\mu}j_{D}{}^{\nu}(\partial_{\rho}B_{\nu\mu}+H_{\rho\nu\mu})\delta(\sigma-\sigma')
\\
=&+\gamma^{AC}\gamma^{BD}\partial_{\sigma}Z^{\rho}j_{C}{}^{\mu}j_{D}{}^{\nu}\partial_{\rho}B_{\nu\mu}\delta(\sigma-\sigma')-f^{AB}{}_{C}j_{\sigma}{}^{C}\delta(\sigma-\sigma') \,\, ,
\end{aligned}
\end{equation}
where in the first step we relabelled the indices on $\gamma$ to combine the two terms, in the second step we introduced $H_{\rho\nu\mu}\equiv \partial_{\rho}B_{\nu\mu}+\partial_{\nu}B_{\mu\rho}+\partial_{\mu}B_{\rho\nu}$ and in the last step we exploited
\begin{equation}
H_{\rho\nu\mu}=j_{\rho}{}^{E}j_{\nu}{}^{F}j_{\mu}{}^{G}H_{EFG}=j_{\rho}{}^{E}j_{\nu}{}^{F}j_{\mu}{}^{G}f_{EFG} \,\, ,
\end{equation}
which on group manifolds is ensured by the Maurer-Cartan equation, and finally exploited that $j_{\mu}{}^{A}$ is a vielbeine. The second contribution is defined, and rearranged, as
\begin{align}
I^{AB}_{2}=& +(\gamma^{AC}\gamma^{BD}-\gamma^{AD}\gamma^{BC})\partial_{\sigma}Z^{\rho}j_{C}{}^{\mu}j_{D}{}^{\nu}j_{E}{}^{\lambda}(B_{\lambda\rho}\partial_{\mu}j_{\nu}{}^{E}-B_{\lambda\mu}\partial_{\rho}j_{\nu}{}^{E})\delta(\sigma-\sigma')
\notag \\
=& +\gamma^{AC}\gamma^{BD}\partial_{\sigma}Z^{\rho}j_{C}{}^{\mu}j_{D}{}^{\nu}j_{E}{}^{\lambda}(B_{\lambda\rho}\partial_{\mu}j_{\nu}{}^{E}-B_{\lambda\mu}\partial_{\rho}j_{\nu}{}^{E}-B_{\lambda\rho}\partial_{\nu}j_{\mu}{}^{E}+B_{\lambda\nu}\partial_{\rho}j_{\mu}{}^{E})\delta(\sigma-\sigma')
\notag \\
=&+\gamma^{AC}\gamma^{BD}\partial_{\sigma}Z^{\rho}j_{C}{}^{\mu}j_{D}{}^{\nu}j_{E}{}^{\lambda}B_{\lambda\rho}(\partial_{\mu}j_{\nu}{}^{E}-\partial_{\nu}j_{\mu}{}^{E})\delta(\sigma-\sigma')
\notag \\
&+\gamma^{AC}\gamma^{BD}\partial_{\sigma}Z^{\rho}j_{C}{}^{\mu}j_{D}{}^{\nu}j_{E}{}^{\lambda}(B_{\lambda\nu}\partial_{\rho}j_{\mu}{}^{E}-B_{\lambda\mu}\partial_{\rho}j_{\nu}{}^{E})\delta(\sigma-\sigma')
\notag \\
=&-f^{AB}{}_{C}\gamma^{CD}j_{D}{}^{\mu}B_{\mu\nu}\partial_{\sigma}Z^{\nu}\delta(\sigma-\sigma')
\notag \\
&+\gamma^{AC}\gamma^{BD}\partial_{\sigma}Z^{\rho}j_{C}{}^{\mu}j_{D}{}^{\nu}j_{E}{}^{\lambda}(B_{\lambda\nu}\partial_{\rho}j_{\mu}{}^{E}-B_{\lambda\mu}\partial_{\rho}j_{\nu}{}^{E})\delta(\sigma-\sigma') \,\, ,
\end{align}
where in the first step we relabelled the indices on $\gamma$ to combine the two terms, in the second step we combined structures of similar type and in the third step we simplified the first type by exploiting the Maurer-Cartan equation in components, namely $\partial_{\mu}j_{\nu}{}^{E}-\partial_{\nu}j_{\mu}{}^{E}=-j_{\mu}{}^{F}j_{\nu}{}^{G}f_{FG}{}^{E}$, finally rearranging using that $j_{\mu}{}^{A}$ is a vielbeine. The last contribution is then defined, and rearranged, as
\begin{align}
I^{AB}_{3}\!\!=\!&+\!\gamma^{AC}\gamma^{BD}j_{C}{}^{\mu}(\sigma)j_{D}{}^{\nu}(\sigma)B_{\nu\mu}(\sigma)\partial_{\sigma}\delta(\sigma\!-\!\sigma')\!-\!\gamma^{AD}\gamma^{BC}j_{C}{}^{\mu}(\sigma')j_{D}{}^{\nu}(\sigma')B_{\nu\mu}(\sigma')\partial_{\sigma'}\delta(\sigma'\!-\!\sigma)
\notag \\
=&+\gamma^{AC}\gamma^{BD}B_{DC}(\sigma)\partial_{\sigma}\delta(\sigma-\sigma')-\gamma^{AD}\gamma^{BC}B_{DC}(\sigma')\partial_{\sigma'}\delta(\sigma'-\sigma)
\notag \\
=&-B^{AB}(\sigma)\partial_{\sigma}\delta(\sigma-\sigma')-B^{AB}(\sigma')\partial_{\sigma'}\delta(\sigma'-\sigma)
\notag \\
=&-B^{AB}(\sigma)\partial_{\sigma}\delta(\sigma-\sigma')+B^{AB}(\sigma')\partial_{\sigma}\delta(\sigma-\sigma')
\notag \\
=&+\partial_{\sigma}B^{AB}(\sigma)\delta(\sigma-\sigma') \,\, ,
\end{align}
where in the first step we converted from curved to flat indices using the inverse vielbeine $j_{A}{}^{\mu}$, in the second step raised the indices with $\gamma$ and exploited antisymmetry of the $B$-field, in the third step we used the second identity in \eqref{def-delta-func-identity-2} on the second term and finally in the fourth step we combined the two pieces using \eqref{def-delta-func-identity-3}. At this point one can combine the three results above to obtain
\begin{align}\label{IAB-rearranged-intermediate}
I^{AB}=&-f^{AB}{}_{C}\gamma^{CD}j_{D}{}^{\mu}B_{\mu\nu}\partial_{\sigma}Z^{\nu}\delta(\sigma-\sigma')-f^{AB}{}_{C}j_{\sigma}{}^{C}\delta(\sigma-\sigma')+\partial_{\sigma}B^{AB}(\sigma)\delta(\sigma-\sigma')
\notag \\
&+\gamma^{AC}\gamma^{BD}\partial_{\sigma}Z^{\rho}j_{C}{}^{\mu}j_{D}{}^{\nu}(\partial_{\rho}B_{\nu\mu}+j_{E}{}^{\lambda}B_{\lambda\nu}\partial_{\rho}j_{\mu}{}^{E}-j_{E}{}^{\lambda}B_{\lambda\mu}\partial_{\rho}j_{\nu}{}^{E})\delta(\sigma-\sigma') \,\, ,
\end{align}
which clearly reproduces the result \eqref{IAB-combination-rearranged} provided that the second line cancels the last term on the first line. This is indeed the case, as one can easily find out by noting that
\begin{equation}
\partial_{\rho}B_{\nu\mu}=\partial_{\rho}(j_{\nu}{}^{E}j_{\mu}{}^{F}B_{EF})=j_{\mu}{}^{F}B_{EF}\partial_{\rho}j_{\nu}{}^{E}+j_{\nu}{}^{E}B_{EF}\partial_{\rho}j_{\mu}{}^{F}+j_{\nu}{}^{E}j_{\mu}{}^{F}\partial_{\rho}B_{EF} \,\, ,
\end{equation}
such that after converting all curved indices to flat ones, the second line in \eqref{IAB-rearranged-intermediate} precisely reduces to $-\partial_{\sigma}Z^{\rho}\partial_{\rho}B^{AB}(\sigma)\delta(\sigma-\sigma')$, which cancels the last term on the first line and leaves us with the desired expression \eqref{IAB-combination-rearranged}.

\paragraph{Calculation of $\{L_{\tau}{}^{A}(\sigma),L_{\tau}{}^{B}(\sigma')\}$.} Using the definition of $L$ in \eqref{PCM-WZ-BA-currents}, the relation \eqref{PCM-WZ-canonical-momentum} to the canonical momentum, and the intermediate brackets \eqref{PCM-WZ-building-block-brackets} one finds
\begin{align}
\{L_{\tau}{}^{A}(\sigma),L_{\tau}{}^{B}(\sigma')\} =&+\tfrac{1}{\hay^2}\gamma^{AC}\gamma^{BD}\{j_{C}{}^{\mu}\pi_{\mu}(\sigma),j_{D}{}^{\nu}\pi_{\nu}(\sigma')\}-\tfrac{\kay}{\hay^2}I^{AB}
\notag\\
&+\tfrac{\kay}{\hay^2}\Bigl( \gamma^{AC}\{j_{C}{}^{\mu}\pi_{\mu}(\sigma),j_{\sigma}{}^{B}(\sigma')\}-\gamma^{BC}\{j_{C}{}^{\mu}\pi_{\mu}(\sigma'),j_{\sigma}{}^{A}(\sigma)\} \Bigr)
\notag\\
=&-\tfrac{1}{\hay^2}f^{AB}{}_{C}\gamma^{CD}j_{D}{}^{\mu}\pi_{\mu}\delta(\sigma-\sigma')
\notag\\
&+\tfrac{\kay}{\hay^2}f^{AB}{}_{C}\gamma^{CD}j_{D}{}^{\mu}B_{\mu\nu}\partial_{\sigma}Z^{\nu}\delta(\sigma-\sigma')+\tfrac{\kay}{\hay^2}f^{AB}{}_{C}j_{\sigma}{}^{C}\delta(\sigma-\sigma')
\notag\\
&-\tfrac{\kay}{\hay^2}f^{AB}{}_{C}j_{\sigma}{}^{C}\delta(\sigma-\sigma')+\tfrac{\kay}{\hay^2}\gamma^{AB}\partial_{\sigma}\delta(\sigma-\sigma')
\notag\\
&+\tfrac{\kay}{\hay^2}f^{BA}{}_{C}j_{\sigma}{}^{C}\delta(\sigma'-\sigma)-\tfrac{\kay}{\hay^2}\gamma^{BA}\partial_{\sigma'}\delta(\sigma'-\sigma) \,\, .
\end{align}
At this point, the two non-ultralocal terms on the last two lines combine using the second identity in \eqref{def-delta-func-identity-2}, while all the other terms reassemble into
\begin{equation}
\tfrac{1}{\hay}L_{\tau}{}^{C}=\tfrac{1}{\hay}j_{\tau}{}^{C}-\tfrac{\kay}{\hay^2}j_{\sigma}{}^{C}=-\tfrac{1}{\hay^2}\gamma^{CD}j_{D}{}^{\mu}\pi_{\mu}+\tfrac{\kay}{\hay^2}\gamma^{CD}j_{D}{}^{\mu}B_{\mu\nu}\partial_{\sigma}Z^{\nu}-\tfrac{\kay}{\hay^2}j_{\sigma}{}^{C} \,\, ,
\end{equation}
as can be noted from $L_{\tau}{}^{C}$ in \eqref{PCM-WZ-BA-currents} and the relation \eqref{PCM-WZ-canonical-momentum}, thus leading to
\begin{equation}
\{L_{\tau}{}^{A}(\sigma),L_{\tau}{}^{B}(\sigma')\} = \tfrac{1}{\hay}f^{AB}{}_{C}L_{\tau}{}^{C}\delta(\sigma-\sigma')+\tfrac{2\kay}{\hay^2}\gamma^{AB}\partial_{\sigma}\delta(\sigma-\sigma') \,\, .
\end{equation}

\paragraph{Calculation of $\{L_{\tau}{}^{A}(\sigma),L_{\sigma}{}^{B}(\sigma')\}$.} Using the definition \eqref{PCM-WZ-BA-currents}, the relation \eqref{PCM-WZ-canonical-momentum} to the canonical momentum, and the intermediate brackets \eqref{PCM-WZ-building-block-brackets} one finds
\begin{align}
\{L_{\tau}{}^{A}(\sigma),L_{\sigma}{}^{B}(\sigma')\}=&-\tfrac{1}{\hay}\gamma^{AC}\{j_{C}{}^{\mu}\pi_{\mu},j_{\sigma}{}^{B}(\sigma')\}+\tfrac{\kay^2}{\hay^3}\gamma^{BC}\{j_{C}{}^{\mu}\pi_{\mu}(\sigma'),j_{\sigma}{}^{A}(\sigma)\}
\notag \\
&-\tfrac{\kay}{\hay^3}\gamma^{AC}\gamma^{BD}\{j_{C}{}^{\mu}\pi_{\mu}(\sigma),j_{D}{}^{\nu}\pi_{\nu}(\sigma')\}+\tfrac{\kay^2}{\hay^3}I^{AB}
\notag\\
=&+\tfrac{1}{\hay}f^{AB}{}_{C}j_{\sigma}{}^{C}\delta(\sigma-\sigma')-\tfrac{1}{\hay}\gamma^{AB}\partial_{\sigma}\delta(\sigma-\sigma')
\\
&-\tfrac{\kay^2}{\hay^3}f^{BA}{}_{C}j_{\sigma}{}^{C}\delta(\sigma'-\sigma)+\tfrac{\kay^2}{\hay^3}\gamma^{BA}\partial_{\sigma'}\delta(\sigma'-\sigma)
\notag \\
&+\tfrac{\kay}{\hay^3}f^{AB}{}_{C}\gamma^{CD}j_{D}{}^{\mu}\pi_{\mu}\delta(\sigma-\sigma')
\notag \\
&-\tfrac{\kay^2}{\hay^3}f^{AB}{}_{C}\gamma^{CD}j_{D}{}^{\mu}B_{\mu\nu}\partial_{\sigma}Z^{\nu}\delta(\sigma-\sigma')-\tfrac{\kay^2}{\hay^3}f^{AB}{}_{C}j_{\sigma}{}^{C}\delta(\sigma-\sigma') \,\, .
\notag
\end{align}
Also in this case the two non-ultralocal terms combine after using the second identity in \eqref{def-delta-func-identity-2}, while the other terms reassemble into
\begin{equation}
\tfrac{1}{\hay}L_{\sigma}{}^{C}=\tfrac{1}{\hay}j_{\sigma}{}^{C}-\tfrac{\kay}{\hay^2}j_{\tau}{}^{C}=\tfrac{1}{\hay}j_{\sigma}{}^{C}+\tfrac{\kay}{\hay^3}\gamma^{CD}j_{D}{}^{\mu}\pi_{\mu}-\tfrac{\kay^2}{\hay^3}\gamma^{CD}j_{D}{}^{\mu}B_{\mu\nu}\partial_{\sigma}Z^{\nu} \,\, ,
\end{equation}
as can be noted from $L_{\sigma}{}^{C}$ in \eqref{PCM-WZ-BA-currents} and the relation \eqref{PCM-WZ-canonical-momentum}, thus leading to
\begin{equation}
\{L_{\tau}{}^{A}(\sigma),L_{\sigma}{}^{B}(\sigma')\} = \tfrac{1}{\hay}f^{AB}{}_{C}L_{\sigma}{}^{C}\delta(\sigma-\sigma')-\tfrac{1}{\hay}(1+\tfrac{\kay^2}{\hay^2})\gamma^{AB}\partial_{\sigma}\delta(\sigma-\sigma') \,\, .
\end{equation}

\paragraph{Calculation of $\{L_{\sigma}{}^{A}(\sigma),L_{\sigma}{}^{B}(\sigma')\}$.}  Using the definition \eqref{PCM-WZ-BA-currents}, the relation \eqref{PCM-WZ-canonical-momentum} to the canonical momentum, and the intermediate brackets \eqref{PCM-WZ-building-block-brackets} one finds
\begin{align}\label{PCM-WZ-AA_bracket-intermediate}
\{L_{\sigma}{}^{A}(\sigma),L_{\sigma}{}^{B}(\sigma')\}=&+\tfrac{\kay^2}{\hay^4}\gamma^{AC}\gamma^{BD}\{j_{C}{}^{\mu}\pi_{\mu}(\sigma),j_{D}{}^{\nu}\pi_{\nu}(\sigma')\}-\tfrac{\kay^3}{\hay^4}I^{AB}
\notag \\
&+\tfrac{\kay}{\hay^2}\Bigl( \gamma^{AC}\{j_{C}{}^{\mu}\pi_{\mu}(\sigma),j_{\sigma}{}^{B}(\sigma')\}-\gamma^{BC}\{j_{C}{}^{\mu}\pi_{\mu}(\sigma'),j_{\sigma}{}^{A}(\sigma)\}\Bigr)
\notag \\
=&-\tfrac{\kay^2}{\hay^4}f^{AB}{}_{C}\gamma^{CD}j_{D}{}^{\mu}\pi_{\mu}\delta(\sigma-\sigma')
\notag\\
&+\tfrac{\kay^3}{\hay^4}f^{AB}{}_{C}\gamma^{CD}j_{D}{}^{\mu}B_{\mu\nu}\partial_{\sigma}Z^{\nu}\delta(\sigma-\sigma')+\tfrac{\kay^3}{\hay^4}f^{AB}{}_{C}j_{\sigma}{}^{C}\delta(\sigma-\sigma')
\notag\\
&-\tfrac{\kay}{\hay^2}f^{AB}{}_{C}j_{\sigma}{}^{C}\delta(\sigma-\sigma')+\tfrac{\kay}{\hay^2}\gamma^{AB}\partial_{\sigma}\delta(\sigma-\sigma')
\notag\\
&+\tfrac{\kay}{\hay^2}f^{BA}{}_{C}j_{\sigma}{}^{C}\delta(\sigma'-\sigma)-\tfrac{\kay}{\hay^2}\gamma^{BA}\partial_{\sigma'}\delta(\sigma'-\sigma) \,\, .
\end{align}
The non-ultralocal terms sum up once again using the second relation in \eqref{def-delta-func-identity-2}, while all the other contributions neither combine into $L_{\tau}$ nor into $L_{\sigma}$, but a combination of the two. Collecting similar terms and matching coefficients, it is not hard to realise what this combination should be: using $L_{\tau}{}^{C},L_{\sigma}{}^{C}$ in \eqref{PCM-WZ-BA-currents} and the relation \eqref{PCM-WZ-canonical-momentum} one finds
\begin{equation}
-\tfrac{2\kay}{\hay^2}L_{\sigma}{}^{C}-\tfrac{\kay^2}{\hay^3}L_{\tau}{}^{C}=(\tfrac{\kay^3}{\hay^4}-\tfrac{2\kay}{\hay^2})j_{\sigma}{}^{C}-\tfrac{\kay^2}{\hay^4}\gamma^{CD}j_{D}{}^{\mu}\pi_{\mu}+\tfrac{\kay^3}{\hay^4}\gamma^{CD}j_{D}{}^{\mu}B_{\mu\nu}\partial_{\sigma}Z^{\nu} \,\, ,
\end{equation}
which is precisely what similar terms in \eqref{PCM-WZ-AA_bracket-intermediate} sum up to. This finally leads to
\begin{equation}
\{L_{\sigma}{}^{A}(\sigma),L_{\sigma}{}^{B}(\sigma')\}=f^{AB}{}_{C}\Bigl( -\tfrac{2\kay}{\hay^2}L_{\sigma}{}^{C}-\tfrac{\kay^2}{\hay^3}L_{\tau}{}^{C}\Bigr)\delta(\sigma-\sigma')+\tfrac{2\kay}{\hay^2}\gamma^{AB}\partial_{\sigma}\delta(\sigma-\sigma') \,\, .
\end{equation}

\bibliographystyle{utphys}
\bibliography{master}

@article{Conti_2019,
   title={The $\mathrm{T}\overline{\mathrm{T}
}$ perturbation and its geometric interpretation},
   volume={2019},
   ISSN={1029-8479},
   url={http://dx.doi.org/10.1007/JHEP02(2019)085},
   DOI={10.1007/jhep02(2019)085},
   number={2},
   journal={Journal of High Energy Physics},
   publisher={Springer Science and Business Media LLC},
   author={Conti, Riccardo and Negro, Stefano and Tateo, Roberto},
   year={2019},
   month=feb }

@article{Ran:2024vgl,
    author = "Ran, Xi-Yang and Hao, Feng and Yamada, Masatoshi",
    title = "{Geometric realization via irrelevant deformations induced by the stress-energy tensor}",
    eprint = "2410.02537",
    archivePrefix = "arXiv",
    primaryClass = "hep-th",
    doi = "10.1103/PhysRevD.111.085033",
    journal = "Phys. Rev. D",
    volume = "111",
    number = "8",
    pages = "085033",
    year = "2025"
}

@article{Caputa:2020lpa,
    author = "Caputa, Pawel and Datta, Shouvik and Jiang, Yunfeng and Kraus, Per",
    title = "{Geometrizing $ T\overline{T} $}",
    eprint = "2011.04664",
    archivePrefix = "arXiv",
    primaryClass = "hep-th",
    reportNumber = "CERN-TH-2020-188",
    doi = "10.1007/JHEP03(2021)140",
    journal = "JHEP",
    volume = "03",
    pages = "140",
    year = "2021",
    note = "[Erratum: JHEP 09, 110 (2022)]"
}

@article{Ran:2025xas,
    author = "Ran, Xi-Yang and Hao, Feng and Ouyang, Hao",
    title = "{Holography for stress-energy tensor flows}",
    eprint = "2508.12275",
    archivePrefix = "arXiv",
    primaryClass = "hep-th",
    doi = "10.1103/v2g6-l983",
    journal = "Phys. Rev. D",
    volume = "112",
    number = "8",
    pages = "L081905",
    year = "2025"
}

@book{Babelon:2003qtg,
    author = "Babelon, Olivier and Bernard, Denis and Talon, Michel",
    title = "{Introduction to Classical Integrable Systems}",
    doi = "10.1017/CBO9780511535024",
    isbn = "978-0-521-03670-2, 978-0-511-53502-4",
    publisher = "Cambridge University Press",
    series = "Cambridge Monographs on Mathematical Physics",
    year = "2003"
}

@article{BREZIN1979442,
title = {Remarks about the existence of non-local charges in two-dimensional models},
journal = {Physics Letters B},
volume = {82},
number = {3},
pages = {442-444},
year = {1979},
issn = {0370-2693},
doi = {https://doi.org/10.1016/0370-2693(79)90263-6},
url = {https://www.sciencedirect.com/science/article/pii/0370269379902636},
author = {E. Brézin and C. Itzykson and J. Zinn-Justin and J.-B. Zuber}
}

@article{Driezen:2021cpd,
    author = "Driezen, Sibylle",
    title = "{Modave Lectures on Classical Integrability in 2d Field Theories}",
    eprint = "2112.14628",
    archivePrefix = "arXiv",
    primaryClass = "hep-th",
    doi = "10.22323/1.404.0002",
    journal = "PoS",
    volume = "Modave2021",
    pages = "002",
    year = "2022"
}

@article{Lacroix:2021iit,
    author = "Lacroix, Sylvain",
    title = "{Four-dimensional Chern\textendash{}Simons theory and integrable field theories}",
    eprint = "2109.14278",
    archivePrefix = "arXiv",
    primaryClass = "hep-th",
    doi = "10.1088/1751-8121/ac48ed",
    journal = "J. Phys. A",
    volume = "55",
    number = "8",
    pages = "083001",
    year = "2022"
}

@article{Drinfeld:1985rx,
    author = "Drinfeld, V. G.",
    title = "{Hopf algebras and the quantum Yang-Baxter equation}",
    journal = "Sov. Math. Dokl.",
    volume = "32",
    pages = "254--258",
    year = "1985"
}

@article{osti_5772113,
  author       = {Pohlmeyer, K and Rehren, K},
  title        = {Reduction of the two-dimensional O(n) nonlinear sigma-model},
  doi          = {10.1063/1.524026},
  url          = {https://www.osti.gov/biblio/5772113},
  journal      = {J. Math. Phys. (N.Y.); (United States)},
  issn         = {ISSN JMAPA},
  volume       = {20:12},
  place        = {United States},
  year         = {1979},
  month        = {12}}

@article{Bernard:1990jw,
    author = "Bernard, Denis",
    title = "{Hidden Yangians in 2-D massive current algebras}",
    reportNumber = "SACLAY-SPH-T-90-109",
    doi = "10.1007/BF02099123",
    journal = "Commun. Math. Phys.",
    volume = "137",
    pages = "191--208",
    year = "1991"
}

@article{Morone:2024ffm,
    author = "Morone, Tommaso and Negro, Stefano and Tateo, Roberto",
    title = "{Gravity and TT flows in higher dimensions}",
    eprint = "2401.16400",
    archivePrefix = "arXiv",
    primaryClass = "hep-th",
    doi = "10.1016/j.nuclphysb.2024.116605",
    journal = "Nucl. Phys. B",
    volume = "1005",
    pages = "116605",
    year = "2024"
}

@Article{MAILLET198654,
title = {New integrable canonical structures in two-dimensional models},
journal = {Nuclear Physics B},
volume = {269},
number = {1},
pages = {54-76},
year = {1986},
issn = {0550-3213},
doi = {https://doi.org/10.1016/0550-3213(86)90365-2},
url = {https://www.sciencedirect.com/science/article/pii/0550321386903652},
author = {Jean-Michael Maillet}
}

@article{Tsolakidis:2024wut,
    author = "Tsolakidis, Evangelos",
    title = "{Massive gravity generalization of $ T\overline{T} $ deformations}",
    eprint = "2405.07967",
    archivePrefix = "arXiv",
    primaryClass = "hep-th",
    doi = "10.1007/JHEP09(2024)167",
    journal = "JHEP",
    volume = "09",
    pages = "167",
    year = "2024"
}

@article{Morone:2024sdg,
    author = "Morone, Tommaso and Tateo, Roberto",
    title = "{Solutions to the Ricci Flow via Einstein Field Equations}",
    eprint = "2411.10265",
    archivePrefix = "arXiv",
    primaryClass = "hep-th",
    month = "11",
    year = "2024"
}

@article{Brizio:2024arr,
    author = "Brizio, Nicol\`o and Morone, Tommaso and Tateo, Roberto",
    title = "{Stress-energy tensor deformations, Ricci flows and black holes}",
    eprint = "2408.06031",
    archivePrefix = "arXiv",
    primaryClass = "hep-th",
    month = "8",
    year = "2024"
}

@article{Babaei-Aghbolagh:2024hti,
    author = "Babaei-Aghbolagh, H. and He, Song and Morone, Tommaso and Ouyang, Hao and Tateo, Roberto",
    title = "{Geometric Formulation of Generalized Root-TT Deformations}",
    eprint = "2405.03465",
    archivePrefix = "arXiv",
    primaryClass = "hep-th",
    doi = "10.1103/PhysRevLett.133.111602",
    journal = "Phys. Rev. Lett.",
    volume = "133",
    number = "11",
    pages = "111602",
    year = "2024"
}

@article{Conti:2022egv,
    author = "Conti, Riccardo and Romano, Jacopo and Tateo, Roberto",
    title = "{Metric approach to a $ \mathrm{T}\overline{\mathrm{T}} $-like deformation in arbitrary dimensions}",
    eprint = "2206.03415",
    archivePrefix = "arXiv",
    primaryClass = "hep-th",
    doi = "10.1007/JHEP09(2022)085",
    journal = "JHEP",
    volume = "09",
    pages = "085",
    year = "2022"
}

@article{Ivanov:2003uj,
    author = "Ivanov, E. A. and Zupnik, B. M.",
    title = "{New approach to nonlinear electrodynamics: Dualities as symmetries of interaction}",
    eprint = "hep-th/0303192",
    archivePrefix = "arXiv",
    doi = "10.1134/1.1842299",
    journal = "Phys. Atom. Nucl.",
    volume = "67",
    pages = "2188--2199",
    year = "2004"
}

@article{Conti:2019dxg,
    author = "Conti, Riccardo and Negro, Stefano and Tateo, Roberto",
    title = "{Conserved currents and $\text{T}\bar{\text{T}}_s$ irrelevant deformations of 2D integrable field theories}",
    eprint = "1904.09141",
    archivePrefix = "arXiv",
    primaryClass = "hep-th",
    doi = "10.1007/JHEP11(2019)120",
    journal = "JHEP",
    volume = "11",
    pages = "120",
    year = "2019"
}

@article{Guica:2019vnb,
      author         = "Guica, Monica",
      title          = "{On correlation functions in $J\bar T$-deformed CFTs}",
      journal        = "J. Phys.",
      volume         = "A52",
      year           = "2019",
      number         = "18",
      pages          = "184003",
      doi            = "10.1088/1751-8121/ab0ef3",
      eprint         = "1902.01434",
      archivePrefix  = "arXiv",
      primaryClass   = "hep-th",
      SLACcitation   = "%%CITATION = ARXIV:1902.01434;%%"
}

@article{Cavaglia:2016oda,
      author         = "Cavagli\`a, Andrea and Negro, Stefano and Sz\'ecs\'enyi, Istv\'an M. and Tateo, Robertoo",
      title          = "{$T \bar{T}$-deformed 2D Quantum Field Theories}",
      journal        = "JHEP",
      volume         = "10",
      year           = "2016",
      pages          = "112",
      doi            = "10.1007/JHEP10(2016)112",
      eprint         = "1608.05534",
      archivePrefix  = "arXiv",
      primaryClass   = "hep-th",
      SLACcitation   = "%%CITATION = ARXIV:1608.05534;%%"
}

@article{Zamolodchikov:2004ce,
      author         = "Zamolodchikov, Alexander B.",
      title          = "{Expectation value of composite field T anti-T in
                        two-dimensional quantum field theory}",
      year           = "2004",
      eprint         = "hep-th/0401146",
      archivePrefix  = "arXiv",
      primaryClass   = "hep-th",
      reportNumber   = "BONN-TH-2004-02",
      SLACcitation   = "%%CITATION = HEP-TH/0401146;%%"
}

@article{Smirnov:2016lqw,
      author         = "Smirnov, F. A. and Zamolodchikov, A. B.",
      title          = "{On space of integrable quantum field theories}",
      journal        = "Nucl. Phys.",
      volume         = "B915",
      year           = "2017",
      pages          = "363-383",
      doi            = "10.1016/j.nuclphysb.2016.12.014",
      eprint         = "1608.05499",
      archivePrefix  = "arXiv",
      primaryClass   = "hep-th",
      SLACcitation   = "%%CITATION = ARXIV:1608.05499;%%"
}

@article{Borsato:2023dis,
    author = "Borsato, Riccardo",
    title = "{Lecture notes on current\textendash{}current deformations}",
    eprint = "2312.13847",
    archivePrefix = "arXiv",
    primaryClass = "hep-th",
    doi = "10.1140/epjc/s10052-024-12966-5",
    journal = "Eur. Phys. J. C",
    volume = "84",
    number = "6",
    pages = "648",
    year = "2024"
}

@article{Hoare:2021dix,
    author = "Hoare, Ben",
    title = "{Integrable deformations of sigma models}",
    eprint = "2109.14284",
    archivePrefix = "arXiv",
    primaryClass = "hep-th",
    doi = "10.1088/1751-8121/ac4a1e",
    journal = "J. Phys. A",
    volume = "55",
    number = "9",
    pages = "093001",
    year = "2022"
}

@article{Borsato:2018idb,
	Archiveprefix = {arXiv},
	Author = {Borsato, Riccardo and Wulff, Linus},
	Doi = {10.1007/JHEP08(2018)027},
	Eprint = {1806.04083},
	Journal = {JHEP},
	Pages = {027},
	Primaryclass = {hep-th},
	Reportnumber = {NORDITA 2018-041, NORDITA-2018-041},
	Slaccitation = {%%CITATION = ARXIV:1806.04083;%%},
	Title = {{Non-abelian T-duality and Yang-Baxter deformations of Green-Schwarz strings}},
	Volume = {08},
	Year = {2018}}

@article{Borsato:2016pas,
	Archiveprefix = {arXiv},
	Author = {Borsato, Riccardo and Wulff, Linus},
	Doi = {10.1103/PhysRevLett.117.251602},
	Eprint = {1609.09834},
	Journal = {Phys. Rev. Lett.},
	Number = {25},
	Pages = {251602},
	Primaryclass = {hep-th},
	Reportnumber = {IMPERIAL-TP-RB-2016-06},
	Slaccitation = {%%CITATION = ARXIV:1609.09834;%%},
	Title = {{Integrable Deformations of $T$-Dual $\sigma$ Models}},
	Volume = {117},
	Year = {2016}}

@article{Yang:1968rm,
	Author = {Yang, Chen-Ning and Yang, C. P.},
	Doi = {10.1063/1.1664947},
	Journal = {J. Math. Phys.},
	Pages = {1115-1122},
	Slaccitation = {%%CITATION = JMAPA,10,1115;%%},
	Title = {{Thermodynamics of one-dimensional system of bosons with repulsive delta function interaction}},
	Volume = {10},
	Year = {1969}}

@article{Abdalla:1982yd,
	Author = {Abdalla, E. and Forger, M. and Gomes, M.},
	Doi = {10.1016/0550-3213(82)90238-3},
	Journal = {Nucl. Phys.},
	Pages = {181},
	Slaccitation = {%%CITATION = NUPHA,B210,181;%%},
	Title = {On the origin of anomalies in the quantum nonlocal charge for the generalized nonlinear sigma models},
	Volume = {B210},
	Year = {1982}}

@article{Arutyunov:2009ga,
	Archiveprefix = {arXiv},
	Author = {Arutyunov, Gleb and Frolov, Sergey},
	Doi = {10.1088/1751-8113/42/25/254003},
	Eprint = {0901.4937},
	Journal = {J. Phys. A},
	Pages = {254003},
	Primaryclass = {hep-th},
	Slaccitation = {%%CITATION = ARXIV:0901.4937;%%},
	Title = {Foundations of the {$AdS_5 \times S^5$} Superstring. Part {I}},
	Volume = {A42},
	Year = {2009}}

@article{Arutyunov:2009mi,
	Archiveprefix = {arXiv},
	Author = {Arutyunov, Gleb and de Leeuw, Marius and Torrielli, Alessandro},
	Doi = {10.1016/j.nuclphysb.2009.03.024},
	Eprint = {0902.0183},
	Journal = {Nucl. Phys.},
	Pages = {319-350},
	Primaryclass = {hep-th},
	Slaccitation = {%%CITATION = 0902.0183;%%},
	Title = {The Bound State {S}-Matrix for {$AdS_5 \times S^5$} Superstring},
	Volume = {B819},
	Year = {2009}}

@article{Beisert:2005fw,
	Archiveprefix = {arXiv},
	Author = {Beisert, Niklas and Staudacher, Matthias},
	Doi = {10.1016/j.nuclphysb.2005.06.038},
	Eprint = {hep-th/0504190},
	Journal = {Nucl. Phys.},
	Pages = {1-62},
	Primaryclass = {hep-th},
	Slaccitation = {%%CITATION = HEP-TH/0504190;%%},
	Title = {Long-range {$PSU(2,2|4)$} {B}ethe ansaetze for gauge theory and strings},
	Volume = {B727},
	Year = {2005}}

@article{Beisert:2005tm,
	Archiveprefix = {arXiv},
	Author = {Beisert, Niklas},
	Eprint = {hep-th/0511082},
	Journal = {Adv. Theor. Math. Phys.},
	Pages = {945-979},
	Primaryclass = {hep-th},
	Slaccitation = {%%CITATION = HEP-TH/0511082;%%},
	Title = {The {$SU(2|2)$} dynamic {$S$}-matrix},
	Volume = {12},
	Year = {2008}}

@article{Beisert:2007ds,
	Archiveprefix = {arXiv},
	Author = {Beisert, Niklas},
	Eprint = {0704.0400},
	Journal = {PoS},
	Pages = {002},
	Primaryclass = {nlin.SI},
	Slaccitation = {%%CITATION = ARXIV:0704.0400;%%},
	Title = {The {S}-matrix of {$AdS/CFT$} and {Y}angian symmetry},
	Volume = {SOLVAY},
	Year = {2006}}

@article{Beisert:2010jr,
    author = "Beisert, Niklas and others",
    title = "{Review of AdS/CFT Integrability: An Overview}",
    eprint = "1012.3982",
    archivePrefix = "arXiv",
    primaryClass = "hep-th",
    reportNumber = "AEI-2010-175, CERN-PH-TH-2010-306, HU-EP-10-87, HU-MATH-2010-22, KCL-MTH-10-10, UMTG-270, UUITP-41-10",
    doi = "10.1007/s11005-011-0529-2",
    journal = "Lett. Math. Phys.",
    volume = "99",
    pages = "3--32",
    year = "2012"
}

@article{Bena:2003wd,
	Archiveprefix = {arXiv},
	Author = {Bena, Iosif and Polchinski, Joseph and Roiban, Radu},
	Doi = {10.1103/PhysRevD.69.046002},
	Eprint = {hep-th/0305116},
	Journal = {Phys. Rev.},
	Pages = {046002},
	Slaccitation = {%%CITATION = HEP-TH/0305116;%%},
	Title = {Hidden symmetries of the {$AdS_5 \times S^5$} superstring},
	Volume = {D69},
	Year = {2004}}

@article{Bernard:1992ya,
	Archiveprefix = {arXiv},
	Author = {Bernard, Denis},
	Doi = {10.1142/S0217979293003371},
	Eprint = {hep-th/9211133},
	Journal = {Int. J. Mod. Phys.},
	Pages = {3517-3530},
	Slaccitation = {%%CITATION = HEP-TH/9211133;%%},
	Title = {An Introduction to {Y}angian Symmetries},
	Volume = {B7},
	Year = {1993}}

@article{Delduc:2013qra,
	Archiveprefix = {arXiv},
	Author = {Delduc, Francois and Magro, Marc and Vicedo, Benoit},
	Doi = {10.1103/PhysRevLett.112.051601},
	Eprint = {1309.5850},
	Journal = {Phys. Rev. Lett.},
	Number = {5},
	Pages = {051601},
	Primaryclass = {hep-th},
	Slaccitation = {%%CITATION = ARXIV:1309.5850;%%},
	Title = {An integrable deformation of the {$AdS{5} \times S^5$} superstring action},
	Volume = {112},
	Year = {2014}}

@article{Dolan:2003uh,
	Archiveprefix = {arXiv},
	Author = {Dolan, Louise and Nappi, Chiara R. and Witten, Edward},
	Eprint = {hep-th/0308089},
	Journal = {JHEP},
	Pages = {017},
	Slaccitation = {%%CITATION = HEP-TH/0308089;%%},
	Title = {A relation between approaches to integrability in superconformal {Y}ang-{M}ills theory},
	Volume = {0310},
	Year = {2003}}

@article{Dorey:1996gd,
	Archiveprefix = {arXiv},
	Author = {Dorey, P.},
	Eprint = {hep-th/9810026},
	Pages = {85-125},
	Primaryclass = {hep-th},
	Slaccitation = {%%CITATION = HEP-TH/9810026;%%},
	Title = {Exact {$S$} matrices},
	Year = {1996}}

@article{Drinfeld:1987sy,
	Author = {Drinfeld, V.G.},
	Journal = {Sov. Math. Dokl.},
	Pages = {212-216},
	Slaccitation = {%%CITATION = SVMDA,36,212;%%},
	Title = {A New realization of {Y}angians and quantized affine algebras},
	Volume = {36},
	Year = {1988}}

@article{Ferko:2022cix,
    author = "Ferko, Christian and Sfondrini, Alessandro and Smith, Liam and Tartaglino-Mazzucchelli, Gabriele",
    title = "{Root-$T \bar T$ Deformations in Two-Dimensional Quantum Field Theories}",
    eprint = "2206.10515",
    archivePrefix = "arXiv",
    primaryClass = "hep-th",
    doi = "10.1103/PhysRevLett.129.201604",
    journal = "Phys. Rev. Lett.",
    volume = "129",
    number = "20",
    pages = "201604",
    year = "2022"
}

@article{Faddeev:1985qu,
	Author = {Faddeev, L.D. and Reshetikhin, N. Yu.},
	Doi = {10.1016/0003-4916(86)90201-0},
	Journal = {Annals Phys.},
	Pages = {227},
	Slaccitation = {%%CITATION = APNYA,167,227;%%},
	Title = {Integrability of the principal chiral field model in {$1+1$}-dimension},
	Volume = {167},
	Year = {1986}}

@article{Frolov:2003qc,
	Archiveprefix = {arXiv},
	Author = {Frolov, S. and Tseytlin, A. A.},
	Doi = {10.1016/S0550-3213(03)00580-7},
	Eprint = {hep-th/0304255},
	Journal = {Nucl. Phys.},
	Pages = {77-110},
	Slaccitation = {%%CITATION = HEP-TH/0304255;%%},
	Title = {Multi-spin string solutions in {$AdS_5 \times S^5$}},
	Volume = {B668},
	Year = {2003}}

@article{Kazakov:2004qf,
	Archiveprefix = {arXiv},
	Author = {Kazakov, V. A. and Marshakov, A. and Minahan, J. A. and Zarembo, K.},
	Eprint = {hep-th/0402207},
	Journal = {JHEP},
	Pages = {24},
	Primaryclass = {hep-th},
	Slaccitation = {%%CITATION = HEP-TH/0402207;%%},
	Title = {Classical/quantum integrability in {AdS/CFT}},
	Volume = {5},
	Year = {2004}}

@article{Cardy:2018sdv,
	Archiveprefix = {arXiv},
	Author = {Cardy, John},
	Doi = {10.1007/JHEP10(2018)186},
	Eprint = {1801.06895},
	Journal = {JHEP},
	Pages = {186},
	Primaryclass = {hep-th},
	Slaccitation = {%%CITATION = ARXIV:1801.06895;%%},
	Title = {{The $ T\overline{T} $ deformation of quantum field theory as random geometry}},
	Volume = {10},
	Year = {2018}}

@article{Klimcik:2002zj,
	Archiveprefix = {arXiv},
	Author = {Klimcik, Ctirad},
	Doi = {10.1088/1126-6708/2002/12/051},
	Eprint = {hep-th/0210095},
	Journal = {JHEP},
	Pages = {051},
	Primaryclass = {hep-th},
	Reportnumber = {IML-02-XY},
	Slaccitation = {%%CITATION = HEP-TH/0210095;%%},
	Title = {{Yang-Baxter sigma models and dS/AdS T duality}},
	Volume = {12},
	Year = {2002}}

@article{Klimcik:2008eq,
    author = "Klimcik, Ctirad",
    title = "{On integrability of the Yang-Baxter sigma-model}",
    eprint = "0802.3518",
    archivePrefix = "arXiv",
    primaryClass = "hep-th",
    doi = "10.1063/1.3116242",
    journal = "J. Math. Phys.",
    volume = "50",
    pages = "043508",
    year = "2009"
}

@article{Klimcik:2014bta,
    author = "Klimcik, Ctirad",
    title = "{Integrability of the bi-Yang-Baxter sigma-model}",
    eprint = "1402.2105",
    archivePrefix = "arXiv",
    primaryClass = "math-ph",
    doi = "10.1007/s11005-014-0709-y",
    journal = "Lett. Math. Phys.",
    volume = "104",
    pages = "1095--1106",
    year = "2014"
}

@book{Yoshida:2021qfl,
    author = "Yoshida, Kentaroh",
    title = "{Yang{\textendash}Baxter Deformation of 2D Non-Linear Sigma Models}: {Towards Applications to AdS/CFT}",
    doi = "10.1007/978-981-16-1703-4",
    isbn = "978-981-16-1702-7, 978-981-16-1703-4",
    publisher = "Springer",
    address = "Singapore",
    series = "SpringerBriefs in Mathematical Physics",
    volume = "40",
    year = "2021"
}

@article{Luscher:1977uq,
	Author = {L{\"u}scher, M.},
	Doi = {10.1016/0550-3213(78)90211-0},
	Journal = {Nucl. Phys.},
	Pages = {1-19},
	Slaccitation = {%%CITATION = NUPHA,B135,1;%%},
	Title = {Quantum Nonlocal Charges and Absence of Particle Production in the Two-Dimensional Nonlinear Sigma Model},
	Volume = {B135},
	Year = {1978}}

@article{MacKay:2004tc,
	Archiveprefix = {arXiv},
	Author = {MacKay, N. J.},
	Eprint = {hep-th/0409183},
	Journal = {Int. J. Mod. Phys.},
	Pages = {7189-7218},
	Slaccitation = {%%CITATION = HEP-TH/0409183;%%},
	Title = {Introduction to {Y}angian symmetry in integrable field theory},
	Volume = {A20},
	Year = {2005}}

@article{Metsaev:1998it,
	Archiveprefix = {arXiv},
	Author = {Metsaev, R. R. and Tseytlin, A. A.},
	Doi = {10.1016/S0550-3213(98)00570-7},
	Eprint = {hep-th/9805028},
	Journal = {Nucl. Phys.},
	Pages = {109-126},
	Primaryclass = {hep-th},
	Slaccitation = {%%CITATION = HEP-TH/9805028;%%},
	Title = {Type {IIB} superstring action in {$AdS_5 \times S^5$} background},
	Volume = {B533},
	Year = {1998}}

@article{Minahan:2002ve,
	Archiveprefix = {arXiv},
	Author = {Minahan, J. A. and Zarembo, K.},
	Eprint = {hep-th/0212208},
	Journal = {JHEP},
	Pages = {013},
	Slaccitation = {%%CITATION = HEP-TH/0212208;%%},
	Title = {The Bethe-ansatz for {$N = 4$} super Yang-Mills},
	Volume = {0303},
	Year = {2003}}

@article{Novikov:1982ei,
	Author = {Novikov, S.P.},
	Journal = {Usp. Mat. Nauk},
	Pages = {3-49},
	Slaccitation = {%%CITATION = UMANA,37N5,3;%%},
	Title = {The {H}amiltonian formalism and a many valued analog of {M}orse theory},
	Volume = {37N5},
	Year = {1982}}

@article{MacKay:1992he,
    author = "MacKay, N. J.",
    title = "{On the classical origins of Yangian symmetry in integrable field theory}",
    reportNumber = "DTP-92-09",
    doi = "10.1016/0370-2693(92)90280-H",
    journal = "Phys. Lett. B",
    volume = "281",
    pages = "90--97",
    year = "1992",
    note = "[Erratum: Phys.Lett.B 308, 444--444 (1993)]"
}

@article{Polyakov:1984et,
	Author = {Polyakov, Alexander M. and Wiegmann, P. B.},
	Doi = {10.1016/0370-2693(84)90206-5},
	Journal = {Phys. Lett.},
	Pages = {223-228},
	Slaccitation = {%%CITATION = PHLTA,B141,223;%%},
	Title = {Goldstone Fields in Two-Dimensions with Multivalued Actions},
	Volume = {B141},
	Year = {1984}}

@article{Torrielli:2010kq,
	Archiveprefix = {arXiv},
	Author = {Alessandro Torrielli},
	Doi = {10.1007/s11005-011-0491-z},
	Eprint = {1012.4005},
	Journal = {Lett. Math. Phys.},
	Pages = {547-565},
	Primaryclass = {hep-th},
	Slaccitation = {%%CITATION = 1012.4005;%%},
	Title = {Review of {AdS/CFT} Integrability, {C}hapter {VI.2}: {Y}angian Algebra},
	Volume = {99},
	Year = {2010}}

@article{Torrielli:2011gg,
	Archiveprefix = {arXiv},
	Author = {Torrielli, Alessandro},
	Doi = {10.1088/1751-8113/44/26/263001},
	Eprint = {1104.2474},
	Journal = {J. Phys.},
	Pages = {263001},
	Primaryclass = {hep-th},
	Slaccitation = {%%CITATION = ARXIV:1104.2474;%%},
	Title = {Yangians, S-matrices and {$AdS/CFT$}},
	Volume = {A44},
	Year = {2011}}

@article{Witten:1983ar,
	Author = {Witten, Edward},
	Doi = {10.1007/BF01215276},
	Journal = {Commun. Math. Phys.},
	Pages = {455-472},
	Slaccitation = {%%CITATION = CMPHA,92,455;%%},
	Title = {Nonabelian Bosonization in Two-Dimensions},
	Volume = {92},
	Year = {1984}}

@article{Yang:1967bm,
	Author = {Yang, Chen-Ning},
	Doi = {10.1103/PhysRevLett.19.1312},
	Journal = {Phys. Rev. Lett.},
	Pages = {1312-1314},
	Slaccitation = {%%CITATION = PRLTA,19,1312;%%},
	Title = {Some exact results for the many body problems in one dimension with repulsive delta function interaction},
	Volume = {19},
	Year = {1967}}

@article{Zamolodchikov:1978xm,
	Author = {Zamolodchikov, Alexander B. and Zamolodchikov, Alexei B.},
	Doi = {10.1016/0003-4916(79)90391-9},
	Journal = {Annals Phys.},
	Pages = {253-291},
	Slaccitation = {%%CITATION = APNYA,120,253;%%},
	Title = {Factorized {S}-matrices in two dimensions as the exact solutions of certain relativistic quantum field models},
	Volume = {120},
	Year = {1979}}

@article{Borsato:2022tmu,
    author = "Borsato, Riccardo and Ferko, Christian and Sfondrini, Alessandro",
    title = "{Classical integrability of root-TT flows}",
    eprint = "2209.14274",
    archivePrefix = "arXiv",
    primaryClass = "hep-th",
    doi = "10.1103/PhysRevD.107.086011",
    journal = "Phys. Rev. D",
    volume = "107",
    number = "8",
    pages = "086011",
    year = "2023"
}

@article{Ferko:2023wyi,
    author = "Ferko, Christian and Kuzenko, Sergei M. and Smith, Liam and Tartaglino-Mazzucchelli, Gabriele",
    title = "{Duality-invariant nonlinear electrodynamics and stress tensor flows}",
    eprint = "2309.04253",
    archivePrefix = "arXiv",
    primaryClass = "hep-th",
    doi = "10.1103/PhysRevD.108.106021",
    journal = "Phys. Rev. D",
    volume = "108",
    number = "10",
    pages = "106021",
    year = "2023"
}

@inproceedings{Ivanov:2002ab,
      author         = "Ivanov, E. A. and Zupnik, B. M.",
      title          = "{New representation for Lagrangians of selfdual nonlinear
                        electrodynamics}",
      booktitle      = "{Supersymmetries and Quantum Symmetries. Proceedings,
                        16th Max Born Symposium, SQS'01: Karpacz, Poland,
                        September 21-25, 2001}",
      year           = "2002",
      pages          = "235-250",
      eprint         = "hep-th/0202203",
      archivePrefix  = "arXiv",
      primaryClass   = "hep-th",
      SLACcitation   = "%%CITATION = HEP-TH/0202203;%%"
}

@article{Zarembo:2017muf,
    author = "Zarembo, K.",
    editor = "Dorey, Patrick and Korchemsky, Gregory and Nekrasov, Nikita and Schomerus, Volker and Serban, Didina and Cugliandolo, Leticia",
    title = "{Integrability in Sigma-Models}",
    eprint = "1712.07725",
    archivePrefix = "arXiv",
    primaryClass = "hep-th",
    reportNumber = "NORDITA-2017-137, UUITP-04-17",
    month = "12",
    year = "2017"
}

@phdthesis{Seibold:2020ouf,
    author = "Seibold, Fiona K.",
    title = "{Integrable deformations of sigma models and superstrings}",
    doi = "10.3929/ethz-b-000440825",
    school = "Zurich, ETH, Zurich, ETH",
    year = "2020"
}

@article{Bielli:2024khq,
    author = "Bielli, Daniele and Ferko, Christian and Smith, Liam and Tartaglino-Mazzucchelli, Gabriele",
    title = "{T Duality and TT-like Deformations of Sigma Models}",
    eprint = "2407.11636",
    archivePrefix = "arXiv",
    primaryClass = "hep-th",
    doi = "10.1103/PhysRevLett.134.101601",
    journal = "Phys. Rev. Lett.",
    volume = "134",
    number = "10",
    pages = "101601",
    year = "2025"
}

@article{Bielli:2024ach,
    author = "Bielli, Daniele and Ferko, Christian and Smith, Liam and Tartaglino-Mazzucchelli, Gabriele",
    title = "{Integrable higher-spin deformations of sigma models from auxiliary fields}",
    eprint = "2407.16338",
    archivePrefix = "arXiv",
    primaryClass = "hep-th",
    doi = "10.1103/PhysRevD.111.066010",
    journal = "Phys. Rev. D",
    volume = "111",
    number = "6",
    pages = "066010",
    year = "2025"
}

@article{Bielli:2024fnp,
    author = "Bielli, Daniele and Ferko, Christian and Smith, Liam and Tartaglino-Mazzucchelli, Gabriele",
    title = "{Auxiliary Field Sigma Models and Yang-Baxter Deformations}",
    eprint = "2408.09714",
    archivePrefix = "arXiv",
    primaryClass = "hep-th",
    month = "8",
    year = "2024"
}

@article{Bielli:2024oif,
    author = "Bielli, Daniele and Ferko, Christian and Smith, Liam and Tartaglino-Mazzucchelli, Gabriele",
    title = "{Auxiliary field deformations of (semi-)symmetric space sigma models}",
    eprint = "2409.05704",
    archivePrefix = "arXiv",
    primaryClass = "hep-th",
    doi = "10.1007/JHEP01(2025)096",
    journal = "JHEP",
    volume = "01",
    pages = "096",
    year = "2025"
}

@article{Bielli:2025uiv,
    author = "Bielli, Daniele and Ferko, Christian and Galli, Michele and Tartaglino-Mazzucchelli, Gabriele",
    title = "{Higher-spin currents and flows in auxiliary field sigma models}",
    eprint = "2504.17294",
    archivePrefix = "arXiv",
    primaryClass = "hep-th",
    doi = "10.1007/JHEP08(2025)078",
    journal = "JHEP",
    volume = "08",
    pages = "078",
    year = "2025"
}

@article{Ferko:2024ali,
    author = "Ferko, Christian and Smith, Liam",
    title = "{Infinite Family of Integrable Sigma Models Using Auxiliary Fields}",
    eprint = "2405.05899",
    archivePrefix = "arXiv",
    primaryClass = "hep-th",
    doi = "10.1103/PhysRevLett.133.131602",
    journal = "Phys. Rev. Lett.",
    volume = "133",
    number = "13",
    pages = "131602",
    year = "2024"
}

@article{Cesaro:2024ipq,
    author = "Ces\`aro, Mattia and Kleinschmidt, Axel and Osten, David",
    title = "{Integrable auxiliary field deformations of coset models}",
    eprint = "2409.04523",
    archivePrefix = "arXiv",
    primaryClass = "hep-th",
    doi = "10.1007/JHEP11(2024)028",
    journal = "JHEP",
    volume = "11",
    pages = "028",
    year = "2024"
}

@article{Cesaro:2025msv,
    author = "Ces{\`a}ro, Mattia and Osten, David",
    title = "{Integrable deformations of dimensionally reduced gravity}",
    eprint = "2502.01750",
    archivePrefix = "arXiv",
    primaryClass = "hep-th",
    doi = "10.1007/JHEP06(2025)064",
    journal = "JHEP",
    volume = "06",
    pages = "064",
    year = "2025"
}

@article{Itsios:2014vfa,
    author = "Itsios, Georgios and Sfetsos, Konstantinos and Siampos, Konstantinos and Torrielli, Alessandro",
    title = "{The classical Yang{\textendash}Baxter equation and the associated Yangian symmetry of gauged WZW-type theories}",
    eprint = "1409.0554",
    archivePrefix = "arXiv",
    primaryClass = "hep-th",
    reportNumber = "DMUS-MP-14-10",
    doi = "10.1016/j.nuclphysb.2014.10.004",
    journal = "Nucl. Phys. B",
    volume = "889",
    pages = "64--86",
    year = "2014"
}

@article{Loebbert:2016cdm,
    author = "Loebbert, Florian",
    title = "{Lectures on Yangian Symmetry}",
    eprint = "1606.02947",
    archivePrefix = "arXiv",
    primaryClass = "hep-th",
    reportNumber = "HU-EP-16-12",
    doi = "10.1088/1751-8113/49/32/323002",
    journal = "J. Phys. A",
    volume = "49",
    number = "32",
    pages = "323002",
    year = "2016"
}

@article{Fukushima:2024nxm,
    author = "Fukushima, Osamu and Yoshida, Kentaroh",
    title = "{4D Chern-Simons theory with auxiliary fields}",
    eprint = "2407.02204",
    archivePrefix = "arXiv",
    primaryClass = "hep-th",
    reportNumber = "RIKEN-iTHEMS-Report-24, STUPP-24-270",
    doi = "10.1007/JHEP09(2025)001",
    journal = "JHEP",
    volume = "09",
    pages = "001",
    year = "2025"
}

@article{Klose:2016qfv,
    author = {Klose, Thomas and Loebbert, Florian and M{\"u}nkler, Hagen},
    title = "{Nonlocal Symmetries, Spectral Parameter and Minimal Surfaces in AdS/CFT}",
    eprint = "1610.01161",
    archivePrefix = "arXiv",
    primaryClass = "hep-th",
    reportNumber = "HU-EP-16-30",
    doi = "10.1016/j.nuclphysb.2017.01.008",
    journal = "Nucl. Phys. B",
    volume = "916",
    pages = "320--372",
    year = "2017"
}

@article{Forger:1991cm,
    author = "Forger, M. and Laartz, J. and Schaper, U.",
    title = "{Current algebra of classical nonlinear sigma models}",
    eprint = "hep-th/9201025",
    archivePrefix = "arXiv",
    reportNumber = "FREIBURG-THEP-91-10",
    doi = "10.1007/BF02102634",
    journal = "Commun. Math. Phys.",
    volume = "146",
    pages = "397--402",
    year = "1992"
}

@article{Forger:1991ty,
    author = "Forger, M. and Bordemann, M. and Laartz, J. and Schaper, U.",
    title = "{The Lie-Poisson structure of integrable classical nonlinear sigma models}",
    eprint = "hep-th/9201051",
    archivePrefix = "arXiv",
    reportNumber = "FREIBURG-THEP-91-11",
    doi = "10.1007/BF02097062",
    journal = "Commun. Math. Phys.",
    volume = "152",
    pages = "167--190",
    year = "1993"
}

@article{Lozano:1995jx,
    author = "Lozano, Y.",
    title = "{NonAbelian duality and canonical transformations}",
    eprint = "hep-th/9503045",
    archivePrefix = "arXiv",
    reportNumber = "PUPT-1532",
    doi = "10.1016/0370-2693(95)00777-I",
    journal = "Phys. Lett. B",
    volume = "355",
    pages = "165--170",
    year = "1995"
}

@article{Sklyanin:1980ij,
    author = "Sklyanin, E. K.",
    title = "{Quantum version of the method of inverse scattering problem}",
    doi = "10.1007/BF01091462",
    journal = "Zap. Nauchn. Semin.",
    volume = "95",
    pages = "55--128",
    year = "1980"
}

@article{Delduc:2013fga,
    author = "Delduc, Francois and Magro, Marc and Vicedo, Benoit",
    title = "{On classical $q$-deformations of integrable sigma-models}",
    eprint = "1308.3581",
    archivePrefix = "arXiv",
    primaryClass = "hep-th",
    doi = "10.1007/JHEP11(2013)192",
    journal = "JHEP",
    volume = "11",
    pages = "192",
    year = "2013"
}

@article{Delduc:2014kha,
    author = "Delduc, Francois and Magro, Marc and Vicedo, Benoit",
    title = "{Derivation of the action and symmetries of the $q$-deformed $AdS_{5} \times S^{5}$ superstring}",
    eprint = "1406.6286",
    archivePrefix = "arXiv",
    primaryClass = "hep-th",
    doi = "10.1007/JHEP10(2014)132",
    journal = "JHEP",
    volume = "10",
    pages = "132",
    year = "2014"
}

@phdthesis{Lacroix:2018njs,
    author = "Lacroix, Sylvain",
    title = "{Integrable models with twist function and affine Gaudin models}",
    eprint = "1809.06811",
    archivePrefix = "arXiv",
    primaryClass = "hep-th",
    reportNumber = "tel-01900498, 2018LYSEN014",
    school = "Lyon, Ecole Normale Superieure",
    year = "2018"
}

@article{Wess:1971yu,
    author = "Wess, J. and Zumino, B.",
    title = "{Consequences of anomalous Ward identities}",
    doi = "10.1016/0370-2693(71)90582-X",
    journal = "Phys. Lett. B",
    volume = "37",
    pages = "95--97",
    year = "1971"
}

@article{Abdalla:1993sn,
    author = "Abdalla, E. and Abdalla, M. C. B. and Branco, O. H. G. and Saltini, L. E.",
    title = "{Current algebra of super WZNW models}",
    eprint = "hep-th/9301108",
    archivePrefix = "arXiv",
    reportNumber = "IFUSP-P-1026",
    doi = "10.1088/0305-4470/27/13/043",
    journal = "J. Phys. A",
    volume = "27",
    pages = "4709--4716",
    year = "1994"
}

@article{Evans:2000hx,
    author = "Evans, J. M. and Hassan, M. and MacKay, N. J. and Mountain, A. J.",
    title = "{Conserved charges and supersymmetry in principal chiral and WZW models}",
    eprint = "hep-th/0001222",
    archivePrefix = "arXiv",
    reportNumber = "PUPT-1908, DAMTP-99-178, IMPERIAL-TP-99-00-15",
    doi = "10.1016/S0550-3213(00)00257-1",
    journal = "Nucl. Phys. B",
    volume = "580",
    pages = "605--646",
    year = "2000"
}

@article{Ferko:2025bhv,
    author = "Ferko, Christian and Galli, Michele and Huang, Zejun and Tartaglino-Mazzucchelli, Gabriele",
    title = "{Soliton surfaces and the geometry of integrable deformations of the $ {\mathbbm{CP}}^{N-1} $ model}",
    eprint = "2509.05081",
    archivePrefix = "arXiv",
    primaryClass = "hep-th",
    doi = "10.1007/JHEP03(2026)144",
    journal = "JHEP",
    volume = "03",
    pages = "144",
    year = "2026"
}

@article{Sfetsos:2013wia,
    author = "Sfetsos, Konstadinos",
    title = "{Integrable interpolations: From exact CFTs to non-Abelian T-duals}",
    eprint = "1312.4560",
    archivePrefix = "arXiv",
    primaryClass = "hep-th",
    reportNumber = "DMUS-MP-13-23, DMUS--MP--13-23",
    doi = "10.1016/j.nuclphysb.2014.01.004",
    journal = "Nucl. Phys. B",
    volume = "880",
    pages = "225--246",
    year = "2014"
}

@article{Luscher:1977rq,
    author = "Luscher, M. and Pohlmeyer, K.",
    title = "{Scattering of Massless Lumps and Nonlocal Charges in the Two-Dimensional Classical Nonlinear Sigma Model}",
    reportNumber = "DESY-77-65",
    doi = "10.1016/0550-3213(78)90049-4",
    journal = "Nucl. Phys. B",
    volume = "137",
    pages = "46--54",
    year = "1978"
}

@article{Wiegmann:1984pw,
    author = "Wiegmann, P. B.",
    title = "{On the Theory of Nonabelian Goldstone Bosons in Two-dimensions: Exact Solution of the O(3) Nonlinear $\sigma$ Model}",
    reportNumber = "NORDITA-84/32",
    doi = "10.1016/0370-2693(84)90205-3",
    journal = "Phys. Lett. B",
    volume = "141",
    pages = "217",
    year = "1984"
}

@article{Ogievetsky:1987vv,
    author = "Ogievetsky, E. and Reshetikhin, N. and Wiegmann, P.",
    title = "{The principal chiral field in two dimensions on classical lie algebras: The Bethe-ansatz solution and factorized theory of scattering}",
    doi = "10.1016/0550-3213(87)90138-6",
    journal = "Nucl. Phys. B",
    volume = "280",
    pages = "45--96",
    year = "1987"
}

@article{Eichenherr:1979ci,
    author = "Eichenherr, H. and Forger, M.",
    title = "{On the Dual Symmetry of the Nonlinear Sigma Models}",
    reportNumber = "FREIBURG-THEP 79/2a",
    doi = "10.1016/0550-3213(79)90276-1",
    journal = "Nucl. Phys. B",
    volume = "155",
    pages = "381--393",
    year = "1979"
}

@article{Eichenherr:1979hz,
    author = "Eichenherr, H. and Forger, M.",
    title = "{More about non-linear sigma models on symmetric spaces}",
    reportNumber = "FREIBURG-THEP-79-7",
    doi = "10.1016/0550-3213(80)90525-8",
    journal = "Nucl. Phys. B",
    volume = "164",
    pages = "528--535",
    year = "1980",
    note = "[Erratum: Nucl.Phys.B 282, 745--745 (1987)]"
}

@inproceedings{Laartz:1992rw,
    author = "Laartz, J. and Bordemann, M. and Forger, M. and Schaper, U.",
    title = "{New classical r matrices from integrable nonlinear sigma models}",
    booktitle = "{19th International Colloquium on Group Theoretical Methods in Physics}",
    eprint = "hep-th/9209108",
    archivePrefix = "arXiv",
    reportNumber = "FREIBURG-THEP-92-20",
    month = "9",
    year = "1992"
}

@article{Forger:1992kn,
    author = "Forger, M. and Laartz, J. and Schaper, U.",
    title = "{The Algebra of the energy momentum tensor and the Noether currents in classical nonlinear sigma models}",
    eprint = "hep-th/9210130",
    archivePrefix = "arXiv",
    reportNumber = "FREIBURG-THEP-92-24",
    doi = "10.1007/BF02102641",
    journal = "Commun. Math. Phys.",
    volume = "159",
    pages = "319--328",
    year = "1994"
}

@article{delaOssa:1992vci,
    author = "de la Ossa, Xenia C. and Quevedo, Fernando",
    title = "{Duality symmetries from nonAbelian isometries in string theory}",
    eprint = "hep-th/9210021",
    archivePrefix = "arXiv",
    reportNumber = "NEIP-92-004",
    doi = "10.1016/0550-3213(93)90041-M",
    journal = "Nucl. Phys. B",
    volume = "403",
    pages = "377--394",
    year = "1993"
}

@article{Alvarez:1993qi,
    author = "Alvarez, E. and Alvarez-Gaume, Luis and Barbon, J. L. F. and Lozano, Y.",
    title = "{Some global aspects of duality in string theory}",
    eprint = "hep-th/9309039",
    archivePrefix = "arXiv",
    reportNumber = "CERN-TH-6991-93, FTUAM-93-28",
    doi = "10.1016/0550-3213(94)90067-1",
    journal = "Nucl. Phys. B",
    volume = "415",
    pages = "71--100",
    year = "1994"
}

@article{Buscher:1987sk,
    author = "Buscher, T. H.",
    title = "{A Symmetry of the String Background Field Equations}",
    reportNumber = "ITP-SB-87-21",
    doi = "10.1016/0370-2693(87)90769-6",
    journal = "Phys. Lett. B",
    volume = "194",
    pages = "59--62",
    year = "1987"
}

@article{Buscher:1987qj,
    author = "Buscher, T. H.",
    title = "{Path Integral Derivation of Quantum Duality in Nonlinear Sigma Models}",
    reportNumber = "ITP-SB-87-61",
    doi = "10.1016/0370-2693(88)90602-8",
    journal = "Phys. Lett. B",
    volume = "201",
    pages = "466--472",
    year = "1988"
}

@article{Rocek:1991ps,
    author = "Rocek, Martin and Verlinde, Erik P.",
    title = "{Duality, quotients, and currents}",
    eprint = "hep-th/9110053",
    archivePrefix = "arXiv",
    reportNumber = "IASSNS-HEP-91-68, ITP-SB-91-53",
    doi = "10.1016/0550-3213(92)90269-H",
    journal = "Nucl. Phys. B",
    volume = "373",
    pages = "630--646",
    year = "1992"
}

@article{Borsato:2017qsx,
    author = "Borsato, Riccardo and Wulff, Linus",
    title = "{On non-abelian T-duality and deformations of supercoset string sigma-models}",
    eprint = "1706.10169",
    archivePrefix = "arXiv",
    primaryClass = "hep-th",
    reportNumber = "NORDITA-2017-065",
    doi = "10.1007/JHEP10(2017)024",
    journal = "JHEP",
    volume = "10",
    pages = "024",
    year = "2017"
}

@article{Osten:2016dvf,
    author = "Osten, David and van Tongeren, Stijn J.",
    title = "{Abelian Yang{\textendash}Baxter deformations and TsT transformations}",
    eprint = "1608.08504",
    archivePrefix = "arXiv",
    primaryClass = "hep-th",
    reportNumber = "HU-EP-16-29",
    doi = "10.1016/j.nuclphysb.2016.12.007",
    journal = "Nucl. Phys. B",
    volume = "915",
    pages = "184--205",
    year = "2017"
}

@article{Alvarez:1994wj,
    author = "Alvarez, Enrique and Alvarez-Gaume, Luis and Lozano, Yolanda",
    title = "{A Canonical approach to duality transformations}",
    eprint = "hep-th/9406206",
    archivePrefix = "arXiv",
    reportNumber = "CERN-TH-7337-94, FTUAM-94-15",
    doi = "10.1016/0370-2693(94)00982-1",
    journal = "Phys. Lett. B",
    volume = "336",
    pages = "183--189",
    year = "1994"
}

@article{Maillet:1985ec,
    author = "Maillet, Jean Michel",
    title = "{Hamiltonian Structures for Integrable Classical Theories From Graded Kac-moody Algebras}",
    reportNumber = "PAR-LPTHE-85/40",
    doi = "10.1016/0370-2693(86)91289-X",
    journal = "Phys. Lett. B",
    volume = "167",
    pages = "401--405",
    year = "1986"
}

@article{Rajeev:1988hq,
    author = "Rajeev, S. G.",
    title = "{NONABELIAN BOSONIZATION WITHOUT WESS-ZUMINO TERMS. 1. NEW CURRENT ALGEBRA}",
    reportNumber = "UR-1075, ER13065-551",
    doi = "10.1016/0370-2693(89)91528-1",
    journal = "Phys. Lett. B",
    volume = "217",
    pages = "123--128",
    year = "1989"
}

@article{Balog:1993es,
    author = "Balog, J. and Forgacs, P. and Horvath, Z. and Palla, L.",
    editor = "Lust, D. and Weigt, G.",
    title = "{A New family of SU(2) symmetric integrable sigma models}",
    eprint = "hep-th/9307030",
    archivePrefix = "arXiv",
    reportNumber = "ITP-502-BUDAPEST",
    doi = "10.1016/0370-2693(94)90213-5",
    journal = "Phys. Lett. B",
    volume = "324",
    pages = "403--408",
    year = "1994"
}

@article{Fridling:1983ha,
    author = "Fridling, B. E. and Jevicki, A.",
    title = "{Dual Representations and Ultraviolet Divergences in Nonlinear $\sigma$ Models}",
    reportNumber = "BROWN-HET-510",
    doi = "10.1016/0370-2693(84)90987-0",
    journal = "Phys. Lett. B",
    volume = "134",
    pages = "70--74",
    year = "1984"
}

@article{Fradkin:1984ai,
    author = "Fradkin, E. S. and Tseytlin, Arkady A.",
    title = "{Quantum Equivalence of Dual Field Theories}",
    reportNumber = "LEBEDEV-84-67",
    doi = "10.1016/0003-4916(85)90225-8",
    journal = "Annals Phys.",
    volume = "162",
    pages = "31",
    year = "1985"
}

@article{Drummond:2008vq,
    author = "Drummond, J. M. and Henn, J. and Korchemsky, G. P. and Sokatchev, E.",
    title = "{Dual superconformal symmetry of scattering amplitudes in N=4 super-Yang-Mills theory}",
    eprint = "0807.1095",
    archivePrefix = "arXiv",
    primaryClass = "hep-th",
    reportNumber = "LAPTH-1257-08, LPT-ORSAY-08-60",
    doi = "10.1016/j.nuclphysb.2009.11.022",
    journal = "Nucl. Phys. B",
    volume = "828",
    pages = "317--374",
    year = "2010"
}

@article{Drummond:2009fd,
    author = "Drummond, James M. and Henn, Johannes M. and Plefka, Jan",
    editor = "Liu, Feng and Xiao, Zhigang and Zhuang, Pengfei",
    title = "{Yangian symmetry of scattering amplitudes in N=4 super Yang-Mills theory}",
    eprint = "0902.2987",
    archivePrefix = "arXiv",
    primaryClass = "hep-th",
    reportNumber = "HU-EP-09-06, LAPTH-1308-09",
    doi = "10.1088/1126-6708/2009/05/046",
    journal = "JHEP",
    volume = "05",
    pages = "046",
    year = "2009"
}

@article{deLeeuw:2008dp,
    author = "de Leeuw, Marius",
    title = "{Bound States, Yangian Symmetry and Classical r-matrix for the AdS$_5 \times $S$^5$ Superstring}",
    eprint = "0804.1047",
    archivePrefix = "arXiv",
    primaryClass = "hep-th",
    reportNumber = "ITP-UU-08-18, SPIN-08-17",
    doi = "10.1088/1126-6708/2008/06/085",
    journal = "JHEP",
    volume = "06",
    pages = "085",
    year = "2008"
}

@article{Li:2026ecl,
    author = "Li, Bo-Rui and He, Song and Liu, Yu-Xiao",
    title = "{Correlators in $T\bar{T}$ and Root-$T\bar{T}$ Deformed CFTs}",
    eprint = "2604.14939",
    archivePrefix = "arXiv",
    primaryClass = "hep-th",
    month = "4",
    year = "2026"
}

@article{Chen:2025jzb,
    author = "Chen, Liangyu and Du, Zhengyuan and Liu, Kangning and Song, Wei",
    title = "{Symmetries and operators in $T\bar{T}$ deformed CFTs}",
    eprint = "2507.08588",
    archivePrefix = "arXiv",
    primaryClass = "hep-th",
    month = "7",
    year = "2025"
}

@article{Bombardelli:2016rwb,
    author = "Bombardelli, Diego and Cagnazzo, Alessandra and Frassek, Rouven and Levkovich-Maslyuk, Fedor and Loebbert, Florian and Negro, Stefano and Sz{\'e}cs{\'e}nyi, Istvan M. and Sfondrini, Alessandro and van Tongeren, Stijn J. and Torrielli, Alessandro",
    title = "{An integrability primer for the gauge-gravity correspondence: An introduction}",
    eprint = "1606.02945",
    archivePrefix = "arXiv",
    primaryClass = "hep-th",
    reportNumber = "CNRS-16-03, DCPT-16-19, DESY-16-083, DMUS-MP-16-09, HU-EP-16-13, NORDITA-2016-33, HU-MATH-16-08",
    doi = "10.1088/1751-8113/49/32/320301",
    journal = "J. Phys. A",
    volume = "49",
    number = "32",
    pages = "320301",
    year = "2016"
}

@article{Demulder:2023bux,
    author = "Demulder, Saskia and Driezen, Sibylle and Knighton, Bob and Oling, Gerben and Retore, Ana L. and Seibold, Fiona K. and Sfondrini, Alessandro and Yan, Ziqi",
    title = "{Exact approaches on the string worldsheet}",
    eprint = "2312.12930",
    archivePrefix = "arXiv",
    primaryClass = "hep-th",
    reportNumber = "NORDITA 2023-083",
    doi = "10.1088/1751-8121/ad72be",
    journal = "J. Phys. A",
    volume = "57",
    number = "42",
    pages = "423001",
    year = "2024"
}

@Book{Dirac:1964:LQM,
  author =       "Paul A. M. Dirac",
  title =        "Lectures on quantum mechanics",
  volume =       "2",
  publisher =    "Belfer Graduate School of Science, New York",
  year =         "1964",
  MRclass =      "81-01",
  MRnumber =     "MR2220894",
  bibdate =      "Sun May 15 17:38:32 2011",
  series =       "Belfer Graduate School of Science Monographs Series"
}

@book{henneaux1992quantization,
  title={Quantization of gauge systems},
  author={Henneaux, Marc and Teitelboim, Claudio},
  year={1992},
  publisher={Princeton university press}
}

@book{mann2018lagrangianandhamiltoniandynamics,
  title={Lagrangian \& Hamiltonian Dynamics},
  author={Mann, Peter},
  year={2018},
  publisher={Oxford University Press}
}

@article{Drinfeld:1986in,
    author = "Drinfeld, V. G.",
    title = "{Quantum groups}",
    doi = "10.1007/BF01247086",
    journal = "Zap. Nauchn. Semin.",
    volume = "155",
    pages = "18--49",
    year = "1986"
}

@article{Baxter:1972hz,
    author = "Baxter, Rodney J.",
    title = "{Partition function of the eight vertex lattice model}",
    doi = "10.1016/0003-4916(72)90335-1",
    journal = "Annals Phys.",
    volume = "70",
    pages = "193--228",
    year = "1972"
}

@article{Bombardelli:2016scq,
    author = "Bombardelli, Diego",
    title = "{S-matrices and integrability}",
    eprint = "1606.02949",
    archivePrefix = "arXiv",
    primaryClass = "hep-th",
    doi = "10.1088/1751-8113/49/32/323003",
    journal = "J. Phys. A",
    volume = "49",
    number = "32",
    pages = "323003",
    year = "2016"
}

@article{deVega:1984wk,
    author = "de Vega, H. J. and Eichenherr, H. and Maillet, J. M.",
    title = "{{Yang-Baxter} Algebras of Monodromy Matrices in Integrable Quantum Field Theories}",
    reportNumber = "PAR-LPTHE-84/05",
    doi = "10.1016/0550-3213(84)90272-4",
    journal = "Nucl. Phys. B",
    volume = "240",
    pages = "377--399",
    year = "1984"
}

@article{deVega:1983ogx,
    author = "de Vega, H. J. and Eichenherr, H. and Maillet, J. M.",
    title = "{Classical and Quantum Algebras of Nonlocal Charges in $\sigma$ Models}",
    reportNumber = "PAR LPTHE 83-9",
    doi = "10.1007/BF01215281",
    journal = "Commun. Math. Phys.",
    volume = "92",
    pages = "507",
    year = "1984"
}

@article{Kawaguchi:2010jg,
    author = "Kawaguchi, Io and Yoshida, Kentaroh",
    title = "{Hidden Yangian symmetry in sigma model on squashed sphere}",
    eprint = "1008.0776",
    archivePrefix = "arXiv",
    primaryClass = "hep-th",
    reportNumber = "KUNS-2286",
    doi = "10.1007/JHEP11(2010)032",
    journal = "JHEP",
    volume = "11",
    pages = "032",
    year = "2010"
}

@article{Kawaguchi:2011mz,
    author = "Kawaguchi, Io and Orlando, Domenico and Yoshida, Kentaroh",
    title = "{Yangian symmetry in deformed WZNW models on squashed spheres}",
    eprint = "1104.0738",
    archivePrefix = "arXiv",
    primaryClass = "hep-th",
    reportNumber = "KUNS-2328, IPMU11-0054",
    doi = "10.1016/j.physletb.2011.06.007",
    journal = "Phys. Lett. B",
    volume = "701",
    pages = "475--480",
    year = "2011"
}

@article{Abdalla:1993sc,
    author = "Abdalla, E. and Abdalla, M. C. B. and Brunelli, J. C. and Zadra, A.",
    title = "{The Algebra of nonlocal charges in nonlinear sigma models}",
    eprint = "hep-th/9307071",
    archivePrefix = "arXiv",
    reportNumber = "IFUSP-P-1061, IC-93-220, CERN-TH-7061-93",
    doi = "10.1007/BF02112321",
    journal = "Commun. Math. Phys.",
    volume = "166",
    pages = "379--396",
    year = "1994"
}

@article{Saltini:1995xr,
    author = "Saltini, L. E. and Zadra, A.",
    title = "{Algebra of nonlocal charges in supersymmetric nonlinear sigma models}",
    eprint = "solv-int/9511007",
    archivePrefix = "arXiv",
    reportNumber = "IFUSP-P-1188",
    doi = "10.1142/S0217751X97000487",
    journal = "Int. J. Mod. Phys. A",
    volume = "12",
    pages = "419--436",
    year = "1997"
}

@article{Hatsuda:2004it,
    author = "Hatsuda, Machiko and Yoshida, Kentaroh",
    title = "{Classical integrability and super Yangian of superstring on AdS(5) x S**5}",
    eprint = "hep-th/0407044",
    archivePrefix = "arXiv",
    reportNumber = "KEK-TH-970",
    doi = "10.4310/ATMP.2005.v9.n5.a2",
    journal = "Adv. Theor. Math. Phys.",
    volume = "9",
    number = "5",
    pages = "703--728",
    year = "2005"
}

@article{Kawaguchi:2012ve,
    author = "Kawaguchi, Io and Matsumoto, Takuya and Yoshida, Kentaroh",
    title = "{The classical origin of quantum affine algebra in squashed sigma models}",
    eprint = "1201.3058",
    archivePrefix = "arXiv",
    primaryClass = "hep-th",
    reportNumber = "KUNS-2379",
    doi = "10.1007/JHEP04(2012)115",
    journal = "JHEP",
    volume = "04",
    pages = "115",
    year = "2012"
}

@inproceedings{Dolan:2004ps,
    author = "Dolan, Louise and Nappi, Chiara R. and Witten, Edward",
    title = "{Yangian symmetry in D = 4 superconformal Yang-Mills theory}",
    booktitle = "{3rd International Symposium on Quantum Theory and Symmetries}",
    eprint = "hep-th/0401243",
    archivePrefix = "arXiv",
    doi = "10.1142/9789812702340_0036",
    pages = "300--315",
    month = "1",
    year = "2004"
}

@article{Dolan:2004ys,
    author = "Dolan, Louise and Nappi, Chiara R.",
    title = "{Spin models and superconformal Yang-Mills theory}",
    eprint = "hep-th/0411020",
    archivePrefix = "arXiv",
    doi = "10.1016/j.nuclphysb.2005.04.006",
    journal = "Nucl. Phys. B",
    volume = "717",
    pages = "361--386",
    year = "2005"
}

@article{Torrielli:2016ufi,
    author = "Torrielli, Alessandro",
    title = "{Lectures on Classical Integrability}",
    eprint = "1606.02946",
    archivePrefix = "arXiv",
    primaryClass = "hep-th",
    reportNumber = "DMUS-MP-16-04",
    doi = "10.1088/1751-8113/49/32/323001",
    journal = "J. Phys. A",
    volume = "49",
    number = "32",
    pages = "323001",
    year = "2016"
}

@article{Costello:2019tri,
    author = "Costello, Kevin and Yamazaki, Masahito",
    title = "{Gauge Theory And Integrability, III}",
    eprint = "1908.02289",
    archivePrefix = "arXiv",
    primaryClass = "hep-th",
    reportNumber = "IPMU19-0110",
    month = "8",
    year = "2019"
}

@article{Vicedo:2019dej,
    author = "Vicedo, Benoit",
    title = "{4D Chern{\textendash}Simons theory and affine Gaudin models}",
    eprint = "1908.07511",
    archivePrefix = "arXiv",
    primaryClass = "hep-th",
    doi = "10.1007/s11005-021-01354-9",
    journal = "Lett. Math. Phys.",
    volume = "111",
    number = "1",
    pages = "24",
    year = "2021"
}

@article{Delduc:2019whp,
    author = "Delduc, Francois and Lacroix, Sylvain and Magro, Marc and Vicedo, Benoit",
    title = "{A unifying 2D action for integrable $\sigma $-models from 4D Chern{\textendash}Simons theory}",
    eprint = "1909.13824",
    archivePrefix = "arXiv",
    primaryClass = "hep-th",
    doi = "10.1007/s11005-020-01268-y",
    journal = "Lett. Math. Phys.",
    volume = "110",
    number = "7",
    pages = "1645--1687",
    year = "2020"
}

@phdthesis{Cole:2024hyt,
    author = "Cole, Lewis",
    title = "{Applications of Chern-Simons Theories to Integrable Sigma Models}",
    doi = "10.23889/SUThesis.69223",
    school = "Swansea U.",
    month = "9",
    year = "2024"
}

@mastersthesis{Liniado:2025sin,
    author = "Liniado, Joaquin",
    title = "{Chern-Simons Theories and Integrability}",
    eprint = "2509.19188",
    archivePrefix = "arXiv",
    primaryClass = "hep-th",
    type = "Other thesis",
    month = "9",
    year = "2025"
}

@article{Boujakhrout:2023qlz,
    author = "Boujakhrout, Y. and Saidi, E. H. and Laamara, R. Ahl and Drissi, L. B.",
    title = "{Superspin chains solutions from 4D Chern-Simons theory}",
    eprint = "2309.04337",
    archivePrefix = "arXiv",
    primaryClass = "hep-th",
    doi = "10.1007/JHEP04(2024)043",
    journal = "JHEP",
    volume = "04",
    pages = "043",
    year = "2024"
}

@article{Cole:2024sje,
    author = "Cole, Lewis T. and Cullinan, Ryan A. and Hoare, Ben and Liniado, Joaquin and Thompson, Daniel C.",
    title = "{Gauging the diamond: integrable coset models from twistor space}",
    eprint = "2407.09479",
    archivePrefix = "arXiv",
    primaryClass = "hep-th",
    doi = "10.1007/JHEP12(2024)202",
    journal = "JHEP",
    volume = "12",
    pages = "202",
    year = "2024"
}

@article{Lacroix:2024wrd,
    author = "Lacroix, Sylvain and Wallberg, Anders",
    title = "{Geometry of the spectral parameter and renormalisation of integrable sigma-models}",
    eprint = "2401.13741",
    archivePrefix = "arXiv",
    primaryClass = "hep-th",
    reportNumber = "CERN-TH-2024-008",
    doi = "10.1007/JHEP05(2024)108",
    journal = "JHEP",
    volume = "05",
    pages = "108",
    year = "2024"
}

@article{Lacroix:2025ias,
    author = "Lacroix, Sylvain and Levine, Nat and Wallberg, Anders",
    title = "{1-loop renormalisability of integrable sigma-models from 4d Chern-Simons theory}",
    eprint = "2505.23890",
    archivePrefix = "arXiv",
    primaryClass = "hep-th",
    reportNumber = "CERN-TH-2025-107",
    doi = "10.1007/JHEP09(2025)153",
    journal = "JHEP",
    volume = "09",
    pages = "153",
    year = "2025"
}

@article{Sakamoto:2025hwi,
    author = "Sakamoto, Jun-ichi and Tateo, Roberto and Yamazaki, Masahito",
    title = "{$T\overline{T}$ and root-$T\overline{T}$ deformations in four-dimensional Chern-Simons theory}",
    eprint = "2509.12303",
    archivePrefix = "arXiv",
    primaryClass = "hep-th",
    reportNumber = "TTI-MATHPHYS-39",
    doi = "10.1007/JHEP01(2026)084",
    journal = "JHEP",
    volume = "01",
    pages = "084",
    year = "2026"
}

@article{Fukushima:2025tlj,
    author = "Fukushima, Osamu and Matsumoto, Takaki and Yoshida, Kentaroh",
    title = "{The Courant-Hilbert construction in 4D Chern-Simons theory}",
    eprint = "2509.22080",
    archivePrefix = "arXiv",
    primaryClass = "hep-th",
    reportNumber = "RIKEN-iTHEMS-Report-25, STUPP-25-288",
    doi = "10.1007/JHEP01(2026)122",
    journal = "JHEP",
    volume = "01",
    pages = "122",
    year = "2026"
}

@article{Cole:2025zmq,
    author = "Cole, Lewis T. and Hoare, Ben",
    title = "{Integrable models from 4d holomorphic BF theory}",
    eprint = "2512.15566",
    archivePrefix = "arXiv",
    primaryClass = "hep-th",
    month = "12",
    year = "2025"
}

@article{Bittleston:2026tdr,
    author = "Bittleston, Roland and Mason, Lionel and Moosavian, Seyed Faroogh",
    title = "{The Self-Duality Equations on a Riemann Surface and Four-Dimensional Chern-Simons Theory}",
    eprint = "2601.05309",
    archivePrefix = "arXiv",
    primaryClass = "hep-th",
    month = "1",
    year = "2026"
}

@article{Stedman:2026awg,
    author = "Stedman, Jake",
    title = "{Hamiltonian Analysis of Doubled 4d Chern-Simons}",
    eprint = "2601.18647",
    archivePrefix = "arXiv",
    primaryClass = "hep-th",
    month = "1",
    year = "2026"
}

@article{Ashwinkumar:2026dwd,
    author = "Ashwinkumar, Meer and Blau, Matthias",
    title = "{Integrable Deformations of the Breitenlohner-Maison Model from 4d Chern-Simons Theory}",
    eprint = "2604.26452",
    archivePrefix = "arXiv",
    primaryClass = "hep-th",
    month = "4",
    year = "2026"
}

@article{Baglioni:2025tsc,
    author = "Baglioni, Nicola and Bielli, Daniele and Galli, Michele and Tartaglino-Mazzucchelli, Gabriele",
    title = "{Relating auxiliary field formulations of $4d$ duality-invariant and $2d$ integrable field theories}",
    eprint = "2512.21982",
    archivePrefix = "arXiv",
    primaryClass = "hep-th",
    month = "12",
    year = "2025"
}
\end{document}